\pdfoutput=1
\documentclass[12pt, dvipsnames]{ociamthesisv2}

\usepackage{xcolor}

\usepackage{setspace}
\usepackage{fullpage}
\usepackage{float}
\usepackage[british]{babel}

\usepackage{lineno}

\usepackage{amsfonts,amssymb,amsbsy,amsthm, stmaryrd}
\usepackage{amsmath, nccmath}
\usepackage{tikz}
\usepgflibrary{patterns}
\usepackage{enumerate}
\usepackage[shortlabels]{enumitem}
\usepackage{layout}
\usepackage[english=british]{hyphsubst}
\usepackage[utf8]{inputenc}
\usepackage[T1]{fontenc}

\usetikzlibrary{patterns,positioning,arrows,decorations.markings,calc,decorations.pathmorphing,decorations.pathreplacing,decorations.fractals}

\usepackage{lscape}
\usepackage{rotating}
\usepackage{tabularx,multirow}
\usepackage{makeidx}
\usepackage{multicol}
\usepackage{mathrsfs}
\usepackage[normalem]{ulem}
\usepackage{relsize}
\usepackage{comment}
\usepackage{subcaption}

\usepackage[ruled,linesnumbered]{algorithm2e}
\usepackage{cleveref}

\usepackage[autostyle=false, style=british]{csquotes}
\MakeOuterQuote{"}

\usepackage{bm}
\usepackage{ulem}

\binoppenalty=\maxdimen
\relpenalty=\maxdimen

\def\titleEn{Intensive use of computing resources for dominations in grids and other combinatorial problems}
\def\titleFr{Utilisation intensive de l'ordinateur : dominations dans les grilles et autres problèmes combinatoires} 

\author{Alexandre Talon}             
\college{Laboratoire de l'Informatique du Parall\'elisme (LIP) - UMR 5668}  

\renewcommand{\submittedtext}{Th\`ese de doctorat de l'universit\'e de Lyon}
\degree{op\'er\'ee par\\
	l'\'Ecole Normale Sup\'erieure de Lyon}     
\degreedate{Soutenue publiquement le 29 novembre 2019 par\\}         

\def\numThese{2019LYSEN079}
\makeindex

\def\VV{\mathcal{V}}
\def\FF{\mathcal{F}}

\def\EE{\mathcal{E}}
\def\SSS{\mathcal{S}}
\def\RRR{\mathrm{R}}
\def\OO{\mathcal{O}}
\def\PP{\mathcal{P}}

\def\NN{\mathbb{N}}

\def\ZZ{\mathbb{Z}}
\def\ZZd{\mathbb{Z}^d}
\def\RR{\mathbb{R}}

\def\dd{\mathrm{d}}

\def\ab{\alpha\beta}
\def\cd{\gamma\delta}
\def\ac{\alpha\gamma}
\def\ad{\alpha\delta}
\def\bc{\beta\gamma}
\def\bd{\beta\delta}

\def\FF{\mathcal{F}}
\def\Gin{G_\mathrm{in}}
\def\Gout{G_\mathrm{out}}

\theoremstyle{remark}
\newtheorem{rk}{Remark}[chapter]

\theoremstyle{definition}
\newtheorem{deff}{Definition}[chapter]
\newtheorem{notation}[deff]{Notation} 
\newtheorem{ppty}{Property}[chapter]
\theoremstyle{plain}
\newtheorem*{question*}{Question}
\newtheorem{question}{Question}
\newtheorem{thm}{Theorem}[chapter]
\newtheorem{prop}[thm]{Proposition}
\newtheorem{lemma}[thm]{Lemma}

\newtheorem*{theorem*}{Theorem}
\newtheorem*{claim*}{Claim}
\newtheorem{claim}{Claim}[chapter]
\newtheorem{fact}{Fact}[chapter]
\newtheorem{conj}{Conjecture}

\definecolor{rouge}{RGB}{255,77,77}
\definecolor{vert}{RGB}{0,178,102}
\definecolor{jaune}{RGB}{255,255,0}
\definecolor{violet}{RGB}{208,32,144}
\definecolor{orange}{RGB}{255,140,0}
\definecolor{bleu}{RGB}{0,0,205}

\newenvironment{enumerate*}%
  {\begin{enumerate}%
    \setlength{\itemsep}{0pt}%
    \setlength{\topsep}{0pt}
    \setlength{\parskip}{0pt}}%
  {\end{enumerate}}

\begin{document}

\setcounter{secnumdepth}{3}
\setcounter{tocdepth}{3}

\maketitle                  

\addcontentsline{toc}{chapter}{Acknowledgements}

\thispagestyle{empty}
\begin{center}
\vspace*{1.5cm}
{\Large \bfseries Acknowledgements}
\end{center}
\vspace{0.5cm}

Here is the place to thank all the people who, directly or not, helped me achieve this PhD. Such help fell into two main categories: science and friendship or social activities.

First, I want to thank Michaël for his guidance and support during all the duration of my PhD. This includes the knowledge he shared with me, as well as the numerous discussions we had about the problems we tackled. I also want to thank him, Guilhem, Nathalie, Silvère, Tero and Mathieu for reading part or the whole of this manuscript, and providing much helpful feedback. I also want to thank Silvère for his involvement in the work we did together, which gave me knowledge about subshifts which are very interesting objects. I thank Karl Dahlke for answering my questions with plenty of details about the work he did on polyominoes. I also want to congratulate him for his project Edbrowse, a text-based editor which was aimed to help blind people access the web. Many thanks to the MALIP team who helped me with administrative tasks, and Alice and Aurore for helping me for the final steps of my PhD.

I now want to thank many other people whom I have met, and who helped me evade from the world of the PhD. First, I want to thank Alice, Alma and Aurore for their very valuable friendship and support throughout these three years. The people I played badminton with, and most specifically Edmond, Juliette, Corentin, Leeroy, Hubert and Selva allowed me to avoid some stress through sport, and we spent very nice moments even outside badminton sessions. Michaël and Laetitia supported me regularly around some beers, helped me rethink part of the world, and we shared a lot of fun moments. Gabrielle also supported me around beers, being for some time in the same adventure as me. To continue in the register of beers, I am very glad I randomly met Alice, with whom I spent lots of great moments, with interesting discussions. I am also grateful to Flore for all the valuable discussions we had, and her integrity of reasoning over various matters, among which climate change. Janelle and Mathieu, with small attentions, helped me a lot at a period when my PhD moments were not so bright. All the messages I received from friends to cheer me up, at the end of my PhD, were very nice and helped me keep up, so thanks to all of you who reached me on this occasion. I also had nice discussions and jokes with the people of the MC2 team.

A lot of other people have been part of my life through several activities. I want to thank the members of ENvertS and the RF for all the projects and things we organised together. I also want to thank the Improfesseurs for all the training sessions and shows we did. Finally, the "club jeux" (board games club) was a big part of my life during these three years; I am grateful to the support the people attending it provided me, particularly from Jean, Valentine and Henry.
I thank my parents who always supported me, and even understood at some point that asking me what I had done during the week was not a good question to raise each week. Finally, I am very happy to have met so many people with whom I enjoyed very good moments. I am glad to have kept contact with older friends too. Even though I cannot put all your names here, I think to all of you while writing this and am happy to know you.

\phantomsection
\addcontentsline{toc}{chapter}{R\'esum\'e}
\begin{resume}
	
Nous cherchons à prouver de nouveaux résultats en théorie des graphes et combinatoire grâce à la vitesse de calcul des ordinateurs, couplée à des algorithmes astucieux. Nous traitons quatre problèmes.\\

Le théorème des quatre couleurs affirme que toute carte d’un monde où les pays sont connexes peut être coloriée avec 4 couleurs sans que deux pays voisins aient la même couleur. Il a été le premier résultat prouvé en utilisant l'ordinateur, en 1989. Nous souhaitions automatiser encore plus cette preuve. Nous expliquons la preuve et fournissons un programme qui permet de la réétablir, ainsi que d'établir d'autres résultats avec la même méthode. Nous donnons des pistes potentielles pour automatiser la recherche de règles de déchargement.\\

Nous étudions également les problèmes de domination dans les grilles. Le plus simple est celui de la domination. Il s'agit de mettre des pierres sur certaines cases d'une grille pour que chaque case ait une pierre, ou ait une voisine qui contienne une pierre. Ce problème a été résolu en 2011 en utilisant l’ordinateur pour prouver une formule donnant le nombre minimum de pierres selon la taille de la grille. Nous adaptons avec succès cette méthode pour la première fois pour des variantes de la domination. Nous résolvons partiellement deux autres problèmes et fournissons des bornes inférieures pour ces problèmes pour les grilles de taille arbitraire.\\

Nous nous sommes aussi penchés sur le dénombrement d’ensembles dominants. Combien y a-t-il d’ensemble dominant une grille donnée ? Nous étudions ce problème de dénombrement pour la domination et trois variantes. Nous prouvons l'existence de taux de croissance asymptotiques pour chacun de ces problèmes. Pour chaque, nous donnons en plus un encadrement de son taux de croissance asymptotique.\\

Nous étudions enfin les polyominos, et leurs façons de paver des rectangles. Il s'agit d'objets généralisant les formes de Tetris : un ensemble de carrés connexe (« en un seul morceau »). Nous avons attaqué un problème posé en 1989 : existe-t-il un polyomino, d'ordre impair ? Il s'agit de trouver un polyomino qui peut paver un rectangle avec un nombre impair de copies, mais ne peut paver de rectangle plus petit. Nous n'avons pas résolu ce problème, mais avons créé un programme pour énumérer les polyominos et essayer de trouver leur ordre, en éliminant ceux ne pouvant pas paver de rectangle. Nous établissons aussi une classification, selon leur ordre, des polyominos de taille au plus 18.

\end{resume}

\clearpage
\phantomsection
\addcontentsline{toc}{chapter}{Abstract}
\begin{abstract}
	Our goal is to prove new results in graph theory and combinatorics thanks to the speed of computers, used with smart algorithms. We tackle four problems.\\

The four-colour theorem states that any map of a world where all countries are made of one part can be coloured with 4 colours such that no two neighbouring countries have the same colour. It was the first result proved using computers, in 1989. We wished to automatise further this proof. We explain the proof and provide a program which proves it again. It also makes it possible to obtain other results with the same method. We give potential leads to automatise the search for discharging rules.\\

We also study the problems of domination in grids. The simplest one is the one of domination. It consists in putting a stone on some cells of a grid such that every cell has a stone, or has a neighbour which contains a stone. This problem was solved in 2011 using computers, to prove a formula giving the minimum number of stones needed depending on the dimensions of the grid. We successfully adapt this method for the first time for variants of the domination problem. We solve partially two other problems and give for them lower bounds for grids of arbitrary size.\\

We also tackled the counting problem for dominating sets. How many dominating sets are there for a given grid? We study this counting problem for the domination and three variants. We prove the existence of asymptotic growths rates for each of these problems. We also give bounds for each of these growth rates.\\

Finally, we study polyominoes, and the way they can tile rectangles. They are objects which generalise the shapes from Tetris: a connected (of only one part) set of squares. We tried to solve a problem which was set in 1989: is there a polyomino of odd order? It consists in finding a polyomino which can tile a rectangle with an odd number of copies, but cannot tile any smaller rectangle. We did not manage to solve this problem, but we made a program to enumerate polyominoes and try to find their orders, discarding those which cannot tile rectangles. We also give statistics on the orders of polyominoes of size up to 18.

\end{abstract}
 
\begin{romanpages} 
\tableofcontents            

\cleardoublepage
\phantomsection
\addcontentsline{toc}{chapter}{Introduction}
\begin{introduccion}
This PhD falls into what is called "computer science", or "informatique" in French. This phrase means the science, hence the study, of computers. One may find it strange that computers are not more part of it. In some countries, a large part of computer science is called "discrete mathematics". It is considered part of this topic since for instance combinatorics and graph theory consists in studying some discrete objects, which intuitively means objects we can count. Rational numbers are also discrete objects whereas real numbers belong to the field of continuous mathematics: there are much too many of them to be able to count them. However, a certain number of people use computers in their research work in computer science, but only a few proved big results thanks to computers.
	
	In this PhD, we put the computers in the limelight and used them as tools to help us find new results. There are mainly two ways of using a computer to prove a result: using its huge computational power, and using its rigour to validate proofs and certify their correctness. In this thesis we focus on the first means: devising and implementing programs to be run by computers, the result of which are then used to prove a theorem. However a computer still has finite resources\footnote{As does the Earth, see \Cref{dom-counting-chapter}.} and many problems involve checking an infinite number of objects. This means that theory is needed to reduce this infinity of cases into a finite number of them, even if that number is huge (but not too huge). This part is fundamental, but may turn out not to be enough. A second step may be to think more to reduce again the number of objects the computer will have to examine, thanks to symmetry reasons for instance.\\
	
    The first big result which was proved by harnessing computers power is the four-colour theorem. It can be formulated as follows: any map in which countries are connected\footnote{In particular, all islands must be a proper country of their own.} can be coloured using only four colours, and such that any two neighbouring countries have different colours. In computer science, it is stated as "every planar graph is four colourable". It was first proved by Appel and Haken in 1976~\cite{appel-haken}. There were some doubts on this proof since it contained a flaw, which was quickly fixed by the authors. In 1996, Robertson and al. proved again this theorem in~\cite{4-col-paper} and their proof was somewhat more concise. In fact, both proofs rely on two lists of graphs (forbidden graphs, and discharging rules) and the proof of Robertson and al. contained much fewer forbidden configurations (633 versus 1476) and rules (32 versus 487). The four-colour theorem was proved using clever ideas and the computational power of computers. Actually, the idea of using discharging was first suggested by Heesch during the 70s, and he had thought about it since the 60s. He might have, with Durre, proved it earlier than Appel and Haken, but they did not have enough resources to run their program. The introduction of computers and the increase of their speed was crucial for the proof of this theorem.
	
	Other problems were solved or tackled using the power of computers, like the ones we talk about in this thesis. For instance, Rao recently closed a problem which had remained open for a century about tiling the plane with convex pentagons in~\cite{rao-pentagons}: he showed that no other families of convex pentagons than the ones already known could tile the plane. However the use of computers is not restricted to computer science nor to graph theory. Some people in cryptography or number theory also resort to the raw power of computers: some try to factor RSA integers, and others try to find, for example, which natural integers can be written as $x^3 + y^3 + z^3$ where $x,y,z \in \ZZ$. We now know\footnote{The answer for the famous number 42 was found very recently.} the answer for all numbers lesser than or equal to 100 and there are only ten numbers lesser than 1000 for which we do not have the answer yet.
	
	In a completely different domain, here geometry with applications in chemistry, a 400-year-old conjecture made by Kepler was proved in 1998 by Hales and Ferguson in a series of papers. The whole proof and some comments may be found in~\cite{kepler-livre}. This conjecture states that any packing of sphere of equal radius has a density at most around 0.7405. It was proved by examining a finite list of configurations to check the conjecture, despite the infinite number of possible configurations. It was Fejes Tóth who showed that this problem could be reduced to checking only a finite number of configurations, and provided them in 1953. The proof by Hales and Ferguson was very talked about since it closed a 400-year-old problem which was in the list of most important problems Hilbert made. Some people were not convinced of it because it was obtained with the help of a computer, like the proofs of the four-colour theorem some had earlier doubted.\\
	
	 The other use for computers we mentioned above belongs more in the fields of compilation or logic (or even arithmetic). It uses the automation of the computers to certify things like correctness of proofs, correctness of programs, or decimals in floating-point arithmetic. Apart from using the raw power of computers, one can use the automation and the ability of a computer to follow some rules strictly. Some logicians try to certify proofs, with software tools like Coq, a proof assistant/engine. This means that they input each step of a proof in a certain fashion to the proof assistant and the software validates these steps, which certifies that the proof contains no flaws. This was done with the four colour theorem by Gonthier~\cite{4-col-coq} in 2005. Concerning the Kepler conjecture, it was in fact checked during four years by a team of twelve reviewers, which is quite exceptional. They said they estimated the probability of the proof to be true to be at least 99\%, and the Kepler conjecture was widely accepted as a theorem. To remove any remaining doubts, in 2014 the Flyspeck project team, headed by Hales, showed in~\cite{kepler-formal} that their proof of the conjecture was correct, combining two proof assistants. People studying compiler design and/or logic also try to certify some aspects of programs: they prove, with the help of a computer, specific properties a program will have, for instance that never a division by zero would occur, or that the program will always terminate\footnote{It may be possible to show that some programs do terminate, however there is no way to decide if any program given as input always terminate.}. These so-called certified programs are used in critical real-time embedded systems, like in planes. We recall that this thesis is not about certification or proof assistants.\\
	 
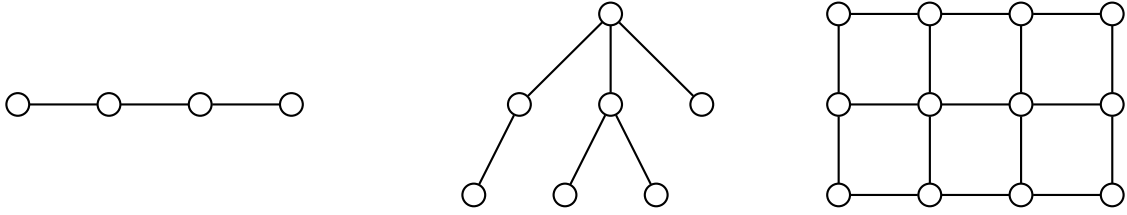
\begin{figure}[H]

\centering
\begin{tikzpicture}[->, scale=0.6, transform shape, thick,place/.style={draw, circle,thick,
inner sep=0pt,minimum size=5mm}]

\node (a) at (-1,0) [place] {};
\node (b) at (1,0) [place] {};
\node (c) at (3,0) [place] {};
\node (d) at (5,0) [place] {};

\draw [-] (a) edge (b);
\draw [-] (b) edge (c);
\draw [-] (c) edge (d);

\node (u) at (12,2) [place] {};
\node (v) at (10,0) [place] {};
\node (w) at (12,0) [place] {};
\node (x) at (14,0) [place] {};
\node (v1) at (9,-2) [place] {};
\node (w1) at (11,-2) [place] {};
\node (w2) at (13,-2) [place] {};

\draw [-] (u) edge (v);
\draw [-] (u) edge (w);
\draw [-] (u) edge (x);

\draw [-] (w) edge (w1);
\draw [-] (w) edge (w2);
\draw [-] (v) edge (v1);

\node (d) at (17,2) [place] {};
\node (dd) at (19,2) [place] {};
\node (ddd) at (21,2) [place] {};
\node (dddd) at (23,2) [place] {}; 

\node (e) at (17,0) [place] {};
\node (ee) at (19,0) [place] {};
\node (eee) at (21,0) [place] {};
\node (eeee) at (23,0) [place] {}; 

\node (f) at (17,-2) [place] {};
\node (ff) at (19,-2) [place] {};
\node (fff) at (21,-2) [place] {};
\node (ffff) at (23,-2) [place] {}; 

\draw[-] (d) -- (dd);
\draw[-] (ddd) -- (dd);
\draw[-] (dddd) -- (ddd);
\draw[-] (e) -- (ee);
\draw[-] (eee) -- (ee);
\draw[-] (eeee) -- (eee);
\draw[-] (f) -- (ff);
\draw[-] (fff) -- (ff);
\draw[-] (ffff) -- (fff);

\draw[-] (d) -- (e);
\draw[-] (e) -- (f);
\draw[-] (dd) -- (ee);
\draw[-] (ee) -- (ff);
\draw[-] (ddd) -- (eee);
\draw[-] (eee) -- (fff);
\draw[-] (dddd) -- (eeee);
\draw[-] (eeee) -- (ffff);

\end{tikzpicture}
\caption{From left to right: a path, a tree and a grid. The vertices are the circles, and the edges are the line segments joining them.}
\label{ex-graph-intro}
\end{figure}
	
In this thesis we study mostly problems on graphs, and some combinatorics. A graph can be seen as a list of relations between objects (see \Cref{ex-graph-intro}). For instance, any social networks would be seen as the set of friendships between its members. The members are called the \emph{vertices} of the graph and the friendships are called the \emph{edges}: they connect pairs of vertices. Graphs are for example used to model a lot of problems occurring in real life: when asking for a route for public transport, the website or application of the transport company uses graphs to model its network. Each stop is a vertex, and two successive stops on a same line are connected by an edge, weighted by the average time between the two stops. In case a bus and the underground both have the two stops as successive stops, we may have multiple edges between them. To answer the customer of the public transport who would like to know the fastest route, we have to study the \emph{shortest-path problem}: given two vertices, what is the path of minimal total weight between them? This problem has been much studied in graph theory and we know efficient algorithms to solve it.

However, not all problems on graphs are easy and a lot of people study graphs in more abstract ways, without direct applications. These questions include categorising graphs into classes the elements of which share a lot of properties. Often the elements of a same class are constructed according to specific patterns and rules. Other questions include studying some problems on graphs, showing that some have a particular property, and designing fast algorithms to solve some problems. Part of the motivation to split graphs into classes is to make it easier to study them. For instance, a certain problem may be very hard on graphs in general, but turn out to be easier on some classes of graphs. Usually paths (a line of vertices connected only to their left and right neighbours) and trees (like a family tree, see \Cref{ex-graph-intro}) are the simplest families of graphs on which to study a problem. Other classes of graphs which have a bounded treewidth, that is "look like trees",  may also be easier to study than general graphs. Another class, which is harder to study because it does not have a bounded treewidth, or even a bounded cliquewidth, is the class of grids: each cell of the grid is a vertex, and it is connected to the four adjacent cells.\\

The main part of this thesis addresses some problems of \emph{domination} on this class of grids. This problem consists in selecting a set of \emph{dominating} vertices such that any vertex not in this set is connected to a dominating vertex. The goal is to select as few vertices as possible while dominating the graph. The domination problem can model real-life problems such as choosing where to build fire stations or hospitals such that every city is not too far away from one. Multiple variants of this problem have been studied. They are obtained by modifying the condition for a vertex to be dominated. The Roman domination, for instance, looks like a game: we put zero, one or two troops of soldiers in each vertex, and the set of troops is dominating if any vertex with zero troops has a neighbouring vertex with two troops. Intuitively, a military strategy could consider that a troop would defend the vertex and two troops could split into two, one of the two leaving to defend a neighbouring vertex\footnote{provided the enemy does not attack too many vertices at once!}. These domination problems are very hard on general graphs, but we managed to solve some and approximate others in the case of grids. The simplest domination problem was already solved in grids by Gonçalves et al.~\cite{rao} in 2011. Among the multiple variants of the dominations, we study here the domination, minimal-domination, 2-domination, Roman-domination, total-domination, minimal-total-domination and distance-two-domination problems. We introduce the meta $k$ domination and the minimal meta $k$ domination, which extends the domination and total-domination problems in a certain direction. We already mentioned that the goal is to find, given a graph, the minimum size of a dominating set. We study this optimisation problem for variants of the domination, in \Cref{domination-chapter}. Another problem which we study here is to try to estimate how many dominating sets there are in a large grid, in \Cref{dom-counting-chapter}. Whereas the minimum dominating set in grids belongs to graph theory and combinatorics, this problem leads us to the side of combinatorics, outside graph theory.\\

As we mentioned above, computers are able to operate in an automatic way and very fast. They may try a lot of possible solutions of a problem, provided we program them to do so by giving a set of instructions in the right syntax. A computer program would for instance solve a Sudoku in a matter of seconds, or even less than one if correctly optimised. The same goes for many puzzles which may take us some, or even a lot of time to solve. One other kind is the Rubik's Cube: people usually follow an algorithm to solve one, hence a computer can solve it way faster than any human would. 

Some people exploited the fact that even a computer cannot solve a problem if the number of configurations to examine is too big, and if no ways to reduce this number are known. The Eternity game was a board game consisting of 209 irregularly shaped small polygons, and a dodecadron board. The goal was to fit all the 209 pieces together to tile the board, and a prize of one million pounds was promised to whomever would solve it first within four years. A pair of mathematicians made it in 2000 and were paid. A a sequel, called Eternity 2, was released in 2007. It was composed of 256 decorated squares: they should be placed on the board so that decorations match on the edges of two adjacent squares. This problem is strongly related to objects called Wang tiles, which are equivalent to the concept of SFTs we discuss in this thesis. A prize of two million dollars was offered if someone could solve within four years, but no one achieved it.

We tackle another "geometric" problem, similar to the first Eternity game, which also belongs the the field of combinatorics. It involves polyominoes, which are generalisations of Tetris pieces: they may contain fewer or more than four squares. 
We are interested in a more complex puzzle: given a polyomino and as many copies of it that you want, try to tile a rectangle with these pieces. The rectangle is not given, it is up to you to find if there exists one tilable by the polyomino or if the polyomino cannot tile any rectangle. Also, if it can tile a rectangle, please find one which requires as few copies of the polyomino as possible as in \Cref{intro-tiling}. It would be awesome if this minimum number of copies could be an odd integer greater than one. Unfortunately, we still do not know if this is possible. Here, contrarily to the games we mentioned, the number of configurations is not finite, for there are infinitely many rectangles we could try. Even trying to tile a single rectangle with about a hundred copies of a polyomino may require a huge amount of configurations to try. Fortunately we may use clever arguments to reduce this number, which makes it possible to find the order of some polyominoes, but not for every one we tested.\\

\begin{figure}

\centering
\includegraphics[scale=1]{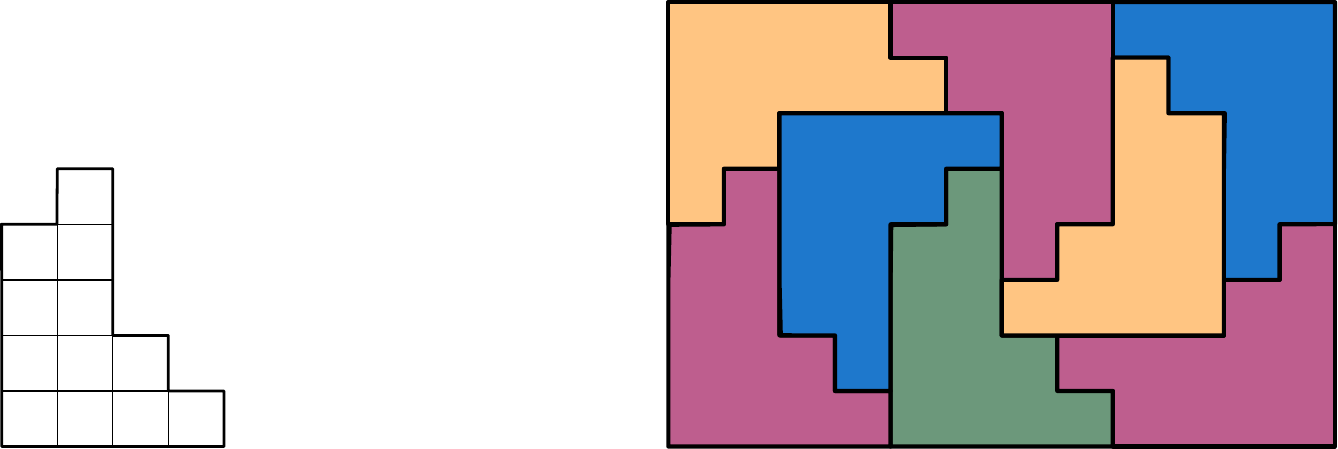}

\caption{A polyomino of size 12 and order 8 and a tiling of the associated rectangle.}
\label{intro-tiling}
\end{figure}

To tackle any problem with the help of a computer, a certain methodology must be followed. The power of computers is huge, all the more so that we can use parallelism and we have increasingly powerful machines. However, there is still a finite and fixed number of operations doable in a certain amount of time with a given machine. This means that if the problem we want to tackle implies examining an infinite number of objects, the first step is to reduce it to a finite number, thanks to theoretical arguments. The second step is to consider the angle of attack: what type of algorithms should we use? Also, what are the optimal the data structures we need? Once this is done, it may be useful to estimate the running time and memory consumption of the program to be devised, if this is possible. If it exceeds the available resources, more care must be given to reducing the number of objects. One classical way to do so is to observe the symmetries of the objects and try to reduce their numbers according to this fact. Indeed, if a group of objects are equivalent, we may as well examine thoroughly one of them and detect that the others do not need to be examined. Another way to use fewer resources can be to think more about the algorithm and data structures, and find better ones. In case the bottleneck is only time, we may do some precomputations and store intermediate results to be reused. This trades some memory for a smaller running time. If on the contrary the amount of memory used is problematic, we may store fewer results and compute them again several times, to benefit from the inverse trade-off.

Nowadays the frequencies of computers no longer increase, or by very little. The way to gather more computing resources is to resort to parallelism. This can usually be done by launching a lot of instances of the same program, at the same time, with different inputs. Another way is to use threads: a program splits itself into several children which run the same code. We split the amount of computations between the threads so that if we have $k$ threads, each would do one $k^\text{th}$ of the computations. To benefit fully from parallelising the code, we must think carefully of how to make parts of the computations independent from one another. Indeed, the speed up will be much higher if the parallel computations do not write to the same locations in memory and do not need to wait for results of another computation to start. One problem is that when using threads, we may tolerate simultaneous accesses to a same location of the memory provided none is writing to it. If one of the simultaneous accesses is a write, the behaviour of the program is undefined: if some thread reads this location, should the value read be the new one or the old one? In the case of two writes, which value should be the one to be stored at the end?

Another hint to achieve better performances is to choose the right programming languages and libraries\footnote{Sorts of programs made by other people, which contains routines to do specific tasks.}. Interpreted languages such as Python may be very elegant and have a nice syntax, but they can be as much as 100 times as slow as C/C++. When the critical operations are classical ones, like matrices products, it may be better to use a library written for that purpose, which will be faster. Using libraries also make it less likely to find bugs in the final program.\\

\begin{figure}[h]
\centering

\includegraphics[scale=0.45]{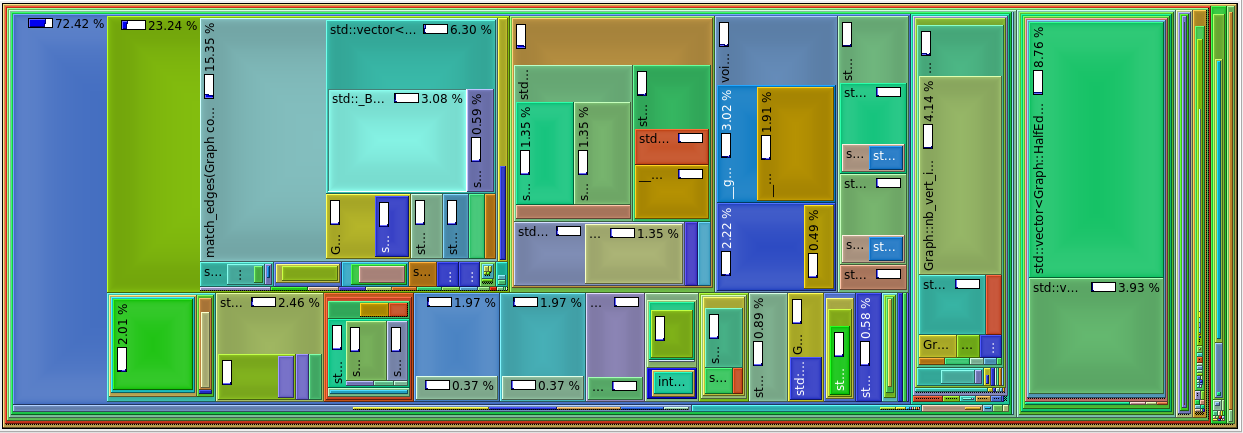}
\caption{A screenshot of a qcachegrind displaying visually some logs of the profiler callgrind. It shows each function as a rectangle whose size is proportional to the time spent in the function.}
\label{screenshot-callgrind}
\end{figure}

Finally, there are some good practices to adopt when programming something rather big. They may take a bit of time and some energy, but they can save a huge amount of time and frustration later on. The first thing is to always write some simple code first, before optimising it. It may be sufficient even if not optimal, and if we need to optimise it we will have a reference to test the new code against. The second thing is then to have tests, if possible automatic ones. For instance, running the first simple code on small inputs can generate test cases for the subsequent versions. Third thing, \textbf{do not overoptimise, or optimise too early}: there is no use optimising some parts (small rectangles in \Cref{screenshot-callgrind}) of the code which takes only one or two percent of the running time. It is only making the code harder to read, maintain, and more prone to bugs. The objective is to find a good balance between a simple code and a fast one. To see which parts of the code are important to optimise we may use tools called \textit{profilers}: they say how many times each function was called, and what portion of the running time these calls took. We used a lot the callgrind software in this PhD (see \Cref{screenshot-callgrind}), with the visualiser qcachegrind.  When bugs occur or are suspected, it is better to use a \textit{debugger}: it enables us to stop at specific points of the program (or at the point where it crashes) and have some information about the values of the variables, the position in the code, which functions were currently being called, and so on. We used GDB for this purpose. For similar goals, we also used valgrind, a piece of software which performs further checks to detect some possible memory errors or misuse. All these low-level tools, though, are to be used in the last resort: they do not replace a careful design and implementation of the algorithms.

\begin{center}
	\textbf{Organisation of the manuscript}
\end{center}
	This thesis deals with four topics: the proof of the four-colour theorem, several dominations numbers of grid graphs, counting some dominating sets in grids, and a tiling problem with polyominoes. The second and third topics are very related: they study similar problems. The last problem is to some extent related to the second topic. All these four topics share the approach to solve them, which resorts at some point to the use of computer programs. This thesis is divided in four chapteris, one dealing with each of the problems we have just listed.\\

	\Cref{discharging-chapter} is about the four-colour theorem and its proof in~\cite{4-col-paper}. The theorem states that any finite planar graph can be coloured properly using only four colours. We first explain the proof from Robertson et al. It works by contradiction: it assumes the existence of a minimal counter-example to the theorem. The minimality forbids some graphs to appear in this counter-example: we generate a list of \emph{forbidden (or reducible) configurations}. The reason the proof works is that no planar graphs can avoid all these reducible configurations. This is shown with the \emph{discharging method}, which consists in assigning weights to the vertices and moving them according to a set of \emph{rules}. This rules are used in an exhaustive search for the hypothetical minimal counter-example. They make it possible to cut all the branches of explorations. Finally, only a finite set of graphs are examined and none can lead to a minimal counter-example, hence the theorem is true. We ported this program for a use in Sage, so that the community may prove some other results with the discharging method, without having to implement everything from scratch. We describe some aspects of the program.\\
	
	\Cref{domination-chapter} tackles variants of the domination number, i.e. the minimum size of a dominating set, in grid graphs. The domination number problem in grids was solved in 2011 by Gonçalves et al.~\cite{rao}. We reuse for the first time their method and apply it to the 2-domination, the Roman domination, the total domination and the distance-2 domination. We explain there the method, and give the results we obtained using a new program we wrote. We solve the 2-domination and the Roman domination, that is we provide closed formulas giving the 2-domination number and the Roman domination number according to the height and the width of the grid. We give, for the distance-2 domination and the total domination, the numbers for up to 15 lines and arbitrary number of columns. For the total loss, this confirms the results of Crevals and Ostergård~\cite{total-dom-article-28}, who went up to 28 lines. We also give, for the total domination, a lower bound provided by our program. We also try to explain why our method does not seem to be able to give a full result for the total domination, which we relate to some covering problems. Most of the results of this chapter are published in~\cite{rao-talon}.\\
	
    The dominations problems can be viewed from another angle: trying to count the number of different dominating sets of a specific graph. In \Cref{dom-counting-chapter}, we attack this problem in grid graphs for several domination variants: the domination, the total domination, and their minimal variants. Using the notion of \emph{subshifts} and some known results about them, we show that the number of dominating sets, for each of these four problems, admits a growth rate which is furthermore \emph{computable}. We show this property by analogy with the \emph{entropy} of the domination subshifts, which we show are \emph{block gluing}. The number of dominating sets in grids is showed to be $\nu^{nm+o(nm)}$ for some constant $\nu$ depending on the problem studied. For each of these constants, we give numerical bounds obtained by our program. We also introduce a new domination family, which generalises the domination and total domination problems: the \emph{meta-$k$-domination} family of problems. We also study one particular property of the subshifts associated to this family. The work of this chapter is a joint work of Silvère Gangloff and myself~\cite{article-counting}.\\
	
	In \Cref{polys-chapter}, we study a problem related to the covering of rectangles we mentioned for \Cref{domination-chapter}. We study polyominoes and the way they can or cannot tile some rectangle. We are interested in \emph{rectifiable} polyominoes: when there exists a rectangle the polyomino can tile. The question we attacked, and has been open for 30 years, is the following: is there a polyomino of \emph{odd order} greater than one ? This means we look for a polyomino which can tile a certain rectangle with an odd number $k > 1$ of copies and which cannot tile any rectangle with fewer copies. We tried to find such a polyomino, unfortunately with no success. We describe several ways to find the order of a polyomino. Apart from these algorithms, we also describe some methods to show, again with the help of a computer program, that a polyomino is not rectifiable. Some ideas come from the work of Karl Dahlke~\cite{dahlke-website}, which we programmed anew, along with some of our own optimisations.\\

    This arxiv version contains the sources of the program developped and used during this PhD. To access them, you need to select the "Other formats" option in the Download section. Then click the "Download source" link.

\end{introduccion}


\end{romanpages}            


\resetlinenumber
\chapter{Discharging and the four colours theorem}
\label{discharging-chapter}
One thing which strongly characterises human beings is their relation to waste. No other species produce that much waste and then do not care about it.
This waste can take many shapes: from the plastic of our packagings, to some residues of the cleaning product we use (or other products used by the industry), buildings and machines left beside when no longer used... and obviously nuclear waste, which can last for  100 000 years. But more shockingly is how we process this waste: some of it is simply discharged at the middle of some nature, illegally\footnote{and the authorities do not seem to put the necessary means to stop it, even when these places are noticed}. Some nuclear waste was just dumped into water: between 1946 and 1993, before it was agreed to stop doing this\footnote{Or maybe not, TEPCO company plans on putting 777 000 tons of contaminated water in the sea following the Fukushima accident.} France alone has put around 15 000 tons of nuclear wast into the sea and the world.\\

As we said in the introduction, the four-colour theorem is very famous because it is the first big problem which was solved with a computer.

\begin{thm}[\cite{appel-haken}]
Every planar graph is four colourable.
\end{thm}

We will explain each term in \Cref{section-def-graphs}, so for the moment we consider an equivalent version.

\begin{thm}
Every planar map can be coloured with four colours.
\end{thm}

Le us assume that we draw a finite number of lines on a sheet of papers, the lines delimiting regions of the sheet. The theorem means that we can colour the regions using only four colours and such that any two adjacent regions (sharing an edge\footnote{non reduced to a point}) have different colours. For instance, it implies that it is not possible for five regions to be pairwise adjacent, for they could not be coloured with only four colours in that case. The theorem is however stronger than this fact.

The conjecture was apparently first proposed when, in 1852, Francis Guthrie was colouring the counties of England and noticed that four colours sufficed in this task to guarantee that two neighbouring counties are given different colours. A lot of people tackled this problem with little or no success until 1976. The first big attempt was made by Kempe~\cite{kempe} in 1879. It was welcomed by his peers until, eleven years later, Heawood~\cite{heawood} showed it was false. Similarly, a claim made by P. Guthrie Tait in 1880 was only disproved in 1891 by Petersen. However, Heawood reused one argument in Kempe's proof, the notions of \textit{Kempe chain} and \textit{Kempe interchange} to prove a weakened version of the four-colour theorem. These notions turned out to be important in later proofs.

\begin{thm}[\cite{heawood}]
Every planar graph is five colourable.
\end{thm}

During the $20^\text{th}$ century, much progress was made and the theorem was eventually proved. In 1913, Birkhoff~\cite{birkhoff} formalised the notion of \textit{reducible configurations}, i.e. subgraphs which cannot appear in a minimal counter-example. He notably showed that the theorem holds for graphs with fewer than 26 vertices. Beginning in the 60's, Heesch worked on the problem and introduced a crucial step towards the resolution of the problem: the \textit{discharging method}. Much excitement went about the problem, and it was thought solved by Shimamoto. He showed that the whole problem could be reduced to checking some property on a particular graph, the "horseshoe". Heesch announced that it was one of the graphs for which they had shown this property, called \textit{D-recucibility}, but it turned out that it was an incorrect result due to a flaw in the program he had devised with Dürre. Finally, Appel and Haken~\cite{appel-haken} proved it in 1976, using the reducibility and discharging methods. A small flaw in the program they used cast some doubts on their proof. This proof was criticised by some sceptical people. The program was quickly fixed by their authors, but this did not end the doubts on their proof. This proof relied on some approximatively 2000 (later reduced to 1476) \emph{reducible configurations} which had to be checked by hand by the authors and Haken's daughter,  Dorothea Blostein. It also relied on 487 \emph{discharging rules}. In 1996, Robserton et al.~\cite{4-col-paper} gave another proof, based on the same principles, but needing much fewer configurations (633) which were checked by a computer program, and also much fewer discharging rules (32). The same authors announced in 1999 an alternate proof, by proving the "snark theorem" which implied the four-colour theorem. The proof is still not completely published. All doubts on the four-colour theorem were disappeared in 2005 when Gonthier~\cite{4-col-coq}, using the Coq proof assistant, made a certified proof of the theorem.

\begin{figure}[h]
\centering
\includegraphics[scale=0.34]{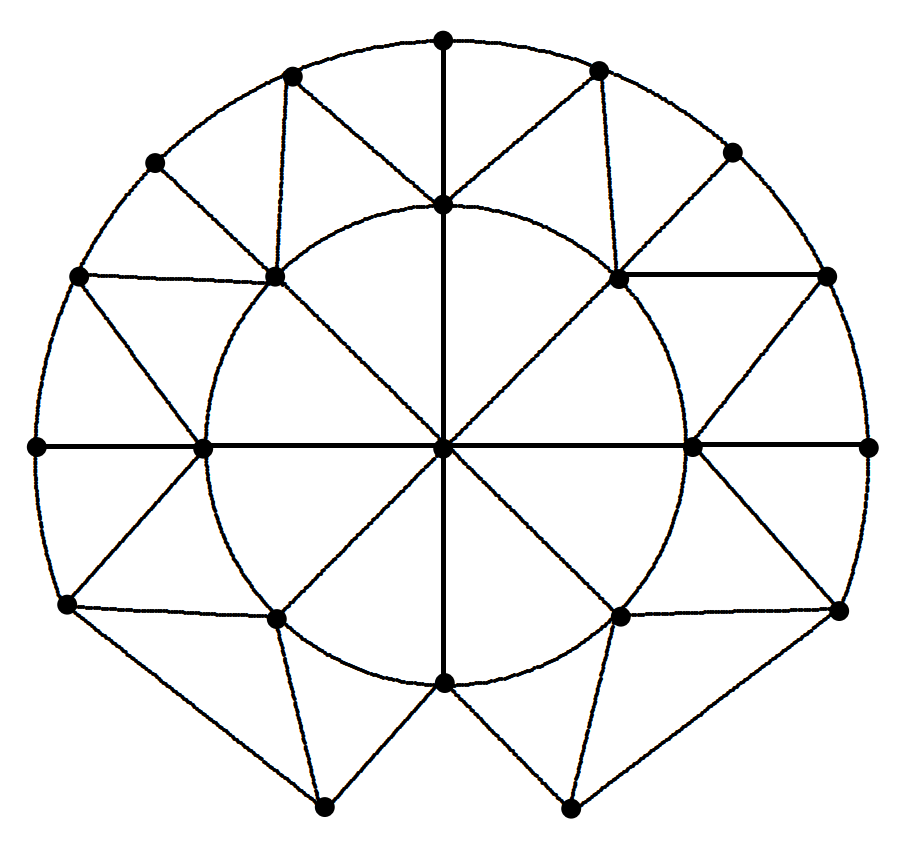}

    \caption{Shimamoto's "horseshoe" graph.\\The figure comes from~\cite{book-horseshoe}.}
\label{fig-horsehoe}
\end{figure}

We begin by giving general definitions about graphs, including planarity, and colourings, in \Cref{section-def-graphs}. \Cref{section-proof} is about the proof of the theorem by Robertson et al. We first give a one-page sketch of the main ideas of their proof. After this, we consider the false proof of Kempe, which introduced a central concept named the \textit{Kempe chains}. The rest of this section is devoted to the main two ingredient of the proofs of the four-colour theorem: \textit{reducible configurations} and \textit{discharging rules}. The former consist in subgraphs which cannot appear in a minimal counter-example, and the latter help us realise that no planar graph can avoid all the reducible configurations. The last section is about the work we did in this topic. It contains some details of the program we developed to reproduce the proof of Robertson et al., and some ideas on how to further automate the proof.

\section{Graph definitions}
\label{section-def-graphs}

In this chapter, \Cref{domination-chapter} and \Cref{dom-counting-chapter} we talk about graphs. We define here what they are, some vocabulary related to them, some of which is specific to this chapter.

\begin{deff}
A \textbf{graph} is a an ordered pair $G = (V,E)$: $V \subset \NN$ is the set of \textbf{vertices} and $E  \subset \{\{u,v\} \,|\, u,v \in V, \; u \neq v\}$ is the set of \textbf{edges}.\\If $V$ is finite, then the graph is \textbf{finite}.
\end{deff}

\begin{rk}
Our definition is the one of \emph{simple} graphs: it forbids loops (an edge between a vertex and itself) and parallel edges (multiple edges having the same endpoints). In all this thesis we only deal with simple graphs.
\end{rk}

\begin{deff}
If $\{u,v\} \in E$ we say that $u$ and $v$ are \textbf{neighbours}, or \textbf{connected}.\\
The number of neighbours of a vertex $u$ is called the \textbf{degree} of $u$ and denoted by $\dd^{\circ}(u)$.
\end{deff}

\begin{deff}
    We define the \textbf{open neighbourhood}, or \textbf{neighbourhood} of a vertex $u$, and denoted by $N(u)$ as the set of neighbours of $u$. The \textbf{closed neighbourhood} of $u$ is $N[u] = N(u) \cup \{u\}$.
\label{def-neighbourhood}
\end{deff}

\begin{figure}[H]

\centering
\begin{tikzpicture}[->, scale=0.6, transform shape, thick,place/.style={draw, circle,thick,
inner sep=0pt,minimum size=7mm}]

\node (a) at (0,0) [place] {$a$};
\node (b) at (6,0) [place] {$b$};
\node (c) at (6,-6) [place] {$c$};
\node (d) at (0,-6) [place] {$d$};
\node (e) at (3,-3) [place] {$e$};
\node (f) at (9,-3) [place] {$f$};
\node (g) at (9,-6) [place] {$g$};

\draw [-] (a) edge (e);
\draw [-] (b) edge (e);
\draw [-] (c) edge (e);
\draw [-] (d) edge (e);
\draw [-] (a) edge (b);
\draw [-] (c) edge (d);
\draw [-] (a) edge (d);
\draw [-] (b) edge (c);
\draw [-] (b) edge (f);
\end{tikzpicture}
\caption{An example of a graph.}
\label{ex-graph}
\end{figure}
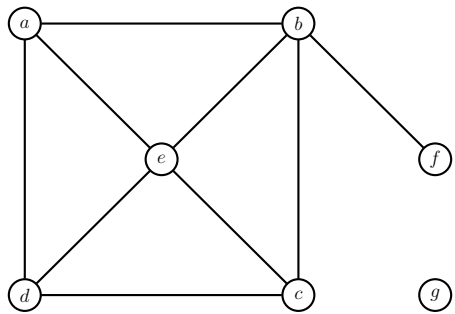

\Cref{ex-graph} illustrates a graph with 7 vertices ($a$ to $g$) and 9 edges. $a$ and $e$ are connected, but $c$ and $a$ are not. $e$ has degree 4, $f$ has degree 1 and $g$ is not connected to any vertex: its degree is 0. The neighbourhood of $b$ is $\{a, e, f\}$ and its closed neighbourhood is $\{a,b,e,f\}$.

\begin{deff}
An \textbf{embedding} of a graph $G$ into the plane is a representation of $G$ on $\RR^2$: vertices are given coordinates and every edge $\{u,v\}$ is drawn as a connected arc whose endpoints are the points assigned to $u$ and $v$.
\end{deff}
Usually when we speak about graphs, in addition to giving $V$ and $E$ we also give embeddings of them, i.e. we draw a representation of them.

\begin{deff}
A graph is said to be \textbf{planar} when it can be embedded into the plane such that no two edges cross.
\end{deff}
This means that there is a way to draw the graph on the plane such that no pairs of edges cross. It does not mean, of course, that any embedding of the graph would respect this property.

\begin{deff}[\Cref{plane-K4}]
Any embedding of a planar graph the edges of which only intersect at the vertices is called a \textbf{plane graph}.
\label{def-plane}
\end{deff}

\begin{rk}
Any planar graph can be embedded on a sphere such that edges only intersect at the vertices.
\end{rk}

\vspace*{-1.5cm}
\begin{figure}[H]
\centering
\begin{tikzpicture}[-, scale=0.5, transform shape, thick,place/.style={draw, circle,thick,
inner sep=0pt,minimum size=6mm}]

\begin{scope}
\node (a) at (0,0) [place] {};
\node (b) at (6,0) [place] {};
\node (c) at (6,-6) [place] {};
\node (d) at (0,-6) [place] {};

\draw [-] (a) edge (b);
\draw [-] (b) edge (d);
\draw [-] (a) edge (c);
\draw [-] (c) edge (d);
\draw [-] (a) edge (d);
\draw [-] (b) edge (c);

\end{scope}

\begin{scope}[xshift=12cm]
\node (a) at (0,0) [place] {};
\node (b) at (6,0) [place] {};
\node (c) at (0,-6) [place] {};
\node (d) at (6,-6) [place] {};

\node (e) at (-2,2) {};

\draw [-] (a) edge (b);
\draw [-] (b) edge (d);
\draw [-] (a) edge (c);
\draw [-] (c) edge (d);
\draw [-] (a) edge (d);
\draw [-] (c) edge [bend left =90, looseness=2.2] (b);
\end{scope}

\end{tikzpicture}
\caption{Two embeddings of the same planar graph ($K_4$). Only the one on the right is a plane embedding.}
\label{plane-K4}
\end{figure}
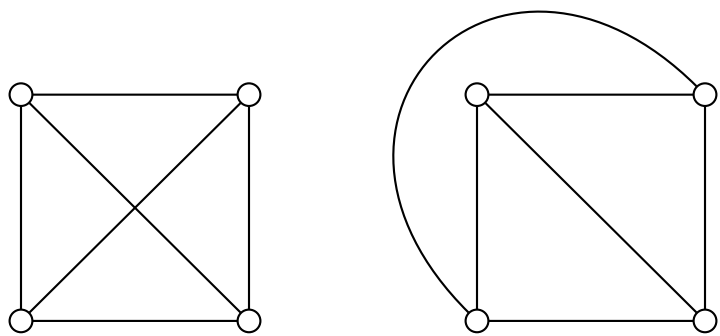

In this chapter we focus on planar graphs, and we reason in fact on plane embeddings of them.

\begin{deff}
A \textbf{face} of a plane graph is a region of the map delimited by vertices and edges, and which contains no vertices and no edges except the ones of its boundary.\\
A \textbf{triangle} is a face delimited by three vertices.
\end{deff}

In \Cref{plane-K4}, the plane graph has 4 faces, all of them being triangles.

\begin{deff}
We say that a plane graph is \textbf{triangulated} when all its faces are triangles. It is \textbf{almost-triangulated} when all its faces except at most one are triangles. If there is one, the non-triangular face is then called the \textbf{outer face}.
\end{deff}

\begin{deff}[see \Cref{ex-colouring}]
$c : V \mapsto \llbracket 1, k \rrbracket$ is called a \textbf{$\bm{k}$-colouring} of $G = 
(V,E)$.\\
A $k$-colouring $c$ is \textbf{proper} when for any $\{u,v\} \in E$ $c(u) \neq c(v)$.\\
A graph is \textbf{$\bm{k}$ colourable} when it admits a proper $k$-colouring.
\end{deff}

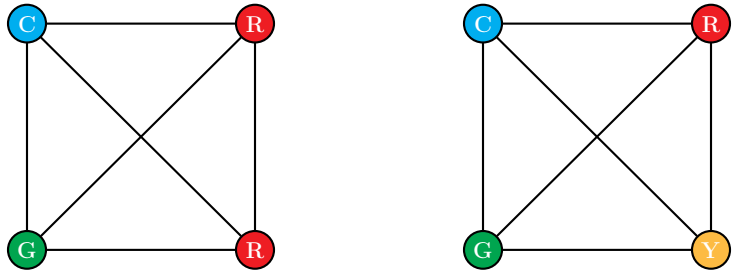
\begin{figure}[H]
\centering
\begin{tikzpicture}[-, scale=0.5, transform shape, thick,place/.style={draw, circle,thick,
inner sep=0pt,minimum size=10mm}]

\begin{scope}
\node (a) at (0,0) [place, fill=Cyan] {\color{white}\Large\textbf{C}};
\node (b) at (6,0) [place, fill=Red] {\color{white}\Large\textbf{R}};
\node (c) at (6,-6) [place, fill=Red] {\color{white}\Large\textbf{R}};
\node (d) at (0,-6) [place, fill=Green] {\color{white}\Large\textbf{G}};

\draw [-] (a) edge (b);
\draw [-] (b) edge (d);
\draw [-] (a) edge (c);
\draw [-] (c) edge (d);
\draw [-] (a) edge (d);
\draw [-] (b) edge (c);

\end{scope}

\begin{scope}[xshift=12cm]
\node (a) at (0,0) [place, fill=Cyan] {\color{white}\Large\textbf{C}};
\node (b) at (6,0) [place, fill=Red] {\color{white}\Large\textbf{R}};
\node (c) at (6,-6) [place, fill=Dandelion] {\color{white}\Large\textbf{Y}};
\node (d) at (0,-6) [place, fill=Green] {\color{white}\Large\textbf{G}};

\draw [-] (a) edge (b);
\draw [-] (b) edge (d);
\draw [-] (a) edge (c);
\draw [-] (c) edge (d);
\draw [-] (a) edge (d);
\draw [-] (b) edge (c);

\end{scope}

\end{tikzpicture}
\caption{Two colourings of $K_4$. Only the one on the right is proper: the two vertices coloured in red on the left are connected.}
\label{ex-colouring}
\end{figure}

A proper colouring means that the vertices of a graph are given colours among a set of $k$ colours, such that any two connected vertices receive different colours. The idea of giving different colours to connected vertices arises from practical problems. In the case of colouring a map, it makes the map clearer by making the frontiers more visible: since adjacent countries have different colours we may not think that they are the same big country. We can also model a room assignment problem with graph colourings. Suppose we have $n$ different lessons to be given and we fixed the schedule of the lessons but would like to know how many rooms we need. We can think of the lessons as vertices, and two vertices are connected if the lessons cannot be put in the same room because their schedules overlap. Each colour we assign to vertices will represent a room, different colours representing different rooms. Then any proper colouring of the graph with $k$ colours gives an assignment of the lessons to $k$ different rooms. The minimum number of colours needed for a proper colouring gives the minimum number of rooms needed.\\

\section{The proof}
\label{section-proof}
\subsection{Scheme of the proof}
\label{5-min-proof}
We give here a five-minute overview of the proof: we abstracted it very much skipping most details and above all technical ones. We boiled down the proof into nine notions, some very short, other more detailed.

The \uline{first idea} is that we make a \emph{proof by contradiction}: we assume that there is indeed a planar graph which is not 4 colourable. Let us call this guy $G_0$.
Actually, we will reason on plane graphs (see \Cref{def-plane}): a planar graph may admit several non-isomorphic plane embeddings.
The \uline{second idea} is to reduce the number of candidates for $G_0$: to show that it cannot exist, we must find contradictory properties it has. So we first assume that we look for a \emph{critical} counter-example: $G_0$ must be 5 colourable. We also restrict ourselves to minimal counter-examples: $G_0$ is chosen such that no graphs smaller than $G_0$ can be counter-examples.
\uline{Idea three} is then to use the previous conditions to show properties our graph has, to restrict the number of candidates for $G_0$. For instance, we can show that $G_0$ must be triangulated (see \Cref{trig-ctr-ex-lemma}) and have no vertex of degree less than five.
The \uline{fourth ingredient} is the notion of \textit{reducible configurations}: they are graphs which cannot appear in $G_0$. There are several reasons for a configuration to be reducible, and they can be tested by a computer. The goal here is to obtain a list of reducible configurations "big enough".
Now, at \uline{step five}, we want to show that no plane graphs can avoid all the reducible configurations we found. To do this, we use the discharging technique. First we decide on a way to assign weights to the vertices of $G_0$, such that the sum of the weights over any plane graph is 120. We then design \textit{discharging rules}: each time some specific configuration\footnote{not the reducible ones, some other ones} appear in a weighted graph $G$, we transfer some weight from some special vertex to another special one, depending on the rule. At the \uline{sixth step}, we know that any application of a rule preserves the total weight of the graph. Hence after applying our set of rules to $G_0$, the sum of the weights is still 120. Therefore, some vertex $v_0$ has a positive weight after applying the rules. We want to show that some reducible configuration necessarily appears in the second neighbourhood of $v_0$.
The \uline{seventh ingredient} is to enumerate all possible neighbourhoods for $v_0$ with the help of a computer. We start with $v_0$ and use a branch-and-bound algorithm to generate these neighbourhoods. At each step of the algorithm we precise a little more our neighbourhood, or we extend it. Before going further, at each step we test our current neighbourhood. If it cannot be triangulated, or if we can see that any extension of this neighbourhood will necessarily contain a reducible configuration or forces $v_0$ to have a non-positive weight after applying our set of rules, we discard it and backtrack.
The \uline{eighth thing to do} is to find a set of discharging rules which is good enough. By this we mean a set which, together with the list of reducible configurations, leads to discarding all possible neighbourhoods and show that $G_0$ does not exist. Robertson et al. found, by trial and error, a sufficient set of discharging rules. Beginning with only a few rules, the exploration program did not seem to finish. By looking at some neighbourhoods which could not be discarded, they designed new discharging rules. A new run of the program led to other new rules, and so on... until at some point the program terminated, proving the theorem.

\subsection{An interesting wrong proof}
We give here a proof which turned out to be false and could not be fixed. However it introduces an important notion in the colouring of planar graphs: the one of \emph{Kempe chains}. It was published in 1879 by Kempe~\cite{kempe}... and the flaw was discovered by Heawood~\cite{heawood} 11 years later! The proof uses some lemmas which are true and reused in the proof by Thomas et al., in particular the concept of \emph{Kempe chains}. 

We us assume we colour our vertices with four colours : $\alpha$, $\beta$, $\gamma$ and $\delta$.

\begin{lemma}
Any minimal counter-example to the four-colour theorem is triangulated.
\label{trig-ctr-ex-lemma}
\end{lemma}

\begin{proof}
Let assume $G_0 = (V_0, E_0)$ is a minimal counter-example for the four-colour theorem and has a face with at least 4 vertices. Then let $u$ and $v$ be two non-adjacent vertices of this face. We create $G_0'$ to be the graph where $u$ and $v$ are identified:
$G'_0 = (V'_0, E'_0)$ with $V'_0 = V \setminus \{u, v\} \cup \mathit{\{uv\}}$ and 
$E'_0 = E_0 \setminus \{\{a,b\} \,|\, a \in \{u,v\}\}
\cup \{\mathit{uv}, b\} \,|\, \{u,v \} \in E_0 \text{ or } \{v, b\} \in E_0\}$.

    Now any colouring $c'$ of $G'_0$ extends itself to a colouring $c$ of $G_0$ by giving to $u$ and $v$ the colour $\mathit{uv}$ has in $c'$. Indeed, identifying $u$ and $v$ preserved the vertices connected to $u$ and $v$: any of their neighbours is neighbour of $\mathit{uv}$. This works because $u$ and $v$ were not neighbours in $G_0$. This implies that the set of colourings of $G'_0$ is included in the one for $G_0$. Since $G_0$ admits no four-colourings, the same applies to $G'_0$. Therefore $G'_0$, which is still planar, is a smaller counter-example, which is a contradiction.
\end{proof}

\begin{lemma}
Any triangulated planar graph has a vertex of degree less than 6.
\end{lemma}

\begin{proof}
We use Euler's formula for planar graphs: $n-m+f = 2$, where $n$ is the number of vertices, $m$ the number of edges and $f$ the number of faces. Besides, since all faces are triangles, the number of edges satisfies $f = 2m/3$: each edge belongs to two faces and the sum of edges over the faces, which are triangles, is $3f$. Euler's formula can then be rewritten to $3n-m = 6$ or $m = 3n-6$. Summing the degrees of the vertices equals $2m$: each edge has two incident vertices. The average degree of the (finite) graph is $2m/n$ so here it is $2m/n = (6n-12)/n < 6$. Since the average degree is lesser than 6, at least one vertex must have a degree at most equal to 5.
\end{proof}

Now we want to show that no counter-examples to the four-colour theorem exist. To show this, Kempe used strong induction: we look for a minimal counter-example $G_0$. By \Cref{trig-ctr-ex-lemma}, we may know that $G_0$ has a maximum number of edges, namely that it is triangulated. We proceed by contradiction, assuming that such a counter-example $G_0$ exists. We proceed by disjunction of cases since $G_0$ has a vertex of degree $d < 6$:

\paragraph{First case: $d < 4$.\\}In this case let $u$ be a vertex of degree $d < 4$. We remove $u$ and its edges, so the remaining graph must be four colourable by the minimality hypothesis. Let us consider the colours its $d < 4$ neighbours receive. There is at least one colour which is not given to them (since we have four colours available). We can extend the four-colouring of $G_0 \setminus u$ to $u$ by colouring it with the remaining colour, hence our graph $G_0$ is four colourable, which is a contradiction.

Note that the graph obtained by removing $u$ might no longer be triangulated. However this is not a problem: $G_0 \setminus u$ is still smaller than $G_0$.

\paragraph{Second case: $d = 4$.}(see \Cref{fig-degree-4-Kempe})\\
We need here to introduce the very useful concept Kempe introduced in his proof.
\begin{deff}
Let $G = (V,E)$ be a four colourable graph and $c : V \rightarrow \{1, 2, 3, 4\}$ a four colouring of it. Let $G_{\{\alpha,\beta\}}$, more concisely written $G_{\alpha,\beta}$, be the subgraph of $G$ whose vertices are coloured with colour $\alpha$ or $\beta$. For any pair of $\alpha \neq \beta \in \{1,2,3,4\}$, any maximal connected component of $G_{\alpha,\beta}$ is called a \textbf{Kempe chain}, or more precisely an \textbf{$\bm{\ab}$-chain} of $c$. 
\end{deff}

\begin{deff}
Given $G$, $c$ and a pair $\alpha \neq \beta$ a \textbf{Kempe interchange} with respect to the colour partition $\{\{\alpha, \beta\}, \{\delta,\gamma\}\}$ is the colouring $c'$ obtained by switching colours $\alpha$ and $\beta$ in one of the $\ab$-chains, or by switching the colours $\gamma$ and $\delta$ in one of the $\cd$-chains.
\end{deff}

One important and useful fact about Kempe chains is the following:
\begin{fact}
Performing a Kempe interchange on a proper four-colouring $c$ always yields a proper four-colouring $c'$.
\end{fact}

Indeed, by definition of a Kempe chain, we are switching two colours of a maximal component of vertices which had these colours. This means that if we exchanged colours $\alpha$ and $\beta$, any vertex which had its colour changed did not have any neighbour coloured $\alpha$ or $\beta$ outside its $\ab$-chain.

\begin{rk}
Performing a Kempe interchange along an $\ab$-chain may create or remove $\ac$-chains, $\ad$-chains, $\bc$-chains and $\bd$-chains. However such a Kempe interchange does not modify the $\ab$-chains and $\cd$-chains.
\end{rk}

Now let us consider a vertex $u$ of degree 4. We remove it from the graph, and colour the resulting (smaller) graph with four colours. If two of $u$'s neighbours share the same colour, then we use one spare colour for $u$. Otherwise, the neighbours have four different colours, let us say Orange, Cyan, Pink and Red. Let us call $v_\mathrm{O}, v_\mathrm{C}, v_\mathrm{P}$ and $v_\mathrm{R}$ the neighbours of $u$ labelled with these colours. We assume that in clockwise direction, the neighbours of $u$ are $v_\mathrm{O}, v_\mathrm{C}, v_\mathrm{P}$ and $v_\mathrm{R}$. We take for instance $v_\mathrm{O}$ and the OP-chain $G_\mathrm{OP}$ which contains $v_\mathrm{O}$. If $v_\mathrm{P}$ does not belong to this subgraph, then we perform a Kempe interchange along this chain. This frees colour O because $v_\mathrm{O}$ gets colour pink. Hence we can colour $u$ with orange. On the contrary, if $v_\mathrm{P}$ belongs to the OP-chain containing $v_\mathrm{O}$ then we consider the RC-chain $G_\mathrm{RC}$ containing $v_\mathrm{R}$. $v_\mathrm{C}$ cannot belong to this chain because the chain $G_\mathrm{OP}$ acts as a barrier: it contains only orange and pink vertices, hence no red or cyan vertex of $G_\mathrm{RC}$. Since the graph is plane, no edges can cross $G_\mathrm{OP}$. This means that we can perform a Kempe interchange along $G_\mathrm{RC}$. $v_\mathrm{R}$ becomes cyan and we can colour $u$ in red.

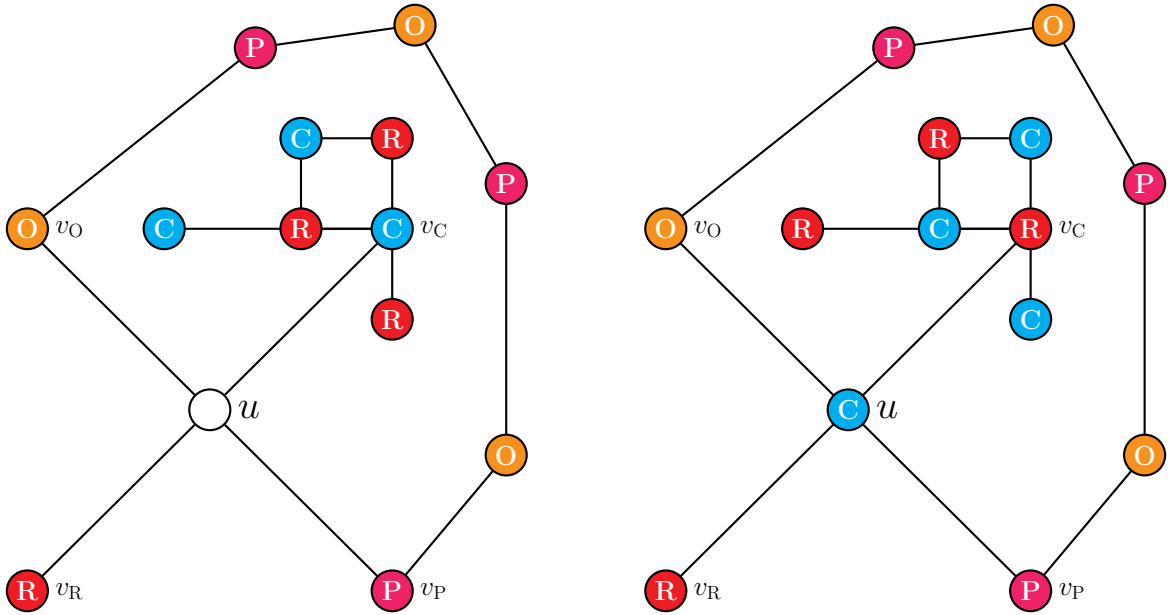
\begin{figure}[H]
\centering
\begin{tikzpicture}[-, scale=0.6, transform shape, thick,place/.style={draw, circle,thick,
inner sep=0pt,minimum size=9mm}]

\begin{scope}
\node (u) at (0,0) [place, label=right:\huge $u$] {};
\node (vC) at (4,4) [place, fill=Cyan, label=right:\Large $v_\mathrm{C}$] {\color{white}\Large \textbf{C}};
\node (vP) at (4,-4) [place, fill=WildStrawberry, label=right:\Large $v_\mathrm{P}$] {\color{white}\Large \textbf{P}};
\node (vR) at (-4,-4) [place, fill=Red, label=right:\Large $v_\mathrm{R}$] {\color{white}\Large \textbf{R}};
\node (vO) at (-4,4) [place, fill=BurntOrange, label=right:\Large $v_\mathrm{O}$] {\color{white}\Large \textbf{O}};

\node (O2) at (1,8) [place, fill=WildStrawberry] {\color{white}\Large \textbf{P}};
\node (O3) at (4.5,8.5) [place,fill=BurntOrange] {\color{white}\Large \textbf{O}};
\node (O4) at (6.5,5) [place, fill=WildStrawberry] {\color{white}\Large \textbf{P}};
\node (O5) at (6.5,-1) [place, fill=BurntOrange] {\color{white}\Large \textbf{O}};

\node (R1) at (2,4) [place, fill=Red] {\color{white}\Large \textbf{R}};
\node (R2) at (2,6) [place, fill=Cyan] {\color{white}\Large \textbf{C}};
\node (R3) at (4,6) [place, fill=Red] {\color{white}\Large \textbf{R}};
\node (R4) at (4,2) [place, fill=Red] {\color{white}\Large \textbf{R}};
\node (R5) at (-1,4) [place, fill=Cyan] {\color{white}\Large \textbf{C}};

\draw [-] (vC) edge (R1);
\draw [-] (R2) edge (R1);
\draw [-] (vC) edge (R4);
\draw [-] (R5) edge (R1);
\draw [-] (R2) edge (R3);
\draw [-] (vC) edge (R1);
\draw [-] (R3) edge (vC);

\draw [-] (u) edge (vC);
\draw [-] (u) edge (vP);
\draw [-] (u) edge (vR);
\draw [-] (u) edge (vO);
\draw [-] (vO) edge (O2);
\draw [-] (O2) edge (O3);
\draw [-] (O4) edge (O3);
\draw [-] (O4) edge (O5);
\draw [-] (vP) edge (O5);

\end{scope}

\begin{scope}[xshift=14cm]
\node (u) at (0,0) [place, fill=Cyan, label=right:\huge $u$] {\color{white}\Large \textbf{C}};
\node (vC) at (4,4) [place, fill=Red, label=right:\Large $v_\mathrm{C}$] {\color{white}\Large \textbf{R}};
\node (vP) at (4,-4) [place, fill=WildStrawberry, label=right:\Large $v_\mathrm{P}$] {\color{white}\Large \textbf{P}};
\node (vR) at (-4,-4) [place, fill=Red, label=right:\Large $v_\mathrm{R}$] {\color{white}\Large \textbf{R}};
\node (vO) at (-4,4) [place, fill=BurntOrange, label=right:\Large $v_\mathrm{O}$] {\color{white}\Large \textbf{O}};

\node (O2) at (1,8) [place, fill=WildStrawberry] {\color{white}\Large \textbf{P}};
\node (O3) at (4.5,8.5) [place,fill=BurntOrange] {\color{white}\Large \textbf{O}};
\node (O4) at (6.5,5) [place, fill=WildStrawberry] {\color{white}\Large \textbf{P}};
\node (O5) at (6.5,-1) [place, fill=BurntOrange] {\color{white}\Large \textbf{O}};

\node (R1) at (2,4) [place, fill=Cyan] {\color{white}\Large \textbf{C}};
\node (R2) at (2,6) [place, fill=Red] {\color{white}\Large \textbf{R}};
\node (R3) at (4,6) [place, fill=Cyan] {\color{white}\Large \textbf{C}};
\node (R4) at (4,2) [place, fill=Cyan] {\color{white}\Large \textbf{C}};
\node (R5) at (-1,4) [place, fill=Red] {\color{white}\Large \textbf{R}};

\draw [-] (vC) edge (R1);
\draw [-] (R2) edge (R1);
\draw [-] (vC) edge (R4);
\draw [-] (R5) edge (R1);
\draw [-] (R2) edge (R3);
\draw [-] (vC) edge (R1);
\draw [-] (R3) edge (vC);

\draw [-] (u) edge (vC);
\draw [-] (u) edge (vP);
\draw [-] (u) edge (vR);
\draw [-] (u) edge (vO);
\draw [-] (vO) edge (O2);
\draw [-] (O2) edge (O3);
\draw [-] (O4) edge (O3);
\draw [-] (O4) edge (O5);
\draw [-] (vP) edge (O5);

\end{scope}

\end{tikzpicture}
\caption{Illustration of the Kempe interchange for the case $d=4$ in Kempe's proof. The PO-chain acts as a barrier: it guarantees that the CR-chain containing $v_\mathrm{C}$ cannot contain $v_\mathrm{R}$. On the right we have performed a CR-interchange so that the colour cyan is free for $u$.}
\label{fig-degree-4-Kempe}
\end{figure}

\paragraph{Third case: $d = 5$.\\}
This last case is the most difficult and uses the ideas of the previous case. As usual, let $u$ be a vertex of degree 5. We four colour the graph $G_0 \setminus \{u\}$. If the neighbours of $u$ are coloured with fewer than four different colours, we use a spare one for $u$. If this is not the case, two vertices have the same colour, say pink, and each of the other three has a colour from the remaining three ones. There are two cases: either the two pink vertices are adjacent when we list the neighbours of $u$ in clockwise order, or they are not. The first case cannot occur since $G_0$ is triangulated: two consecutive neighbours of $u$ are connected themselves, hence they cannot receive the same colour in $G$ nor in $G \setminus u$.

So we deal with the case when the two pink vertices are not adjacent. We may assume that the vertices, listed in clockwise directions are (their names reflect their colours): $v_\mathrm{P}$,  $v_\mathrm{O}$, $v'_\mathrm{P}$, $v_\mathrm{R},$ and $v_\mathrm{C}$. We consider the OC-chain $G_\mathrm{OC}$ which contains $v_\mathrm{O}$. If it does not contain $v_\mathrm{C}$ then for the same reason as in the case $d=4$ we are done. So we assume $G_\mathrm{OC}$ contains $v_\mathrm{C}$. We then consider the OR-chain containing $v_\mathrm{O}$. We again assume that we are in the worst case: $v_\mathrm{R} \in G_\mathrm{OR}$. We now consider two new Kempe chains: the PR-chain $G_\mathrm{PR}$ which contains $v_\mathrm{P}$ and the PC-chain $G_\mathrm{PC}$ which contains $v'_\mathrm{P}$. Like before, $G_\mathrm{OR}$ and $G_\mathrm{OC}$ act as barriers. Therefore, using the same arguments we conclude that both $v_\mathrm{C} \notin G_\mathrm{PC}$ and $v_\mathrm{R} \notin G_\mathrm{PR}$. We can then perform a Kempe interchange on $G_\mathrm{PR}$ and one on $G_\mathrm{PC}$ such that $v_\mathrm{P}$ becomes red and $v'_\mathrm{P}$ becomes cyan. We finish by colouring $u$ with pink.\footnote{QED}

\begin{figure}[h]
\centering
\begin{tikzpicture}[-, scale=0.6, transform shape, thick,place/.style={draw, circle,thick,
inner sep=0pt,minimum size=9mm}]

\begin{scope}

\node (u) at (0,0) [place, label=right:\huge $u$] {};
\node (vR) at (0,-5) [place, fill=Red, label=right:\Large $v_\mathrm{R}$] {\color{white}\Large \textbf{R}};
\node (vC) at (-5,-2) [place, fill=Cyan, label=below:\Large $v_\mathrm{C}$] {\color{white}\Large \textbf{C}};
\node (vP) at (-3,4) [place, fill=WildStrawberry, label=right:\vspace*{-2cm}\Large $v_\mathrm{P}$] {\color{white}\Large \textbf{P}};
\node (vO) at (3,4) [place, fill=BurntOrange, label=right:\Large $v_\mathrm{O}$] {\color{white}\Large \textbf{O}};
\node (vP2) at (5,-2) [place, fill=WildStrawberry, label=right:\Large $v'_\mathrm{P}$] {\color{white}\Large \textbf{P}};

\node (O2) at (3,9) [place, fill=BurntOrange] {\color{white}\Large \textbf{O}};
\node (R2) at (1,6.5) [place, fill=Red, label=left:\Large $x_\mathrm{R}$] {\color{white}\Large \textbf{R}};
\node (C2) at (5,6.5) [place, fill=Cyan, label=right:\Large $x_\mathrm{C}$] {\color{white}\Large \textbf{C}};

\draw [-] (u) edge (vC);
\draw [-] (u) edge (vP);
\draw [-] (u) edge (vR);
\draw [-] (u) edge (vO);
\draw [-] (u) edge (vP2);

\draw [-] (vO) edge (R2);
\draw [-] (vO) edge (C2);
\draw [-] (R2) edge (C2);
\draw [-] (R2) edge (O2);
\draw [-] (O2) edge (C2);
\draw [-] (vP2) edge (C2);
\draw [-] (vP) edge (R2);

\end{scope}

\begin{scope}[xshift=14cm]
\node (u) at (0,0) [place, label=right:\huge $u$] {};
\node (vR) at (0,-5) [place, fill=Red, label=right:\Large $v_\mathrm{R}$] {\color{white}\Large \textbf{R}};
\node (vC) at (-5,-2) [place, fill=Cyan, label=below:\Large $v_\mathrm{C}$] {\color{white}\Large \textbf{C}};
\node (vP) at (-3,4) [place, fill=Red, label=right:\Large $v_\mathrm{P}$] {\color{white}\Large \textbf{R}};
\node (vO) at (3,4) [place, fill=BurntOrange, label=right:\Large $v_\mathrm{O}$] {\color{white}\Large \textbf{O}};
\node (vP2) at (5,-2) [place, fill=Cyan, label=right:\Large $v'_\mathrm{P}$] {\color{white}\Large \textbf{C}};

\node (O2) at (3,9) [place, fill=BurntOrange] {\color{white}\Large \textbf{O}};
\node (R2) at (1,6.5) [place, fill=WildStrawberry, label=left:\Large $x_\mathrm{R}$] {\color{white}\Large \textbf{P}};
\node (C2) at (5,6.5) [place, fill=WildStrawberry, label=right:\Large $x_\mathrm{C}$] {\color{white}\Large \textbf{P}};

\draw [-] (u) edge (vC);
\draw [-] (u) edge (vP);
\draw [-] (u) edge (vR);
\draw [-] (u) edge (vO);
\draw [-] (u) edge (vP2);

\draw [-] (vO) edge (R2);
\draw [-] (vO) edge (C2);
\draw [-] (R2) edge (C2);
\draw [-] (R2) edge (O2);
\draw [-] (O2) edge (C2);
\draw [-] (vP2) edge (C2);
\draw [-] (vP) edge (R2);

\end{scope}
\end{tikzpicture}
\caption{Illustration of the flaw in Kempe's proof for the case $d=5$. When performing the two interchanged mentioned in the proof, $x_\mathrm{R}$ and $x_\mathrm{C}$ both receive the colour pink, hence the colouring is not proper.}
\label{fig-kempe-fail}
\end{figure}
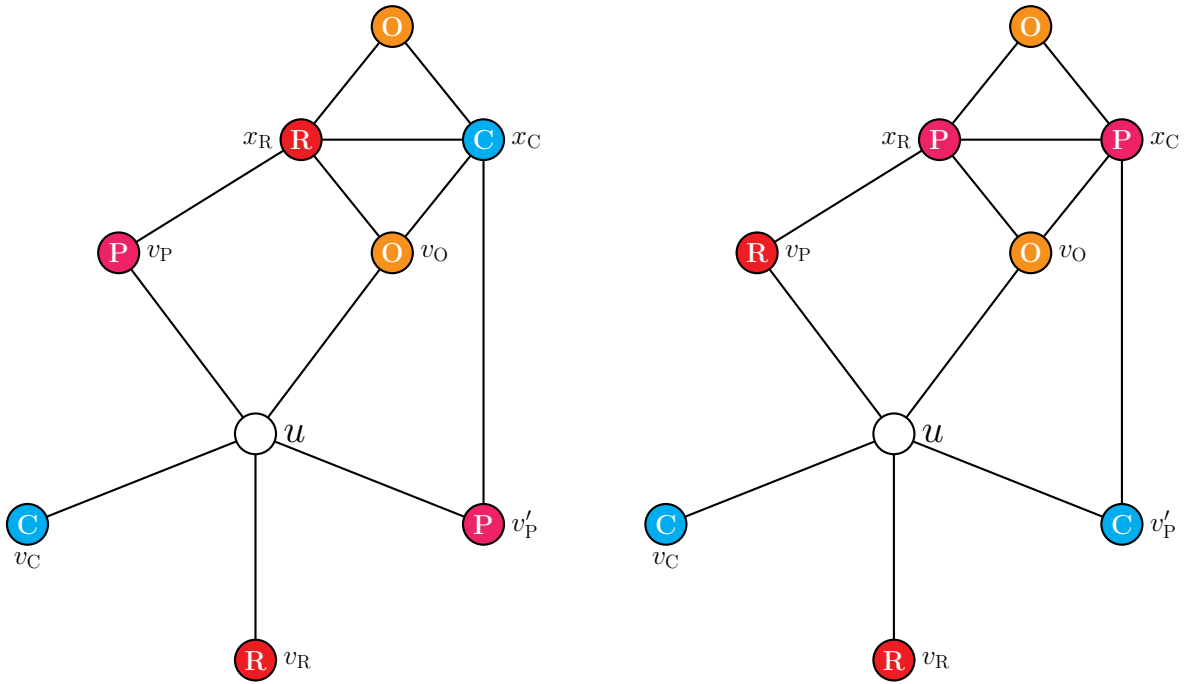

\paragraph{A fourth case for the prosecution.} (see \Cref{fig-kempe-fail})\\
There is subtle flaw in this proof, which you may have overlooked. Indeed, $v_\mathrm{O}$ may have one neighbour $x_\mathrm{R}$ coloured in red and $x_\mathrm{C}$ coloured in cyan. Let us assume that $x_\mathrm{R}$ and $x_\mathrm{C}$ are connected. It is moreover possible that $x_\mathrm{C}$ belongs to $G_\mathrm{PC}$ and that $x_\mathrm{R}$ belongs to $G_\mathrm{PR}$. In this case, after performing the Kempe interchanges in $G_\mathrm{PC}$ and in $G_\mathrm{PR}$, a problem arises: $x_\mathrm{R}$ and $x_\mathrm{C}$ are both pink, which makes the colouring invalid.

In his PhD manuscript, in 1921, Errera\footnote{who bears the same first name as Kempe} found a graph with 17 vertices and 45 edges, known as the Errera Graph, on which Kempe's proof fails. Therefore, the proof cannot be fixed.

However, Kempe's proof can be used to show that any planar graph is five colourable as Heawood did in \cite{heawood}. Indeed, the only non-trivial case in the proof of this result is when $d=5$. We can use the arguments of the case $d=4$ of Kempe's proof: taking two disjoint Kempe chains and using one as a barrier.

\label{section-kempe}

\subsection{Forbidden subgraphs}
To reduce the number of possible counter-examples to explore, one way is to find a list of subgraphs which cannot be contained in any triangulated minimal counter-example: the forbidden configurations. This way, if at some point our current partial graph contains any graph in the list of forbidden configurations, we may discard it because any final graph we would obtain from this one would contain the forbidden configuration, hence could not be a counter-example. These forbidden configurations can be obtained thanks to the properties we imposed on our counter-example $G_0$: triangulation and minimality. We begin by giving a property shared by minimal counterexamples. We then define what a configuration is, and afterwards explain how to find some of the forbidden configurations.

The idea behind it is, assuming the correctness of the theorem, that if we list enough configurations then no plane graphs forbidding them all can exist. \Cref{section-discharging} describes a tool to realise "quickly" that indeed no plane graphs can exclude all the given forbidden configurations.

We already proved something with Kempe's false proof.

\begin{lemma}
The minimum degree of a minimal counter-example is five.
\label{min-degree-five}
\end{lemma}

\begin{proof}
The proof is simple. Let $G$ be a minimal counter-example. Let us assume that $G$ has a vertex $u$ of degree less than 5. By minimality of $G$, $G\setminus u$ is four colourable. We showed in the correct part of Kempe's false proof that there exist a 4-colouring of $G\setminus u$ which can be extended to a colouring of $G$. Therefore, $G$ is not a counter-example, which concludes the proof.
\end{proof}

\begin{deff}
A \textbf{separating short circuit} $C$ of a plane graph $G$ is a cycle of size at most five such that: if $C$ is of length 3 or 4 then each of the two open\footnote{excluding $C$} regions bounded by $C$ contains at least one vertex, and if $C$ is of length 5 then both open regions contains at least 2 vertices.
\end{deff}

\begin{lemma}[\cite{birkhoff}]
Any minimal counter-example contains no short cycles.
\end{lemma}
We prove this lemma in the cases of cycles of lengths 3 and 4 to illustrate a bit the concept of reducibility. We leave aside the case with a short cycle of length 5, which is longer to prove and does not bring much more understanding.

\begin{proof}[Partial proof]
Let us assume that $G_0$ is a minimal counter-example to the theorem and that $C$ is a separating short circuit of $G_0$ of length 3 or 4. We define $G_\textit{in}$ and $G_\textit{out}$ be the two closed regions bounded by $C$. Note that they both contain $C$.

Let us assume that $C$ is a triangle with vertices $a,b$ and $c$. Since $\Gin$ and $\Gout$ are smaller than $G_0$, each of them admits a four-colouring, and the vertices of $C$ must receive different colours since they are pairwise connected. Up to renaming the colours, we may assume that in both colourings $a$ receives colour 1, $b$ colour 2 and $c$ colour 3. The two colourings agree on the vertices of $C$ so that they can be combined to form a four colouring of $G_0$, which is a contradiction.

    We now assume that $C$ has four vertices: $a,b,c$ and $d$. We define in the same way $\Gin$ and $\Gout$. For the same reason as before, each of them is four colourable. Up to renaming the colours, each colouring must be $(1,2,3,4)$, $(1,2,1,3)$, $(1,2,1,2)$ or $(1,2,3,2)$: either all vertices receive different colours, only two opposite vertices receive the same colour, or each pair of opposite vertices receive the same colour. If both $\Gin$ and $\Gout$ admit a colouring of the shape $(1,2,3,4)$ then we are done: each of this colouring can be extended to $G_0$ as for the previous case.

    We then assume that $\Gout$ does not admit $(1,2,3,4)$ as a proper colouring. We will show that it admits both $(1,2,1,3)$ and $(1,2,3,2)$ as proper colourings, and that $\Gin$ admits one of the two. First, we prove that $\Gout$ must admit colourings of the shape $(1,2,1,3)$ and $(1,2,3,2)$. Indeed,  the graph $\Gout$ to which we add the edge $\{a,c\}$ is still smaller than $G_0$, hence it admits a proper 4-colouring giving $a$ and $c$ different colours. This colouring is a proper colouring for $\Gout$. The same argument applies if we instead add the edge $\{b,d\}$. Now we know that $\Gout$ admits both $(1,2,1,3)$ and $(1,2,3,2)$ as proper colourings. If $\Gin$ admits one of them, we are done. If it is not the case then $\Gin$ only admits colourings of the shape $(1,2,3,4)$. We can then look at the $13$-chain containing $c$: if it does not contain $a$ we may perform a Kempe interchange and obtain $(1,2,1,4)$. Up to renaming the colours, we may assume it is the desired $(1,2,1,3)$. Otherwise, as in the case $d=4$ of the Kempe proof
in \Cref{section-kempe} we know that the $24$-chain containing $d$ does not contain $c$, hence we can obtain $(1,2,3,2)$, which concludes the proof for a short circuit of length 4.
\end{proof}

We now introduce the notion of configuration used by the correct proofs of the four-colour theorem. It enables us to define "partial" graphs, which will be shown to be excluded from any minimal counter-example. We recall that an almost-triangulated plane graph is an embedding in which every face is a triangle, except for at most one.

\begin{deff}
A \textbf{configuration} is a couple $C = (H, \gamma)$ where $H = (V,E)$ is an almost-triangulated plane graph and $\gamma : V \rightarrow \NN$. It verifies $\gamma(v) = \dd^\circ(v)$ except for the vertices of at most one face, which has to be the non-triangular face if any. This face is called the outer face, and its vertices verify $\gamma(v) > \dd^\circ(v)$.
\label{def-configuration}
\end{deff}

\begin{figure}[h]
\centering

\includegraphics[scale=0.4]{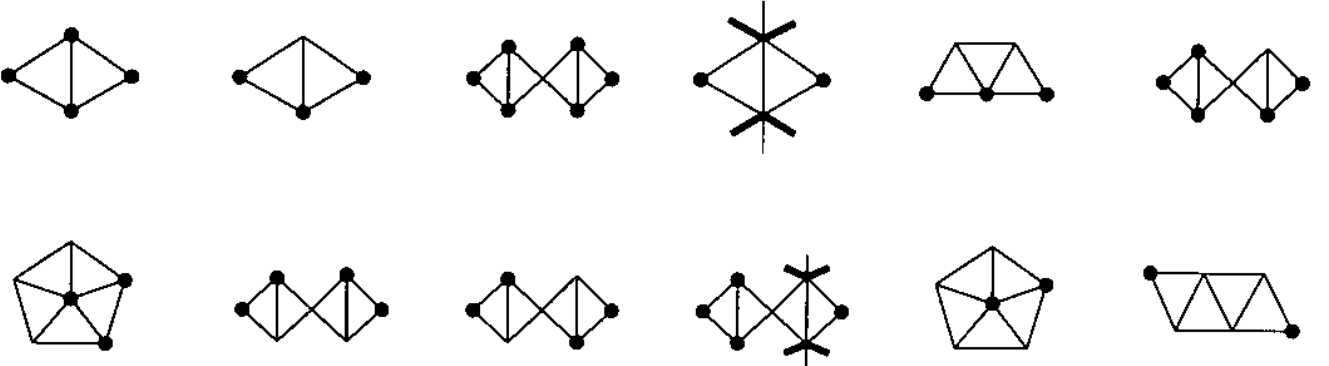}
\caption{An example of reducible configurations found by Robertson et al. The vertices represented by a black circle have a $\gamma$-value of 5, and the others a $\gamma$-value of 6.\\
    The figure comes from~\cite{4-col-paper}.}
\label{fig-reducible}
\end{figure}

We can see an example of configurations in \Cref{fig-reducible}. Each configuration is drawn such that its contour is the outer face. We can see that, in the second line, two configurations have a vertex which does not belong to the outer face. In both cases it has a $\gamma$ value of five, and its degree in the configuration is indeed five. The other vertices have a degree in the configurations at most equal to their $\gamma$ value.

\begin{deff}
A configuration $C = (H, \gamma)$ \textbf{appears} or \textbf{is contained} in a triangulated graph $G$ when:
\begin{itemize}[noitemsep, topsep=0pt]
\item $H$ is an induced subgraph of $G$;
\item except for the outer face, any face of $H$ is mapped to a face in $G$;
\item if $u_H \in H$ is mapped to $u_G \in G$ then $\dd^\circ(u_G) = \gamma(u_H)$: $\gamma$ represents the degrees of the vertices once mapped in a graph.
\end{itemize}
\end{deff}

\begin{deff}
When a configuration $C$ cannot appear in any minimal triangulated counter-example, we say that $C$ is a \textbf{forbidden configuration}, or that $C$ is \textbf{reducible}.
\end{deff}

\Cref{fig-reducible} gives an example of twelve reducible configurations from \cite{4-col-paper}. In that paper, great care is brought to the list of reducible configurations. Their configurations have stronger structural properties than in \Cref{def-configuration}. For instance, if removing a vertex splits the graph into several components then there are at most two of them and this vertex has exactly two other neighbours than the one appearing in the configurations (its $\gamma$-value is 2 more than its degree in the configuration). If a vertex is not incident to the region which is not a triangle, then its gamma value is its degree in the configuration. Some other properties are true for their list of reducible configurations, one of which saying that the sum of the $\gamma(v)-\dd^{\circ}(v)-1$ over some specific sets of vertices is at least two.

One way to show that a configuration $G = (H, \gamma)$ is reducible is to show that if a graph $G_0$ contains $G$ and is not four colourable, then there exists a graph $G'_0$, still not colourable, but with fewer vertices than $G_0$. For instance, $G'_0$ could be constructed from $G_0$ by removing $H$ and replacing it with a smaller subgraph. We describe a bit further how to detect this case, as well as another way to find that a configuration is reducible.\\

We consider a configuration $G$, almost-triangulated, and denote by $L$ the cycle defining its outer face. We call $L$ the \emph{crown} of the configuration. We generate the list $X_\mathrm{true}$ of all 4-colourings of $L$ which induce a proper colouring on $G$, i.e. all colourings pf $L$ which can be extended into a proper colouring of $G$. If $G$ appears in some counter-example, then every proper 4-colouring of $G$ is guaranteed to be non-extendible to the whole graph it would appear in, for otherwise the counter-example would have a 4-colouring. We study the restrictions of the colourings of $G$ to its crown because if $G$ appears in a counter-example, its crown is its "interface" with the other vertices. We now try to find a smaller configuration $G'$ which has the same crown bounding its outer face, with some property on its colourings. If, denoting the set of the restrictions of its proper 4-colourings to the crown $L$ by $X'_\mathrm{true}$, we find that $X'_\mathrm{true} \subseteq X_\mathrm{true}$, we deduce that $G$ is reducible. Indeed, let $G_0$ be a counter-example containing $G$. Let us replace $G$ by $G'$ in $G_0$ to obtain $G'_0$. $G'_0$ is smaller than $G_0$ and admits no 4-colourings. Indeed, any 4-colouring of $G'_0$ would induce a 4-colouring of $G_0$: this 4-colouring restricted to the crown $L$ would also be proper for $G$, hence a proper colouring for $G_0$. We would have achieved our goal: any time $G$ would appear in a counter-example, it could be replaced by $G'$, hence this would show its reducibility.

However, having to enumerate all smaller configurations with the same crown (and its 4-colourings inducing a proper 4-colouring of $G$!) is very costly in time. Also, this technique was not efficient enough for the proof of Robertson et al.: it misses a lot of reducible configurations. Indeed, when examining a configuration, it leaves aside the completions of the configuration into a plane graph. It may be possible to show that some configuration forces any graph extending it to be four colourable. There are indeed others techniques to show that a configuration is reducible: in fact the researchers who worked on the four-colour conjecture categorised the reducibility property into several classes. One of them, used in the proof of Robertson et al. is the \emph{D-reducibility}. It once again uses an argument based on Kempe chains. Let again $G$ be a configurations which we want to show is reducible. Let us call $X_\mathrm{false}$ the list of the 4-colourings of the crown $L$ of $G$ which do not induce a proper 4-colouring of a configuration. Let us assume that $G$ appears in $G_0$. We then know that $G' = (G_0 \setminus G) \cup L$ is four colourable, because it is smaller than $G_0$. Let us call $X'_\mathrm{true}$ the list of the 4-colourings of $L$ which induce a proper 4-colouring of $G'$. $G$ appears in $G_0$ and $G_0$ is not four colourable, therefore like previously $X'_\mathrm{true} \subseteq X_\mathrm{false}$. Our goal is to show that in fact $X'_\mathrm{true} = \varnothing$, which is a contradiction: $G'$, being smaller than $G_0$, is four colourable. 
To do so, we try to find a maximal \textit{consistent} set $X'_\mathrm{C}$ of allowed 4-colourings for $G'$. We begin by setting $X'_\mathrm{C} = X_\mathrm{false}$. Then, for each colouring in $X'_\mathrm{C}$, we first look for other colourings of the crown which should be allowed for $G'$ based on Kempe interchange arguments. If, for some colouring $c \in X'_\mathrm{C}$, one such colouring $c_1$ is not allowed, i.e. not in  $X'_\mathrm{C}$ then this implies that $c$ is in fact not allowed either. This is what is called the consistency of a set: if a colouring is allowed, but not some colourings obtained by Kempe interchanges, then the set of allowed colourings is inconsistent. Removing $c$ may in turn lead to some other removals of possibly valid colourings. We iterate this process until the set $X'_\mathrm{C}$ stabilises. If at the end $X'_\mathrm{C} = \varnothing$, this means that the configuration is D-reducible. We do not dive into the details of how to deduce which colourings should also be valid, given that some $c$ belongs to the current $X'_\mathrm{C}$. The algorithm 
to compute these is quite complex. In our code, the search for a maximal consistent set of colourings included in some $X_\mathrm{false}$ is done by the function \texttt{get\_maximal\_consistent\_colouring\_subset}.

We may notice that we can combine this method with the previous one, when we enumerated smaller subgraphs. Indeed, we looked for some smaller configuration $G'$ such that $X'_\mathrm{true} \subseteq X_\mathrm{true}$, or, equivalently, such that $X_\mathrm{false} \subseteq X'_\mathrm{false}$. This condition is more frequently met if we reduce the set $X_\mathrm{false}$ to a maximal consistent subset. Robertson et al. used this D-reducibility notion, but they also used another type of reducibility (the C-reducibility) which we do not define here.

\subsection{The discharging method}
\label{section-discharging}
We detail here more the idea of discharging, mentioned in \Cref{5-min-proof}\footnote{We hope we were not too efficient so that there are still things to learn or understand here.}. It helps us realise that the class of minimal counter-examples is empty. Discharging requires us to work on \textit{weighted} graphs.

\begin{deff}
A \textbf{weighted} graph is a couple $(G,w)$ where $G=(V,E)$ is a graph and $w:V \mapsto \ZZ$ is the weight function.
\end{deff}

The weight function assigns a weight to each vertex of the graph. The discharging method will consist in locally moving parts of the weights between certain vertices and some of their neighbours whenever certain conditions are met. Note that here we only put weights on vertices, but it is possible to also put weights on edges and faces.

\begin{deff}
A discharging rule is a quadruplet ($F$, $u$, $v$, $q$): $F$ is a configuration, $u$ (the \textbf{source}) and $v$ (the \textbf{sink}) are vertices of $F$, and $q$ is the weight of the rule.
\end{deff}

\begin{deff}
Applying rule ($F$, $u$, $v$, $q$) to a weighted graph $(G,w)$ results in the weighted graph $(G,w')$ where for all $u \in V$, $w'(u) = w(u) + (a-b)q$ if, in $G$, $F$ appears in $a$ different times with $u$ as a sink and $b$ different times with $u$ as a source.
\end{deff}

Applying a rule consists in, each time $F$ appears in $G$ such that $u \in F$ is matched with $u' \in G$ and $v \in F$ is matched with $v' \in G$, transferring a weight of $q$ from $u'$ to $v'$.

\begin{rk}
Note that, like for weighting a graph, we put weights on the vertices, but other uses of the discharging method can also weight edges and faces. In this case, the discharging rules would also transfer weights between vertices, edges and faces.
\end{rk}

\begin{fact}
\label{fact1-discharging}
Applying any discharging rule to a graph does not change its total weight.
\end{fact}

\begin{fact}
\label{fact2-discharging}
If a graph has a positive total weight then it has a vertex of positive weight.
\end{fact}

These two facts are trivial. Now, let us assume that the class of minimal counter-examples to the four-colour theorem is not empty. \textbf{Let $\bm{G_0}$ be a minimal counter-example.} We assign the weights in $G_0$ in the following way:
\[ w(v) = 10(6-\dd^{\circ}(v)).\]

\begin{claim}
If $G$ is triangulated, then the sum of the weights of $(G,w)$ is 120.
\end{claim}

\begin{proof}
To show this claim, we use again Euler's formula for planar graphs: $n-m+f=2$ where $n,m$ and $f$ are respectively the number of vertices, edges, and faces of our graph.

Since our graph is triangulated, we have $3f = 2m$: each face has three edges, but each edge is shared by two faces. Now:
\begin{align}
\sum_{v \in V}{w(v)} = \sum_{v \in V}{10(6-\dd^{\circ}(v))} = 10(\sum_{v \in V}{6} - \sum_{v \in V}{\dd^{\circ}(v)}) &= 10(6n - 2m)\\ &= 10(6m-6f+12-2m)\\ &= 10(12+4m-6f)\\ &= 120.
\end{align}
(1.2) comes from Euler's formula and (1.4) from the relations between $f$ and $m$ we mentioned above.
\end{proof}

\begin{figure}[H]
\centering

\includegraphics[scale=0.9]{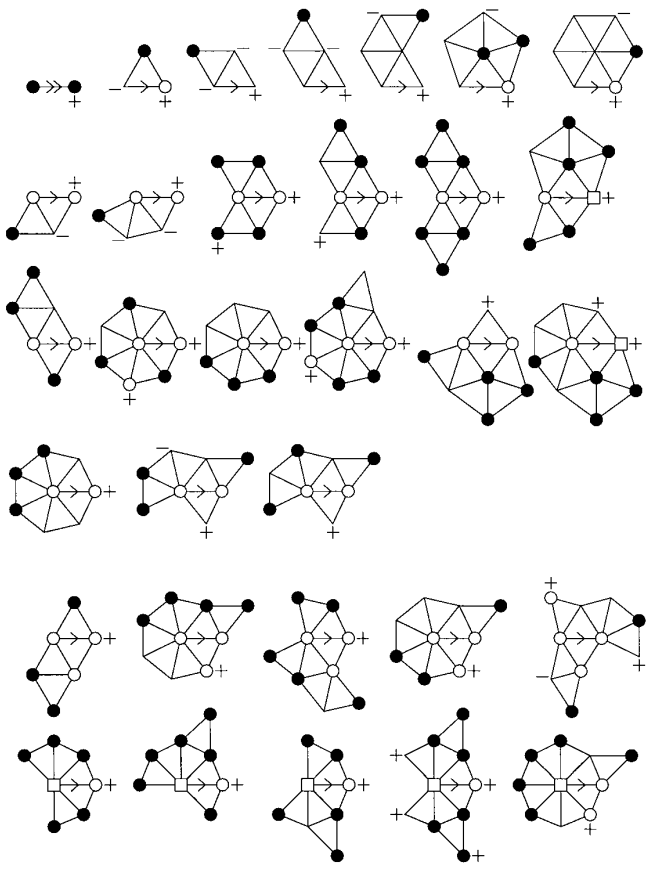}
\caption{The 32 rules used in the proof by Robertson et al. The degrees are given by symbols: black circles, points, white circles, white squares and white triangles respectively represent integers from 5 to 9.
A '-' (resp. '+') sign as exponent means "degree at most" (resp. "degree at least"). All rules have weight 1 except for the first one. The arc in the configuration originates at the source of the rule and ends at its sink.\\
    The figure comes from~\cite{4-col-paper}.}
\label{fig-rules}
\end{figure}

Since $G_0$ has a total weight of 120, and from Facts 1.2 and 1.3, 
\textbf{there is a vertex which has a positive weight after \uline{applying all the discharging rules once}}. \textbf{We name it $\bm{v_0}$}. Our goal will be to show that it cannot exist: any neighbourhood for $v_0$ will necessarily contain reducible configurations. We will enumerate the possible neighbourhood with the help of a computer. This enumeration is detailed in \Cref{discharging-algo}. We show in \Cref{fig-rules} the 32 discharging rules of the proof of Robertson et al.

\section{Contribution}
We began by proving again the four-colour theorem, using the data (the lists of forbidden subgraphs and discharging rules) from the proof of Robertson et al.~\cite{4-col-paper}. We explain here the different parts of the program and some of the algorithms we implemented. A modification of this program is released with a Python interface. It can be included with the Sage 
software, a free mathematical assistant. It enables anyone to try to prove a result using the discharging method: the user has to provide the forbidden subgraphs and discharging rules for their problem, as well as a weight function. Then our program tries to prove that no plane graphs avoiding all the forbidden subgraphs exist. We provide some flexibility to the user since they can provide a python function which, given a partial current graph, decide whether or not it should be discarded: some properties may be better coded than expressed in terms of forbidden subgraphs.

\subsection{The program}
\label{discharging-algo}
We describe here several parts of the program, like the scheme of the branch-and-bound enumeration of the possible neighbourhoods for $v_0$. We recall that $G_0$ is a minimal counter-example we assume the existence, and $v_0$ is one of its vertex which has positive weight after applying all the rules once.

\paragraph{The enumeration.\\}
We start the search for $G_0$ with $v_0$: it is the first vertex we build. We then explore the possible neighbourhoods for $v_0$. In the exploration, we can do several things: add vertices and edges, or choosing the final degree of a vertex (at first, they have some degree interval). By doing so, we will indeed explore every possible neighbourhood for $v_0$, until building $G_0$ or showing that $G_0$ does not exist\footnote{The course of History has taught us that it is this option which happens.}.

Let us describe the algorithm with more details. We first start with the vertex $v_0$, or in fact a triangle containing $v_0$ since $G_0$ is triangulated. The degree of each vertex of the triangle has degree between 5 and 12, except for $v_0$ which has degree between 6 and 11. These restrictions were shown by Robertson et al. (we already showed each degree must be at least 5).

Let us assume that we are at some step with a partial neighbourhood of $v_0$ we call $G$. We have two options. We may choose a vertex $u$ whose current degree lies between $k_1$ and $k_2$ with $k_1 < k_2$. We subdivide its degree interval into two smaller ones: $G'$ derived from $G$ by fixing the degree of $u$ to be $k_1$ and $G''$ in which the degree of $u$ is between $k_1+1$ and $k_2$. We will continue the exploration first with $G'$, then with $G''$. The second option is to create a new edge from a vertex $u$ whose degree interval is not $\llbracket 0; 0 \rrbracket$. This edge may create a new vertex, or its other endpoint may be an existing vertex. In both cases, we create every free triangle we can: if a vertex has degree one, we know it will belong to a triangle we can describe.

Sometimes creating free triangles may fail: a triangle cannot be constructed because a vertex cannot accept new neighbours for instance. When this occurs, we may discard our current graph and backtrack. Before calling recursively our exploration function on a graph $G'$, we do some checks. If $G'$ contains a reducible configuration, we may also discard it. Finally, we also apply all the discharging rules which involve $v_0$ and obtain an interval for the final weight of $v_0$. For instance, for the upper bound we apply the rules we are sure apply, and apply also the rules which may apply and contribute to increasing the weight of $v_0$. If the upper bound is non-positive, this is a contradiction\footnote{We chose $v_0$ such that its weight is positive after applying the rules.} and we may also discard $G'$ and backtrack. Besides, we may use the information about some configurations which almost appeared, or some rule which almost applied to choose how to expand our current graph or which degree to refine. Good heuristics for this choice lead to reducing the number of graphs we explore, hence reducing the running time.

\paragraph{Storing plane graphs.\\}
We manipulate partial plane graphs all the time in the program so they must be stored efficiently for our uses. First, before storing them we had to read the reducible configurations and rules from Robertson et al. The way the reducible configurations are encoded is available at \url{http://people.math.gatech.edu/~thomas/FC/ftpinfo.html}. To parse them requires a good comprehension of their proof. Each one is given as the coordinates of an embedding of the configurations. The format of the rules is described in~\cite{rules-col-paper} and, is also not that handy. Our program uses a different format which we find easier and more convenient. Since a plane graph is completely determined by the \textbf{directed} list of the edges of each vertex, we decided to store a graph as the directed adjacency lists of each of its vertices. We chose the counter-clockwise direction in our program. As we mentioned, each half-edge contains the information about the minimum and maximum number of future edges between itself and the next known half-edge of the vertex.

Concerning the storing of the graphs properly speaking, it essentially boils down to a doubly-linked list of "half-edges" for each vertex. We call them half-edges because an edge $\{a,b\}$ is both stored for $a$ and for $b$. Also, any edge defines an angle: the interval giving the number of other edges between this edge and the next known edge of the vertex. This means that the half-edge joining $a$ to $b$ and the one joining $b$ to $a$ are not the same: they define different angles. For easy triangle making and detection we also store, for each half-edge, a link to its other half.

\begin{figure}[h!]

\includegraphics[scale=0.5]{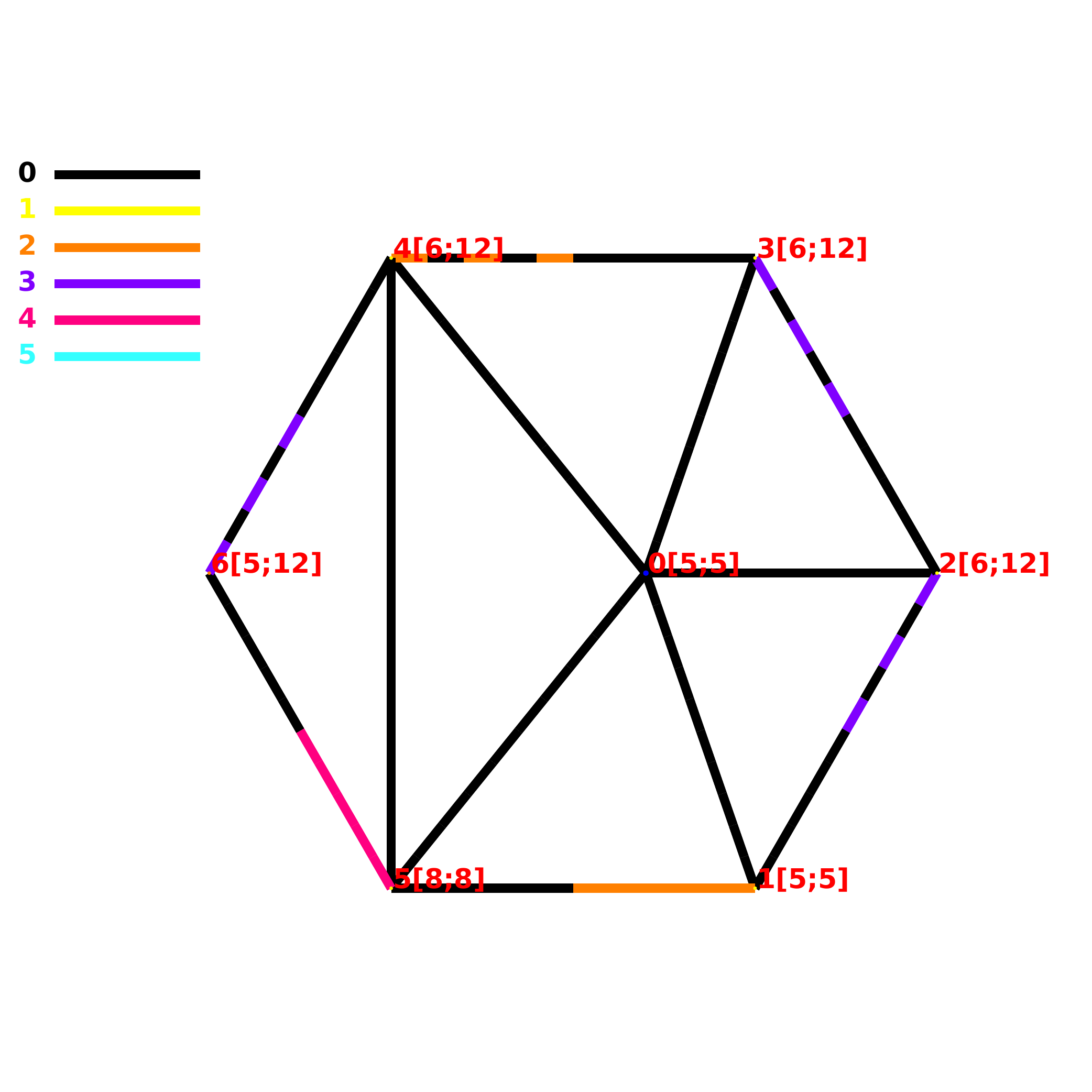}
\centering    
    \vspace*{-1.2cm}
\caption{A screenshot of our program showing a partial graph during the exploration. The colours are indicated on the top left corner. The label of each vertex is of the form "$a[b;c]$": vertex $a$ has degree between $b$ and $c$. Each edge is split in two in the middle: the two different half-edges. The colour of the half-edges indicated the degree interval of the angle until the next edge (in counter-clockwise direction). If a half-edge has two colours, the smallest one is the lower bound on the degree of the angle, the other one is the upper bound.}
\label{screen-prog-discharging}
\end{figure}

\paragraph{Detecting if a configuration appears.\\}
To detect when reducible configurations appear and which rules apply, we need some subgraph detection routine. This is called the induced subgraph isomorphism problem, and is known to be very hard for general graphs. However, here it is on plane graphs that we want to detect some type of isomorphism, which modifies the rules of the problem, and make it polynomial.

Indeed, let us assume that $G$ and $C$ are plane graphs and not configurations: each vertex has a degree instead of a degree interval, and all vertices and edges are known. We want to know if, or even (for the discharging rules) where $C$ appears in $G$.We may choose an edge $e_C$ and try to match it to every edge $e_G$. Now, to verify if one matching of this edge extends to an isomorphism between $C$ and $G$, we may do a graph traversal by walking from faces to faces. We are forced to map the faces which contain $e_C$ to the faces which contain $e_G$. There are two faces, and as soon as one is chosen, then the rest of the isomorphism is forced: there is always one choice (or 0 if $G$ does not contain $C$) for matching every other vertex and edge.

Yet, here we do not test for exact isomorphism because in the rules and reducible configurations, vertices of the crowns do not have every neighbour instantiated. This happens a bit in $C$, and more in $G$. We precise here how a configuration appears in such "partial" graphs $G$ we manipulate here: a configuration appears if whatever the completion of $G$ into any $G_\text{final}$, $C$ appears in $G_\text{final}$. Due to this problem we sometimes cannot conclude: $C$ might for instance be a subgraph but we cannot decide yet if it will be an induced one.

In order to optimise the induced subgraph detection we may first look at the maximum and minimum degree of $C$: each of these two vertices must be assigned a vertex in $G$ with the same degree. After this, we may look at the list of degrees in $C$ and in $G$. By choosing to match first an edge belonging to a vertex of $G$ whose degree appears the least often in the list of degrees of $G$, we may reduce the running time of the check. Indeed, ideally the vertex would only be matched to a single vertex of $G$, if it is alone to verify the degree condition.

\paragraph{A Python/Sage library to use the discharging method.\\}
We decided, since we had programmed it, to let the community use our engine for discharging. We created a Python interface to our program. The user has to provide a file containing the reducible configurations, and a file containing the discharging rules\footnote{for the moment, only with weights on vertices}. They must also provide a function to compute the weights of the vertices. They can provide a python function which has access to a current neighbourhood in the exploration, and decides whether it should be discarded based on some other properties. The program then enumerate the possible neighbourhoods of a vertex which should have a positive weight after applying the rules. It can print the neighbourhoods which could not be discarded (see \Cref{screen-prog-discharging}).\\

Making this program available to Sage meant porting it to Python. Instead of writing the code from scratch in Python, which would in addition make it slower, we decided to interface it with Python. We used the Boost.Python library, which allows us to expose some C++ functions to a python program. It creates a Python module including the specified functions. To do so, we needed to add a bit of code: tell Boost.Python how to translate the functions, i.e. what types they should have. Some problems occurred during this translation, in particular with C++ functions using some C++ features. We sometimes had to define a new C++ function to wrap a problematic one into something with a simpler header (the types of the arguments of the function and its return type). For instance, no functions exposed to Python could contain pointers.

We also said that the user could provide a function to check for properties easier to check with some code than by forbidding graphs. They also need to provide a function to compute the weight of a vertex. This means the C++ code\footnote{called from a python code} should be able to execute a Python function\footnote{If you followed carefully, it means that a python code would execute (some C++ code which executes) some Python code!}. This was harder than running C++ code from Python. We managed to do it with the same Boost.Python library, which provided some "python object" type, which could be anything passed by Python to the C++ program. It also provides functions to specify this object, that is for instance specify its type, and if it is a function. In the latter case, it is possible to run the Python function and retrieve its output. After some time testing how all this was working, we achieved our goal. A good thing is that the modifications involved adding a bit more code but are rather transparent: the exploration function calls the function \texttt{compute\_weights} which can either be a C++ function (as we did ourselves for the four-colour theorem) or a Python function. Our exploration function contains a unique call to the \texttt{compute\_weights} function, which means that the syntax of the calls was not modified and that a C++ function and a Python function can be unified in a single C++ Boost.Python type.

\subsection{A possible more automated approach}
The original goal which motivated us to study and redo the proof of the four-colour theorem was to automatise the proof a bit further. The set of reducible configurations was generated by a computer, so this part is fine. The set of discharging rules, however, was found by hand, by error and trial. This part is somewhat tedious because, as we will see, it is repetitive. It is the part we wanted to improve by making it more automatic, so this section gives ideas on how achieve this.

As we have just written, looking for discharging rules is repetitive. We start with a limited set of discharging rules, and run the program. If it seems not to terminate, it is possible to look at some of the graphs it did not manage to discard, notably (for the four-colour theorem) the ones going further than a neighbourhood of distance two from the first vertex, which was supposed to be non-necessary. Then it is possible to devise some rules, and rerun the program again. Once more, if the program does not terminate in a reasonable amount of time, we can have a look at the graphs which were not discarded. This makes it possible to craft new rules, taking care of not going backwards on the ones which were introduced earlier: a graph which was discarded before should still be discarded after adding a new rule.

The objective was to find a way to generate automatically new discharging rules: the program would be the one to analyse the non-discarded graphs, and find some "optimal" rules. A rule would for instance be optimal if it makes it possible to discard "as many graphs as possible". One obvious way to generate rules would be to enumerate a certain number of "small" configurations and generate rules by defining sources and sinks. Then we could replay the beginning of the run with each rule tested as the new one to be adopted. The one which would lead to discarding the most graphs would be kept.

Another idea could be to try to optimise an existing set of rules: we assume we already have a few rules. We may discard more graphs without adding a new rule by modifying the weight of each rule. A linear program may model our needs: given the list of graphs we have explored so far (the ones we already discarded and some others we would like to discard), we may associate to each one an equation. Let us consider a graph $G$ and $x_{i, G}$ be the net effect of Rule $i$ on our vertex $v_0$ in $G$. If Rule $i$ leads to transferring 4 towards $v_0$ and 7 from $v_0$, then its net effect would be $4-7 = -3$. Now we obtain the condition, for each $G$:
\[w(v_0) + \sum_{i}{x_{i,G} \cdot w_i} \leq 0.\]
The $w_i$'s are proper to the rules, independent from the explored graphs. If solving the linear program shows that some solution exists, it means that with the same set of rules, only by tweaking the coefficients, we may discard more graphs.

We did not investigate much the possible methods to automate the search for good discharging rules. However, finding some automatic ways to obtain them would be a big step forward in the field of automatic proofs for planar graphs. It would make it a lot faster and easier to try the discharging method to attack some other problems on planar graphs.

\resetlinenumber
\chapter{The domination numbers in grid graphs}
\label{domination-chapter}

Societies have always (or for long, at least) seen a group of dominant people emerge and set the rules. Nowadays, it is both publicly said by some people and conveyed by some media that some groups dominate the majority of people. This is refuted or ignored by many in the dominating groups, and little is done to correct the state of things. Such groups may include (very) rich people and countries, men, cisgender\footnote{a person whose (social) gender is the same as the biological sex which was assigned to them at their birth} people, valid\footnote{not disabled} people, white people and so on. Let us state a trivial thing: not all of these people, who enjoy some privileges due to the colour of their skin, their gender or other factors, consciously oppress women or some minorities. Not every person in these groups takes part in the wrongdoings, but many do or take part in the system, which is oppressing. Thus the overall behaviour of these groups is bad. Well, yes, a lot of women contribute to sexism, because society framed them into doing so.  But guess what? Women suffer from this problem which comes from, and is mostly maintained by men: the ones who refuse to see the problem, and those who do nothing to try to fix this state of things. The dominating groups come from all sort of problems such as patriarchy and sexism, queerphobia\footnote{including, but not limited to transphobia and homophobia: (when) did you learn about asexuality and/or stopped thinking this was a disease or a problem?}, disabilism, and a lot of other systemic discriminations. For instance, society leads to the invisibility of disabled people in society: how many times did you watch a sport played by disabled people, either on TV or in real life, or read about it somewhere? Do you realise that in some cities, some disabled people cannot take the underground or the bus, or go to some buildings, because they have mandatory stairs? In Paris for instance, only one line and a dozen stations (out of 303) of the underground are accessible to people in a wheelchair. Patriarchy and sexism also induce a lot of problems: in France women who have a job are four times as likely as men to have a part-time job. They earn overall 18.5\% less than men, 16.3\% less than men when restricting to full-time employees, and they are still paid 12.8\% less than men when considering equivalent positions. They also compose only 23\% of the people in the French parliament. In addition to this, there were 149 people killed in domestic violence in France in 2018: 128 of them were women. 25\% of women between 20 and 65 years declared having suffered some violence in a public place during the last year, and this rate jumps to more than 60\% if we consider women between 20 and 24 years. Among these violences, 1 million women declare having, over the last year, suffered harassment or sexual harassment. During the year 2018 in France, 1905 acts of LGBTQ-phobia\footnote{LGBT stands for Lesbian, Bisexual, Transgender, Queer. The acronym designates the union of people in these non-exclusive sets.} have been recorded, among which 231 involving physical violence. These numbers keep increasing and a poll suggested that only 27\% of the victims of physical violence report it to the police. Two thirds of LGBT people have at some point avoided holding the hands or kissing their companion and 12\% have considered moving to another city to avoid being harassed or assaulted. This means that a majority of LGBT people are denied the right of walking freely in the street without fearing for themselves. \\

However, we study here some forms of domination which makes no one suffer... except maybe the people studying these problems. Indeed, the domination number problem, that is the problem of finding the minimum size of a dominating set\footnote{or to be more formal, the decision problem associated to this optimisation problem} is a NP-complete problem for general graphs. This informally means that if we ask, "Does $G$ admit a dominating set of size less than or equal to $k$?", then it is easy to check that one solution we are given is indeed of size at most $k$ and dominating (the problem is in NP); however we do not know any systematic and "fast" algorithmic way  to prove the "no" answer (the problem is NP-hard: it is as hard as other NP-complete problems).

The basic domination problem consists in selecting a set of vertices in a graph such that any other vertex has a neighbour in that set. As we will see, many variants of this problem exist. These problems can be used to model optimising problems arising in real life, as the Roman domination problem illustrates: it is said to have been used as a model by the Romans to defend their territory. The domination problem can also be used in other contexts, such as some public services: where to put hospitals, fire stations, and other critical places such that every person in a country can benefit from it.

The domination number problem is one of many problems which are hard for general graphs, but are easy to solve for graphs of bounded treewidth. Indeed, Courcelle~\cite{courcelle} showed in 1990 that a particular class of properties, the one being expressible in monadic second-order logic, are decidable in linear time on the class of bounded-treewidth graphs. The treewidth can be understood, intuitively, as a measure of how much a graph "looks like" a tree. The grids are among the simplest graphs which neither have a bounded treewidth nor a bounded cliquewidth (another graph parameter), and for which these kinds of problems are usually difficult to tackle.

The first values (for a number of lines $n \in \{2,3,4\}$) of the domination number in grids were discovered by Jacobson and Kinch in 1983. Then, ten years later Chang and Clark found the ones for 5 and 6 lines. All these results were found without using any computer. Also in 1993, Chang~\cite{chang} conjectured that the domination number for a grid graph of arbitrary size, that is the minimum size of a dominating set, was $\gamma(G_{n,m}) = \left\lceil \frac{(n+2)(m+2)}{5}\right\rceil-4$. He also showed that this was actually an upper bound. Fisher~\cite{fisher}, using computer resources, found the values for $n \leq 19$ and showed that these values were conform to the conjecture. He also found a method to detect and prove the periodicity of the domination number, so as to establish formulas for a fixed number of lines and arbitary number of columns. After a few papers by a few other people, Gonçalves et al.~\cite{rao} finally proved Chang's conjecture in 2011 by showing that his formula was also a lower bound, improving a bound by Guichard~\cite{guichard}.
After 2011, several papers gave formulas for small number of lines for various domination problems. Several generalisations of the domination problem have also been studied in the literature (see for example~\cite{bon}).

The goal of this chapter is to solve the 2-domination and Roman domination problems. We achieve this by giving closed formulas computing the minimum cost of the respective dominating sets for any size of grids as in~\cite{rao-talon}. The formulas we give are simple: they involve multiplications, additions, division and rounding to lower or upper integer. This shows that, like the domination number, the 2-domination and the Roman domination numbers problems are solvable in constant time in grids.\\

This chapter is organised as follows. After giving some definitions in \Cref{section-def-dom}, we will explain in \Cref{section-dom-method} how to solve the 2-domination problem on grid graphs: first on fixed-height grids and then on arbitrary grids, relying on the notion of \emph{loss}. We also use the \textit{Rauzy graphs} to give some complexity information. In \Cref{section-other-problems-dom}, we will study the distance-two domination problem and the total domination problem, but are only able to give formulas for grids of small number of lines. We also give a lower bound for the total loss when both the number of lines and the one of columns are arbitrary. We continue by giving in \Cref{section-conj-dom} some insight on why the method for arbitrary-size grids what we believe the method works in some cases and seems not to in others. We define some properties like the \emph{fixed-height-border-fixing} one which we conjecture explains when the method for arbitrary grids works. In \Cref{section-experimental-dom}, we finally explain some of the optimisations we made and give some implementation details and statistics.

\section{Basic definitions and notations}
\label{section-def-dom}
We define here formally the different types of problems we will study in this chapter. We list them by increasing complexity.

Since we work with two dimensions in this chapter and the following, we will try to be coherent all along. There may be some differences on the order of the indices between these chapters and the associated papers. In what follows, $n$ will always be the number of lines and $m$ the number of columns. When speaking of a grid $G_{n,m}$ or the domination number $\gamma(n,m)$ we first put the number of lines. When using coordinates, we will use the standard order $(x,y)$. The indices will be $j$ when referring to the columns numbers, and $i$ when we refer to the lines indices.

In this thesis, the $x$-values are increasing from left to right, and the $y$-values are when going from top to bottom. In this chapter, the indices and coordinates will always begin at 0 (and not 1), and $(0,0)$ are the coordinates of the top-left cell of a graph or rectangle. 

We denote by $G_{n,m}$ the grid graph with $n$ lines and $m$ columns. In the illustrations to come, the vertices of a grid will be its cells (and not the intersections of the lines).

\begin{deff}
A set $S$ of vertices of $G$ is \textbf{dominating} when any vertex not in $S$ has at least one neighbour in $S$.
\label{def-dominating-set}
\end{deff}

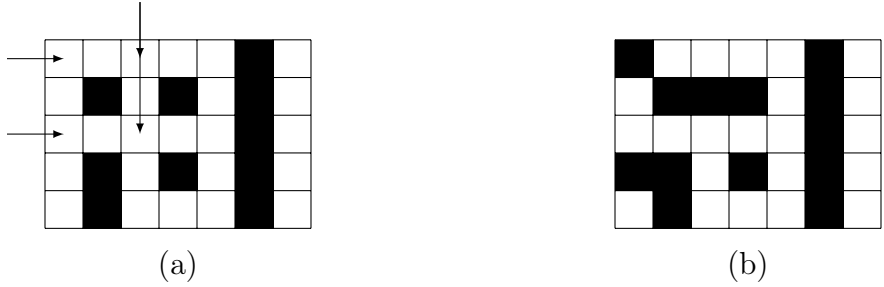
\begin{figure}[h!]
\centering
\begin{tikzpicture}[scale=0.5]
\begin{scope}

\fill (1,0) rectangle (2,2);
\fill (1,3) rectangle (2,4);
\fill (3,3) rectangle (4,4);
\fill (3,1) rectangle (4,2);
\fill (5,0) rectangle (6,5);

\draw (0,0) grid (7,5);

\draw[-latex] (-1,2.5) -- (0.5, 2.5);
\draw[-latex] (-1,4.5) -- (0.5, 4.5);
\draw[-latex] (2.5,6) -- (2.5, 4.5);
\draw[-latex] (2.5,6) -- (2.5, 2.5);

\node at (3.5,-1) {(a)};
\end{scope}

\begin{scope}[xshift=15cm]

\fill (1,0) rectangle (2,2);
\fill (1,3) rectangle (2,4);
\fill (2,3) rectangle (3,4);
\fill (3,3) rectangle (4,4);
\fill (3,1) rectangle (4,2);
\fill (5,0) rectangle (6,5);
\fill (0,4) rectangle (1,5);
\fill (0,1) rectangle (1,2);
\draw (0,0) grid (7,5);

\node at (3.5,-1) {(b)};
\end{scope}

\end{tikzpicture}

\caption{Illustration of a dominating set on $G_{5,7}$:\\
(a) the cells pointed by arrows are not dominated;\\
(b) the set of black cells is dominating.}
\label{fig-dom-intro}
\end{figure}

\begin{deff}
A set $S$ of vertices of $G$ is \textbf{2-dominating} when any vertex not in $S$ has at least two neighbours in $S$.
\end{deff}

\begin{figure}[h!]
\centering
\begin{tikzpicture}[scale=0.5]
\begin{scope}

\fill (1,0) rectangle (2,2);
\fill (1,3) rectangle (2,4);
\fill (3,3) rectangle (4,4);
\fill (3,0) rectangle (4,1);
\fill (4,0) rectangle (5,5);
\fill (6,0) rectangle (7,5);

\draw (0,0) grid (7,5);

\draw[-latex] (-1,2.5) -- (0.5, 2.5);
\draw[-latex] (-1,3.5) -- (0.5, 3.5);
\draw[-latex] (-1,4.5) -- (0.5, 4.5);
\draw[-latex] (-1,0.5) -- (0.5, 0.5);
\draw[-latex] (-1,1.5) -- (0.5, 1.5);

\draw[-latex] (1.5,6) -- (1.5, 4.5);
\draw[-latex] (2.5,-1) -- (2.5, 1.5);
\draw[-latex] (2.5,6) -- (2.5, 4.5);
\draw[-latex] (2.5,6) -- (2.5, 2.5);

\node at (3.5,-1) {(a)};
\end{scope}

\begin{scope}[xshift=15cm]

\fill (0,1) rectangle (1,2);
\fill (0,3) rectangle (1,4);
\fill (1,4) rectangle (2,5);
\fill (2,2) rectangle (3,3);
\fill (3,4) rectangle (4,5);

\fill (1,0) rectangle (2,2);
\fill (1,3) rectangle (2,4);
\fill (3,3) rectangle (4,4);
\fill (3,0) rectangle (4,1);
\fill (4,0) rectangle (5,5);
\fill (6,0) rectangle (7,5);

\draw (0,0) grid (7,5);
\node at (3.5,-1) {(b)};
\end{scope}

\end{tikzpicture}

\caption{Illustration of a 2-dominating set on $G_{5,7}$:\\
(a) the cells pointed by arrows are not 2-dominated: they have 0 or 1 black neighbour;\\
(b) the set of black cells is 2-dominating.}
\label{fig-2dom-intro}
\end{figure}
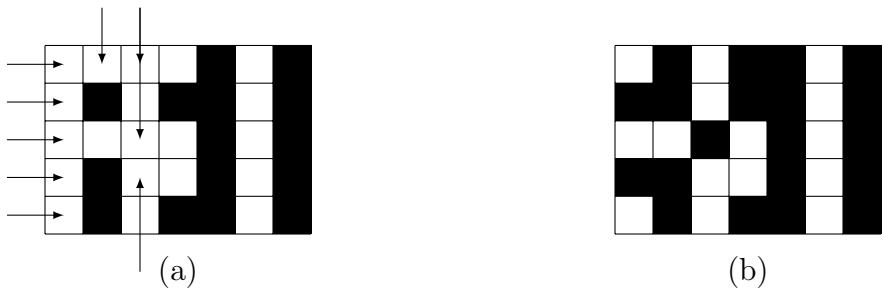

\begin{deff}
A set $S$ of vertices of $G$ is \textbf{distance-2 dominating}\footnote{As for the other types of domination, a point of English grammar arises. We write: "This child is six years old." but "This is a six-year-old child." From that we infer that while there is no need to hyphenate the group when following "to be", we have to when the whole group is an adjective which precedes a noun. We also hyphenate when necessary to avoid confusions.} when any vertex $v$ not in $S$ has a neighbour at distance at most two in $S$: either a neighbour of $v$ or a neighbour of one of $v$'s neighbours.
\end{deff}

\begin{figure}[h!]
\centering
\begin{tikzpicture}[scale=0.5]
\begin{scope}

\fill (1,0) rectangle (2,1);
\fill (1,5) rectangle (2,6);
\fill (2,2) rectangle (3,3);
\fill (5,2) rectangle (6,4);
\fill (6,0) rectangle (7,1);

\draw[-latex] (-1, 3.5) -- (0.5, 3.5);
\draw[-latex] (3.5, 7) -- (3.5, 4.5);
\draw[-latex] (4.5, 7) -- (4.5, 5.5);
\draw[-latex] (6.5, 7) -- (6.5, 5.5);
\draw[-latex] (7.5, 7) -- (7.5, 5.5);
\draw[-latex] (7.5, 7) -- (7.5, 4.5);

\draw (0,0) grid (8,6);

\node at (4,-1) {(a)};
\end{scope}

\begin{scope}[xshift=15cm]

\fill (1,0) rectangle (2,1);
\fill (1,5) rectangle (2,6);
\fill (2,2) rectangle (3,3);
\fill (5,2) rectangle (6,4);
\fill (6,0) rectangle (7,1);

\fill (2,3) rectangle (3,4);
\fill (7,3) rectangle (8,4);
\fill (6,5) rectangle (7,6);

\draw (0,0) grid (8,6);
\node at (4,-1) {(b)};
\end{scope}

\end{tikzpicture}

\caption{Illustration of a distance-2-dominating set on $G_{6,8}$:\\
(a) the cells pointed by arrows are not distance-2 dominated: the closest black cell is at distance at least three;\\
(b) the set of black cells is distance-2 dominating.}
\label{fig-dist-two-intro}
\end{figure}
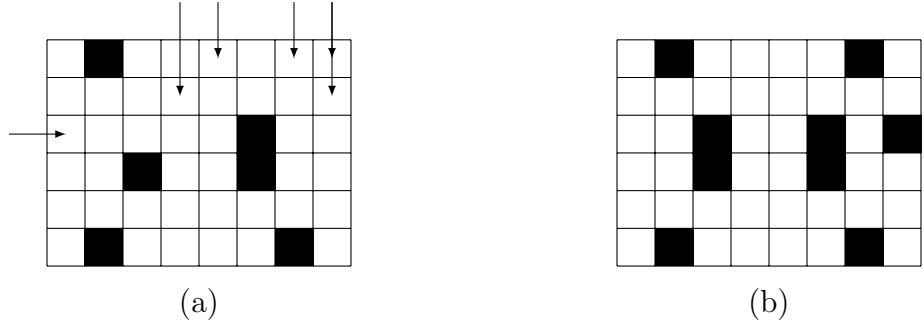

It is easy to note that for instance any 2-dominating set is dominating. Similarly, any dominating or 2-dominating set is distance-2 dominating.

\begin{deff}
A set $S$ of vertices of a graph $G$ is \textbf{total dominating} when any vertex $v \in G$ has at least one neighbour in $S$.
\label{def-total-dominating}
\end{deff}
Notice that, by contrast with the domination, even dominant vertices must have a neighbour which dominates them. In the case of grids, since there are no loops, a vertex is never its own neighbour, therefore each $v \in S$ must be connected to some vertex in $S \setminus \{v\}$.

\begin{figure}[h!]
\centering
\begin{tikzpicture}[scale=0.5]
\begin{scope}

\fill (1,0) rectangle (2,2);
\fill (1,3) rectangle (2,4);
\fill (2,3) rectangle (3,4);
\fill (3,3) rectangle (4,4);
\fill (3,1) rectangle (4,2);
\fill (5,0) rectangle (6,5);
\fill (0,4) rectangle (1,5);
\fill (0,1) rectangle (1,2);

\draw[-latex] (-1.5, 4.5) -- (0, 4.5);
\draw[-latex] (3.5, 6) -- (3.5, 2);
\draw (0,0) grid (7,5);

\node at (3.5,-1) {(b)};
\end{scope}

\begin{scope}[xshift=15cm]

\fill (1,0) rectangle (2,2);
\fill (1,3) rectangle (2,4);
\fill (2,3) rectangle (3,4);
\fill (3,3) rectangle (4,4);
\fill (3,1) rectangle (4,2);
\fill (5,0) rectangle (6,5);

\fill (0,1) rectangle (1,2);

\fill (3,0) rectangle (4,1);
\fill (0,3) rectangle (1,4);
\draw (0,0) grid (7,5);

\node at (3.5,-1) {(b)};
\end{scope}
\end{tikzpicture}

\caption{Illustration of a total dominating set on $G_{5,7}$:\\
(a) the set is not total dominated: the two cells pointed by arrows are not dominated;\\
(b) the set of black cells is total dominating.}
\label{fig-total-intro}
\end{figure}
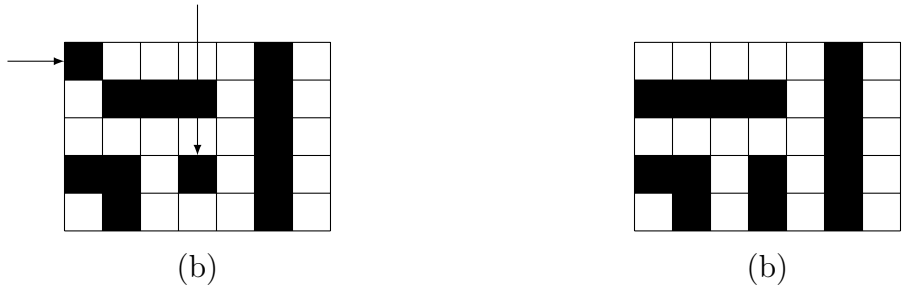

\begin{deff}
A \textbf{Roman-dominating "set"} is a pair $(S_1, S_2)$ such that every vertex $v \notin S_1\cup S_2$ has at least one neighbour in $S_2$.
\end{deff}
Informally, a Roman-dominating set consists in placing troops of soldiers on the vertices. We can either put no troops, one troop or two troops. A single troop can defend the vertex it is placed on while two troops placed on a vertex defend both it and its neighbours. The cost is the total number of troops.

\begin{figure}[h!]
\centering
\begin{tikzpicture}[scale=0.5]
\begin{scope}

\fill (0,0) rectangle (1,1);
\fill (0,4) rectangle (1,5);
\fill (2,1) rectangle (3,2);
\fill (2,3) rectangle (3,4);
\fill (4,0) rectangle (5,1);
\fill (5,2) rectangle (6,3);
\fill (6,0) rectangle (7,1);

\fill[gray] (0,2) rectangle (1,3);
\fill[gray] (3,2) rectangle (4,3);
\fill[gray] (3,4) rectangle (4,5);
\fill[gray] (6,4) rectangle (7,5);

\draw[-latex] (-1, 2.5) -- (1.5, 2.5);
\draw[-latex] (4.5, 6) -- (4.5, 4.5);
\draw[-latex] (4.5, 6) -- (4.5, 3.5);
\draw[-latex] (5.5, 6) -- (5.5, 4.5);
\draw[-latex] (6.5, 6) -- (6.5, 3.5);

\draw (0,0) grid (7,5);

\node at (3.5,-1) {(a)};
\end{scope}

\begin{scope}[xshift=15cm]

\fill (0,0) rectangle (1,1);
\fill (0,4) rectangle (1,5);
\fill (2,1) rectangle (3,2);
\fill (2,3) rectangle (3,4);
\fill (4,0) rectangle (5,1);
\fill (5,2) rectangle (6,3);
\fill (6,0) rectangle (7,1);

\fill (4,4) rectangle (5,5);

\fill[gray] (0,2) rectangle (1,3);
\fill[gray] (1,2) rectangle (2,3);
\fill[gray] (3,2) rectangle (4,3);
\fill[gray] (3,4) rectangle (4,5);
\fill (6,4) rectangle (7,5);

\draw (0,0) grid (7,5);
\node at (3.5,-1) {(b)};
\end{scope}

\end{tikzpicture}

\caption{Illustration of a Roman-dominating set on $G_{5,7}$ (ells with two troops are black and cells with one troop are grey):\\
(a) the cells pointed by arrows are not dominated.\\
(b) the set of black and grey cells is Roman dominating}
\label{fig-roman-intro}
\end{figure}

\begin{deff}
The cost of a dominating, 2-dominating or distance-2-dominating set $S$ is its size: $|S|$. For a Roman-dominating set, the cost is $|S_1|+2|S_2|$.
The \textbf{2-domination number} of a graph $G$, denoted by $\bm{\gamma_2(G)}$, is the \uline{minimum cost} of a 2-dominating set of $G$.\\
We define similarly the total-domination number $\bm{\gamma_\textbf{T}}$, the distance-2-domination number $\bm{\gamma_{\textbf{d2}}}$ and the Roman-domination number by $\bm{\gamma_\textbf{R}}$.\\
\end{deff}

Note that in the previous figures illustrating the various domination problems, the correct dominating sets are not necessarily of minimum size. For some, we may even trivially remove some vertices from the dominating set without breaking the domination property. For instance in \Cref{fig-roman-intro} one grey cell has a black cell neighbour. The grey cell can be removed because it is still dominated by its black neighbour and it does not dominate any cell.

\begin{notation}
For concision we will denote $\gamma_2(G_{n,m})$ by the shorter notation $\bm{\gamma_2(n,m)}$. The same applies to the other domination problems.
\end{notation}

\section{Method for finding the 2-domination number for grids of arbitrary size}
\label{section-dom-method}

Here, we set up a framework for finding the minimum size of dominating sets in a grid. We first explain how to proceed when the number of lines is fixed, then we introduce the notion of \emph{loss} to extend it, when possible, to grids of arbitrary height. The first method is only usable when the number of lines is small, because the number of objects we examine with our computer program becomes too big at some point. The second method begins to work when the number of lines is big enough, but does not always work, depending on some characteristics of the problem. Also, even if it should work on a problem, it may fail by lack of computational power. In fact, the number of lines for which the method works needs to be small enough to make it possible to run the method on a computer. We first explain one after the other the two methods, on the 2-domination problem (see~\cite{rao} for an alternate explanation of the loss method, applied to the domination problem). Both methods rely on the notions of \emph{states} and \emph{compatibility relations}, which are translated into transfer matrix products in the $(\min, +)$-algebra (it can also be viewed as a dynamic algorithm). The $\bm{(\min, +)}$\textbf{-algebra} consists in substituting, in the computations, the operator + by the operator $\min$ as well as replacing the multiplication by the addition. It is a standard method. The second method uses the more recent and more complex method of \emph{loss} introduced by Gonçalves et al.~\cite{rao}. We are the first ones to use this method to find other results. The optimisations we made and some details on the manner they were implemented will also be discussed.

Throughout the explanations of this section, we will prove the following theorem, which confirms the results found by \cite{lu-xu, shaheen} for $n \leq 4$ and slightly corrects the result by \cite{shaheen} for $n = 5$:

\begin{thm}[\cite{rao-talon}]
For all $1 \leq n \leq m$, the 2-domination number equals:\\
\[\gamma_2(n,m) =    \left\{
\setstretch{1.25}
\begin{array}{ll}
      \left\lceil\frac{m+1}{2}\right\rceil & \quad\textrm{if }n = 1 \\
      m & \quad\textrm{if }n=2 \\
      m+\left\lceil\frac{m}{3} \right\rceil & \quad\textrm{if } n=3 \\
      2m-\left\lfloor \frac{m}{4} \right\rfloor& \quad\textrm{if } n=4\textrm{ and } m\mod 4 = 3 \\
      2m-\left\lfloor \frac{m}{4} \right\rfloor+1& \quad\textrm{if } n=4\textrm{ and } m\mod 4 \neq 3 \\
      2m+\left\lceil\frac{m}{7}\right\rceil+1& \quad\textrm{if } n=5\textrm{ and } m\mod 7 \in \{0,6\}\\
      2m+\left\lceil\frac{m}{7}\right\rceil& \quad\textrm{if } n=5\textrm{ and } m\mod 7 \notin \{0,6\}\\
      2m+\left\lfloor\frac{6m}{11}\right\rfloor+1& \quad\textrm{if } n=6\textrm{ and } m\mod 11 \in \{0,2,6\}\\
      2m+\left\lfloor\frac{6m}{11}\right\rfloor+2& \quad\textrm{if } n=6 \textrm{ and } m\mod 11 \notin \{0,2,6\}\\

      3m-\left\lfloor\frac{m}{18}\right\rfloor+1& \quad\textrm{if } n=7 \textrm{, } m > 9\textrm{, }m\mod 18 \leq 9 \textrm{ and } m\mod 18 \neq 7\\
      
      3m-\left\lfloor\frac{m}{18}\right\rfloor& \quad\textrm{if } n=7 \textrm{ and } (m \leq 9\textrm{ or }m\mod 18 > 9 \textrm{ or } m\mod 18 = 7)\\
      
      3m+\left\lfloor\frac{m}{3}\right\rfloor& \quad\textrm{if } n=8 \textrm{ and } m \mod 3 = 1\\
      3m+\left\lfloor\frac{m}{3}\right\rfloor+1& \quad\textrm{if } n=8 \textrm{ and } m \mod 3 \neq 1\\
      \left\lfloor \frac{(n+2)(m+2)}{3}\right\rfloor-6 & \quad\textrm{if } n\geq 9.
\end{array}
\right. \]
\label{th-2dom}
\end{thm}
\subsection{Fixed (small) height and width}

We present here the technique of transfer matrices, a well-known method to solve many problems of this kind. We adapt it here to establish the 2-domination values for grids of fixed height and width. The technique follows a dynamic programming approach. One parameter of this approach (the number of states) is exponential in the number of lines, which is why the latter needs to be small enough: we want the computations we run on a computer to finish within a reasonable\footnote{i.e. less than the time until the end of the PhD... or more truthfully a few hours up to a couple of days on a machine with a hundred of cores} time.

In this paragraph, we introduce the notion of column \emph{states}, which enables us to "enumerate" the partial (i.e. the first columns of) 2-dominating sets by remembering only their "fingerprint" on their last column. This way, when we consider adding a new column, we do not need to remember the whole partial dominating set but only some information stored in the current column. This is possible because the domination problems we study have a \textbf{local characterisation}\footnote{We speak a bit about this notion in \Cref{dom-counting-chapter}}: while the property is global (defined on a whole object, which can be arbitrarily large), it is possible to make local queries (querying neighbours at most at a fixed distance from the vertex we consider) in such a way that we can check if the global property is verified by querying local information at each vertex. This is most helpful for us because we can enumerate a dominating set column by column and, when looking at column $j$, we can forget columns with index less than $j-2$: designing our algorithm smartly, we do not recall the full columns with indices $j-1$ and $j-2$ but only the necessary information.

Before defining the notions, we indicate that \uline{the number of lines $n$ is supposed fixed}. We show in \Cref{prop-primitive} that all the states and compatibility relations we will define implicitly appear when the number of columns $m$ is at least some number. We assume here that \uline{$m$ is big enough} in this respect, that is 8 for the 2-domination. We define $\SSS = \{$ \textsc{stone}, \textsc{need\_one}, \textsc{ok} $\}$ to be the set of \textbf{cell values}. \textsc{stone} means that the cell belongs to the 2-dominating set $D$: in the rest, we refer to the cells of $D$ as \textbf{"containing a stone"}. \textsc{ok} means that the previous and current columns suffice to 2-dominate the cell, before adding the next column. Finally, \textsc{need\_one} means that the cell had so far one neighbour in $D$, so it needs a new one in the next column. If $D$ is a 2-dominating set of cells of the grid $G_{n,m}$ let $f(D) \in (\SSS^n)^m$ be such that\footnote{$f(D)$ is considered as a matrix so $i$ is the index for the lines and appears first.} $f(D)[i][j]$ is \textsc{stone} if\footnote{Here $(j,i)$ are coordinates, hence their inverted order compared to indexing $f(D)$.} $(j,i) \in D$, \textsc{ok} if at least two among $(j-1,i)$, $(j, i-1)$ and $(j, i+1)$ are in $D$, or \textsc{need\_one} otherwise. Note that, since $D$ is 2-dominating, a cell is \textsc{need\_one} if exactly one among $(j-1,i)$, $(j, i-1)$ and $(j, i+1)$ is in $D$. Note that the value of a cell does not depend on the values of the cells of the next column(s). We also define, for $0 \leq j < n$, $f_j(D)$ to be the vector containing the values of column $j$ of $f(D)$: for all $i\in\llbracket 0 , n-1 \rrbracket$, $f_j(D)[i]=f(D)[i][j]$.\\

We now define the set of (column) \textbf{states} $\VV = \cup_{0\leq j < m}{\{f_j(D) : D \textrm{ is 2-dominating}\}}$. $\VV$ is the set of states which appear in some 2-dominating set. Among these states, we define the set of \textbf{first states} $\FF = \{f_0(D) : D \textrm{ is a 2-dominating set}\}$. Finally, we define the set of \textbf{end} (or \textbf{dominated}) \textbf{states} $\EE = \{f_{m-1}(D) : D \textrm{ is a 2-dominating set}\}$. $\FF$ is the set of states which can be the first column of a 2-dominating set, that is whose entries only depend on themselves and not on a presupposed previous column. $\EE$ is the set of states which can be the last column of a dominating set because they do not need a next column to be 2-dominated: they are dominated by themselves and their previous column.\\

We now define the \textbf{relation of compatibility} $\RRR$: we say that a state $S'\in \VV$ is \textbf{compatible} with $S\in \VV$, and write $S\RRR S'$ if there exist a 2-dominating set $D$ and some $j \in \llbracket 0 , m-2 \rrbracket$ such that $\textrm{f}_j(D) = S$ and $\textrm{f}_{j+1}(D) = S'$. 
 Defining these states enables us to use the principles of dynamic programming: instead of enumerating all possible 2-dominating sets, we realise that the information conveyed in $f(D)$ is enough, and that we only need the information at a column $j$ to continue to column $j+1$. In particular, we do not need to know what happened in previous columns with indices less than $j$. This corresponds to the notion of \emph{nearest-neighbour} recoding in the world of SFTs which we will define in \Cref{dom-counting-chapter}: a translation \textit{subshift} into an equivalent one on a different alphabet such that the forbidden patterns are only of size two. In our case, it means that we only forbid couples of columns states $S,S'$ when $S'$ cannot be put just after $S$.

To illustrate these concepts, we give the rules defining the sets $\VV$, $\FF$, $\EE$ and the relation $\RRR$. These explicit definitions correspond exactly to the implicit way we defined these concepts in the previous paragraph. It is easy to convince ourselves of this fact as they are translations of the implicit definitions of $\VV$, $\FF$, $\EE$ and $\RRR$ considered along with the definition of the cell values \textsc{stone}, \textsc{need\_one} and \textsc{ok}. The concept of states and compatibility relation are shown on \Cref{fig-ex-states}.
\newpage

$S \in \VV$ if and only if for all $ i \in \llbracket 0 , n-1 \rrbracket$: 
\begin{itemize}[topsep=0pt, noitemsep]
\item if $S[i] =$ \textsc{need\_one} then at most one among $S[i-1]$, $S[i+1]$ is \textsc{stone};
\item if $S[i] =$ \textsc{ok} then at least one among $S[i-1]$ and $S[i+1]$ is \textsc{stone}.
\end{itemize}
\vspace{5pt}

$S \in \FF$ if and only if for all $i \in \llbracket 0 , n-1 \rrbracket$:
\begin{itemize}[topsep=0pt]
\setlength{\itemsep}{0pt}
\item if $S[i] =$ \textsc{need\_one} then exactly one among $S[i-1]$, $S[i+1]$ is \textsc{stone};
\item if $S[i] =$ \textsc{ok} then both $S[i-1]$ and $S[i+1]$ are \textsc{stone} (so $1 \leq i < m-1$).
\end{itemize}
\vspace{5pt}
A state $S$ belongs to $\EE$ if and only if $S\in \VV$ and none of its entries is \textsc{need\_one}.

Finally, $S \RRR S'$ if and only if for all $i \in \llbracket 0 , n-1 \rrbracket$:
\begin{itemize}
\setlength{\itemsep}{0pt}
\item if $S[i] =$ \textsc{need\_one} then $S'[i] =$ \textsc{stone};
\item if $S'[i] =$ \textsc{need\_one} then exactly one among $S'[i-1]$, $S'[i+1]$ and $S[i]$ is \textsc{stone};
\item if $S'[i] =$ \textsc{ok} then at least two among $S'[i-1]$, $S'[i+1]$ and $S[i]$ are \textsc{stone}.
\end{itemize}

\begin{figure}[h]
\centering
\includegraphics[scale=0.3]{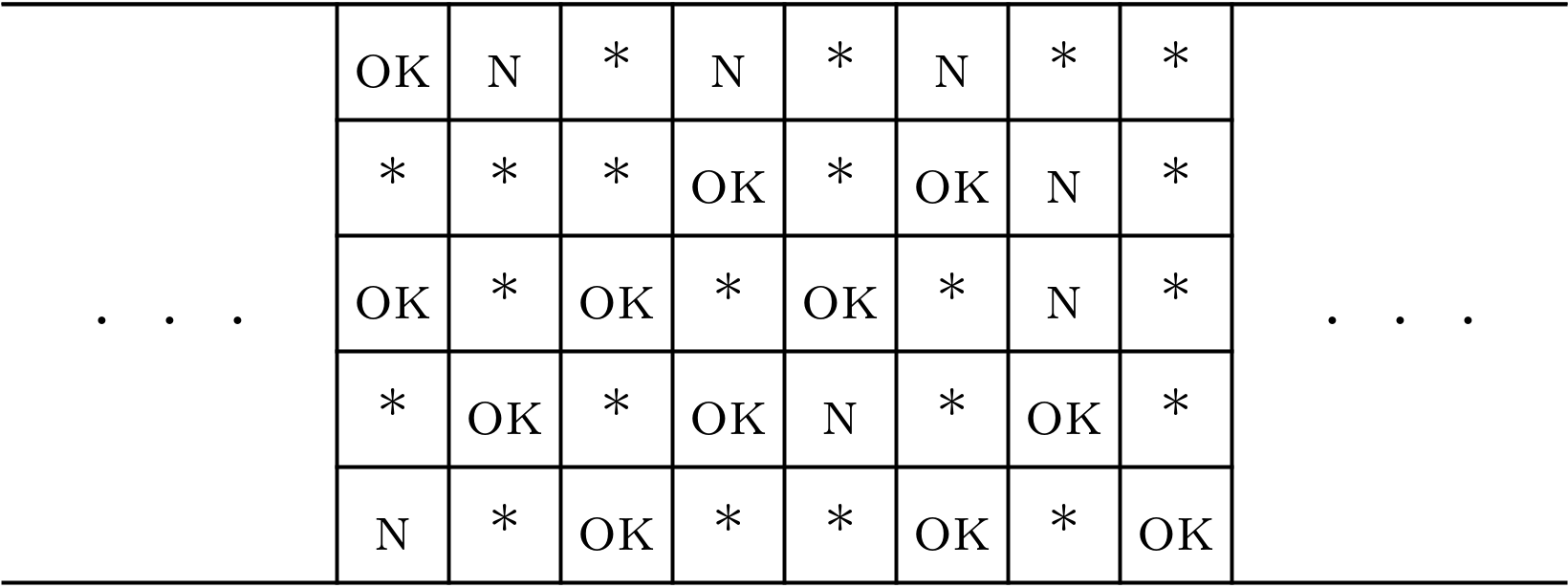}
\caption{Illustration of the states and compatibility relations for the 2-domination problem. '\textsc{n}' corresponds to the state value \textsc{need\_one} and '\textsc{*}' corresponds to \textsc{stone}. The data on one column only depends on this column and the columns at its left.}
\label{fig-ex-states}
\end{figure}

We said above that the number of states is exponential in the number of lines. We can in fact compute the growth rate of the number of states thanks to a concept named the \emph{Rauzy graph} of a language. We will explain in \Cref{section-rauzy}. It allows us to find the number of states in our technique.

\begin{fact}
There are $\Theta(x_0^n) \approx \Theta(2.485584^n)$ states in $\VV$ for $n$ lines, where $x_0$ is the real root of the polynomial $x^3 - 2x^2 - 3$.
\end{fact}

\begin{notation}
If $V$ is a vector, then we denote by $|V|_\mathit{foo}$ the number of its entries equal to $\mathit{foo}$.
\end{notation}

\begin{claim}
Let $F$ be the vector of size $|\VV|$ such that $F[S] = |S|_{\text{\textsc{stone}}}$ if $S\in \FF$ or $+\infty$ otherwise.\\ Let $E$ be the vector of size $|\VV|$ such that $E[S] = 0$ if $S \in \EE$ or $+\infty$ otherwise.\\ Let $T$ be the square matrix with $|\VV|$ lines such that $T[S][S'] = |S'|_{\text{\textsc{stone}}}$ if $S \RRR S'$ or $+\infty$ otherwise.\\Then for any $m > 0$, $\gamma_2(n,m) = F^{\mathsmaller T}T^{m-1}E$.\\
(We recall that the products of matrices are done in the $(\min, +)$-algebra.)
\label{claim-exact}
\end{claim}

\begin{proof}
Let $j \geq 1$. $V_j = F^{\mathsmaller T}T^{j-1}$ is a vector such that if $S \in \VV$ then $V_j[S]$ is the minimum size of a set $X$ which 2-dominates the subgrid with $n$ lines and $j-1$ columns, and such that the last column of $X$ is in state $S$. However, we are interested in a 2-dominating set, therefore the last column (assumed to be in state $S$) should be 2-dominated as well. Thus $V_m E = \min_{S \in \EE}{V_m[S]}$ gives us the minimum size of any 2-dominating set.
\end{proof}

This claim leads to a simple algorithm to compute $\gamma_2(n,m)$: generate the different sets and the compatibility relation, and then compute the exponentiation of matrices, and two matrix-vector products. The matrix $T$ is a transfer matrix, whose exponentiation propagates the fact of being 2-dominated one column further.
This is enough to compute the 2-domination numbers for fixed $n$ and $m$, but our goal here is to find all the numbers for fixed $n$ and arbitrary $m$. The next section fills this hole.\\

\subsection{Fixed number of lines but arbitrary number of columns}
The method we have just presented works for a fixed number of lines. Since it is linear in the number of columns, this number must be finite. We use here a less known method to obtain values for arbitrary width of the grid, establishing recurrence relations whose existence are guaranteed by some properties of the transfer matrices. The computations lead to closed formulas for a fixed number of lines, hence we obtain a constant-time algorithm at the end, exploiting this formula. We also, on a theoretical point, make the connection between the effectiveness of the method
and the \emph{primitivity} of the transfer matrix we use.\\

By observing that the transfer matrix we manipulate is \emph{primitive}, we show that this periodicity was expected. That is, even without finding the recurrence relation, it is possible to prove that the formula follows some recurrence relation without running any program. Furthermore, this observation extends to any problem like this one. We can, for instance, guarantee that the distance-three-domination (any vertex not in $S$ must have a neighbour at distance at most three in $S$) number satisfies a recurrence relation. The same applies to all the dominating problems we study in this chapter.

\begin{deff}
A matrix $M$ is \textbf{primitive} (in the $(\min, +)$-algebra)\footnote{(almost) last reminder that we do not work in the standard algebra} when there exists an integer $k > 0$ such that $\max_{i,j}(\{ M^k[i][j] \}) < +\infty$.
\end{deff}

\begin{prop}
$T$ is primitive.
\label{prop-primitive}
\end{prop}
This proposition means that given two states $S_1$ and $S_2$, it is always possible to find some 2-dominating set $D$ and some $i_1$ and $i_2 \leq i_1 + k$ such that $f_{i_1}(D) = S_1$ and $f_{i_2}(D) = S_2$. This can be related to the notion of \emph{irreducibility} of Markov chains. As we mentioned above, it certifies that any possible state appears in some dominating set of a grid if there are $k$ columns separating it from the special first states and the special end states. So any state appears if there are at least $2k+1$ columns. Now, to be sure to see every compatibility relations, we must add another column: each possible state appears at the ${k+1}^\text{th}$ column, and any next state must be separated from the end states by $k$ columns. Therefore, $m \geq 2k+2$ is enough, which means 8 for the 2-domination (see the proof just below).

\begin{proof}
Let $S_0, S_2 \in \VV$. Let $S_*$ be the state whose entries all are \textsc{stone}. There exists some $S_1 \in \VV$ such that $S_0 \RRR S_*$, $S_* \RRR S_1$ and $S_1 \RRR S_2$. We leave the construction of $S_1$ to the reader: put the necessary stones and fill the rest accordingly. We conclude that $T^3 < +\infty$.
\end{proof}

We now prove an interesting property that primitive matrices have, in the\\(min, +)-algebra: the series of their powers satisfy some recurrence relation. This is true for any primitive matrix in this algebra.

\begin{thm}
Let $M$ be a primitive matrix with coefficients in $\NN \cup \{+\infty\}$. Let $k$ be such that $\max(\{ M^k[i][j] \}) < +\infty$. Then there exist some $l_0$, $p$ and $r$ such that for all $l \geq l_0$, $M^{l+r} = M^r + p$, where $M^r+p$ means\footnote{We recall that the replacement of $(+,\times)$ by $(\min,+)$ only applies in the internal operations of products between matrices and matrices or vectors.} that we add $p$ to every element of $M^r$.
\label{th-primitive-rec}
\end{thm}
\begin{proof}
Notice that, since $\max(\{ M^k[i][j] \}) < +\infty$, each line of $M$ has at least one element different from $+\infty$. The same goes for the columns. We denote by $\alpha$ the maximum value of $M^k$. Let $l > k$. For every $i$ and $j$,
\begin{equation}
M^l[i][j] = \min_{u}(M^{l-k}[i][u]+M^k[u][j]) \leq \min_{u}(M^{l-k}[i][u]+\alpha) \leq \min_{u}(M^{l-k}[i][u])+ \alpha.
\label{eq-primitive-max}
\end{equation}
This quantity is finite since $M^{l-k}$, like $M$, has in each row and each column one element different from $+\infty$.

Now let us bound the value of $M^l[i][j]$ from below:
\begin{equation}
M^l[i][j] = \min_{u}(M^{l-k}[i][u]+M^k[u][j]) \geq \min_{u}(M^{l-k}[i][u])
\label{eq-primitive-min}
\end{equation}

By subtracting \Cref{eq-primitive-min} from \Cref{eq-primitive-max}, we obtain
\begin{equation*}
0 \leq \max_{i,j}(M^l[i][j]) - \min_{i,j}(M^l[i][j]) \leq \alpha.
\end{equation*}

This allows us, for any $l > k$, to define $M'_l \in \llbracket 0; \alpha \rrbracket^{nm}$ by the following decomposition:
\begin{equation*}
M^l = \min(M^l)+ M'_l.
\end{equation*}
This means that each $M^l$ is decomposed into a matrix which has all its coefficients equal to $\min(M^l)$ plus a matrix whose coefficients are bounded by a value independent from $l$. Since there are finitely many matrices in $\llbracket 0; \alpha \rrbracket^{nm}$, we conclude that there are two equal matrices $M_{l_0} = M_{l_1}$ for $l_1 > l_0 > k$.
By letting $p = \min(M^{l_1}) - \min(M^{l_0})$ we obtain
\begin{equation}
M^{l_1} = M^{l_0} + p.
\end{equation}

Now we choose $r = l_1-l_0$, and for $l \geq l_0$:
\begin{equation*}
M^{l+r} = M^{l-l_0}M^{l_0+r} = M^{l-l_0}(M^{l_0}+p) = M^{l-l_0}M^{l_0}+p = M^{l}+p.
\end{equation*}

Note that in our tropical algebra, if $A,B$ and $C$ are matrices, and '+' denotes, as before, the standard plus operation, then $A(B+C) = AB+C$.
\end{proof}

A direct implication of \Cref{th-primitive-rec} is that the transfer matrix $T$ verifies some recurrence relation. Thanks to \Cref{claim-exact}, the relations we obtain for the transfer matrix $T$ directly apply to the 2-domination number. Here are the relations we obtain for $n \leq 12$:
\begin{multicols}{2}
\begin{itemize}
\item $\forall m \geq 3, \gamma_2(1,m) = \gamma_2(1,m-2) + 1$;
\item $\forall m \geq 3, \gamma_2(2,m) = \gamma_2(2,m-1) + 1$;
\item $\forall m \geq 5, \gamma_2(3,m) = \gamma_2(3,m-3) + 4$;
\item $\forall m \geq 8, \gamma_2(4,m) = \gamma_2(4,m-4) + 7$;
\item $\forall m \geq 14, \gamma_2(5,m) = \gamma_2(5,m-7) + 15$;
\item $\forall m \geq 20, \gamma_2(6,m) = \gamma_2(6,m-11) + 28$;
\item $\forall m \geq 31, \gamma_2(7,m) = \gamma_2(7,m-18) + 53$;
\item $\forall m \geq 16, \gamma_2(8,m) = \gamma_2(8,m-3) + 10$;
\item $\forall m \geq 17, \gamma_2(9,m) = \gamma_2(9,m-3) + 11$;
\item $\forall m \geq 14, \gamma_2(10,m) = \gamma_2(10,m-1) + 4$;
\item $\forall m \geq 16, \gamma_2(11,m) = \gamma_2(11,m-3) + 13$;
\item $\forall m \geq 17, \gamma_2(12,m) = \gamma_2(12,m-3) + 14$.
\end{itemize}
\end{multicols}

These relations were obtained by running our program on a computer. For instance, finding the relation for $n=12$ takes around 17.5 seconds on a personal laptop using only one CPU, and it uses 56Mib of RAM. There are around 26.5k states and 10M compatible pairs between states.

Thanks to these relations, and to the first values we obtain for each $n$, we deduce the formulas for $\gamma_2(n,m)$, for $1 \leq n \leq 12$. For instance, for $n = 5$ we only need to know the recurrence relation, plus the first twelve\footnote{$20=14+6$} values. We stopped here at $n=12$ here because the method for arbitrarily large $n$ works from $n \geq 13$.

We may observe that these figures are the ones reflecting the periodicity of the transfer matrix: for $n=12$, this begins when $m \geq 17$. However, once we have this relation we can "go back in time" and check from what point the relations begins to be valid for the $\gamma_2$ numbers we computed. By doing so, we find for instance that the recurrence relation for $\gamma_2(12, m)$ is actually valid from $m = 13$ instead of 17.
\\

Now thanks to \Cref{th-primitive-rec} we can prove a meta-theorem for a class of problems.

\begin{thm}
Any minimisation problem on grids which admits a local characterisation and a primitive transfer matrix has, when the number of lines is fixed, a closed formula for answer.
\label{th-meta-grids}
\end{thm}

\begin{proof}
It suffices to apply the same method as here. The recurrence relation the transfer matrix satisfies guarantees a recurrence relation on the parameter studied, hence the answer is a closed-form expression. This expression involves basic arithmetic operations (+,-,*,/) and the remainder of a Euclidean division.
\end{proof}

\begin{rk}
We can notice that any problem to which we can associate an SFT with the \emph{block-gluing} property (see \Cref{dom-counting-chapter}) will have a primitive transfer matrix and hence falls into the hypothesis of \Cref{th-meta-grids}.
\end{rk}

This method can work for other classes of graphs, as long as we find an "acyclic" way of enumerating parts of the dominating sets, like for the fasciagraphs defined by Bouznif et al.~\cite{fasciagraphs-paper}. Our method, and the theory behind stated by \Cref{th-meta-grids} can for instance be used to prove parts of what they investigate in~\cite{fasciagraphs-paper}. We used it to prove that some local problems with a primitivity condition admits a closed formula as answer, hence can be answered in constant time. In fact, the same applies to rotagraphs: instead of computing $F^{\mathsmaller T}T^{m-1}E$ (for a fasciagraph, with a first and last "column"), we compute $\min_{S\in \VV}(T^{m}[S][S]$ (for $m > 1$): the first and last state must be the same. In their paper, they prove similar results, using a different technique and a result similar to the one we prove for the recurrence of powers of primitive matrices.

The method we described can work also in these structures for other types of problems. We show in \Cref{dom-counting-chapter} that we can also count some types of dominating sets, using different arguments (and working in the standard algebra). Basically, this approach will work when the problem has some sort of local characterisation or local property.

\label{section-primitive-dom}

\subsection{The number of states}
\label{section-rauzy}
We present here a technique to find the growth rate of the number of words of a \emph{factorial} language: a set of words which is defined by forbidding specific patterns to be \emph{factors}. This means for instance that we choose to allow only words  which do not contain "an" or "on": "tomato" and "cherry" would be authorised whereas "banana" and "cinnamon" would not. The method consists in constructing the Rauzy graph of the language, and take its largest eigenvalue, which turns out to be the desired growth rate.

Let $\mathcal{L}$ be a language of words forbidding patterns of sizes up to $k$. Without loss of generality, we may assume that all the forbidden factors are of length equal to $k$. The (directed) Rauzy graph $R_i(\mathcal{L}) = (V_i, E_i)$ of $\mathcal{L}$ of order $i$ is defined by: $V_i$ is the set of factors of size $i$ appearing in words of $\mathcal{L}$, and $(u, v) \in E_i$ (note that the relation is not symmetric) when there exist $u'$  and two letters $a$ and $b$ such that $u = au'$ and $v = u'b$. We will refer to this relation as $R_i$ in the rest of this subsection. The way this graph works is similar to the one in \Cref{counting-numerical-section}, in which we use transfer matrices in the standard algebra to count dominating sets. The number of arcs going to some vertex $u$ of a Rauzy graph is the number of ways to obtain this factor by adding a letter to a smaller authorised factor of $\mathcal{L}$. This way, $R_i^n[u]$ counts the number of words of size $n$ the suffix of which is $u$, forbidding all forbidden factors of $\mathcal{L}$ of size up to $i+1$. In the case when the matrix is primitive (in the standard algebra), the Perron-Frobenius theorem states that it admits a largest eigenvalue which is positive and simple. Let us call it $\lambda_i$ for the graph $R_i$.
This $\lambda_i$ is a lower bound on the growth rate. However, for $i \geq k-1$, this $\lambda_i$ becomes the growth rate of the number of words of the language: all forbidden factors of size up to $k-1+1=k$ are guaranteed not to appear.

Unfortunately here our languages of the valid states are not factorial: for the first and last lines of a state, we forbid specific factors which are not forbidden in the "middle" of the word. Still, if we remove this limitation, the factorial language which we would obtain would have the same growth rate as our language of real valid states. As we showed in \Cref{section-primitive-dom}, the matrices of compatibility between the states are primitive. Therefore we could compute the matrices of the associated factorial languages (without the limitations at the extremities of the states). However, we proceeded here by smartly reusing the existing code: we take, as vertices of the Rauzy graph $R_i$, the factors of size $i$ of our languages which are far enough away from the beginning and from the end of the states. Their behaviour is similar to the behaviour of the middle parts of the states, if the distance to $i=0$ and $i=n-1$ is big enough. With the help of the Sage software, we could compute an approximate form of the $\lambda_i$'s by providing the adjacency matrix of the Rauzy graph, computed by our program. We also obtain a polynomial there are root of (the minimal polynomial of the matrix). We give each time the minimal polynomial of the Rauzy graph, or more precisely this polynomial divided by its lowest degree monomial.

\subsection{Arbitrary height}
\label{section-big-number-lines}

We adapt here the method Gonçalves et al. introduced in~\cite{rao}. This method is much less known than the previous ones, and to our knowledge it is the first time it is adapted. Quite some people provided formulas for some domination problems when the number of lines is small and fixed, but no one provided a formula for the real 2D case.

The idea is to assume that  for a sufficiently large grid, there are always dominating sets of minimum size in which the positions of the stones would be a projection of an optimal dominating set (i.e. one with minimum ratio of stones, one fifth here) of $\ZZ^2$, except on a fixed-height border of the grid. The border are special because not every cell has 4 neighbours at the frontier of a finite grid. This assumption turns out to also be true for the 2-domination problem. The height of the border which needs some rearrangement is a constant: it does not depend on the size of the grid we consider (provided $n$ and $m$ are large enough).

We need to count how many stones we need to make up for the problematic fact that the cells of the border have fewer than four neighbours. We use the dual concept of \emph{loss} to integrate this number to the number of stones at the centre of the grid. The loss denotes how much "influence" produced by the stones of a 2-dominating set is wasted. For instance, two neighbouring stones would cause a loss of 2: each stone cell dominates its stone neighbour which did not need to be dominated by another cell since it has a stone. Instead of computing the minimum number of stones needed, we will compute a lower bound on the minimum loss possible on the border. It happens that this lower bound gives us a direct lower bound on the 2-domination number, and that these bounds are sharp.\\

More formally, given a 2-dominating set $D$ of the $n\times m$ grid, we define the \textbf{loss} to be \[\ell(D,n,m) = 4|D|-2(nm-|D|).\] The idea behind this formula is simple: each stone contributes to the domination of its four neighbours, and each cell not in $D$ should be dominated twice. The difference between these two quantities is the influence of stones that was "lost", or "wasted", i.e. not necessary. The loss function should have several characteristics: while it should be "easy" to compute, it should also be reversible so that given the loss, we can deduce the 2-domination number.

The loss for the 2-domination problem can be computed in the following way. Let us consider a 2-dominating set $S$ and a cell which belongs to $S$. If it is a corner of the grid, it induces a loss of 2 because of this; else, if it is on the border, it induces a loss of 1. If a stone cell has $p$ neighbours also in $S$, then its loss is increased by $p$. Now if a cell does not belong to $S$ and has $p$ neighbours in this set, then it has a loss of $p-2$. We illustrate these computations in \Cref{fig-loss-2dom}. For instance, the top-left cell  has a loss of two because it is a corner: its top and left neighbours do not exist. Two cells at the right, we have a loss of two: this cell in on the border and has a neighbour in $S$. The cell with a three on the one before last column is in $S$ and has $p=3$ neighbours in $S$. Last example: the cell on the second column and one before last line has a loss of two: it does not belong to $S$ but has $p = 4$ neighbours in $S$: two of them are enough to 2-dominate it, hence it has a loss of two.

\begin{figure}[h]
\centering
\includegraphics[scale=0.4]{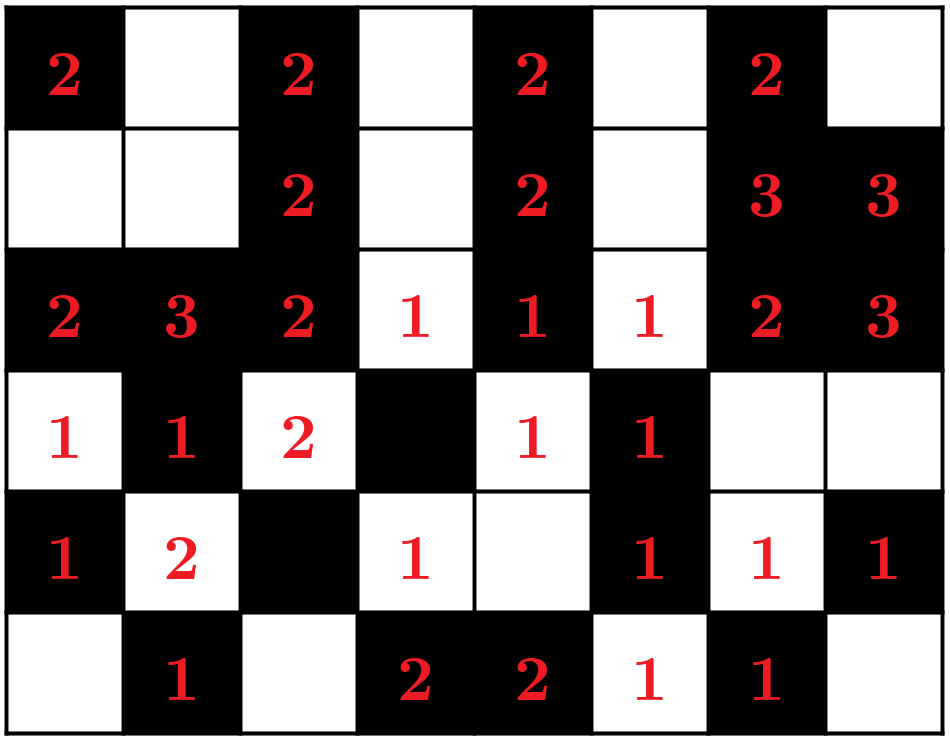}

\caption{Illustration of the local computation of the loss for the 2-domination problem. The red numbers indicate the loss induced by the cell they are on. The cells with no number have a loss of 0.}
\label{fig-loss-2dom}
\end{figure}

\begin{notation}
We denote by $\ell(n,m)$ the \textbf{minimum possible loss} over every 2-dominating set of $G_{n,m}$.
\end{notation}

Here, we can indeed reverse the formula as we wanted: $|D| = (2nm+\ell(D,n,m))/6$.  We then obtain
\begin{equation}
\gamma_2(n,m) = \frac{2nm+\ell(n,m)}{6}.
\label{eq-reversed-loss-2dom}
\end{equation}
The method solves the problem, that is the lower bound we obtain by the loss technique matches the 2-domination numbers: we prove later that it is also an upper bound. Computing the minimum loss over a big grid seems very hard, but we managed to obtain the right values by computing a lower bound for $\ell(n,m)$ which happens to be equal to it: we achieve this by computing the loss only on the border of size $h$ (where $2h < \min(n,m)$, see \Cref{figure-loss}).

\begin{deff}[see \Cref{figure-loss}]
The \textbf{border} of height $h$ of an $n \times m$ grid is the set of cells $(j,i)$ such that either $\min(i, n-1-i) < h$ or $\min(j, m-1-j) < h$.
We define the \textbf{corners} as the four connected parts of the grid composed of cells $(j,i)$ such that both $\min(i, n-1-i) < h$ and $\min(j, m-1-j) < h$.
The remaining four connected parts of the border are called the \textbf{bands}.
\end{deff}

\begin{figure}[H]
\centering

\includegraphics{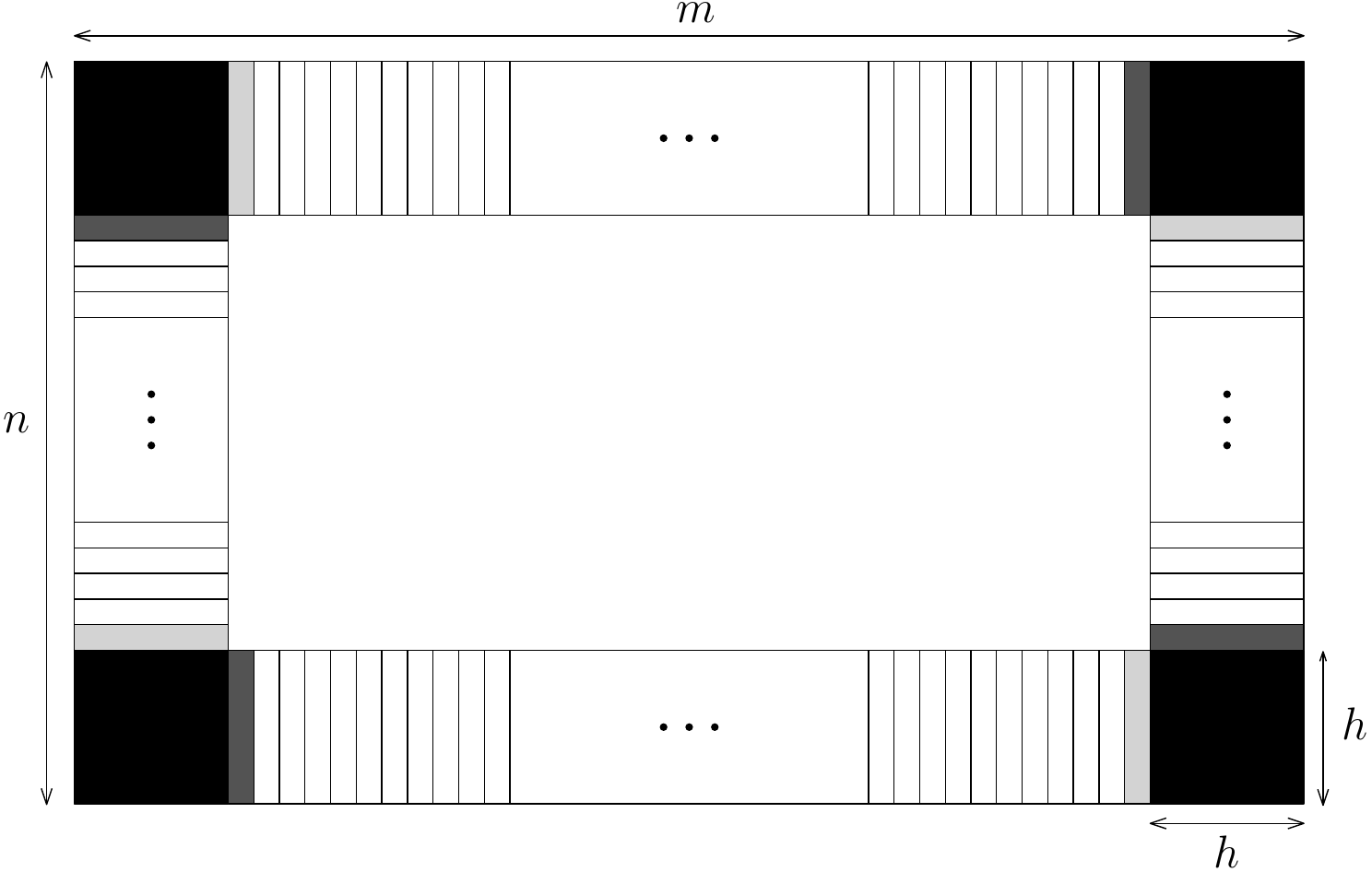}
\caption{The borders of size $h$ of the grid. The four parts coloured in black are the \emph{corners}, and the white parts are the \emph{bands}. The grey cells belong to both the bands and the corners: the ones filled with light grey are the output of a band and input of the corner next to it; the ones in dark grey are the output of a corner and input of the following band.}
\label{figure-loss}
\end{figure}

Once again, we use an algorithm which is faster than an exhaustive search over the dominating sets by working with the notion of states. However, we need to adapt the sets $\VV, \FF, \EE$ and the relation $\RRR$ we worked with.

Let us begin with computing the loss over the bottom band. We will, using transfer matrices as in the beginning of the chapter, compute the loss column by column. We remember, for each state placed on each column, what the minimum possible loss is so that the column is in this state. Thus we need to adapt the sets $\VV, \EE$ and the relation $\RRR$.  In the following we assume, as mentioned before, that $2h < \min(n,m)$.

\begin{notation}
We define the function $\hat{f_j}$ such that, if $D$ is a 2-dominating set,\\$\hat{f_j}(D) = f_j(D)[0], \cdots f_j(D)[h-1]$. $\hat{f_j}$ consists of the bottom $h$ lines of $f_j$.
As previously, $\hat{f_j}(D)$ denotes the column $j$ of $\hat{f}(D)$.
\end{notation}

We continue by defining the set of \textbf{almost valid states}
\[\VV_a = \bigcup_{0\leq j < m}{\{\hat{f_j}(D) : D \textrm{ is a 2-dominating set}\}}.\]
This set contains the states we will enumerate to compute the loss. It is in fact almost the same set as $\VV$: $S \in \VV_a$ if and only if, for $i \in \llbracket 0 , h-1 \rrbracket$:
\begin{itemize}[noitemsep, topsep=0pt]
\item if $S[i] =$ \textsc{need\_one} then at most one among $S[i-1]$, $S[i+1]$ is \textsc{stone};
\item if $S[i] =$ \textsc{ok} then $i = 0$ or at least one among $S[i-1]$ and $S[i+1]$ is \textsc{stone}.
\end{itemize}
\vspace{5pt}
We  need neither first nor dominated states to compute the loss, so we will neither define a set $\FF_a$ nor $\EE_a$. We define the relation of \textbf{almost-compatibility}, which differs a bit from $\RRR$. The main difference with the relation of compatibility we used earlier concerns the top line (of the bottom border) $i = 0$: the first cell will have a neighbour above it, in the centre of the grid, so we need to consider the case when this neighbour has a stone. More explicitly, if $S, S'\in \VV_a$, $S \RRR_a S'$ if and only if for $i \in \llbracket 0 , n-1 \rrbracket$:
\begin{itemize}[noitemsep, topsep=0pt]
\item if $S[i] =$ \textsc{need\_one} $S'[i] =$ \textsc{stone};
\item if $S'[i] =$ \textsc{need\_one} and $i \neq 0$ then exactly one among $S'[i-1]$, $S'[i+1]$ and $S[i]$ is \textsc{stone};
\item if $S'[0] =$ \textsc{need\_one} then at most one among $S'[1]$ and $S[0]$ is \textsc{stone};
\item if $S'[i] =$ \textsc{ok} and $i \neq 0$ then at least two among $S'[i-1]$, $S'[i+1]$ and $S[i]$ are \textsc{stone};
\item if $S'[0] =$ \textsc{ok} then at least one among $S'[1]$ and $S[0]$ is \textsc{stone}.
\end{itemize}
\vspace{5pt}
As said just before, we allow a state whose first cell has value \textsc{need\_one} not to be dominated by the next column: in this case, we consider that its neighbour above dominates it. This is a conservative assumption: we may underestimate the loss this way, but this is fine since we compute a lower bound on the loss.\\

We use again the exponentiation of a transfer matrix to compute the minimum loss over some border of a grid. We define the \textbf{band matrix $\bm{T_\mathrm{a}}$} such that $T_\mathrm{a}[S][S']$ contains the loss induced by putting state $S'$ after state $S$. By exponentiating the matrix $T_\mathrm{a}$ we can compute the minimum loss over a border, excluding the loss induced by the first state alone on itself.

The next step is to compute the loss for the corners. A corner is composed of an $h$ by $h$ square, plus an input column and an output column. Let us consider the bottom right corner of \Cref{figure-loss}. The last column of the bottom band is coloured in light grey: it is the input column of the square (and the output column of the band). At the other side of the square, the horizontal "column" filled with dark grey is the output column of the square (and the input column of the next border). Suppose that the input column of the square is in state $A$ and its output column is in state $B$. The \textbf{loss over the corner} is the sum of:
\begin{itemize}[noitemsep, topsep=0pt]
	\item the loss on the corner by $A, B$ and the corner itself;
	\item the loss on $B$ by the corner and $B$ itself;
	\item the loss on $A$ by the corner.
\end{itemize}
\vspace{5pt}
The explanation is simple: the input state was fixed by the loss computation on the band (so its loss so far was already counted) and the corner provides the first state for the next band (so we have to compute its loss so far).
Similarly to $T_\mathrm{a}$, we define the \textbf{corner matrix $\bm{C_\mathrm{a}}$} for a corner: $C_\mathrm{a}[S][S']$ contains the minimum loss over a corner whose input state is $S$ and output state is $S'$ as defined just before (we do not count the loss induced by $S$ alone on itself).

\begin{lemma}
The minimum loss over the border is \[ \min_{S \in \VV_a}((T_\mathrm{a}^{m-2h-1}C_\mathrm{a}T_\mathrm{a}^{n-2h-1}C_\mathrm{a})^2[S][S]).\]
\label{lemma-compute-loss-matrix}
\end{lemma}
Before diving into the proof, we mention that, since we are in the (min,+)-algebra, the formula in the above lemma can be rewritten as $\textrm{Tr}((T_\mathrm{a}^{m-2h-1}C_\mathrm{a}T_\mathrm{a}^{n-2h-1}C_\mathrm{a})^2)$.

\begin{proof}
$T_\mathrm{a}^{m-2h-1}$ is the minimum loss over a band starting on the output column of the bottom-left corner and ending on the input column of the bottom-right corner. Hence $T_\mathrm{a}^{m-2h-1}C_\mathrm{a}$ means computing the minimum loss on the bottom band we have just described, and extending it to the output state of the bottom right corner. As mentioned above, in the corner loss we take the input state of the corner as it is (which is exactly what $T^{m-2h-1}$ provides: the loss on the last state by itself and its preceding column was already computed). Since $C_\mathrm{a}$ includes the loss induced by the corner onto the output state of the corner, $T_\mathrm{a}^{m-2h-1}C_\mathrm{a}T_\mathrm{a}^{n-2h-1}$ extends the loss to the right band. Now, $T_\mathrm{a}^{m-2h-1}C_\mathrm{a}T_\mathrm{a}^{n-2h-1}C_\mathrm{a}$ corresponds to the loss from the output of the bottom-left corner to the output of the top-right corner, that is the loss of half the border. By squaring this matrix, we obtain the minimum losses over the whole border of the grid: $(T_\mathrm{a}^{m-2h-1}C_\mathrm{a}T_\mathrm{a}^{n-2h-1})^2[S][S]$ means that we compute the minimum loss by starting from the leftmost column of the bottom band, excluding the bottom-left corner, which we suppose is in state $S$, and leaving it in state $S$ after the bottom-left corner computation, after going through all the band in counter-clockwise direction.
\end{proof}

\begin{lemma}
Let $\ell_h(n,m)$ be the minimum loss over the border of height $h$ on a $n \times m$-grid for the 2-domination. Then $\lceil \frac{2nm+\ell_h(n,m)}{6} \rceil$ is a lower bound on $\gamma_2(n,m)$.
\label{lemma-lower-bound-loss-dom-number}
\end{lemma}

\begin{proof}
When computing the loss over some part of the grid, we obtain a lower bound on the loss of the whole grid. The same applies when we compute the minimum loss. We then replace $\ell(n,m)$ by $\ell_h(n,m)$ in \Cref{eq-reversed-loss-2dom}. Since $\gamma_2(n,m)$ is an integer, we can take the upper bound.
\end{proof}

We now try to find some $h$ for which the minimum loss over the border of height $h$ matches the minimum loss over the grid. In the rest of this section, \textbf{we consider that $\mathbf{h = 6}$}, unless explicitly specified otherwise. This value is sufficient to obtain the correct bounds with our program. Here again, we have the problem of computing the minimum loss over borders of arbitrary widths. However, we may notice that, if we let $H(n,m) = (T_\mathrm{a}^{m-13}C_\mathrm{a}T_\mathrm{a}^{n-13})^2$, then there exist some $j_0, k$ and $p$ such that $\forall\; r \geq r_0, T_\mathrm{a}^{r+k} = T_\mathrm{a}^{r}+p$, so that $H(n+i, m+j) = H(n,m) + 2(i+j)p$ for all $ n,m \geq 13+r_0$. Indeed, the matrix $T_\mathrm{a}$ is primitive for the same reasons as for the transfer matrix of \Cref{th-primitive-rec}. The factor 2 before $(i+j)p$ comes from the fact that the matrix $T_\mathrm{a}$ appears twice for the horizontal bands and twice for the vertical onew.

With the program considering a band of height 6, we find that
\begin{equation}\text{for all }r \geq 20,\; T_\mathrm{a}^{r+3} = T_\mathrm{a}^r + 6.
\label{rec-loss-2dom}
\end{equation}

Note that if we choose a border of height 7, the same recurrence relation on $T_\mathrm{a}$ is true from $r \geq 17$. From \Cref{rec-loss-2dom} and because $\ell(n,m)$ is symmetric, we deduce:

\begin{claim}
$\ell_6(n+3,m) = \ell_6(n,m+3) = \ell_6(n,m)+12$ for every $n,m \geq 33$.
\label{claim-loss-rec}
\end{claim}

The recurrence begins at 33 because in \Cref{lemma-compute-loss-matrix}, $T_\mathrm{a}$ is put to the power $n-2h-1$ for instance, so that to have $n-2h-1 \geq 20$ we need to have $n \geq 20+2h+1 = 33$ for $h=6$. Despite \Cref{rec-loss-2dom}, it is 12 that we add and not 6, because in the formula in \Cref{lemma-compute-loss-matrix}: as we wrote just above, there are two vertical bands and two horizontal ones. The formula can be rewritten as $\ell_6(n,m) = \ell_6(n,m-3)+12$ for every $n \geq 36$.
So, to complete the proof, we must check that the formula holds for $n,m \leq 36$ to initialise the recurrence. We check with the method for fixed height the values for $9 \leq n \leq 12$ because the loss method can only by used when $2h < \min(n,m)$. This means that, when the height is 6, we must check the $\gamma_2$ values for $n < 13$ by another method. Then, we can compute our lower bound for $\gamma_2(n,m)$ values for $13 \leq n,m \leq 36$ using the loss method with height 6. These values match the upper bound we show below, so we obtain the values of $\gamma_2(n,m)$.

Now the recurrence relation on $T_\mathrm{a}^{r}$ from \Cref{rec-loss-2dom} and exploited in \Cref{claim-loss-rec} completes the proof of the theorem. Indeed, if $13 \leq n,m \leq 33$ then for all $k \in \NN$:
\begin{align*}
\gamma_2(n,m+3k) \geq \frac{2n(m+3k)+\ell(n,m+3k)}{6} \geq & \frac{2n(m+3k)+\ell_6(n,m+3k)}{6} \\ \geq & \frac{2nm+\ell_6(n,m)}{6} + nk+2k\\ \geq & \left\lfloor \frac{(n+2)(m+2)}{3}-6 \right\rfloor +nk +2k \\ \geq & \left\lfloor \frac{(n+2)(m+3k+2)}{3}-6 \right\rfloor.
\end{align*}

The first inequality comes from \Cref{lemma-lower-bound-loss-dom-number}. The transition from the first line to the second one comes from \Cref{claim-loss-rec}.
The transition from the second to the third line comes from the fact that we checked values of $\gamma_2$ for $n$ and $m$ between 13 and 33 and they match the equality we use to substitute $(2nm+\ell_6(n,m))/6$.
This proves the lower bound for every $13 \leq n \leq 36$ and $m \in \NN$. To prove it for $n > 33$, it suffices to do the exact same computation, namely computing $\gamma_2(n+3k, m)$ for any $n \geq 36$ and any $m \geq 2 \cdot 6 +1 = 13$.

\begin{figure}[h!]
\begin{center}
\includegraphics[scale=0.38]{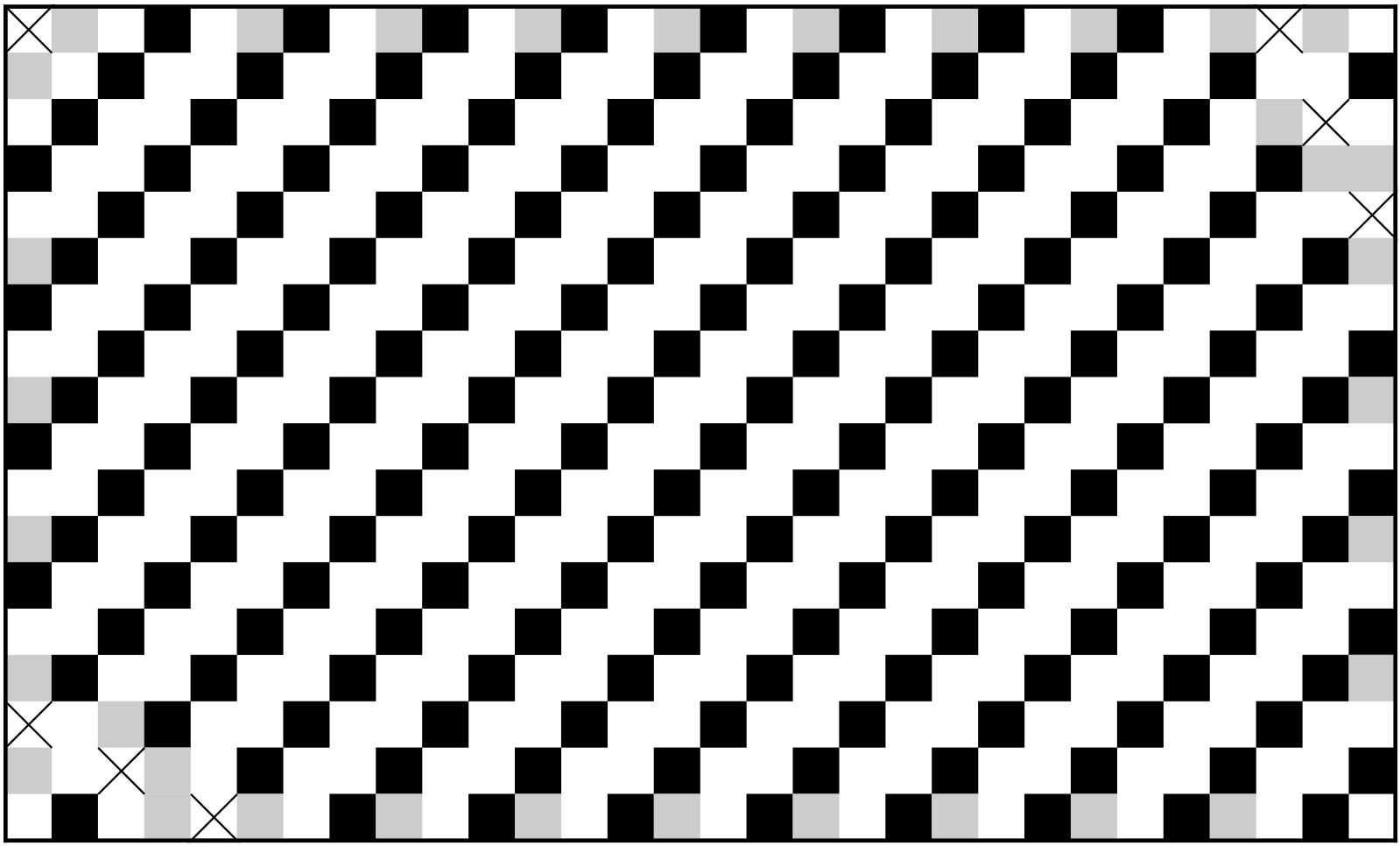}
\caption{Example of an optimal 2-dominating set $D$ on a $18 \!\times\! 30$ grid. $D$ is the set of cells which are filled with grey or black. The black cells and the cells with a cross constitute the projection of a minimal 2-dominating set on the grid $\ZZ^2$.}
\label{ex-dominating-set}
\end{center}
\end{figure}

\begin{figure}[H]
\centering
\includegraphics[scale=1]{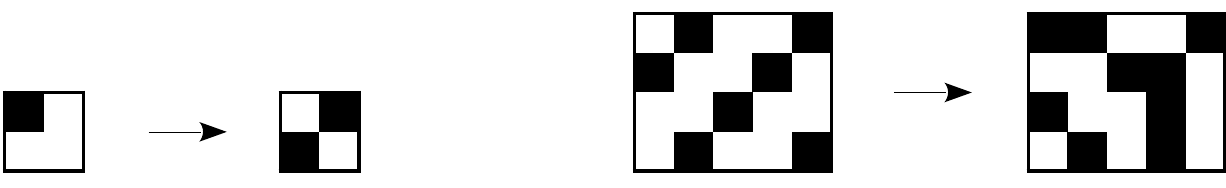}
\caption{The two rules used to convert a restriction of an optimal 2-dominating set of $\ZZ^2$ into optimal 2-dominating set of a rectangle. For each corner, if one of the two patterns before the arrows appear, we replace them by the version on the right of the arrow.}
\label{figure-rules-2dom}
\end{figure}

To show that this bound is sharp, we show that our lower bound is also an upper bound, by giving general 2-dominating sets of the right sizes. To construct these 2-dominating sets we consider, for the infinite grid $\ZZ^2$, the 2-dominating set $D = \{(j,i) : i+j \mod 3 = 0\}$ and its rotations. We then take all the different restrictions of these 2-dominating sets for $\ZZ^2$ into a finite $n \!\times\! m$ grid. For each restriction, we modify each corner of size 6 according to two rules which depend on the pattern of that corner. The two rules are shown in \Cref{figure-rules-2dom}. A rule corresponds to removing some cells and adding some other cells to $D$ in that corner. For instance, Rule~1 could be stated as follows: if the cell at the angle of the grid is in the dominated set, we remove it from the set and add instead its two neighbours. Finally, we put a stone on the cells of the first and last rows and columns which are not 2-dominated. One can show that for $13 < n \leq m$ one of the resulting 2-dominating $D_{n,m}$ set has the right size. We can see an example of such a $D_{n,m}$ for a $18 \!\times\! 30$ grid in \Cref{ex-dominating-set}. The first rule is used in the top-left corner and the second rule is used in the top right and bottom left corner. No modifications need to be done in the bottom right corner. By counting the number of stones in the regular pattern (black and crossed cells in \Cref{ex-dominating-set}), removing the number of crossed cells, and adding the number of grey cells, we get $D = \frac{nm+2n+2m}{3}-5$, which is equal to the number in \Cref{th-2dom} when $n$ and $m$ are multiple of 3.

The grid we show is of size $18 \!\times\! 30$, but it extends immediately to any $n\!\times\! m$ grid when $n$ and $m$ are both greater than 14 and multiple of 3. Applying the same method for $14 \leq n,m$ when the two numbers have other congruences modulo 3 leads to 2-dominating sets having the right size.

\section{Adaptation to other problems and results}
\label{section-other-problems-dom}

In this section we explain the adaptations needed to make the method presented in the previous section work for one other type of domination, namely the \emph{Roman domination}. We will give the corresponding theorem obtained by applying this method. We will also give partial results on some dominations problems, namely the \emph{distance-2 domination} and the \emph{total domination}, for which we could not get the method to fully work. We will conjecture the possible reasons explaining the partial failure of the method to obtain a closed formula for arbitrary heights and widths of the grid. We also give a lower bound for the total domination.

\subsection{The Roman domination}

We explain here how the code was adapted to the Roman domination to obtain the following result.

\begin{thm}[\cite{rao-talon}]
The Roman domination number is such that, for all $1 \leq n \leq m$,\\
\[\gamma_\emph{R}(n,m) =    \left\{
\setstretch{1.5}
\begin{array}{ll}
      \left\lceil\frac{2m}{3}\right\rceil & \quad\textrm{if }n = 1 \\
      m+1 & \quad\textrm{if }n=2 \\
      \left\lceil \frac{3m}{2} \right\rceil& \quad\textrm{if } n=3\textrm{ and } m\mod 4 = 1 \\
      \left\lceil \frac{3m}{2} \right\rceil+1& \quad\textrm{if } n=3\textrm{ and } m\mod 4 \neq 1 \\
      2m+1& \quad\textrm{if } n=4\textrm{ and } m = 5 \\
      2m& \quad\textrm{if } n=4\textrm{ and } m > 5 \\
      \left\lfloor\frac{12m}{5}\right\rfloor+2& \quad\textrm{if } n=5\\
      \left\lfloor\frac{14m}{5}\right\rfloor+2& \quad\textrm{if } n=6\textrm{ and } m\mod 5 \in \{0,3,4\}\\
      \left\lfloor\frac{14m}{5}\right\rfloor+3& \quad\textrm{if } n=6\textrm{ and } m\mod 5 \notin \{0,3,4\}\\
      \left\lfloor\frac{16m}{5}\right\rfloor+2& \quad\textrm{if } n=7 \textrm{ and } m = 7 \textrm{ or }m\mod 5 = 0\\
      \left\lfloor\frac{16m}{5}\right\rfloor+3& \quad\textrm{if } n=7 \textrm{ and } (m > 7\textrm{ and }m\mod 5 \neq 0)\\
      \left\lfloor\frac{18m}{5}\right\rfloor+4& \quad\textrm{if } n=8 \textrm{ and } m \mod 5 = 3\\
      \left\lfloor\frac{18m}{5}\right\rfloor+3& \quad\textrm{if } n=8 \textrm{ and } m \mod 5 \neq 3\\
      
      \left\lfloor \frac{2(n+1)(m+1)-2}{5}\right\rfloor-1 & \quad\textrm{if } n\geq 9 \textrm{ and } n\mod 5 = 4 \textrm{ and } m \mod 5 = 4\\
      \left\lfloor \frac{2(n+1)(m+1)-2}{5}\right\rfloor & \quad\textrm{if } n\geq 9 \textrm{ and } n\mod 5 \neq 4 \textrm{ or } m \mod 5 \neq 4
\end{array}
\right. \]
\label{th-Roman-domination}
\end{thm}

 First, for this problem the set of values for a cell is 
$\SSS = \{$ \textsc{two\_stones}, \textsc{stone}, \textsc{ok}, \textsc{need\_one} $\}$. \textsc{stone} means that we put a troop (here we will talk of stones instead) on the cell, so it does not need to be dominated by another cell. \textsc{two\_stones} means that we put two troops on the cell, so it dominates its neighbours. Now a state $S$ is in $\VV$ if and only if for each $i \in \llbracket 0, n-1 \rrbracket$:

\begin{itemize}[noitemsep, topsep=0pt]
\item if $S[i] = $ \textsc{need\_one} then neither $S[i-1]$ nor $S[i+1]$ is \textsc{two\_stones};
\item if $S[i] = $ \textsc{stone} then neither $S[i-1]$ nor $S[i+1]$ is \textsc{two\_stones} or \textsc{stone};
\end{itemize}
Note that the second rule is not required for the consistency of the state, but it is an optimisation which allows us to reduce a lot the number of states, as we will see in \Cref{section-rauzy}. It is justified by the fact that in a minimum Roman dominating set, we can always remove any stone neighbouring a cell with two stones, and if there are two neighbouring cells with a stone each we still have a dominating set of same value by removing one of the stones and putting a second stone on the other cell. This implies that there exist minimum Roman dominating sets matching the extra rules we enforce.

\begin{fact}[see \Cref{section-rauzy}]
There are $\Theta(x_0^n) \approx \Theta(2.956295^n)$ states in $\VV$ for $n$ lines, where $x_0$ is the largest real root of the polynomial $x^4 - 3x^3 - x^2 + 3x + 1$.
\end{fact}

A state $S \in \VV$ is in $\FF$ if and only if for every $i \in \llbracket 0, \ldots, n-1 \rrbracket$, if $S[i] =$ \textsc{ok} then at least one among $S[i-1]$ and $S[i+1]$ is \textsc{two\_stones}.

$(S, S')$ is a compatible pair if and only if for $i \in \llbracket 0, \ldots, n-1 \rrbracket$:
\begin{itemize}[noitemsep, topsep=0pt]
\item if $S[i] =$ \textsc{need\_one} then $S'[i] =$ \textsc{two\_stones};
\item if $S'[i] =$ \textsc{need\_one} then $S[i] \neq$ \textsc{two\_stones};
\item if $S'[i] =$ \textsc{ok} then at least one among $S[i]$, $S'[i-1]$ and $S'[i+1]$ is \textsc{two\_stones};
\item if $S[i] \in \{\text{\textsc{two\_stones}, \textsc{stone}}\}$ then $S'[i] \neq$ \textsc{stone};
\item if $S[i] =$ \textsc{stone} then $S'[i] \neq$ \textsc{two\_stones}.
\end{itemize}
\vspace{5pt}
Finally, a state $S \in \VV$ is in $\EE$ if and only if none of its entry is \textsc{need\_one}. We may notice that we do not enumerate every $f_j(D)$: we forbid for instance two neighbouring cells when each contains one or two stones. This is possible because there exist Roman-dominating sets of minimum size which do not contain this pattern. We will discuss this optimisation in \Cref{section-experimental-dom}. Here we can always avoid this pattern by removing the single stone of one cell. If the other one also had a single stone, we add one stone to it..\\

We now need to adapt the loss. We define, for the Roman domination, \[\ell(n,m) = 5|S_2|+\frac{5}{2}|S_1|-nm = \frac{5}{2}(2|S_2|+|S_1|) -nm.\] Indeed, each cell with two stones dominates five cells, and each cell in $S_1$ dominates only itself, but we add to it an additional loss of $3/2$ to penalize its bad ratio of number of dominated cells compared to number of stones used. This allows us to obtain \[\gamma_\textrm{R}(n,m) \geq \frac{2}{5}(\ell(n,m)+nm).\] Note that in the program, what we compute is actually $2\ell(n,m)$ to avoid to manipulate fractions or floating numbers.
Let us define the almost-valid states which, for this problem, coincide with the valid states: $\VV_a = \VV$.
Now if $S, S' \in \VV_a$, $S \RRR_a S'$ if and only if for $i \in \llbracket 0, n-1 \rrbracket$:
\begin{itemize}[noitemsep, topsep=0pt]
\item if $S[i] =$ \textsc{need\_one} then $S'[i] =$ \textsc{two\_stones};
\item if $S'[i] =$ \textsc{need\_one} then $S[i] \neq$ \textsc{two\_stones};
\item if $S'[i] =$ \textsc{ok} and $i \neq 0$ then at least one among $S[i]$, $S'[i-1]$ and $S'[i+1]$ is \textsc{two\_stones};
\item if $S[i] \in \{\text{\textsc{two\_stones}, \textsc{stone}}\}$ then $S'[i] \neq$ \textsc{stone};
\item if $S[i] =$ \textsc{stone} then $S'[i] \neq$ \textsc{two\_stones}. 
\end{itemize}

Here we do not give complete details on how we compute the loss. Each cell with two stones having $k < 4$ neighbours with one or two stones contributes for $k$, and each cell dominated by $k > 1$ cells (with two stones) also contributes for $k-1$. Finally, each stone with one cell contributes for $3/2$. All these contributions sum up to make the loss. We recall that in the program we compute twice these values.\\

As in the previous section, we get exact values for "small" values of $n$, and a lower bound for bigger values of $n$. This time, we obtain the following recurrence relation:
\begin{equation} \text{for all }r \geq 12, T_\mathrm{a}^{r+5} = T_\mathrm{a}^r + 5.
\label{rec-loss-roman}
\end{equation}
We prove the formula for an arbitrary number of columns and lines the same way as for the 2-domination problem, using \Cref{rec-loss-roman} and the definition of the loss function. We conclude that our lower bound is the exact value of $\gamma_\textrm{R}$ thanks to the thesis of Currò (see ~\cite[Chapter 4, Theorem 10]{curro}). Indeed, he showed some upper bound for the Roman-domination number. The lower bound we find is the same as his upper bound, hence both are sharp and are the Roman-domination number.

\subsection{The total domination}
We present here the details for the total domination. Unfortunately, for reasons we will discuss in \Cref{section-conj-dom}, we were not able to find the values for grids of arbitrary size, so we give only partial results. The total domination in grids was studied by Gravier~\cite{gravier}. He gives the values for up to 4 lines, and provides some lower and upper bound. We improve his lower bound. Crevals and Ostergård~\cite{total-dom-article-28}, on their side, gave values up to 28 lines, which is more than we do (up to 15 lines for us).\\

First, we can see that the domination, 2-domination and total domination are in fact part of a more general class of problems:

\begin{deff}
A set $S \subset V$ is $\bm{(a,b)}$\textbf{-dominating} a graph $G = (V,E)$ when any vertex $v$ in $S$ has at least $a$ neighbours in $S$ and every vertex outside $S$ has at least $b$ neighbours in $S$.
\end{deff}

It is clear that the domination is the (0,1)-domination while the total domination is the $(1,1)$-domination. Given this fact, storing information about whether a cell has a stone, or whether one or zero of its "current" neighbours have one is no longer enough. Fortunately, we can encode in a cell value the number of neighbouring stones, or alternately the number of stones it lacks to be $(a,b)$-dominated, plus the knowledge of whether or not it contains a stone. Depending on the problem we consider, we may need more or less information: for the total domination for instance, we also need to store for the cells with a stone whether or not they are dominated by another cell.

One way to encode the necessary information for an $(a,b)$-domination problem is to have the following set for the cells values: \textsc{stone\_prev, stone, none\_prev, none}. The "\textsc{\_prev}" suffix means that the neighbouring cell from the previous column also contains a stone. "\textsc{none}" means that the cell does not contain a stone itself. We then compute which states are valid and the compatibility relations just from that. Indeed, we can recover how many times a cell is dominated from this piece of information.

However, as we will see in details in \Cref{section-rauzy}, we can make the computations a lot faster by choosing carefully how to encode the necessary information. So the set of cell values is $\SSS = \{$\textsc{stone\_ok, stone\_need\_one, stone\_need\_two, ok, need\_one, need\_two}$\}$.
Some of these values are not necessary: for the computations of the exact values, when the number of lines is fixed, any cell would have at most one new neighbour (in the next column). This makes the values ending by '\textsc{\_two}' useless: they are used only for the loss computation. We do not give explicitly the set of valid states, first states, ending states and compatibility relation: the logic behind them is very similar to the one for the domination, and they are present in the code (see respectively functions \texttt{is\_state\_valid, can\_cell\_neighbourhood\_be\_first, is\_state\_dominated} and \texttt{are\_state\_compatible} in the source code).

\begin{fact}[see \Cref{section-rauzy}]
There are $\Theta(x_0^n) \approx \Theta(2.618034^n)$ states in $\VV$ for $n$ lines, where $x_0$ is the real root of the polynomial $x^4 - 3x^3 + 3x - 1$.
\end{fact}

As we warned above, we did not manage to get a closed formula which would work for every value of $n$ and $m$. Hence we give here values for small number of lines (up to 15), and some bounds on the quantity $\gamma_\textrm{T}$ when the number of lines is arbitrary. The lower bound is obtained by the same loss method, and again the transfer matrix we use for the bands verifies some recurrence property, hence we can extend our lower bound for an arbitrary number of lines and columns. However, this lower bound does not seem to match the actual value: it seems to increase little by little, and we exhaust the computing resources while it still wants to grow\footnote{like a poor vegetable running out of water}.

\begin{thm}
For $1 \leq n \leq 15$ and any $m \geq n$, the following equalities about the total domination number hold:
\allowdisplaybreaks
\begin{fleqn}
\begin{align*}
\gamma_\emph{T}(1,m) &= \left\lbrace\setstretch{1.5}
		\begin{array}{ll} 
      \left\lfloor\frac{m}{2}\right\rfloor & \text{ if } m \mod 4 = 0 \\
      \left\lfloor\frac{m}{2}\right\rfloor+1 & \text{ otherwise}\\
      \end{array}
      \right.\\
\gamma_\emph{T}(2,m) &= \left\lbrace\setstretch{1.5}
		\begin{array}{ll} 
      \left\lfloor\frac{2m+2}{3}\right\rfloor+1 & \text{ if } m \mod 3 = 1 \\
      \left\lfloor\frac{2m+2}{3}\right\rfloor & \text{ otherwise}\\
      \end{array}
      \right.\\
\gamma_\emph{T}(3,m) &= n\\
\gamma_\emph{T}(4,m) &= \left\lbrace\setstretch{1.5}
		\begin{array}{ll} 
      \left\lfloor\frac{6m+3}{5}\right\rfloor+2 & \text{ if } m \mod 5 \in \{0,3\} \\
      \left\lfloor\frac{6m+3}{5}\right\rfloor+1 & \text{ otherwise}\\
      \end{array}
      \right.\\
\gamma_\emph{T}(5,m) &= \left\lbrace\setstretch{1.5}
		\begin{array}{ll} 
      \left\lfloor\frac{6m+3}{4}\right\rfloor+2 & \text{ if } m \mod 4 = 0 \\
      \left\lfloor\frac{6m+3}{4}\right\rfloor+1 & \text{ otherwise}\\
      \end{array}
      \right.\\
\gamma_\emph{T}(6,m) &= \left\lbrace\setstretch{1.5}
		\begin{array}{ll} 
      \left\lfloor\frac{12m}{7}\right\rfloor+4 & \text{ if } m \mod 7 = 5 \\
      \left\lfloor\frac{12m}{7}\right\rfloor+3 & \text{ if } m \mod 7 \in \{1,2,3\} \\
      \left\lfloor\frac{12m}{7}\right\rfloor+2 & \text{ otherwise}\\
      \end{array}
      \right.\\
\gamma_\emph{T}(7,m) &= \left\lbrace\setstretch{1.5}
		\begin{array}{ll} 
      2m+2 & \text{ if } n \mod 2 = 0 \text{ or } m \in \{9,11,15,21\} \\
      2m+1 & \text{ otherwise}\\
      \end{array}
      \right.\\
\gamma_\emph{T}(8,m) &= \left\lbrace\setstretch{1.5}
		\begin{array}{ll} 
      \left\lfloor\frac{20m+6}{9}\right\rfloor+4 & \text{ if } m \mod 9 \in \{0,7\} \text{ and } n \notin \{9,16\} \\
      \left\lfloor\frac{20m+6}{9}\right\rfloor+3 & \text{ if } m \mod 9 \in \{2,3,4,5\}\\
      \left\lfloor\frac{20m+6}{9}\right\rfloor+2 & \text{ otherwise}\\
      \end{array}
      \right.\\
\gamma_\emph{T}(9,m) &= \left\lbrace\setstretch{1.5}
		\begin{array}{ll} 
      \left\lfloor\frac{10m+3}{4}\right\rfloor+3 & \text{ if } m \mod 4 = 2 \\
      \left\lfloor\frac{10m+3}{3}\right\rfloor+2 & \text{ otherwise}\\
      \end{array}
      \right.\\
\gamma_\emph{T}(10,m) &= \left\lbrace\setstretch{1.5}
		\begin{array}{ll} 
      \left\lfloor\frac{30m+1}{11}\right\rfloor+6 & \text{ if } m \mod 11 = 9 \text{ and } m \neq 20 \\
      \left\lfloor\frac{30m+1}{11}\right\rfloor+5 & \text{ if } m \mod 11 \in \{2,5,7\} \text{ and } m \notin \{13,18 \} \\
      \left\lfloor\frac{30m+1}{11}\right\rfloor+4 & \text{ if } m \mod 11 \in \{0,1,3,6\} \text{ or } m = 20 \\
      \left\lfloor\frac{30m+1}{11}\right\rfloor+3 & { otherwise}\\
      \end{array}
      \right.\\
\gamma_\emph{T}(11,m) &= \left\lbrace\setstretch{1.2}
		\begin{array}{ll} 
      3m+4 & \text{ if } m \in \{12,22\}\\
      3m+3& \text{ if } m  \in \{13,15,17,19,23,27,29,33,37,43,47,57\}\\
      3m+2 & \text{ otherwise}\\
      \end{array}
      \right.\\
\gamma_\emph{T}(12,m) &= \left\lbrace\setstretch{1.5}
		\begin{array}{ll} 
      \left\lfloor\frac{42m+9}{13}\right\rfloor+6 & \text{ if } m \mod 13 \in \{0,11\} \text{ and } m \notin \{13,24,26,37\} \\
      \left\lfloor\frac{42m+9}{13}\right\rfloor+5 & \text{ if } m \mod 13 \in \{2,4,7,9\} \text{ and } m \notin \{15,17,20\} \\
      \left\lfloor\frac{42m+9}{13}\right\rfloor+4 & \text{ if } m \mod 13 \in \{3,5,6,8\} \text{ or } m \in \{13,24,26,37\} \\
      \left\lfloor\frac{42m+9}{13}\right\rfloor+3 & \text{ otherwise}\\
      \end{array}
      \right.\\
\gamma_\emph{T}(13,m) &= \left\lbrace\setstretch{1.5}
		\begin{array}{ll} 
      \left\lfloor\frac{14m+3}{4}\right\rfloor+5 & \text{ if } m \in \{14,26\}\\
      \left\lfloor\frac{14m+3}{4}\right\rfloor+4& \text{ if } m \mod 4 = 0 \text{ or } m  = 19 \\
      \left\lfloor\frac{14m+3}{4}\right\rfloor+3 & \text{ otherwise}\\
      \end{array}
      \right.\\
\gamma_\emph{T}(14,m) &= \left\lbrace\setstretch{1.5}
		\begin{array}{ll} 
      \left\lfloor\frac{56m+2}{15}\right\rfloor+8 & \text{ if } m \mod 15 = 13 \text{ and } m \notin \{28, 43\} \\
      \left\lfloor\frac{56m+2}{15}\right\rfloor+7 & \text{ if } m \mod 15 \in \{2,9,11\} \text{ and } m \notin \{17,24,26,32,41\} \\
      \left\lfloor\frac{56m+2}{15}\right\rfloor+6 & \text{ if } (m \mod 15 \in \{0,5,6,7\} \text{ and } m \notin \{15,21,22,30\})\\ &\quad \text{ or } m \in \{28, 43\} \\
            \left\lfloor\frac{56m+2}{15}\right\rfloor+5 & \text{ if } m \mod 15 \in \{1,3,4,10\} \text{ or } m \in \{17,24,26,32,41\} \\
      \left\lfloor\frac{56m+2}{15}\right\rfloor+4 & \text{ otherwise}\\
      \end{array}
      \right.\\
      \gamma_\emph{T}(15,m) &= \left\lbrace\setstretch{1.2}
\begin{array}{ll} 
       4m+6  & \text{ if } m \in \{16,30\} \\
       4m+5  & \text{ if }  m \in \{21,23\} \\
      4m+4  & \text{ if } (m \mod 2 = 0 \text{ and } m \notin \{16, 30\}) \\
      & \quad\text{ or } m \in \{17,19,25,27,31,35,37,39,41,45,49,53,55,59,63,\\
      &\hspace{2cm}\;\;67,73,77,81,91,95,109 \} \\
     4m+3 & \text{ otherwise}\\
      \end{array}
      \right.\\
\end{align*}
\end{fleqn}

\label{th-total-domination}
\end{thm}
\Cref{th-total-domination} confirms the results from Crevals and Ostergård~\cite{total-dom-article-28} for values up to 15. In their paper they managed to go up to $n = 28$, using another approach. They do not enumerate full dominating sets as we do, but have like us some notion of states. However, instead of enumerating the states of the columns, they only enumerate the states of partial dominating sets: they store only the elements of the minimal dominating set. 
Using clever arguments they also manage to cut down on the number of states they enumerate, and it turns out this number grows less fast than in our method.

The loss can be adapted to the generic $(a,b)$-domination problem: each stone contributes to dominating its 4 neighbouring cells. The dominating $|D|$ cells need to be dominated $a$ times and the cells $nm-|D|$ cells not in $D$ need to be dominated $b$ times. When reversing the formula we obtain:
\begin{equation} 
\gamma_{a,b}(n,m) = \frac{b \cdot nm+\ell(n,m)}{4-a+b}. 
\label{formule-a-b-dom}
\end{equation}

With our program to compute the loss, we find that, for a band of height 10:

\begin{equation}
\text{for all }r \geq 31,\; T_\mathrm{a}^{r+22} = T_\mathrm{a}^{r} + 10.
\label{rec-loss-total}
\end{equation}

From this, as for the other problems we studied above, we deduce:

\begin{claim}
$\ell_{10}(n+22,m) = \ell_{10}(n,m+22) = \ell_{10}(n,m)+10$ for every $n,m \geq 52$.
\label{claim-loss-rec-total}
\end{claim}

First, we may notice that our lower bounds obtained thanks to the loss are not too far from the actual values Crevals and Ostergård found. For 28 lines, they find 416 for $m=56$ and 427 for $m=58$ when we find that it must respectively be at least 411 and 426.

Also, the bounds we obtain agree to some extent to the conjecture of~\cite{total-dom-article-28}. Indeed, their formulas imply that, when $m \mod 4 \in \{1,3\}$ then $\gamma_\mathrm{T}(n,m) = \Theta(\frac{nm+n+m}{4})$ and our values imply 
\begin{equation}
\label{total-lower-bound}
\gamma_\mathrm{T}(n,m) \geq \frac{nm+10/11(n+m)}{4}+O(1).
\end{equation}
We can even go a bit further on the constant after this equivalent. By using \Cref{claim-loss-rec-total} we deduce that $\ell_{10}(n,m) \geq 2\frac{10}{22}(n+m) + c = \frac{10}{11}(n+m)+c(n \mod 22, m \mod 22)$, where the $c(i,j)$ for $0 \leq i,j \leq 22$ are some constants depending on the actual values of $\ell_{10}(n,m)$. The factor 2 comes from the fact that a rectangle has two vertical bands and two horizontal bands, as we mentioned for the 2-domination. We determine, thanks to the values of the loss of height 10 for $52 \leq n,m \leq 74$, a lower bound on the $c(i,j)$ constants.
From what just precedes and by \Cref{formule-a-b-dom}, we obtain:
\begin{prop} 
\[\gamma_\mathrm{T}(n,m) \geq \frac{nm+\frac{10}{11}(n+m)}{4}-1.\]
\end{prop} 

We may even notice that the fraction which is multiplied by $(n+m)$ in \Cref{total-lower-bound} increases when the height increases and may converge towards $1$, which is the values from their conjecture. Indeed, the fractions we obtain equal $6/7$ for 6 and 7 lines, $8/9$ for 8 and 9 lines, and $10/11$ for 10 lines. It may even be possible that it takes all the values of the shape $2l/(2l+1)$ when the height becomes arbitrarily big.

\begin{conj}
\[\gamma_\mathrm{T}(n,m) = \frac{nm+n+m}{4}+O(1).\]
\end{conj}

\subsection{The distance-2 domination}
As was written in the introduction, a grid is distance-2 dominated by $S$ when any vertex not in $S$ is at distance at most two of an element of $S$. The more general distance-$k$-domination problem was studied by Farina and Grez~\cite{distance-k-domination} who proved some upper bound on the associated domination number.\\

Here again we did not manage to get a closed formula for arbitrary numbers of lines and columns. The problem is not the same as for the total domination: we more likely just lacked of a bit a computing resources instead of the problem being much more difficult to tackle. Indeed, since now a vertex may be dominated by a vertex at distance two, we need to store more information on previous columns: not just some information about the previous columns, but also some about the one even before.

The set of cell values is $\SSS = \{$ \textsc{stone\_prev, stone, ok\_prev, ok, need\_dist\_two},\linebreak \textsc{need\_dist\_one} $\}$. The "prev" suffix, here again, means that the cell of the previous column has a stone. \textsc{ok} means that the cell is dominated, whereas \textsc{need\_dist\_two} means that the cell is not dominated so that it requires a cell at distance at most two with a stone. \textsc{need\_dist\_one} is similar but it means here that the cell in the previous column was not dominated, hence the next cell needs to have a stone to dominate this ante-predecessor cell. As for the total domination we do not detail the other special sets and the compatibility relation, which are only a matter of logic and optimisations, and can be found in the source code (see the functions \texttt{is\_state\_valid, can\_cell\_neighbourhood\_be\_first, is\_state\_dominated} and \texttt{are\_state\_compatible}).

\begin{fact}[see \Cref{section-rauzy}]
There are $\Theta(x_0^n) \approx \Theta(2.958770^n)$ states in $\VV$ for $n$ lines, where $x_0$ is the largest real root of the polynomial $x^{24} - 5x^{23} + 6x^{22} + 2x^{21} - 7x^{20} + 6x^{19} - 8x^{18}\\ + 8x^{17} + 4x^{16} - 13x^{15} + 5x^{14} + 6x^{13} + 8x^{12} - 14x^{11} - 17x^{10} + 14x^{9} - 8x^{8} - 10x^{7} - 8x^{6} + 5x^{5}\\ + 9x^{4} - x^{3} - 5x^{2} + 4x - 1$.
\end{fact}

As for the total domination, we give the formulas for small values of $n$.

\begin{thm}

For $1 \leq n \leq 14$ and any $m \geq n$, the following equalities about the distance-two-domination number hold:
\allowdisplaybreaks
\begin{fleqn}
\begin{align*}
\gamma_\emph{d2}(1,m) &=  \left\lceil \frac{m}{5} \right\rceil\\
\gamma_\emph{d2}(2,m) &=  \left\lfloor \frac{m}{4}+1 \right\rfloor\\
\gamma_\emph{d2}(3,m) &=  \left\lceil \frac{m}{3} \right\rceil\\
\gamma_\emph{d2}(4,m) &= \left\lbrace\setstretch{1.5}
		\begin{array}{ll} 
      \left\lfloor \frac{3m}{7} \right\rfloor+1 & \quad\text{if } m \mod 7 \in \{0,1,3,5\} \\
      \left \lfloor \frac{3m}{7} \right\rfloor+2& \quad\text{otherwise}  \\
      \end{array}
      \right.\\
\gamma_\emph{d2}(5,m) &= \left\lbrace\setstretch{1.5}
		\begin{array}{ll} 
      \left \lfloor \frac{m+1}{2} \right\rfloor& \quad\textrm{if } m\mod 6 = 1 \\
      \left \lfloor \frac{m+1}{2}\right\rfloor + 1& \quad\textrm{otherwise} \\
      \end{array}
      \right.\\
\gamma_\emph{d2}(6,m) &= \left\lbrace\setstretch{1.5}
		\begin{array}{ll} 
      \left \lfloor \frac{3m}{5} \right\rfloor+1& \quad\textrm{if }  m \mod 5 \neq 3 \\
      \left \lfloor \frac{3m}{5} \right\rfloor+2& \quad\textrm{otherwise}\\      
      \end{array}
      \right.\\
\gamma_\emph{d2}(7,m) &= \left\lbrace\setstretch{1.5}
		\begin{array}{ll} 
      7& \quad\textrm{if } m = 9 \\
      \left \lfloor \frac{2m}{3} \right\rfloor+2& \quad\textrm{otherwise} \\
      \end{array}
      \right.\\
\gamma_\emph{d2}(8,m) &= \left\lbrace\setstretch{1.5}
		\begin{array}{ll} 
      12& \quad\textrm{if } m = 13 \\
      \left \lfloor \frac{3m}{4} \right\rfloor+1& \quad\textrm{if } m \mod 8 \in \{4, 7\}\\
            \left \lfloor \frac{3m}{4} \right\rfloor+2& \quad\textrm{otherwise}\\
      \end{array}
      \right.\\
\gamma_\emph{d2}(9,m) &= \left\lbrace\setstretch{1.5}
		\begin{array}{ll} 
      \left \lfloor \frac{5m}{6} \right\rfloor+1& \quad\textrm{if } m \in \{ 11, 18\} \\
      \left \lfloor \frac{5m}{6} \right\rfloor+3& \quad\textrm{if } m \mod 6 \in \{2,3,9\} \text{ and } m \notin \{ 4,15,20,21,27,32,39\} \\
      \left \lfloor \frac{5m}{6} \right\rfloor+2& \quad\textrm{ otherwise} \\
      \end{array}
      \right.\\
\gamma_\emph{d2}(10,m) &= \left\lbrace\setstretch{1.5}
		\begin{array}{ll} 
      \left \lfloor \frac{10m}{11} \right\rfloor+3& \quad\textrm{if } m \mod 11 \in \{ 2,3,5,8\} \\
      \left \lfloor \frac{10m}{11} \right\rfloor+2& \quad\textrm{otherwise} \\
      \end{array}
      \right.\\   
      \gamma_\emph{d2}(11,m) &= \left\lbrace\setstretch{1.5}
		\begin{array}{ll} 
      m+1& \quad\textrm{if } m \mod 30 \in \{1,4,6,7,9,11,14,16,17,19,21,24,26,27,29,30\} \\
      m+2& \quad\textrm{otherwise} \\
      \end{array}
      \right.\\   
\gamma_\emph{d2}(12,m) &= \left\lbrace\setstretch{1.5}
		\begin{array}{ll} 
      \left \lfloor \frac{15m}{14} \right\rfloor+4& \quad\textrm{if } m \mod 14 = 11 \textrm{ and } m \neq 25 \\
      \left \lfloor \frac{15m}{14} \right\rfloor+2& \quad\textrm{if } m \mod 14 \in \{1,4,7\} \textrm{ or } m \in \{14,17,19,28\} \\
      \left \lfloor \frac{15m}{14} \right\rfloor+3 & \quad\textrm{otherwise} \\
      \end{array}
      \right.\\
\gamma_\emph{d2}(13,m) &= \left\lbrace\setstretch{1.5}
		\begin{array}{ll} 
      \left \lfloor \frac{15m}{13} \right\rfloor+4& \quad\textrm{if } m \mod 13 = 5 \text { and } n \neq 31\\
      \left \lfloor \frac{15m}{13} \right\rfloor+3& \quad\textrm{if } (m \neq 13 \text{ and } \mod 13 \in \{0,2,3,6,8,10,11,12\}) \text { or } n = 31\\
      \left \lfloor \frac{15m}{13} \right\rfloor+2 & \quad\textrm{otherwise} \\
      \end{array}
      \right.\\   
\gamma_\emph{d2}(14,m) &= \left\lbrace\setstretch{1.5}
		\begin{array}{ll} 
      \left \lfloor \frac{21m}{17} \right\rfloor+2& \quad\textrm{if } m \mod 17 = 1 \text { or } n \in \{23,30,47\}\\
            \left \lfloor \frac{21m}{17} \right\rfloor+4& \quad\textrm{if }m = 36 \text{ or } (m > 46 \text{ and } n \mod 17 \in \{2,3,8,11,14,16\} \\
       & \quad\quad\quad\quad\quad\quad\quad\quad\quad\quad\;\text{ and } n \notin \{54,59,71\})  \\
      \left \lfloor \frac{21m}{17} \right\rfloor+3 & \quad\textrm{otherwise} \\
      \end{array}
      \right.\\  
\gamma_\emph{d2}(15,m) &= \left\lbrace\setstretch{1.5}
		\begin{array}{ll} 
      \left \lfloor \frac{21m}{16} \right\rfloor+2& \quad\textrm{if } m \mod 16 \in \{1,4,7\} \\
            \left \lfloor \frac{21m}{16} \right\rfloor+4& \quad\textrm{if } m \mod 16 \in \{2,3,5,8\} \text{ and } n \notin \{19,21\}\\
      \left \lfloor\frac{21m}{16} \right\rfloor+3 & \quad\textrm{otherwise}\\
      \end{array}
      \right.\\
\end{align*}
\end{fleqn}

\label{th-dist2-domination}

\end{thm}

\newpage
We adapt again the loss function: now a cell contributes to the domination of 12 cells, and each cell not in the dominating set needs to be dominated once. We then obtain:
\[\gamma_{\mathrm{2d}}(n,m) = \frac{nm+\ell(n,m)}{13}.\]

However, we were not able to find good bounds with the method we used for the other problems.

\section{Conjectures about why the method works}
\label{section-conj-dom}
As we said earlier, the method for a fixed number of lines should work for any problem the properties of which can be checked locally, that is by the means of a finite list of forbidden patterns.
Some authors investigated the problem of domination in Cartesian products of cycles (see for instance~\cite{kla, pav}). The first part of the technique (when $n$ is fixed and small) may be adapted, as stated by \Cref{th-meta-grids} but the second part (for arbitrary number of lines) does not apply directly since a crucial property is that the loss can be concentrated inside the \emph{borders} of the grids.

Some necessary and sufficient conditions for the loss method (the one described in \Cref{section-big-number-lines}, for arbitrary large number of lines) to work are yet to be discovered.
We try here to infer what these conditions might be by giving some properties we believe to be related to the effectiveness of the method. 

We believe that the reasons why the method gives sharp bounds can be expressed as some tiling properties.
Indeed, the domination problems are related to covering problems. For instance, \Cref{domination-shape} shows the shape associated to the domination problem. A smallest dominating set in a grid is equivalent to a smallest covering set of the rectangle with this shape. The method of Gonçalves et al. works thanks to the fact that the shape has the following two properties. First, it can tile (that is, cover without overlaps) the infinite plane. Second, we can find optimal solutions which consist of projecting a tiling of the plane, cropping it and modifying only tiles at bounded distance from the border.

\begin{figure}
\centering
\includegraphics[scale=0.3]{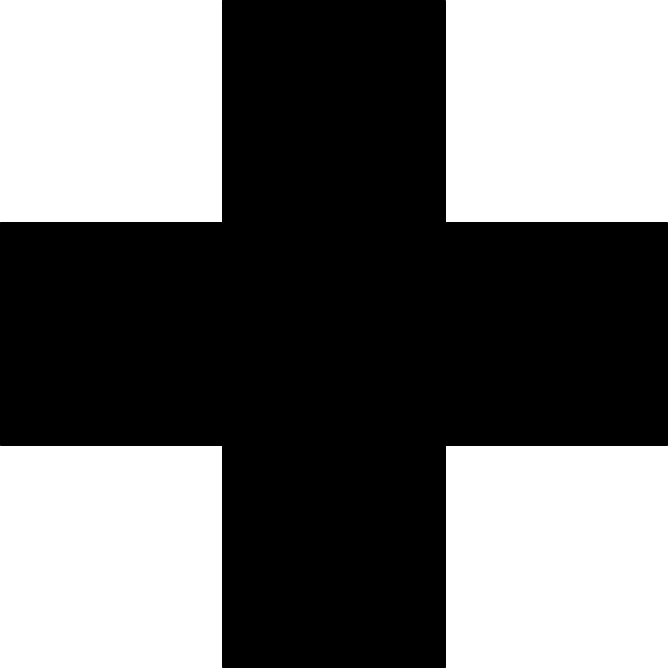}
\caption{The shape corresponding to the domination tiling problem.}
\label{domination-shape}
\end{figure}
In the case of the 2-domination and the Roman domination, it is not properly speaking a covering problem, but a generalised covering problem with some weights (see \Cref{2-roman-figures}). The properties we write below are rather focused on standard tilings than on generalised tiling.
\begin{figure}[h]
\hspace{-1.7cm}
\begin{subfigure}[b]{0.6\textwidth}
\includegraphics[scale=0.6]{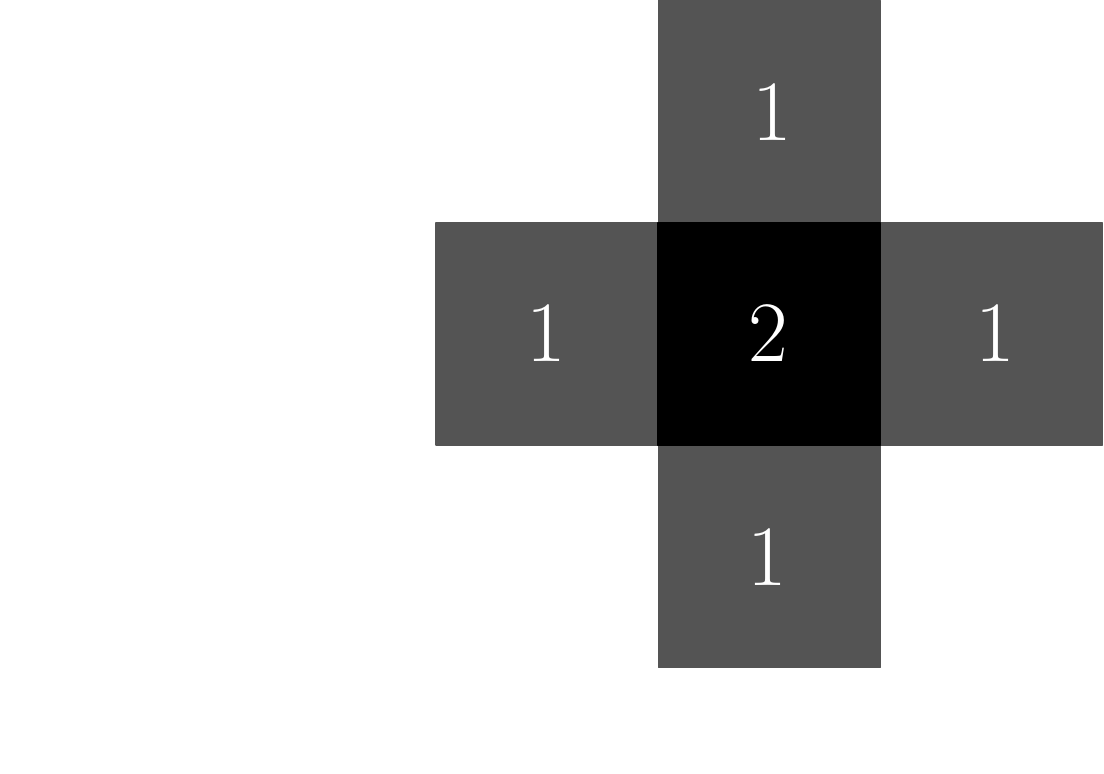}
\caption{The 2-domination shape.}
\end{subfigure}
\begin{subfigure}[b]{0.4\textwidth}
\includegraphics[scale=0.6]{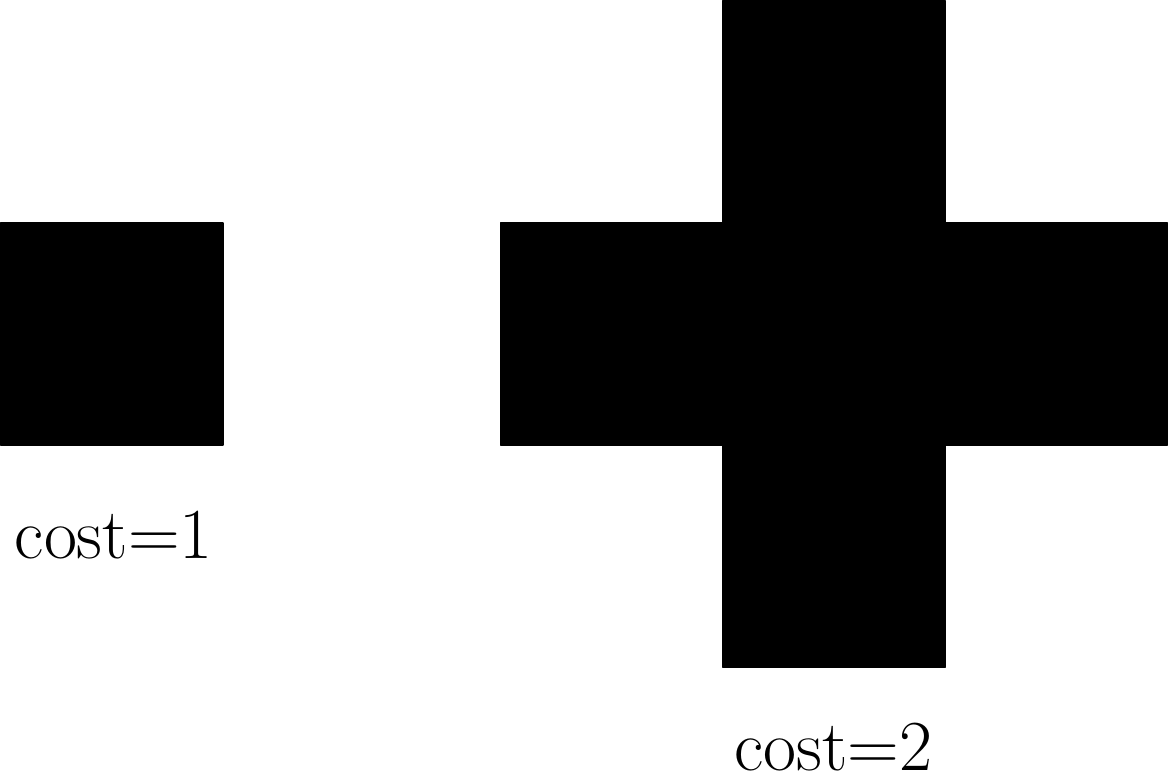}
\caption{The Roman-domination shapes.}
\end{subfigure}
\caption{The shapes for the 2-domination and the Roman domination. For the 2-domination, we look for a covering such that the sum of the weights (in white) on a cell is at least 2. For the Roman domination, we cover with two tiles, but they have different costs. We are interested in a covering of minimum weight.}
\label{2-roman-figures}
\end{figure}

One crucial point is the following property. 
\begin{ppty}[Fixed-height border-fixing]
Let $X$ be a shape. $X$ has the fixed-height border-fixing property if there exist $k, n_0, m_0$ such that, for any $n \geq n_0$ and $m \geq m_0$, there exists an optimal \underline{covering} of the $n \times m$ rectangle whose cells at distance greater than $k$ of the border are included in a tiling (so with no overlaps) of the plane.
\end{ppty}
For instance, the 2-domination shape has this property for $k=3$: any optimal solution to the 2-domination problem can be obtained from an infinite optimal 2-domination set of which we modify only cells at distance at most 3 from the border. Note that, due to the automation feature of the algorithm, this is indeed $k=3$ here even if the program needs to explore borders of size 6 to find the correct bounds.

The fixed-height border-fixing property implies that the bounds given by the method are sharp for some  constant height band, independent of the size of the rectangle. This seems to be related to the following property.

\begin{ppty}[Crystallisation]
Let $X$ be a shape. We say that $X$ has the crystallisation property if there exists $k \in \NN$ such that for every partial tiling of size $k$ with the shape $X$, either this tiling cannot be extended to tile the plane, or there is a unique way to do so up to rotation/symmetry.
\end{ppty}

For instance, the domination shape has this property for $k=2$: any two cross shapes put on a grid either cannot be extended into a tiling of $\ZZ^2$ or can be completed into only one such tiling. On the contrary, the total domination does not have this property.
The total-domination problem has been studied a lot in other graphs (see~\cite{review} for example), but remains open for grids. It is related to the shapes in \Cref{total-shape-fig}. The small one corresponds to the influence of one "stone": note that the centre cell does not dominate itself. The big ones are the unions of two copies of the small one. One can see that tiling the plane with the small shape is equivalent to tiling the plane with the set of the two big shapes: in the small shape, the middle cell must be dominated. As shown, the big shape can be vertical or horizontal. The problem with our technique is that a tiling of the plane can, with a certain degree of freedom, mix the vertical and the horizontal big shapes. This probably leads to some non-zero loss in the centre of a big grid to be necessary for a covering to be of minimum size. In this case, the assumption of the loss on the centre of the grid being zero would be false, making our technique unusable.

\begin{figure}
\centering
\includegraphics[scale=1]{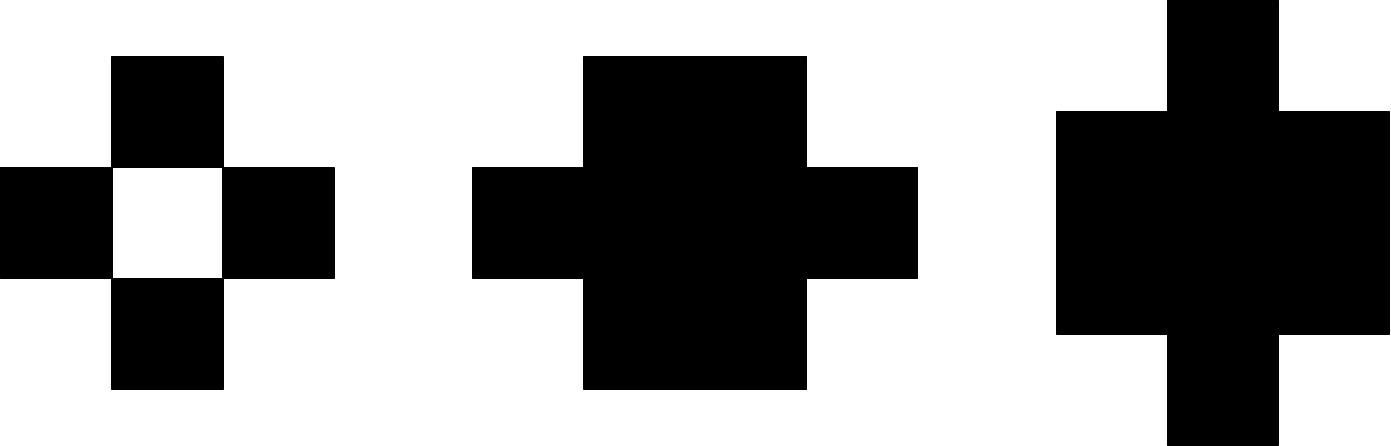}
\caption{The shapes associated to the total domination. The big ones are the two different unions of two copies of the small one. Tiling the plane with the small one boils down to tiling the planes with the two big ones.}
\label{total-shape-fig}
\end{figure}

\begin{conj}
If a shape $X$ tiles the plane and has the crystallisation property then it also has the fixed-height border-fixing property.
\end{conj}

These properties could also be used on tiling problems with other shapes, even if they have no relation with any domination problem on grids.

\section{Experimental details: implementation and optimisations}
\label{section-experimental-dom}
All of the results we obtained required a lot of computing resources. Some of them were obtained after quite a lot of optimisations, and some partial results could have been improved if the memory consumption and the running time had been smaller. We explain here some optimisations we applied, some of them being more theoretical while others are technical.

\paragraph{Problem-dependent optimisations.\\}

As we said before, we made optimisations to reduce the number of states in some problems, using their properties. For instance for the Roman domination, there exist minimum dominating sets which never have one stone next to a cell with either one (case a) or two stones (case b). We can simply remove the stone in case a, or make it a two stones cell and remove its neighbours with one stone in case b. If we remove the optimisations for the states of the Roman domination, the growth rate of the number of states becomes $3.561553$ instead of $2.956295$. We achieved a very good improvement. For instance, for 10 lines the optimised version generates 32.5k states and 6.2M elements in the transfer matrix, while the unoptimised version generates 182k states and 2500M elements in the transfer matrix, which is very significant.

Another part where the improvements were big was about the $(a,b)$ domination. The first version we programmed encoded in the states whether or not each cell had a stone, and whether its predecessor had a stone. The set of states of this version is: $\SSS = \{$ \textsc{stone\_prev, stone, none\_prev, none} $\}$. The other version we programmed after, which we will call the fast version, instead encoded for each cell the fact of whether it has a stone and how many stones it needs (that is 0 or 1 for computing the domination numbers, and 0, 1 or 2 for the computation of the loss). The states of this version for a fixed number of lines, is included\footnote{If $a$ or $b$ equals zero, some states are no longer useful.} in $\{$ \textsc{stone\_ok, stone\_need\_one, stone\_need\_two, ok, need\_one, need\_two} $\}$. This version is faster because if for instance a cell is dominated by its state, we do not need to convey the information of whether or not its predecessor cell had a stone. We will compare the two versions for the total domination, i.e. the (1,1) domination. In the slow version, the growth rate of the number of states is 4, whereas in the fast version, it is $2.618034$, which is a huge gain. Using 95 cores, computing the transfer matrix for 10 lines uses 150s for the slow version compared to 1.3s for the fast one. In the slow version there are 525k states and 252M compatible pairs in the transfer matrix; in the fast version there are 10.4k states and 2.2M compatible pairs, once again a very important gain.

\paragraph{Pruning symmetries in the states.\\}

In order to reduce the number of states, we can do a simple observation: take a dominating set for any domination problem, and apply a horizontal symmetry to it. It remains dominating. This means that we do not need to store all the states: when they are not symmetrical we may keep only one representative out of the two since they both have the same number of stones. Also, let us assume that $S$ and 
$S'$ are compatible, and let us denote the symmetric of $S$ by $\mathrm{rev}(S)$. Then let us assume that $S$ and $S'$ are compatible, but we only stored $\mathrm{rev}(S)$ and not $S'$. This is not a problem since $S$ and $S'$ being compatible implies that $\mathrm{rev}(S)$ and $\mathrm{rev}(S')$ are compatible. When computing the transfer matrix, it suffices to check for compatibility with the reversed versions of one of the two states. If a state can be put first, then its symmetrical can also. The same applies for the end states. This implies that, by using only one representative for each pair of symmetrical states, we compute the same domination numbers as when we store all states. This optimisation saves us around half of the states, half of the compatible pairs.
This optimisation does not, however, change the growth rate of the number of states of a problem: it would otherwise make it decrease exponentially when the number of lines increase instead of roughly dividing it by two.

For the 2-domination with 12 lines, without pruning symmetrical states we have around 142k states and 10M compatible pairs, and the computation of the transfer matrix with 95 cores takes around 59s. When we prune symmetrical states it takes 43.5s: we have 71k states and 4.6M compatible pairs. This a slight improvement, however the optimised version has a memory peak at 236MB whereas the other uses 493MB at peak. This may make a difference, for some problems, when computing the transfer matrices: it may make it possible to go one line further. Unfortunately, this optimisation does not apply to the computation of the loss matrices. Indeed, the top cells have a neighbour above which can contain a stone

\paragraph{Checking validity while generating states.\\}

To generate the valid states, we may recursively enumerate all the states in $\SSS^n$ and each time we complete a state (by choosing its bottom value) we check whether or not the state is valid. However, in the program the generation of the valid states was not parallelised and so it took a non-negligible part of the total running time when using a machine with a lot of cores. This justified trying to optimise this part of the program. The optimisation we programmed was to partially check the validity of a state while generating the state: each time we fix the value of the next cell we check whether or not this beginning is valid instead of waiting until the state is complete.

This means that in many cases we can cut some branches of exploration of the states values before reaching the end. In the minimal-domination problem we investigate in \Cref{dom-counting-chapter}, this helped reduce the time a lot: it takes 5.1s instead of 9.7s for a height of 7 and 46.4s instead of 121.4s for a height of 8. We did not implement this feature for the problems in this chapter since the generation of states did not take that long compared to the other parts of the computation. Computing the loss matrix is, as we will see, cubic in the number of states.

\paragraph{Pre-computing a subset of the potentially compatible states.\\}

The previous paragraph was about precomputing the validity of possible states while we are generating the states. The idea here is similar: when considering a state and checking which other states might be compatible with it, we waste some time trying a lot of states which have no chance of being compatible. 
For instance, $S$ may be incompatible with $S'$ because of the values of the first three cells... in that case, any state $S''$ with the same first three cells will not be compatible with $S'$. For some problems, we may precompute a list of \emph{candidates} to the compatibility relation with a specific state $S'$: any state outside this list will not be compatible, but not every state in the list will be compatible. For other problems, like the simple domination or the 2-domination, precomputing the list of candidates to the compability slows down the program: the ratio of incompatible "pairs" is less than the one of the minimal domination, for instance. This method will be used in \Cref{dom-counting-chapter} for the minimal versions of the domination problems.

 In some problems, mainly the ones storing a lot of things in the past: the distance-2 domination, and especially the minimal and minimal total domination that we study in \Cref{dom-counting-chapter} are the best examples to illustrate this fact. Since the latter stores whether or not the current cell has a stone as well as the information about the most recent three columns, compatibility rules follow some structure. Indeed, if $S'$ can be put after $S$ then the information about columns "-1", "-2", and "-3" must respectively be the same as the information in $S$ about its relative current column and columns "-1" and "-2". So we can discard any state where the information differs. Since each value is stored in four bits (one for each column we store information about), the list of possible next states $S'$ is obtained by shifting each value of $S$ by one bit. This means that we test only one eights of the number of states as possible successors for a state $S$, as opposed to the whole set of states we would test otherwise. To simplify, for instance, 1110 cannot be compatible with 1111 because its one-before-last bit should match the last bit of 1110 (i.e. a zero). 1110 might, on the contrary, be compatible with 1100 or 1101: the zero has shifted towards the left.
 
For the minimal domination of height 7 (around 6.2 millions states), using a machine with 20 cores, the computation of the transfer matrix takes 8 seconds with the optimisation and more than 22 minutes without (we did not wait for it to end).

\paragraph{Optimising matrices products.\\}

The matrices we handle are big, or even huge: as many lines and columns as there are states, which can be around one million (or a bit more) for some problems in this chapter. So the products between matrices take a lot of time and this is a concern of primary importance. Yet, our matrices are very sparse, the ratio of the number of states squared divided by the number of actual (different from $+\infty$) values are the following:

\begin{itemize}[noitemsep]
\item for the 2-domination, it is 32 for 10 lines and 153 for 14 lines;
\item for the Roman domination, it is 314 for 10 lines and 2899 for 14 lines;
\item for the distance-2-domination, it is 226 for 10 lines and 1884 for 14 lines.
\end{itemize} This is very fortunate because the algorithm we use have a running time of roughly $\OO(nbState^2)$ when the number of lines is fixed and $\OO(nbState^3)$ for the loss method. Since the matrices are sparse, the running time is in fact proportional to the number of compatible pairs, that is the number of values of the matrix.
Still, the matrices products take a very long time. One first observation is that since we know that our matrices are primitive (i.e. there exists $k$ such that $M^k$ has its full $nbState^2$ coefficents) products between matrices rapidly take a very very long time. This can be avoided, for the first method (fixed number of lines) by doing only matrix-vector products when computing the domination numbers. This way, we do a number of operations which is around the number of values in the matrix which are different from $+\infty$.

\paragraph{Ordering the corner computations.\\}

The final optimisation we talk about here is done at a higher level: it is more a design detail of the algorithm. To compute the loss, we need to compute the corner transfer matrix (see \Cref{section-big-number-lines}) $C_\mathrm{a}$. $C_\mathrm{a}[a][b]$ contains the minimum loss inside the corner if $a$ and $b$ are the "internal" edges of the corner. So the standard algorithm is basically the following: 
for any two states $a$ and $b$, compute the minimum loss over the corner (i.e. compute the loss over $h$ columns, where $h$ is the height of the border we chose). This leads, letting $nbState = |\VV_a|$, to $nbState^2$ possibilities for the choices of $a$ then $b$, multiplied by a factor of $h \cdot |\RRR_a|$ where $|\RRR_a|$ is the number of almost compatible pairs. Roughly, this would be some $\OO(nbState^4)$. However, things can be more finely tuned: let us assume we compute the loss over the bottom-right corner, let us fix $b$ as the output state of the corner (the "upper horizontal" state). We then compute the minimum loss inside the corner, and finally obtain, for every possible input state of the corner ("left vertical" state) $a$ the value of $C_\mathrm{a}[a][b]$ by considering all the possible states $a$ as a $(h+1)^\text{th}$ compatibility computation. This leads to a roughly $\OO(nbState^3)$ algorithm, which is much better: we handle matrices with $nbState = 20000$, so we save a factor around 20000 in this case.

\paragraph{Using threads.\\}

When you lack time, you may want to save up time eating, so a brunch is a good compromise. Or rather, "brunch" was the name of the machine with the most (efficient) cores in our laboratory: it has 96 cores. So this was a great incentive to parallelise the code. As we mentioned in the introduction, the key point when parallelising a code is to identify how to make the parallel computations independent. Here we were lucky: computing the transfer matrix (i.e. the compatibility and almost-compatibility relations) can be done fairly easily in parallel: we divide the states into $p$ share of equal size and each thread computes, for each state $S$ in its share, the set of states which are compatible with $S$. The good thing is that all threads write the list of compatible states at different points of memory, so there is no risk of writing at a location of the memory at the same time as another write or a read. They may read the same variable at the same time but this is not an issue. Due to the fact that not all portions are easily parallelisable and to a few technical issues the speedup is not 1: if we use $p$ cores, the running time is not divided by $p$. Some of the problems involve context switching (to a minor impact): sometimes a thread is executed in one core and then switched to another so that all the values in the registers need to be copied, and all the thread do not have exactly the same amount of work to do: when multiplying two matrices and giving a portion of the multiplication to each core, some may be "lucky" and finish well before the others. For instance, computing the transfer matrix of height 13 for the 2-domination problem takes 167.6 seconds using one thread. On a machine with 96 cores, using 20 threads, it takes 15.2s and using the 100 threads, it takes 4.8s.\\

\emph{It is beautiful to create a lot of threads and then see the usage of CPU decrease little by little, as though the finishing times of the threads followed a normal law.}

\newcommand{\aarg}{\textrm{arg}}
\tikzstyle{mydashed}=[dash pattern=on 1.5pt off 2pt]
\pgfdeclarepatternformonly{my crosshatch dots}{\pgfqpoint{-1pt}{-1pt}}{\pgfqpoint{5pt}{5pt}}{\pgfqpoint{1.5pt}{1.5pt}}%
{
    \pgfpathcircle{\pgfqpoint{1pt}{0pt}}{.4pt}
    \pgfpathcircle{\pgfqpoint{1pt}{0pt}}{.4pt}
    \pgfusepath{fill}
}
\resetlinenumber
\chapter{Asymptotic growth of the number of dominating sets}
\label{dom-counting-chapter}
We live in a world whose main concept is the one of growth. All countries try to achieve the maximum possible economic growth so that their people may enjoy a better life because they have more money and comfort. However, alongside with the growing number of people on the Earth, such an economical and technological growth leads to huge problems, some of which are climate change and the loss of biodiversity. For instance, we have lost 52\% of the existing animal species between 1970 and 2012, according to the WWF. Some animals are small yet very useful to humans, like bees. They also have suffered a big loss over the past decades, and still decrease now. Concerning the resources and energy that we \textbf{consume}, we can cite the IEA\footnote{International Energy Agency} which says that between 1971 and 2017 the total of produced energy \textbf{increased by more than 2.5 times}. Our consumption of resources increased a lot as well, which is not sustainable. Indeed, in 2019 the \textbf{Earth Overshoot Day} was July 29th\footnote{It happened to be the precise day I wrote this paragraph.}: between January 1st and this date, the world consumed the amount of resources that the Earth can renew in one year. This means that we consume at least 1,42 times what the earth renews each year... asymptotically, this leads to the total depletion of the available resources.\\

However, in this chapter the only thing which grows is the number of inoffensive\footnote{yes, remember \Cref{domination-chapter}} dominating sets. In the previous chapter, we studied these sets from an optimisation point of view: trying to minimise the size of a dominating set. It is also of interest to study the related counting problem: how many different dominating sets are there? We will tackle this question on various domination problems: the domination, total domination, minimal domination and minimal total domination. This means that, for each of these problems, we will show the existence of some constant $\nu$, depending on the problem, such that the number of appropriate dominating sets is $\nu^{nm+o(nm)}$ when both dimensions $n$ and $m$ of the grid tend to infinity. We will also give bounds and estimates on the growth rate $\nu$. We will do the same for the other domination problems listed just above. In this journey, we will see that counting  dominating sets is related to some problems in dynamical systems, and particularly to the notion of entropy in subshifts of finite type. This notion of subshift, which we define later, encompasses the problems which can be defined as colouring $\ZZ^d$ forbidden a specific set of patterns to appear. In dimension 2, the subshifts can be seen as factor-free languages of bidimensional words.\\

We begin in \Cref{section-def-counting} by giving the missing definitions of some dominating sets. We then introduce in \Cref{SFT-local-characterisation} the simple concept of \emph{local characterisation} and the notion of \emph{subshifts (of finite type)}, and explain their similarities with the domination problems. These concepts are tools we will use to obtain our results: our domination problems can be viewed as some subshifts of finite type (SFT). In \Cref{section-SFT-dom-comparison}, we explain the link between these tools and our objectives: counting the patterns of a certain size which can appear in one of our SFTs helps us count the number of dominating sets associated to this SFT. We then introduce in \Cref{section-entropy} the notion of \emph{entropy} of a subshift, which, in our cases, turns out to be equal to the growth rates we are looking for. In one subsection we show that some of the subshifts we have defined are \emph{block gluing}. This fact implies that their entropies are computable numbers, hence the same is true for their asymptotic growth rates. In \Cref{section-bounds-growth-rates} we finally show how to obtain bounds on the growth rates. We give the numerical approximations and bounds we obtained for each problem, thanks to computer resources, using a program similar to the one of \Cref{domination-chapter}. Finally, we introduce a family of subshifts with particular properties in \Cref{section-meta-k-domination}. This family generalises the domination and total domination, and we provide a second family of problems to do the same for their minimal counterparts. We study the block-gluing property for this second family and show that each of them are block gluing, but the block-gluing constant is a function of the parameter of the family.

The work of this chapter was done with Silvère Gangloff. Most of the results can be found in~\cite{article-counting}. Beware that some notations, like $D_{n,m}$, differ a bit from the article in order to keep a coherency with \Cref{domination-chapter}.

\section{Basic definitions and notation}
\label{section-def-counting}
In this chapter, $G$ will be an undirected graph of which we examine the diverse dominating sets, and we will focus on grid graphs only. Contrarily to the previous chapter, the graphs here may be infinite when we specify it. We recall here \Cref{def-total-dominating}: a set $S$ of vertices of a graph $G$ is \textbf{total dominating} when any vertex $v \in V$ has at least one neighbour in $S$

\begin{deff}
A \textbf{minimal dominating} set $S$ is a dominating set which is inclusionwise minimal: any $S' \subsetneq S$ is not dominating. Or equivalently: for each $v \in S$, the set $S \setminus \{v\}$ is not dominating.\\
Likewise, a \textbf{minimal total dominating} set is  a total dominating set which is inclusionwise minimal.
\end{deff}

\begin{notation}
In the following, for all integers $n$ and $m$, we denote by $\bm{D_{n,m}}$, $\bm{M_{n,m}}$, $\bm{T_{n,m}}$ and $\bm{\mathit{MT}_{n,m}}$ respectively the number of dominating sets of the grid $G_{n,m}$, the number of its minimal dominating sets, the number of its total dominating sets and the number of its minimal total dominating sets.
\label{notation-counting-numbers}
\end{notation}

To familiarise the reader with the notions of domination we study here, some of them are illustrated in \Cref{figure.domination.notions}.

\begin{figure}[h!]
\centering
\begin{tikzpicture}[scale=0.5]
\begin{scope}

\fill[gray!90] (0,0) rectangle (2,2);
\fill[gray!90] (1,3) rectangle (2,4);
\fill[gray!90] (2,2) rectangle (3,3);
\fill[gray!90] (3,3) rectangle (4,4);
\fill[gray!90] (3,1) rectangle (4,2);
\draw (0,0) grid (4,4);
\node at (1.5,-1) {(a)};
\node at (-0.8,0) {\scriptsize (1,1)};
\end{scope}

\begin{scope}[xshift=6cm]

\fill[gray!90] (1,3) rectangle (2,4);
\fill[gray!90] (0,1) rectangle (1,2);
\fill[gray!90] (1,0) rectangle (2,1);
\fill[gray!90] (3,1) rectangle (4,3);

\draw (0,0) grid (4,4);
\node at (1.5,-1) {(b)};
\end{scope}

\begin{scope}[xshift=12cm]

\fill[gray!90] (0,1) rectangle (2,2);
\fill[gray!90] (1,0) rectangle (2,1);
\fill[gray!90] (1,3) rectangle (3,4);
\fill[gray!90] (3,1) rectangle (4,3);

\draw (0,0) grid (4,4);

\node at (1.5,-1) {(c)};
\end{scope}

\end{tikzpicture}
\caption{Illustration on $G_{4,4}$:\\
(a) a dominating set which is neither minimal dominating nor total dominating;\\
(b) a minimal dominating set which is not total dominating (the bottom-left dominant vertices are not dominated);\\
(c) a minimal total dominating set.}
\label{figure.domination.notions}
\end{figure}

\section{Local characterisations and relation with SFTs}
\label{SFT-local-characterisation}
In this section we recall, and for completeness 
prove, the local characterisations 
of some notions of dominating sets. This means that 
one can check if a set $S$ is dominating (or minimal dominating, and so on) by examining, for each vertex, whether or not this vertex and its (possibly extended) neighbourhood 
are in $S$. This was what enabled \Cref{domination-chapter} to work: checking properties locally makes it possible to encode only some information about the (few) previous column(s). This allows us to enumerate the dominating sets without keeping (and hence enumerating) them fully in memory. We then introduce and define the \emph{subshifts of finite type} (SFT) which are dynamical systems objects strongly linked to the local characterisation property. Each domination problem will be associated to its SFT counterpart, with which it shares some properties like the growth rate (called \emph{entropy} in the world of SFTs). We will later use the framework around the SFTs and prove some properties they have to deduce the counterpart in the domination growth rates.

\subsection{Local characterisations}
\label{local-charac-section}

\begin{deff}
\label{definition.local.domination.notions}
When a set $S$ of a graph $G$ is fixed, a vertex is called a \textbf{dominant} element of $G$ when it is in $S$, and a \textbf{dominated} element when it has a neighbour in $S$.\\$w$ is said to be a \textbf{private neighbour} of a dominant element $v$ when $v$ is the only neighbour of $w$ in the set $S$.
\end{deff}

\begin{fact}
\label{fact.dominating.sets}
Let $S$ be a set of vertices of a graph $G$. 
Then for all vertices $v$ and $w$ such that $w$ is not a neighbour of $v$, $w$ is dominated by $S$ if and only if it is dominated by $S \setminus \{v\}$.
\end{fact}

\begin{deff}
Let $S$ be a set of vertices of a graph $G$. 
We say that a dominant element is \textbf{isolated} in $S$ when it has no neighbours 
in $S$.
\end{deff}
In the previous three definitions, $S$ is a dominating set for any of the variant of domination we study.

\begin{prop}
Let $S$ be a dominating set of a graph $G$. $S$ is minimal dominating if and only if every element of $S$ is isolated in $S$ or has a private neighbour not in $S$.
\end{prop}
\newpage

\begin{proof} \leavevmode
    \begin{itemize}
	\item $(\Rightarrow)$: Let us assume that $S$ is minimal dominating. Every vertex not in $S$ has a neighbour in $S$ because $S$ is a dominating set. Now let us fix $v \in S$. 	From \Cref{fact.dominating.sets} and by definition of a minimal dominating set, any $w$ which is not in the neighbourhood of $v$ is dominated by $S \setminus \{v\}$. Since $S \setminus \{v\}$ is not dominating, it means that: 
	\begin{enumerate}
		\item $v$ is not dominated by $S \setminus \{v\}$, which means that $v$ is isolated in $S$;
		\item or there exists some $u \notin S$ connected to $v$ which is not dominated by $S \setminus \{v\}$, hence $u$ is a private neighbour of $v$ which is not in $S$.
		
	\end{enumerate}
	\item $(\Leftarrow)$: Conversely, let us fix some dominating set $S$ such that every $v \in S$ has a private neighbour not in $S$, or is isolated. Fix some $v \in S$. If it has a private neighbour $u$, then $u$ is not dominated by $S \setminus \{v\}$, and thus $S \setminus \{v\}$ is not dominating. If it has no private neighbours, then it is isolated. This means that $v$ is not dominated by $S \setminus \{v\}$, therefore the set is not dominating. In both cases, we conclude that $S$ is minimal dominating.
	\end{itemize}
\end{proof}

With a similar proof, we obtain the following:

\begin{prop}
A total dominating set $S$ of a graph $G$ is minimal total dominating if and only if any $v \in S$ has a private neighbour.
\end{prop} 

\subsection{Subshifts of Finite Type (SFTs)}
We now introduce the notion of SFTs: they intuitively correspond to sets of possible colourings of the infinite grid $\ZZ^d$ which avoid some fixed finite set of forbidden patterns. We give the definitions in the general context of arbitrary dimension $d$, but in this chapter we will mostly work in dimension two, and a bit in dimension one in \Cref{section-nearest}.

\begin{deff}
Let $\mathcal{A}$ be a finite set, and $d \ge 1$ an integer. A \textbf{pattern} $p$ on alphabet $\mathcal{A}$ is an element of $\mathcal{A}^{\mathbb{U}}$ for some finite $\mathbb{U} \subset \mathbb{Z}^d$. The set $\mathbb{U}$ is called the \textbf{support} of $p$, and is denoted $\text{supp}(p)$.
\end{deff}

\begin{deff}
Given an alphabet $\mathcal{A}$, any colouring of $\ZZ^d$ with values in $\mathcal{A}$, that is any element of $\mathcal{A}^{\ZZd}$, is called a \textbf{configuration}.
\end{deff}

\begin{deff}
Let $C_1$ and $C_2$ be two configurations. Let $S_n$ be the square of size $2n+1$ centred in $(0,0)$. Let $n$ be the maximum integer such that the pattern on support $S_n$ of $C_1$ and the one on same support of $C_2$ coincide, or $+\infty$ if $C_1 = C_2$.

We define the distance on the set of configurations to be $d(C_1, C_2) = 2^{-n}$ if $C_1 \neq C_2$ or 0 otherwise.

\label{def-distance-colouring}
\end{deff}

\begin{notation}
For a configuration $x = (x_\textbf{u})_{\textbf{u} \in \mathbb{Z}^d}$ of $\mathcal{A}^{\mathbb{Z}^d}$ (resp. for a pattern $p \in \mathcal{A}^{\mathbb{U}}$ for some $\mathbb{U} \subset \mathbb{Z}^d$), we denote by $x_{|\mathbb{V}}$ the restriction of $x$ to some subset $\mathbb{V} \subset \mathbb{Z}^d$ (resp. the restriction of $p$ to $\mathbb{V} \subset \mathbb{U}$).
\end{notation}

\begin{deff}
Let $\mathcal{A}$ be a finite set, and $d \ge 1$ integer. A \textbf{$d$-dimensional subshift} on alphabet $\mathcal{A}$ is a subset of $\mathcal{A}^{\mathbb{Z}^d}$ defined by a set of forbidden patterns. Formally, a subset 
$X$ of $\mathcal{A}^{\mathbb{Z}^d}$ is a subshift when there exist some finite sets $\mathbb{U} \subset \mathbb{Z}^d$ and $\mathcal{F} \subset \mathcal{A}^{\mathbb{U}}$ such that:
\[X= \left\{x \in \mathcal{A}^{\mathbb{Z}^d}: 
\forall\, \textbf{u} \in \mathbb{Z}^d, x_{|\textbf{u}+\mathbb{U}} \notin \mathcal{F}\right\}.\]
The elements of $\mathcal{F}$ are called the \textbf{forbidden patterns}.
\end{deff}

\begin{deff}
Let $X$ be a subshift defined by a set of forbidden patterns $\mathcal{F}$. When $\mathcal{F}$ is finite, then $X$ is called a \textbf{subshift of finite type} or SFT.
\end{deff}

\begin{deff}
For a subshift $X$, a \textbf{globally-admissible} pattern of size $\llbracket 1, n \rrbracket ^d$ is some pattern $p \in \mathcal{A}^{\llbracket 1,n\rrbracket ^d}$ which appears in a configuration of $X$, that is when $x_{|\llbracket 1,n\rrbracket ^d}=p$.
When $d = 2$, we extend the definition to patterns $p \in \mathcal{A}^{\llbracket 1,n\rrbracket \times \llbracket 1, m \rrbracket}$ when there exists a configuration $x$ of $X$ such that $x_{|\llbracket 1,n\rrbracket \times \llbracket 1, m \rrbracket}=p$.
\end{deff}

Although here we limit ourselves to SFTs, the world of subshifts contains many subshifts which are not of finite type. We will give a few examples in dimension one, on alphabet $\{a,b\}$. The set $X_{\leq 1}$ of words with \emph{at most} one 'a' is not of finite type: an infinite word may contain two 'a's at arbitrarily large distance from each other, which we cannot forbid with finite forbidden patterns. However, there exists a SFT on an alphabet with three symbols which, after remaping one to $a$ and the other two to $b$, gives the same language as $X_{\leq 1}$: we say that $X_{\leq 1}$ is \emph{sofic}. Now, by modifying slightly the condition, we define $X_{=1}$ to be the set of words with \emph{exactly} one 'a'. This set is not a subshift. It is due to the fact that the desired 'a' may be arbitrarily "far away". No sets of forbidden patterns  can enforce that an 'a' is present. We can also consider the equivalent definition of subshifts to see that $X_{=1}$ is not one:

\begin{deff}[Alternate definition of subshift]
A set $X$ of configurations of $\ZZ^d$ with values in $\mathcal{A}$ is a subshift when it is closed and stable by translation.
\end{deff}

We can see that the word containing only 'b's is in the closure of $X_{=1}$ but does not belong to $X_{=1}$, which shows that $X_{=1}$ is not closed, hence neither is it a subshift.

\subsection{The domination subshifts}
\label{subshifts-examples}
For our domination subshifts, the alphabet is 
$\mathcal{A}_0 =\left\{\begin{tikzpicture}[scale=0.3]
\draw (0,0) rectangle (1,1);
\end{tikzpicture},\begin{tikzpicture}[scale=0.3]
\fill[gray!90] (0,0) rectangle (1,1);
\draw (0,0) rectangle (1,1);
\end{tikzpicture}\right\}$, and $d=2$. We work in dimension two since we study grids.

\begin{deff}
The \textbf{domination} (resp. \textbf{minimal domination}, \textbf{total domination} and \textbf{minimal total domination}) subshift denoted by $X^\mathrm{D}$ (resp. $X^\mathrm{M}$, $X^\mathrm{T}$ and $X^\mathrm{MT}$), is the set of elements $x \in \mathcal{A}_0^{\mathbb{Z}^2}$ such that 
$\{\textbf{u} \in \mathbb{Z}^2: x_{\textbf{u}} = \begin{tikzpicture}[scale=0.3]
\fill[gray!90] (0,0) rectangle (1,1);
\draw (0,0) rectangle (1,1);
\end{tikzpicture}\}$
is a dominating (resp. minimal dominating, total dominating and minimal total dominating) set of the infinite square grid $\mathbb{Z}^2$.
In all these cases, a configuration $x$ of the subshift is called a \textbf{dominated  configuration}. We also say that $\textbf{u}$ is a \textbf{dominant position} of the configuration $x$ when $x_\textbf{u}$ is grey. Likewise, a \textbf{private neighbour} is still a position which is dominated by exactly one dominant position.
\end{deff}

The local characterisations for each type of dominating sets we gave earlier can straightforwardly be translated into finite sets of forbidden patterns. We then obtain the following result:

\begin{prop}
The sets $X^\mathrm{D}$, $X^\mathrm{M}$, $X^\mathrm{T}$ and $X^\mathrm{MT}$ are subshifts of finite type.
\end{prop}

As we just mentioned, the key point is that these domination problems have a local characterisation, that is a set of forbidden patterns of bounded sizes. All the problems which enjoy this property can also be associated to SFTs. For instance, the stable set problem for graphs consists in selecting vertices such that no two vertices are connected. In the world of SFTs, it is called the \textit{hard-square} problem, and is associated to the Fibonacci subshift of dimension two. This SFT has been studied a lot, as well as its growth rate. We have good approximations of the growth rate (see for instance~\cite{stable-set-sft}), which corresponds to the \textit{entropy} of the subshift. However we do not know its exact value, nor do we have closed formulas accounting for the number of stable sets of the finite grid $G_{n,m}$. These problems of counting turn out to be very hard to solve exactly, even if we only want to find the exact entropy. This chapter is a first step in studying the number of dominating sets for a few variants of domination, by proving that they have some interesting properties and approximating their entropies.

\section{Comparing the growth of SFTs with the growth of dominating sets}
\label{section-SFT-dom-comparison}
\begin{notation}
For a subshift of finite type $X$, we denote by $N_n (X)$ the number of globally-admissible patterns of size $\llbracket 1,n\rrbracket ^d$. When $d=2$, we extend the notation and denote by resp. $N_{n,m}(X)$ the number of globally-admissible patterns of size $\llbracket 1,n \rrbracket \times \llbracket 1,m \rrbracket$.
\end{notation}

\begin{deff}
The \textbf{topological entropy} of a subshift of finite type is the number
\[h(X) = \inf_{n \rightarrow +\infty} \frac{\log_2 (N_n (x))}{n^d}.\]
    \label{def-entropy}
\end{deff}

We will simply refer to this notion as the entropy in the rest of the manuscript. The following two lemmas are well known (see for instance~\cite{Lind-Marcus}).

\begin{lemma}
\label{lemma-entropy-limit}
The infimum in the definition of $h(X)$ is in fact a limit: 
\[h(X) = \lim_n \frac{\log_2 (N_n (x))}{n^d}.\]
\end{lemma}

\begin{notation}
Let us denote by $\sigma$ the $\mathbb{Z}^d$-\textbf{shift} action on $\mathcal{A}^{\mathbb{Z}^d}$ defined such that for all $\textbf{u},\textbf{v} \in \mathbb{Z}^d$, 
\[(\sigma^{\textbf{u}} x)_{\textbf{v}} = x_{\textbf{v}+\textbf{u}}\texttt{.}\]
Informally, $\sigma$ acts on a configuration by translating it by the vector $\textbf{u}$.
\end{notation}

\begin{deff}
A \textbf{conjugation} between two $d$-dimensional subshifts of finite type $X$ and $Z$ is an invertible 
map $\varphi: X \rightarrow Z$ such that,
for all $\textbf{u} \in \mathbb{Z}^d$ and $x \in X$,
$\varphi(\sigma^{\textbf{u}}.x) = \sigma^{\textbf{u}}.\varphi (x)$.
In this case, $X$ and $Z$ are said to be \textbf{conjugated}.
\end{deff}

\begin{lemma}
If two subshifts of finite type $X$ and $Z$ 
are conjugated, then $h(X)=h(Z)$.
\label{lemma-conjugation}
\end{lemma}

The entropy is one parameter which is preserved by conjugation. It may be a way to show that two subshifts are not conjugated, by showing that their entropies are different. We will use this lemma and the concept of conjugation only in \Cref{section-nearest}.

\begin{lemma}
Let $X$ be a bidimensional subshift of finite type. Then:
\[h(X) = \lim_{n,m} \frac{\log_2 (N_{n,m} (X))}{nm}.\]
\end{lemma}

This limit is to be understood as letting both $n$ and $m$ tend towards infinity at the same time (but possibly at different speeds). We can see that this limit exists and corresponds to the entropy by decomposing a square into rectangles and doing the other way around. This allows us to bound by below and by above this limit by numbers which tend towards the entropy $h(X)$.

We now need to show, for each variant of the domination, that the SFT and its associated dominating set problem grow at the same speed, that is the number of $n \times m$ patterns appearing in the SFT grows at the same pace as the number of dominating sets of the $n \times m$ grid. This is not trivial, since the dominating sets of a finite grid $G_{n,m}$ do not correspond exactly to the globally-admissible patterns on the same grid of the corresponding SFT type presented in \Cref{subshifts-examples}. 
Indeed, in such a pattern, the positions of the border may for instance be dominated by a position outside the pattern in a configuration 
in which the pattern appears. Nonetheless, we will see that we can compare the number of globally-admissible patterns of size $n \times m$ for $X^\mathrm{D}$ (resp. $X^\mathrm{M}$, $X^\mathrm{T}$ and $X^\mathrm{MT}$) to the number of dominating sets (resp. minimal dominating sets, total dominating set and minimal total dominating sets) of $G_{n,m}$. We will later use this to prove the existence of an asymptotic growth rate for the grid, turning out to be equal to the entropy of the corresponding SFT.

For concision, we will assimilate the set of vertices of $G_{n,m}$ to any translate of \linebreak $\llbracket 1,n \rrbracket \times \llbracket 1,m \rrbracket$. We will assimilate any dominating (for any domination problem) set $S$ of vertices of a finite grid $G_{n,m}$ with the pattern $p$ of $\mathcal{A}^{\mathbb{Z}^2}$ on $\llbracket 1,n \rrbracket \times \llbracket 1,m \rrbracket$ defined by $p_{\textbf{u}}$ being grey if and only  if $\textbf{u} \in S$.

\begin{deff}
If $\mathbb{U}$ is a subset of $\mathbb{Z}^2$, 
we define the (extended) \textbf{neighbourhood} of $\mathbb{U}$ as
\[\mathcal{N} (\mathbb{U}) = \bigcup_{\textbf{u} \in \mathbb{U}} \left( \textbf{u} + \llbracket -1,1\rrbracket ^2\right).\]
Using iterates of the function $\mathcal{N}$ we also define, for all $n,m \ge 1$ and $k \geq 1$, the \textbf{border}
\[\mathbb{B}_{n,m,k} = \mathcal{N}^{k} (\llbracket 1,n \rrbracket \times \llbracket 1,m \rrbracket) \setminus \mathcal{N}^{k-1} (\llbracket 1,n \rrbracket \times \llbracket 1,m \rrbracket).\] 

For convenience, we extend the notation to $\mathbb{B}_{n,m,0}
= \llbracket 1,n \rrbracket \times \llbracket 1,m \rrbracket \setminus \llbracket 2,n-1 \rrbracket \times \llbracket 2,m-1 \rrbracket$.

\[\begin{tikzpicture}[scale=0.4]

\fill[gray!60] (-3,-3) rectangle (12,7);
\fill[pattern=north west lines] (-2,-2) rectangle (11,6);
\draw (-2,-2) rectangle (11,6);
\draw (-3,-3) rectangle (12,7);
\node at (16.5,1.8) {$\mathbb{B}_{n,m,k+1}$};
\draw[-latex] (14,2) -- (11.5,2);

\draw[fill=gray!40] (0,0) rectangle (9,4);
\draw[fill=black] (0,0) circle (4pt) ;

\draw[-latex] (-4.5,0) -- (-0.4,0);
\node[scale=0.7] at (-6.2,0) {$(1,1)$};

\draw[-latex] (8,-4.5) -- (4.5,-4.5) -- (4.5,-1);
\node at (13,-4.5) {$\mathcal{N}^k (\llbracket 1,n \rrbracket \times 
\llbracket 1,m \rrbracket)$};

\node[scale=0.9] at (4.5,2) {$\llbracket 1,n \rrbracket \times \llbracket 1,m \rrbracket$};

\end{tikzpicture}\]
\label{def-border}
\end{deff}

\begin{lemma}
For all $n,m \ge 3$, the following inequalities hold: \[N_{n-2,m-2} (X^\mathrm{D}) \le D_{n,m} \le N_{n,m} (X^\mathrm{D}).\]
\end{lemma}

\begin{proof} \leavevmode
\begin{enumerate}
\item For all $n,m \ge 1$, any dominating set of $G_{n,m}$ 
can be extended into a configuration of $X^\mathrm{D}$ 
by defining the symbol of 
any position outside $\llbracket 1,n \rrbracket \times \llbracket 1,m \rrbracket$ to be grey. As a consequence, 
any dominating set of $G_{n,m}$ is globally 
admissible in $X^\mathrm{D}$ and thus $D_{n,m} \le N_{n,m} (X^\mathrm{D})$.
\item Any pattern of $X^\mathrm{D}$ on $\llbracket 1,n\rrbracket \times \llbracket 1,m \rrbracket$ can be turned into 
a dominating set of\linebreak $\llbracket 0,n+1\rrbracket
\times \llbracket 0,m+1 \rrbracket$
by extending it with grey symbols. Hence we obtain
the inequality $N_{n,m} (X^\mathrm{D}) \le D_{n+2,m+2}$ 
for all $n,m \ge 1$. 
\end{enumerate}
\end{proof}

Using very similar arguments, we obtain the same inequality for the total domination.
\begin{lemma}
For all $n,m \ge 2$, the following inequalities hold: \[N_{n-2,m-2} (X^\mathrm{T}) \le T_{n,m} \le N_{n,m} (X^\mathrm{T}).\]
\end{lemma}

We then address the minimal and minimal total domination. As we will see, the proofs of the following inequalities are more complex.

\begin{lemma}
\label{lemma.comparison.minimal}
For all $n,m \ge 1$, the following inequalities hold: 
\[\frac{1}{2^{6(n+m)}} N_{n,m} (X^\mathrm{M}) \le M_{n,m} \le N_{n,m} (X^\mathrm{M}).\]
\end{lemma}
\newpage
\begin{proof} \leavevmode
\begin{enumerate}
\item \textbf{Second inequality.}
\begin{enumerate}
\item \textbf{A completion algorithm 
of a minimal dominating set into a configuration of 
$X^\mathrm{M}$.} 

Let $S$ be a minimal dominating set of $\llbracket 1,n\rrbracket \times \llbracket 1,m \rrbracket$. Let us extend it into a configuration $x$ of $X^\mathrm{M}$ using the following algorithm: successively for every 
$k \ge 0$, we extend the current pattern 
into a pattern on $\mathcal{N}^{k+1} (\llbracket 1,n \rrbracket \times \llbracket 1,m \rrbracket)$ using the following operations, for all $\textbf{u} \in \mathbb{B}_{n,m,k+1}$:

\begin{enumerate}
	\item if $\textbf{u}$ is a corner then $x_{\textbf{u}}$ is white;
	\item if $\textbf{u}$ is a neighbour of a corner in one of the vertical sides of $\mathbb{B}_{n,m,k+1}$ then $x_{\textbf{u}}$ is white;
	\item for every other $\textbf{u}$, $x_{\textbf{u}}$ is grey if and only if its neighbour in $\mathcal{N}^k (\llbracket 1,n \rrbracket \times \llbracket 1,m \rrbracket) $ is neither dominated by an element in this set, nor a dominant element.
\end{enumerate}

This algorithm is illustrated in 
\Cref{figure.completion.crowns.minimal}.

\begin{figure}[h!]
\[\begin{tikzpicture}[scale=0.25]
\fill[gray!10] (0,0) rectangle (5,5);
\draw (0,0) rectangle (5,5);

\begin{scope}[xshift=12cm]
\fill[gray!10] (0,0) rectangle (5,5);
\draw (0,0) rectangle (5,5);
\draw (5,5) rectangle (6,6);
\draw (0,0) rectangle (-1,-1);
\draw (0,5) rectangle (-1,6);
\draw (5,0) rectangle (6,-1);

\draw (5,5) rectangle (6,4);
\draw (5,0) rectangle (6,1);
\draw (0,0) rectangle (-1,1);
\draw (0,5) rectangle (-1,4);

\end{scope}

\begin{scope}[xshift=24cm]
\fill[gray!10] (0,0) rectangle (5,5);
\fill[white] (1,4) rectangle (4,5);
\fill[white] (2,3) rectangle (3,4);
\draw (1,4) grid (4,5);
\draw (2,3) rectangle (3,4);
\draw[fill=gray!90] (2,5) rectangle (3,6);

\fill[white] (1,0) rectangle (3,1);
\fill[gray!90] (3,0) rectangle (4,1);
\fill[white] (2,1) rectangle (3,2);
\draw (1,0) grid (4,1);
\draw (2,1) rectangle (3,2);
\draw (2,-1) rectangle (3,0);

\draw (0,0) rectangle (5,5);
\draw (5,5) rectangle (6,6);
\draw (0,0) rectangle (-1,-1);
\draw (0,5) rectangle (-1,6);
\draw (5,0) rectangle (6,-1);

\draw (5,5) rectangle (6,4);
\draw (5,0) rectangle (6,1);
\draw (0,0) rectangle (-1,1);
\draw (0,5) rectangle (-1,4);
\end{scope}
\end{tikzpicture}\]
\caption{Illustration of the completion algorithm in $X^\mathrm{M}$: steps of the algorithm are applied successively from left to right.}
\label{figure.completion.crowns.minimal}
\end{figure}
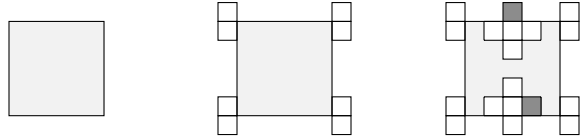
\item \textbf{The output obtained by repeating the algorithm 
is a configuration of $X^\mathrm{M}$.}

\begin{itemize}
\item \textbf{Every position is dominated.}

This is verified for the positions in
$\llbracket 2,n-1\rrbracket \times \llbracket 2,m-1\rrbracket$ because we start from a minimal dominating set. Outside this set, if a position in some $\mathcal{N}^k (\llbracket 1,n \rrbracket \times \llbracket 1,m \rrbracket)$ (for $k \ge 0$) is not dominated 
before extending the configuration 
on $\mathcal{N}^{k+1} (\llbracket 1,n \rrbracket \times \llbracket 1,m \rrbracket)$, then it gets dominated at 
this step by Rule iii and stays that way afterwards.

\item \textbf{Every dominant position 
is isolated or has a private neighbour.}
 \linebreak
Let us consider a dominant position $\textbf{u}$ which 
is not isolated. If it lies in $\llbracket 2,n-1\rrbracket \times \llbracket 2,m-1\rrbracket$, then it has a private neighbour since the pattern 
on $\llbracket 1,n\rrbracket \times \llbracket 1,m\rrbracket$ is a minimal dominating set of $G_{n,m}$. Otherwise, it lies in some $\mathbb{B}_{n,m,k}$ for some 
$k \ge 0$ and there are two cases:

\begin{itemize}
\item \textbf{$\textbf{u}$ is not a corner.}

Its neighbour $\textbf{v} \in \mathbb{B}_{n,m,k+1}$ 
is white by the application 
of the algorithm. Also, since its neighbours
in $\mathbb{B}_{n,m,k}$ are thus dominant or 
dominated, their neighbours in $\mathbb{B}_{n,m,k+1}$ are white. In addition, the neighbour of $\textbf{v}$ in $\mathbb{B}_{n,m,k+2}$ is thus white. 
This is illustrated in \Cref{figure.proof.private.neighbour.minimal}.
As a consequence, $\textbf{v}$ is a private neighbour for $\textbf{u}$.

\item \textbf{$\textbf{u}$ is a corner.}

This case can only happen for the corners of $\llbracket 1,n \rrbracket \times \llbracket 1,m \rrbracket$: the other corners are white by Rule i. We apply a similar reasoning: $u$ has two neighbours, so we may use the previous proof by considering for instance the one at its left or at its right, which is left white by Rule ii.

\end{itemize}

\begin{figure}[h!]
\[\begin{tikzpicture}[scale=0.2]
\fill[color=gray!10] (0,0) rectangle (7,7);
\fill[white] (2,6) rectangle (5,7);
\fill[gray!90] (3,6) rectangle (4,7);
\draw (2,6) grid (5,7);
\draw (0,0) rectangle (7,7);

\begin{scope}[xshift=13cm]
\fill[color=gray!10] (0,0) rectangle (7,7);
\fill[white] (2,6) rectangle (5,7);
\fill[gray!90] (3,6) rectangle (4,7);
\draw (2,6) grid (5,7);
\draw (0,0) rectangle (7,7);
\draw (2,7) grid (5,8);
\end{scope}

\begin{scope}[xshift=26cm]
\fill[color=gray!10] (0,0) rectangle (7,7);
\fill[white] (2,6) rectangle (5,7);
\fill[gray!90] (3,6) rectangle (4,7);
\draw (2,6) grid (5,7);
\draw (0,0) rectangle (7,7);
\draw (2,7) grid (5,8);
\draw (3,8) rectangle (4,9);
\end{scope}
\end{tikzpicture}\]
\caption{Illustration of the proof of a private neighbour for a non-isolated position. Steps of the completion algorithm for $X^\mathrm{M}$ applied from left to right.}
\label{figure.proof.private.neighbour.minimal}
\end{figure}
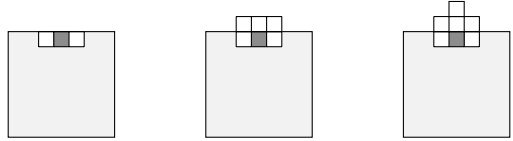

\end{itemize}
\end{enumerate}

\item \textbf{First inequality.} 

\begin{enumerate}
\item \textbf{Transforming patterns of $X^\mathrm{M}$ into minimal dominating sets.}

Let us define an application $\phi_{n,m}$ 
which, to each pattern of $X^\mathrm{M}$ on $\llbracket 1,n\rrbracket \times \llbracket 1, m\rrbracket$, associates a minimal 
dominating set of $G_{n,m}$ defined by: 

\begin{enumerate}
\item suppressing 
any dominant position in $\mathbb{B}_{n,m,0}$ (see \Cref{def-border} for the definition of the borders) which has no private neighbours in $G_{n,m}$ and which is dominated by an element of $G_{n,m}$;

\item changing successively any non-dominant position of $\mathbb{B}_{n,m,0}$ which is still not 
dominated into
a dominant one;
\item successively, for every dominant position $\textbf{u} \in \mathbb{B}_{n,m,0}$: if one of $\textbf{u}$'s neighbours $\textbf{v}$ is the only private neighbour of a position $\textbf{w}$ which is not isolated in $G_{n,m}$ then change $\textbf{w}$ into a non-dominant position.

This step is illustrated in \Cref{figure.third.step.minimal}.

\begin{figure}[h!]
\[\begin{tikzpicture}[scale=0.3]
\fill[gray!10] (0,0) rectangle (7,7);
\fill[white] (2,4) rectangle (5,7);
\fill[gray!90] (2,4) rectangle (4,5);
\fill[white] (3,3) rectangle (4,4);
\fill[gray!90] (4,3) rectangle (5,4);
\fill[gray!90] (5,4) rectangle (6,5);
\draw (0,0) rectangle (7,7);
\draw (2,4) grid (5,7) ;
\draw (3,3) grid (5,4);
\draw (5,4) rectangle (6,5);
\node at (3.5,-1.5) {After Step i.};

\begin{scope}[xshift=12cm]
\fill[gray!10] (0,0) rectangle (7,7);
\fill[white] (2,4) rectangle (5,7);
\fill[gray!90] (2,4) rectangle (4,5);
\fill[white] (3,3) rectangle (4,4);
\fill[gray!90] (3,6) rectangle (4,7);
\fill[gray!90] (4,3) rectangle (5,4);
\fill[gray!90] (5,4) rectangle (6,5);
\draw (0,0) rectangle (7,7);
\draw (2,4) grid (5,7) ;
\draw (3,3) grid (5,4);
\draw (5,4) rectangle (6,5);
\node at (3.5,-1.5) {After Step ii.};
\node at (3.5,5.5) {\textbf{v}};
\node at (3.5,4.5) {\textbf{w}};
\node at (3.5,6.5) {\textbf{u}};
\end{scope}

\begin{scope}[xshift=24cm]
\fill[gray!10] (0,0) rectangle (7,7);
\fill[white] (2,4) rectangle (5,7);
\fill[gray!90] (2,4) rectangle (3,5);
\fill[white] (3,3) rectangle (4,4);
\fill[gray!90] (3,6) rectangle (4,7);
\fill[gray!90] (4,3) rectangle (5,4);
\fill[gray!90] (5,4) rectangle (6,5);
\draw (0,0) rectangle (7,7);
\draw (2,4) grid (5,7) ;
\draw (3,3) grid (5,4);
\draw (5,4) rectangle (6,5);
\node at (3.5,-1.5) {After Step iii.};
\end{scope}
\end{tikzpicture}\]
\caption{Illustration of the second and then third steps of the algorithm defining $\phi_{n,m}$ for $X^\mathrm{M}$, from left to right. $\textbf{u},\textbf{v}$ and $\textbf{w}$ are instances of the positions described in Rule iii.}
\label{figure.third.step.minimal}
\end{figure}
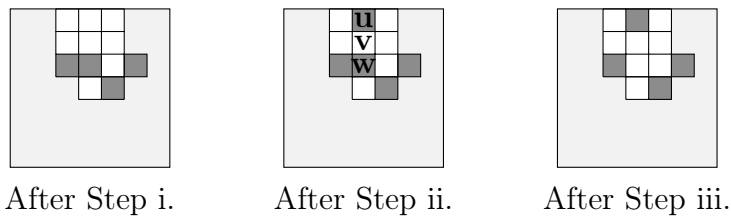
\end{enumerate}

\item \textbf{Verifying that images of $\phi_{n,m}$ 
are minimal-dominating sets.} 

Let us consider a globally-admissible pattern $p$ of $X^\mathrm{M}$ on $\llbracket 1, n\rrbracket \times \llbracket 1,m\rrbracket$. Let us show that the set $\phi_{n,m} (p)$ is 
a minimal dominating set of $G_{n,m}$:

\begin{itemize}
\item \textbf{Any vertex of $G_{n,m}$ is dominated
or dominant in $\phi_{n,m} (p)$.}

Before Step ii, if a position is not dominant 
and not dominated, it becomes dominant during 
this step. Every position in $\llbracket 2, n-1 \rrbracket \times \llbracket 2, m-1 \rrbracket$ was already dominated since the pattern was in $X^\mathrm{M}$.
Moreover, during Step iii, any position 
which is modified to no longer be dominant is necessarily a dominated position ("$\textbf{w}$" is chosen among non-isolated positions).

\item \textbf{Any non-isolated dominant position 
has a private neighbour.}

Before applying $\phi_{n,m}$, only the positions on 
the border $\mathbb{B}_{n,m,0}$ might not have any private neighbour. After Step i, every dominant position on $\mathbb{B}_{n,m,0}$ is isolated, or has a private neighbour. After Step ii, some positions may be dominant, non-isolated, and have no private neighbours. Such positions obtain a private neighbour in Step iii.

\end{itemize}

\item For all $n,m$, the number of preimages 
of $\phi_{n,m}$ for any minimal dominating set 
of $G_{n,m}$ is bounded by $2^{3(2n+2m)}$:
any symbol modified by the application 
is at distance at most two from $\mathbb{B}_{n,m,0}$, and there are $3(2n+2m)$ such symbols. Therefore $N_{n,m} (X^\mathrm{M}) \le 2^{6(n+m)} M_{n,m}$.

\end{enumerate}
\end{enumerate}
\end{proof}

\begin{lemma}
\label{lemma.comparison.minimal.total}
For all $n$, the following bounds hold: 
\[\frac{1}{2^{8(m+n)}} N_{n,m} (X^\mathrm{MT}) \le \mathit{MT}_{n,m} \le N_{n,m} (X^\mathrm{MT}).\]
\end{lemma}

For readability, we reproduce the structure of the proof of \Cref{lemma.comparison.minimal}, but simplify the arguments and refer to this proof.

\begin{proof}
\leavevmode
\begin{enumerate}
\item \textbf{Second inequality.} 

\begin{enumerate}
\item \textbf{A completion algorithm of a minimal total dominating set into a configuration of $X^\mathrm{MT}$.} 

Let us consider a minimal total dominating set of $G_{n,m}$. Any element in \linebreak$\llbracket 2, n-1 \rrbracket \times \llbracket 2, m-1 \rrbracket$ 
is dominated 
by an element of $\llbracket 1, n \rrbracket \times \llbracket 1, m \rrbracket$, and any dominant 
element in $\llbracket 2, n-1 \rrbracket \times \llbracket 2, m-1 \rrbracket$ is not isolated and has a private 
neighbour in $\llbracket 1, n \rrbracket \times \llbracket 1, m \rrbracket$ (which may or may not 
be a dominant position). Let us extend this set
into a configuration $x$ of $X^\mathrm{MT}$ using 
an algorithm very similar to the one in the corresponding 
point in the proof of \Cref{lemma.comparison.minimal}. The condition in the third point is different:
\begin{enumerate}
	\item if $\textbf{u}$ is a corner then $x_{\textbf{u}}$ is white;
	\item if $\textbf{u}$ is a neighbour of a corner in one of the vertical sides of $\mathbb{B}_{n,m,k+1}$ then $x_{\textbf{u}}$ is white;
	\item for every other $\textbf{u}$, $x_{\textbf{u}}$ is grey if and only if its neighbour in $\mathcal{N}^k (\llbracket 1,n \rrbracket \times \llbracket 1,m \rrbracket) $ is not dominated by an element in this set.
\end{enumerate}

\newpage
\item \textbf{The result of the algorithm 
is a configuration of $X^\mathrm{MT}$.}

\begin{itemize}
\item \textbf{Every position is dominated.} 

Similar to the corresponding point in 
the proof of \Cref{lemma.comparison.minimal}. This implies that no dominant positions are isolated.
\item \textbf{Every dominant position 
has a private neighbour.}

Let us consider a dominant position $\textbf{u}$. 
If it is in $\llbracket 3,n-2\rrbracket \times \llbracket 3,m-2\rrbracket $, 
since the pattern on $\llbracket 1,n\rrbracket \times \llbracket 1,m\rrbracket$ is a minimal total dominating set of $G_{n,m}$, we know that it has a private neighbour. Otherwise, it lies
in some $\mathbb{B}_{n,m,k}$ for $k \ge 0$, 
or in $\llbracket 2,n-1\rrbracket \times \llbracket 2,m-1\rrbracket$. Then there are two cases: 

\begin{itemize}
\item \textbf{$\textbf{u}$ is not a corner.} If it has no dominant neighbours in $\mathcal{N}^k (\llbracket 1,n \rrbracket \times \llbracket 1,m \rrbracket)$, let us call $\textbf{v}$ its neighbour in $\mathbb{B}_{n,m,k+1}$. Note that, depending on whether or not $\textbf{u}$ is dominated inside $\mathcal{N}^k (\mathbb{B}_{n,m,0})$, \textbf{v} may be white or grey. Since the neighbours of $\textbf{u}$ in $\mathbb{B}_{n,m,k}$ are dominated, $\textbf{v}$'s neighbours in $\mathbb{B}_{n,m,k+1}$ are white. Finally, since $\textbf{v}$ is dominated by $\textbf{u}$, its neighbour in $\mathbb{B}_{n,m,k+2}$ is white, hence $\textbf{v}$ is a private neighbour for $\textbf{u}$.

\item \textbf{ $\textbf{u}$ is a corner.} We apply a similar reasoning.

\end{itemize}

\end{itemize}
\end{enumerate}

\item \textbf{First inequality:} 

\begin{enumerate}
\item \textbf{A transformation of patterns of $X^\mathrm{MT}$ into minimal-total-dominating sets.} 

Let us define once again an application $\phi_{n,m}$ 
which, to each pattern of $X^\mathrm{MT}$ on $\llbracket 1, n\rrbracket \times \llbracket 1, m\rrbracket$, associates a minimal total 
dominating set of $G_{n,m}$. It is defined in a similar
way as in the corresponding point 
in the proof of \Cref{lemma.comparison.minimal}, but the proof is more complex.

\begin{enumerate}
\item Suppress any dominant position on the border 
$\mathbb{B}_{n,m,0}$ which has no private neighbours in $G_{n,m}$.
\item Successively, for every non-corner undominated position $\textbf{u}$ on 
the border $\mathbb{B}_{n,m,0}$, do the following:

\begin{itemize}
\item Consider the position $\textbf{v}$, neighbour of $\textbf{u}$ in $\llbracket 2,n-1\rrbracket \times \llbracket 2,m-1 \rrbracket$. For each dominant 
position $\textbf{w}$ in the neighbourhood of $\textbf{v}$, 
and for each dominant position $\textbf{w'}$ in the 
neighbourhood of $\textbf{w}$, if $\textbf{w}$ 
is the only private neighbour of $\textbf{w'}$, 
then change $\textbf{w'}$ into a non-dominant 
position.
\item Change $\textbf{v}$ into a dominant 
position.
\end{itemize}

Then do the same operations for the corners 
of $\mathbb{B}_{n,m,0}$, except that 
$\textbf{v}$ is replaced by any neighbour of the corner.

\end{enumerate}

This Step is illustrated 
on \Cref{figure.third.step.minimal.total}.

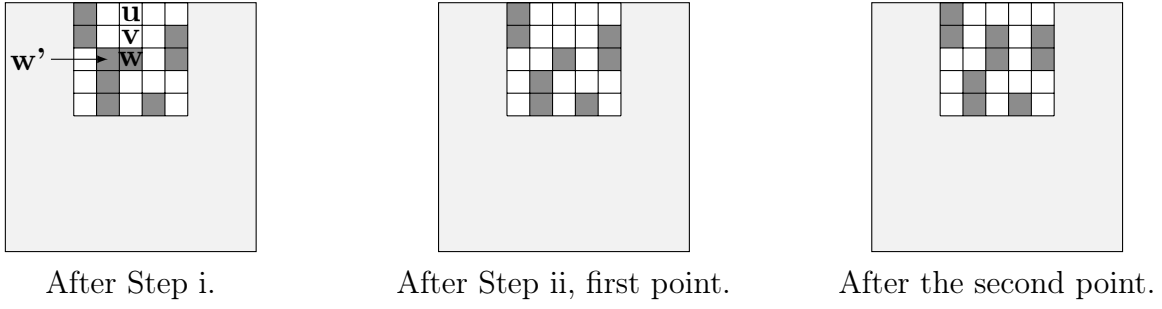
\begin{figure}[H]
\[\begin{tikzpicture}[scale=0.3]
\fill[gray!10] (-2,-2) rectangle (9,9);
\fill[white] (1,4) rectangle (6,9);
\fill[gray!90] (2,6) rectangle (4,7);
\fill[gray!90] (5,6) rectangle (6,7);
\fill[gray!90] (1,7) rectangle (2,9);
\fill[gray!90] (2,4) rectangle (3,6);
\fill[gray!90] (4,4) rectangle (5,5);
\fill[gray!90] (5,6) rectangle (6,8);
\draw (-2,-2) rectangle (9,9);
\draw (1,4) grid (6,9) ;
\node at (3.5,6.5) {$\textbf{w}$};
\node at (-1,6.5) {$\textbf{w'}$};
\draw[-latex] (0,6.5) -- (2.5,6.5);
\node at (3.5,7.5) {$\textbf{v}$};
\node at (3.5,8.5) {$\textbf{u}$};
\node at (3.5,-3.5) {After Step i.};

\begin{scope}[xshift=19cm]

\fill[gray!10] (-2,-2) rectangle (9,9);
\fill[white] (1,4) rectangle (6,9);
\fill[gray!90] (3,6) rectangle (4,7);
\fill[gray!90] (5,6) rectangle (6,7);

\fill[gray!90] (1,7) rectangle (2,9);
\fill[gray!90] (2,4) rectangle (3,6);
\fill[gray!90] (4,4) rectangle (5,5);
\fill[gray!90] (5,6) rectangle (6,8);
\draw (-2,-2) rectangle (9,9);
\draw (1,4) grid (6,9) ;
\node at (3.5,-3.5) {After Step ii, first point.};
\end{scope}

\begin{scope}[xshift=38cm]
\fill[gray!10] (-2,-2) rectangle (9,9);
\fill[white] (1,4) rectangle (6,9);
\fill[gray!90] (3,6) rectangle (4,7);
\fill[gray!90] (5,6) rectangle (6,7);
\fill[gray!90] (1,7) rectangle (2,9);
\fill[gray!90] (2,4) rectangle (3,6);
\fill[gray!90] (4,4) rectangle (5,5);
\fill[gray!90] (5,6) rectangle (6,8);

\fill[gray!90] (3,7) rectangle (4,8);
\draw (-2,-2) rectangle (9,9);
\draw (1,4) grid (6,9) ;
\node at (3.5,-3.5) {After the second point.};
\end{scope}
\end{tikzpicture}\]

\caption{Illustration of the second and then third steps 
of the algorithm defining $\phi_{n,m}$ for $X^\mathrm{MT}$, from 
left to right. $\textbf{u},\textbf{v},\textbf{w}$ and $\textbf{w'}$ are instances 
of the positions described in Rule ii.}
\label{figure.third.step.minimal.total}
\end{figure}

\item \textbf{Verification that images of $\phi_{n,m}$ 
are minimal total dominating sets.} 

Consider a pattern $p$ of $X^\mathrm{MT}$ on $\llbracket 1,n\rrbracket \times \llbracket 1, m\rrbracket$. The set $\phi_{n,m} (p)$ is 
a minimal total dominating set of $G_{n,m}$:

\begin{itemize}
\item \textbf{Any vertex of $G_{n,m}$ is dominated
in $\phi_{n,m} (p)$.}

Any (dominant or not) position which was dominated before applying Rule i is still dominated afterwards: if some position $\textbf{u}$ lies in the neighbourhood of a dominant position $\textbf{v}$ suppressed by 
Rule i, then since $\textbf{v}$ had no private neighbours in $G_{n,m}$, $\textbf{u}$ is dominated by another position.
For similar reasons, no positions become undominated after the application of Rule ii: only the neighbours of some $\textbf{w'}$ could be affected and if \textbf{w'} becomes non-dominant it means that they were dominated by other positions, so that they stay dominated.
Since all the positions inside $\llbracket 2,n-1\rrbracket \times \llbracket 2,m-1 \rrbracket$ were dominated before applying the rules, it only remains to show that the positions inside $\mathbb{B}_{n,m,0}$ are dominated after applying Rule ii. This is true thanks to this rule: any undominated position $\textbf{u}$ inside the border sees its neighbour $\textbf{v}$ inside  $\llbracket 2,n-1\rrbracket \times \llbracket 2,m-1 \rrbracket$ become dominant. The same applies to the corners, except that the neighbour comes from the border.

\item \textbf{Any dominant position 
has a private neighbour.}

At the end of Step i, any dominant position has a private neighbour. Only the creation of 
        a dominant position $\textbf{v}$ during the execution of Rule ii on position $\textbf{u}$ could affect this property, by disabling the private neighbour of a position $\textbf{w}$ in its neighbourhood, or by not having any private neighbour itself. The first case cannot happen since any dominant position $\textbf{w'}$ having $\mathbf{w}$ as its unique private neighbour is suppressed. The second one also never happens since the position $\textbf{u}$ is a private neighbour for $\textbf{v}$. 
\end{itemize}

\item For all $n$ and $m$, the number of preimages 
of $\phi_{n,m}$ for any minimal dominating set 
of $G_{n,m}$ is bounded by $2^{4(2m+2n)}$, 
since any symbol modified by the application 
is at distance at most $4$ of the border of $\llbracket 1,n\rrbracket \times \llbracket 1,m \rrbracket$. 
As a consequence, $N_{n,m} (X^\mathrm{MT}) \le 2^{8(m+n)} \mathit{MT}_{n,m}$.

\end{enumerate}
\end{enumerate}
\end{proof}

\begin{thm}[Asymptotic behaviour]
There exists some $\nu_\mathrm{D} \geq 0$ (resp. 
$\nu_\mathrm{M}$, $\nu_\mathrm{T}$ and $\nu_\mathrm{MT}$) 
such that

\[ D_{n,m} = \nu_\mathrm{D}^{ nm + o(nm)}\] 
(resp. $M_{n,m} = \nu_\mathrm{M}^{ nm + o(nm)}$, $T_{n,m} = \nu_\mathrm{T}^{nm + o(nm)}$ and $\mathit{MT}_{n,m} = \nu_\mathrm{MT}^{ nm + o(nm)}$).
\end{thm}

\begin{proof}
Let us prove this for the sequence $(M_{n,m})$ (the proof is similar for the other sequences).

As a consequence of 
\Cref{lemma.comparison.minimal}, for all $n,m$: 
\[- \frac{6(m+n)}{nm} + \frac{\log_2 (N_{n,m} (X^\mathrm{M}))}{n m}
\le \frac{\log_2(M_{n,m})}{nm} \le \frac{\log_2 (N_{n,m} (X^\mathrm{M}))}{nm}.\]

As a consequence, 
\[\frac{\log_2(M_{n,m})}{nm} \rightarrow h(X^\mathrm{M}).\]
This means that $M_{n,m} = 2^{h(X^\mathrm{M}) \cdot nm+ o(nm)} = \nu_\mathrm{M}^{nm+o(nm)}$, 
where $\nu_\mathrm{M} = 2^{h(X^\mathrm{M})}$.
\end{proof}

\section{Computability of the entropy: the block-gluing property}
\label{section-entropy}
In this section, we prove that the growth rates $\nu_\mathrm{D}$ (resp. $\nu_\mathrm{M}$, $\nu_\mathrm{T}$ and $\nu_\mathrm{MT}$) are computable numbers, meaning that there 
exists an algorithm which computes 
approximations of these numbers with 
arbitrary given precision. For this purpose, we rely on the \emph{block-gluing} property. If a subshift of finite type has this property then it allows us to compute it with the known \Cref{algo-computing-entropy} on page 73. We will finally prove that $X^\mathrm{D}$ (resp. $X^\mathrm{M}$, $X^\mathrm{T}$ and $X^\mathrm{MT}$) are block gluing.

\subsection{Definition and properties}

For two finite subsets $\mathbb{U},\mathbb{V}$ 
of $\mathbb{Z}^2$, we write
\[\delta(\mathbb{U},\mathbb{V})=\min_{\textbf{u} \in \mathbb{U}}( \min_{\textbf{v} \in \mathbb{V}} ||\textbf{v}-\textbf{u}||_{\infty}).\]
This corresponds to the shortest distance between a point in $\mathbb{U}$ and one in $\mathbb{V}$. Note that a square of length one contains four points in $\ZZ^2$ if placed on integer coordinates. This means that if there is one column between $\mathbb{U}$ and $\mathbb{V}$ then the distance between them is one (and not zero). The usual definition of the block-gluing property is the following one.

\begin{deff}
For a fixed integer $c \geq 0$, we say that a bidimensional subshift of finite type $X$ on alphabet $\mathcal{A}$ is $c$-block-gluing when, for every $n \geq 0$ and any two globally-admissible patterns $p$ and $q$ of $X$ on support 
$\llbracket 1,n \rrbracket ^2$, for all $\textbf{u},\textbf{v} \in \mathbb{Z}^2$ such that\linebreak $\delta(\textbf{u}+\llbracket 1,n \rrbracket ^2, \textbf{v}+\llbracket 1,n \rrbracket ^2) \ge c$, 
there exists a configuration $x \in X$ 
such that $x_{|\textbf{u}+\llbracket 1,n \rrbracket ^2} = p$ 
and $x_{|\textbf{v}+\llbracket 1,n \rrbracket ^2} = q$.\\
We say that $X$ is block gluing if it is $c$-block-gluing for some integer $c$.
\end{deff}

Informally, this means that any pair of rectangular patterns placed at whatever positions can be completed 
into a configuration of $X$ provided that there are at least $c$ lines or columns separating the two patterns.

\begin{notation}
For any subshift of finite type $X$, 
we denote by $c(X)$ the smallest $c$ such that 
$X$ is $c$-block-gluing. If $X$ is not block gluing for any integer $c$, we write $c(X)=+\infty$.
\end{notation}

In the following, we will use the notations $\mathbb{Z}_{-} = \,\rrbracket - \infty,0 \rrbracket$ and $\mathbb{Z}_{+} = \,\rrbracket 0,+ \infty \llbracket$.
We give in \Cref{proposition.semi.plane} an equivalent characterisation of the block-gluing property:

\begin{rk}
    In \Cref{proposition.semi.plane} we extend the notion of \emph{patterns} to \textbf{infinite supports}. This will be the case whenever their supports are infinite.
\end{rk}

\begin{prop}
\label{proposition.semi.plane}
Let $c \ge 0$ be an integer.
A bidimensional subshift $X$ is $c$-block-gluing
if and only if for all $k \ge c$ and 
$p$ and $q$ globally-admissible patterns on 
supports $\mathbb{Z}_{-} \times \mathbb{Z}$ (resp. $\mathbb{Z} \times \mathbb{Z}_{-}$) and 
$\mathbb{Z}_{+} \times \mathbb{Z}$ (resp. $\mathbb{Z} \times \mathbb{Z}_{+}$), there 
exists a configuration $x \in X$ such that 
$x_{|\mathbb{Z}_{-} \times \mathbb{Z}} = p$ 
and $x_{|(k,0)+\mathbb{Z}_{+} \times \mathbb{Z}}=q$ (resp. $x_{|\mathbb{Z} \times \mathbb{Z}_{-}} = p$ 
and $x_{|(0,k)+\mathbb{Z} \times \mathbb{Z}_{+}}=q$).
\end{prop}

Informally, this means that in order to check the block-gluing property, it is sufficient to prove that any two patterns on half-planes glued at arbitrary distance, provided it is greater than $c$, can be completed into a configuration of $X$.

\begin{proof} \leavevmode

\begin{itemize}

\item $(\Leftarrow)$: Let us assume that $X$ satisfies the second hypothesis. Let us consider some integer $n$, and two globally-admissible patterns $\overline{p},\overline{q}$ of $X$ on support $\llbracket 1,n \rrbracket^2$. 
We prove the result in the case of columns separating the two patterns: the proof when the patterns are separated by lines is completely symmetrical. We place $\overline{p}$ and $\overline{q}$ such that $k \geq c$ columns separate them.
Since $\overline{p}$ and $\overline{q}$ are globally 
admissible, there exist $p$ and $q$ globally-admissible infinite patterns of $X$ on respective supports 
$\mathbb{Z}_{-} \times \mathbb{Z}$ 
and $(k,0)+ \mathbb{Z}_{+} \times \mathbb{Z}$
which extend respectively $\overline{p}$ 
and $\overline{q}$. 
By hypothesis, there exists some configuration 
$x \in X$ which extends $p$ and $q$ separated by $k$ columns, hence $X$ is $c$-block-gluing.

\item $(\Rightarrow)$: Let us assume that the first hypothesis on $X$ is true, and let $p$ and $q$ be two patterns on supports $\mathbb{Z}_{-} \times \mathbb{Z}$ and $(k,0)+\mathbb{Z}_{+} \times \mathbb{Z}$ 
for some $k \ge c$ (the case with $\ZZ \times \ZZ_{-}$ and $\ZZ \times \ZZ_{+}$ is proved in the same way). From the block-gluing 
property, for all $n > 0$
one can extend the restriction of $p$ on 
$\llbracket 1, n \rrbracket^2 - (n,0)$ and the restriction 
of $q$ on $\llbracket 1, n \rrbracket^2 - (n,0)+(k+1,0)$
into a configuration $x_n \in X$.
By compactness of the subshift $X$ for
the product of the discrete topology, this sequence admits 
a subsequence which converges to some $x \in X$. This $x$ satisfies the equalities 
$x_{|\mathbb{Z}_{-} \times \mathbb{Z}}=p$ 
and $x_{|(k,0)+\mathbb{Z}_{+} \times \mathbb{Z}}=q$.
\end{itemize}
\end{proof}

\subsection{\label{subsection.computability} Algorithmic computability of the entropy}

\begin{deff}
    A function $f: \NN \rightarrow \NN$ is \textbf{computable} when there exists a Turing machine which, on any integer $n$ written in binary\footnote{Note that all binary digits are either 0 or 1. Sex, however, is not binary: there are intersex people who, on \textbf{biological sex} criterion, are neither male nor female. Gender also is not binary: some people categorise themselves as men, some as women, some as agender, gender-fluid, transgender, and a lot of other genders.} on its input tape, terminates with $f(n)$, also written in binary, on its output tape.
\end{deff}

\begin{deff}
Let $f: \mathbb{N} \rightarrow \mathbb{N}$ be a computable function.
A real number $x$ is said to be \textbf{computable} with rate $f$ when there exists an algorithm which, given an integer $n$ as input, outputs in at most $f(n)$ steps a rational number $r_n$ 
such that $|x-r_n|\le \frac{1}{n}$.
\end{deff}

This definition corresponds to Definition 1.3 in~\cite{Pavlov-Schraudner}. The following theorem is Theorem~1.4 in the same reference. Its proof provides an algorithm to compute $h(X)$.

\begin{thm}[\cite{Pavlov-Schraudner}]
\label{theorem.computability.entropy}
Let $X$ be a block-gluing bidimensional subshift of finite type. 
Then $h(X)$ is computable with rate $n \mapsto 2^{O(n^2)}$.
\end{thm}

\begin{rk}
Let us note that in general the entropy of a bidimensional subshift of finite type is not computable at all (see Theorem 1.1 in~\cite{Hochman-Meyerovitch} and the existence of non-computable right-recursively-enumerable numbers). This motivates our detour by the block-gluing property, which is needed for \Cref{algo-computing-entropy}.
\end{rk}

\begin{deff}
Given an SFT, every pattern which does not contain any forbidden pattern is called \textbf{locally admissible}.
\end{deff}

Any globally-admissible pattern is locally admissible: it can be extended to a configuration of the subshift in which no forbidden patterns appear. However the contrary is not true: if for instance we define a one-dimensional SFT on alphabet $\{a,b\}$ by forbidding patterns $aa$, $bb$ and $aba$ then $a$ is locally admissible but not globally admissible since any word of $\ZZ^{\{a,b\}}$ extending it will necessarily contain a forbidden pattern.

The following lemma gives us an effective way to compute the entropy of a SFT. It relates the entropy to the locally-admissible patterns instead of using globally-admissible patterns the way it is defined (\Cref{def-entropy}). Indeed, it is possible to check that a pattern is locally admissible by verifying that no forbidden patterns appear. On the contrary, there is no automatic way to check that a given pattern is globally admissible: one way would be to find an extension of it to $\ZZ^2$, but this might not end in finite time.

\begin{lemma}[\cite{Pavlov-Schraudner}]
\label{lemma.counting.block.gluing}
Let $X$ be a $c$-block-gluing 
bidimensional subshift of finite type 
on alphabet $\mathcal{A}$. For all $k \ge 1$, 
the number $N_k (X)$ is equal to the number of $k \times k$
 patterns which appear in a $\left(|\mathcal{A}|^{2c+1} \cdot (c+k)+1 \right) \times (2c+k+2)$ locally-admissible rectangular pattern whose restrictions on the two extremal vertical (resp. 
 horizontal) edges are equal.
\end{lemma}

The algorithm is as follows:
\begin{center} \begin{algorithm}[H]
\SetAlgoLined
\SetKwData{Left}{left}
\SetKwData{This}{this}
\SetKwData{Up}{up}
\SetKwFunction{Union}{Union}
\SetKwFunction{FindCompress}{FindCompress}

\SetKwInput{Input}{Input}
\SetKwInOut{Output}{Output}
\Input{An integer $n$, an alphabet $\mathcal{A}$ and a set of patterns $\mathcal{F}$ of $\mathcal{A}^{\mathbb{U}}$ for 
some finite $\mathbb{U} \subset \mathbb{Z}^2$}
\Output{A rational approximation of $h(X)$ up to $1/n$, where $X$ is the SFT on alphabet $\mathcal{A}$ defined by the set of forbidden patterns $\mathcal{F}$}
 $k \leftarrow 0$\\
 $r \leftarrow +\infty$\\
 \While {$r \geq 1/2n$}
 {
 $k \leftarrow k + 1$\\
 $m \leftarrow N_k (X)$ (this is a sub-procedure 
 using \Cref{lemma.counting.block.gluing}).
 
 $r \leftarrow $ some rational approximation up to $1/2k$ of $\frac{\log_2 (N_k(X))}{k^2} - \frac{\log_2 (N_k(X))}{(k+c)^2}$
 } 
 
 Return a rational approximation up to $1/2n$ of $\log_2 (N_k (X))/k^2$ 
 \caption{Computing the entropy of a $c$-block-gluing bidimensional SFT.}
 \label{algo-computing-entropy}
\end{algorithm}
\end{center}

\subsection{Some dominating subshifts are block-gluing}

It is straightforward to check that 
the domination subshift $X^\mathrm{D}$ and the total domination subshift $X^\mathrm{T}$
satisfy the block-gluing property, 
with $c(X^\mathrm{D})=1$ (just fill every cell with grey). In 
this section, we prove that 
$X^\mathrm{M}$ and $X^\mathrm{MT}$ also satisfy this property. However, to show that this problem is not trivial, we first show that not all domination problems have the block-gluing property.

\begin{deff}
Let $\sigma$ and $\rho$ be two sets of integers (subsets of $\llbracket 0, 4 \rrbracket$ in the case of grids). $S$ is said to be \textbf{$\bm{(\sigma, \rho)}$-dominating} when \begin{enumerate}
\item the number of neighbours in $S$ of each vertex in $S$ belongs to $\sigma$;
\item the number of neighbours in $S$ of each vertex outside $S$ belongs to $\rho$.
\end{enumerate}
\end{deff}

For instance, the domination is the $(\{0,1,2,3,4\}, \{1,2,3,4\})$-domination and the total domination is the $(\{1,2,3,4\}, \{1,2,3,4\})$-domination.

\begin{prop}
The $(\{3\}, \{1\})$-domination subshift is not block-gluing.
\end{prop}

\begin{notation}
In the following, for all $j \in \mathbb{Z}$, 
we denote by $C_j$ the column $\{j\} \times \mathbb{Z}$ of $\mathbb{Z}^2$.
\end{notation}

\begin{proof}
We prove this result by providing a half-plane pattern which can be glued with itself only if the number of columns between them is of the form $4k+2$.

We use \Cref{proposition.semi.plane} so that instead of finite pattern we may use a half-plane pattern. We take the half-plane pattern $p$ whose rightmost two columns (i.e. $C_0$ and $C_{-1}$) are filled with grey (i.e. dominant elements), the two at their left are filled with white, the two at their left with grey, and so on: $C_j$ for $j \leq 0$ is grey if $-j \mod 4 \in \{0,1\}$ and is white otherwise (see \Cref{figure-not-block-gluing-3-1}). We take $q$ as the vertical symmetric of $p$, as in \Cref{figure-not-block-gluing-3-1}. Since in the column called $C_0$ in the figure, each vertex has exactly 3 neighbours in the dominating set, its neighbour in $C_1$ must be white. Every element in $C_1$ having a neighbour which dominates them, each neighbours of an element of $C_1$ must be white, hence $C_2$ is also white. Now the elements of $C_2$ are not dominated, so $C_3$ must be grey. Every element of $C_2$ now has only 2 grey neighbours, so that $C_4$ must be grey, and so on. $q$ forces every column $C_{4k-1}$ and $C_{4k}$ to be grey, and every $C_{4k+1}$ and $C_{4k+2}$ to be white. However, $q$ forces the exact opposite, which means that for any $k \geq 0$, $p$ and $q$ cannot be glued with a gap of size $4k$ for instance, hence the subshift is not block gluing.

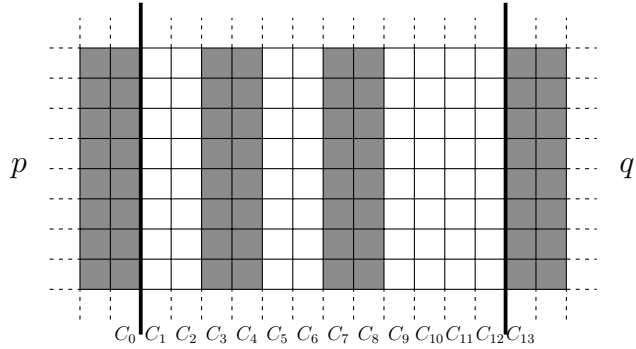
\begin{figure}[h!]
\centering
\begin{tikzpicture}[scale=0.4]
\fill[gray!90] (1,0) rectangle (3,8);
\fill[gray!90] (5,0) rectangle (7,8);
\fill[gray!90] (9,0) rectangle (11,8);
\fill[gray!90] (15,0) rectangle (17,8);

\foreach \x in {1,...,17} {\draw[mydashed] (\x, -1) -- (\x, 0);}
\foreach \x in {1,...,17} {\draw[mydashed] (\x, 8) -- (\x, 9);}

\foreach \x in {0,...,8} {\draw[mydashed] (0,\x) -- (1,\x);}
\foreach \x in {0,...,8} {\draw[mydashed] (17,\x) -- (18,\x);}
\draw (1,0) grid (17,8); 

\node at (-1,4) {$p$};
\node at (19,4) {$q$};

\foreach \x in {0,...,13} {\node[scale=0.625] at (2.5+\x,-1.5) {$C_{\x}$};}

\draw[line width =0.5mm] (3,-1.5) -- (3,9.5);
\draw[line width =0.5mm] (15,-1.5) -- (15,9.5);
\end{tikzpicture}
	
    \caption{The $\{\{3\},\{1\}\}$ domination is not block gluing. The half-plane pattern $p$ forces every $C_{4k-1}$ and $C_{4k}$ to be grey and the rest to be white, whereas $q$ forces the opposite. Hence they cannot be glued whenever $4k$ columns separate them.}
\label{figure-not-block-gluing-3-1}
\end{figure}
\end{proof}

\begin{rk}
\Cref{proposition.semi.plane} is very useful to avoid some technicalities, which we briefly mention here for the sake of insight. If we only considered finite patterns $p$ and $q$ as they are in \Cref{figure-not-block-gluing-3-1}, the proof would have been longer. Let us consider in this remark that $p$ and $q$ are patterns with 8 lines an 2 columns as in the figure. Then the top visible cell of $C_2$ could have been dominated by either its upper neighbour or the one in $C_3$. So we would have had to argue that this is not possible, or that it is not a problem: the sagging of the grey cells is linear in the size of the gap, so taking arbitrarily large (but finite) patterns will show that any gap of constant size would not work.
\label{rk-sagging}
\end{rk}

\begin{thm}
\label{minimal-gluing-th}
The minimal domination subshift is block gluing and $c(X^\mathrm{M})=5$.
\end{thm}

\noindent \textbf{Idea of the proof:}
\textit{In order to simplify the proof of the block-gluing property, we rely again on \Cref{proposition.semi.plane}. The proof of the block-gluing property for two half-plane patterns consists in determining successively how to fill the $k$ intermediate columns from the patterns towards the "centre" (chosen, for concision, to be column $C_{k-2}$). The completion follows an algorithm which enforces that, when the number of intermediate columns is large enough, any added dominant element is isolated or has a private neighbour in an already-constructed column. This ensures that the rules of the subshift are not broken. We now give the full proof.}

\begin{proof} \leavevmode
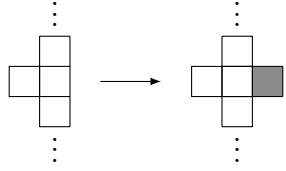
\begin{figure}[h!]
\centering
\begin{tikzpicture}[scale=0.4]
\begin{scope}
\draw (0,0) grid (1,3);
\draw (-1,1) rectangle (0,2);
\node at (0.5,-0.5) {$\vdots$};
\node at (0.5,4) {$\vdots$};

\draw[-latex] (2,1.5) -- (4,1.5);

\begin{scope}[xshift=6cm]
\node at (0.5,-0.5) {$\vdots$};
\node at (0.5,4) {$\vdots$};
\fill[gray!90] (1,1) rectangle (2,2);
\draw (0,0) grid (1,3);
\draw (-1,1) grid (2,2);
\end{scope}
\end{scope}

\end{tikzpicture}
\caption{Illustration of the rule for filling the non-central intermediate columns for $X^\mathrm{M}$. It is also applied symmetrically if the neighbour in $C_{k-1}$ is not dominated.} 
\label{figure.intermediate.columns}
\end{figure}

\begin{figure}[h!]
\centering
\begin{tikzpicture}[scale=0.3]

\fill[gray!90] (0,3) rectangle (1,8);
\fill[gray!90] (0,0) rectangle (1,1);
\fill[gray!90] (1,2) rectangle (2,3);

\fill[gray!90] (10,3) rectangle (11,8);
\fill[gray!90] (8,1) rectangle (10,2);

\foreach \x in {-1,...,2} {\draw[mydashed] (\x, -1) -- (\x, 0);}
\foreach \x in {-1,...,2} {\draw[mydashed] (\x, 8) -- (\x, 9);}
\foreach \x in {8,...,11} {\draw[mydashed] (\x, -1) -- (\x, 0);}
\foreach \x in {8,...,11} {\draw[mydashed] (\x, 8) -- (\x, 9);}

\foreach \y in {0,...,8} {\draw[mydashed] (-1,\y) -- (-2,\y);}
\foreach \y in {0,...,8} {\draw[mydashed] (11,\y) -- (12,\y);}

\draw (-1,0) grid (2,8);
\draw (8,0) grid (11,8);
\node at (-3,4) {$p$};
\node at (13,4) {$q$};
\node[scale=0.625] at (2.6,-2) {$C_1$};
\node[scale=0.625] at (1.6,-2) {$C_0$};
\node[scale=0.625] at (8.95,-2) {$C_{k+1}$};
\node[scale=0.625] at (7.7,-2) {$C_k$};
\node[scale=0.7] at (5,-2) {$\hdots$};

\draw[line width =0.5mm] (2,-1.5) -- (2,9.5);
\draw[line width =0.5mm] (8,-1.5) -- (8,9.5);

\node[scale=1.25] at (5,-3.5) {$(0)$};

\begin{scope}[xshift=20cm]

\fill[gray!90] (0,3) rectangle (1,8);
\fill[gray!90] (0,0) rectangle (1,1);
\fill[gray!90] (1,2) rectangle (2,3);
\fill[gray!90] (10,3) rectangle (11,8);
\fill[gray!90] (8,1) rectangle (10,2);

\fill[gray!90] (3,3) rectangle (4,7);
\fill[gray!90] (3,1) rectangle (4,2);
\fill[gray!90] (7,3) rectangle (8,7);

\foreach \x in {-1,...,2} {\draw[mydashed] (\x, -1) -- (\x, 0);}
\foreach \x in {-1,...,2} {\draw[mydashed] (\x, 8) -- (\x, 9);}
\foreach \x in {8,...,11} {\draw[mydashed] (\x, -1) -- (\x, 0);}
\foreach \x in {8,...,11} {\draw[mydashed] (\x, 8) -- (\x, 9);}

\foreach \x in {3,4,5,6,7} {\draw[mydashed] (\x, -1) -- (\x, 0);}
\foreach \x in {3,4,5,6,7} {\draw[mydashed] (\x, 8) -- (\x, 9);}

\foreach \y in {0,...,8} {\draw[mydashed] (-1,\y) -- (-2,\y);}
\foreach \y in {0,...,8} {\draw[mydashed] (11,\y) -- (12,\y);}
\draw (-1,0) grid (2,8); 
\draw (8,0) grid (11,8);
\draw (2,0) grid (5,8);
\draw (6,0) grid (8,8);

\draw[white, line width=0.5mm] (3.03,0) -- (3.96,0); 
\draw[white, line width=0.5mm] (3.03,8) -- (3.96,8); 
\draw[white, line width=0.5mm] (5.03,8) -- (5.96,8); 
\draw[white, line width=0.5mm] (5.03,0) -- (6.96,0); 
\draw[white, line width=0.5mm] (6,-1) -- (6,0.95);

\draw[mydashed] (6, -1) -- (6, 1);

\draw[line width =0.5mm] (2,-1.5) -- (2,9.5);
\draw[line width =0.5mm] (8,-1.5) -- (8,9.5);

\node[scale=0.7] at (6,-2) {$C_{k-2}$};
\node[scale=1.25] at (5,-3.5) {$(1)$};

\end{scope}

\begin{scope}[yshift=-15cm]

\fill[gray!90] (0,3) rectangle (1,8);
\fill[gray!90] (0,0) rectangle (1,1);
\fill[gray!90] (1,2) rectangle (2,3);
\fill[gray!90] (10,3) rectangle (11,8);
\fill[gray!90] (8,1) rectangle (10,2);

\fill[gray!90] (3,3) rectangle (4,7);
\fill[gray!90] (3,1) rectangle (4,2);
\fill[gray!90] (7,3) rectangle (8,7);

\fill[gray!90] (5,1) rectangle (6,3);

\foreach \x in {-1,...,2} {\draw[mydashed] (\x, -1) -- (\x, 0);}
\foreach \x in {-1,...,2} {\draw[mydashed] (\x, 8) -- (\x, 9);}
\foreach \x in {8,...,11} {\draw[mydashed] (\x, -1) -- (\x, 0);}
\foreach \x in {8,...,11} {\draw[mydashed] (\x, 8) -- (\x, 9);}

\foreach \x in {3,4,5,6,7} {\draw[mydashed] (\x, -1) -- (\x, 0);}
\foreach \x in {3,4,5,6,7} {\draw[mydashed] (\x, 8) -- (\x, 9);}

\foreach \y in {0,...,8} {\draw[mydashed] (-1,\y) -- (-2,\y);}
\foreach \y in {0,...,8} {\draw[mydashed] (11,\y) -- (12,\y);}
\draw (-1,0) grid (2,8); 
\draw (8,0) grid (11,8);
\draw (2,0) grid (8,8);

\draw[white, line width=0.5mm] (3.03,0) -- (3.96,0); 
\draw[white, line width=0.5mm] (3.03,8) -- (3.96,8); 
\draw[white, line width=0.5mm] (5.03,8) -- (5.96,8); 
\draw[white, line width=0.5mm] (5.03,0) -- (6.96,0); 
\draw[white, line width=0.5mm] (6,-1) -- (6,0.95);

\draw[mydashed] (6,-1) -- (6,1);
\draw[line width =0.5mm] (2,-1.5) -- (2,9.5);
\draw[line width =0.5mm] (8,-1.5) -- (8,9.5);

\node[scale=1.25] at (5,-3) {$(2)$};
\end{scope}

\begin{scope}[yshift=-15cm,xshift=20cm]

\fill[gray!90] (0,3) rectangle (1,8);
\fill[gray!90] (0,0) rectangle (1,1);
\fill[gray!90] (1,2) rectangle (2,3);
\fill[gray!90] (10,3) rectangle (11,8);
\fill[gray!90] (8,1) rectangle (10,2);

\fill[gray!90] (3,3) rectangle (4,7);
\fill[gray!90] (3,1) rectangle (4,2);

\fill[gray!90] (7,3) rectangle (8,7);

\fill[gray!90] (5,1) rectangle (6,3);
\fill[gray!90] (5,4) rectangle (6,5);
\fill[gray!90] (5,6) rectangle (6,7);

\foreach \x in {-1,...,2} {\draw[mydashed] (\x, -1) -- (\x, 0);}
\foreach \x in {-1,...,2} {\draw[mydashed] (\x, 8) -- (\x, 9);}
\foreach \x in {8,...,11} {\draw[mydashed] (\x, -1) -- (\x, 0);}
\foreach \x in {8,...,11} {\draw[mydashed] (\x, 8) -- (\x, 9);}

\foreach \x in {3,4,5,6,7} {\draw[mydashed] (\x, -1) -- (\x, 0);}
\foreach \x in {3,4,5,6,7} {\draw[mydashed] (\x, 8) -- (\x, 9);}

\foreach \y in {0,...,8} {\draw[mydashed] (-1,\y) -- (-2,\y);}
\foreach \y in {0,...,8} {\draw[mydashed] (11,\y) -- (12,\y);}
\draw (-1,0) grid (2,8); 
\draw (8,0) grid (11,8);
\draw (2,0) grid (8,8);

\draw[white, line width=0.5mm] (3.03,0) -- (3.96,0); 
\draw[white, line width=0.5mm] (3.03,8) -- (3.96,8); 
\draw[white, line width=0.5mm] (5.03,8) -- (5.96,8); 
\draw[white, line width=0.5mm] (5.03,0) -- (6.96,0); 
\draw[white, line width=0.5mm] (6,-1) -- (6,0.95);

\draw[mydashed] (6,-1) -- (6,1);

\draw[line width =0.5mm] (2,-1.5) -- (2,9.5);
\draw[line width =0.5mm] (8,-1.5) -- (8,9.5);

\node[scale=1.25] at (5,-3) {$(3)$};
\end{scope}
\end{tikzpicture}

\caption{Illustration of the algorithm
filling the intermediate columns between two half-plane patterns $p$ and $q$ for the minimal domination.\\
(0) Initial setting of the two patterns.\\
(i) After Step i of the algorithm\\
Since we show only finite sub-parts of $p$ and $q$, the values of some cells in between remain unknown. We left them 
non filled and remove one of their boundary. We chose $k=6$ for the illustration, still the proof works from $k=5$ upwards.}
\label{figure.completing.algorithm}
\end{figure}

\begin{itemize}
\item \textbf{Filling the intermediate columns 
between two half-plane patterns.}

Let $p$ and $q$ be two patterns respectively on 
$\mathbb{Z}_{-} \times \mathbb{Z}$ and 
$\mathbb{Z}_{+} \times \mathbb{Z}$ (the 
proof for the vertical case is similar). Let 
us determine a configuration of $\mathcal{A}_0^{\mathbb{Z}^2}$ such that $x_{|\mathbb{Z}_{-} \times \mathbb{Z}}=p$ and $x_{|(k,0)+\mathbb{Z}_{+} \times \mathbb{Z}}=q$. 
The intermediate 
columns $C_1,...,C_k$ are determined by the following algorithm. It intuitively puts grey cells only when it is really necessary to dominate the neighbour in the previous column:

\begin{enumerate}
\vspace{-0,1cm}
\setlength{\parskip}{0pt}
\setlength{\itemsep}{3pt}
\item \textbf{Filling the intermediate columns, from $C_1$ to $C_{k-3}$, then $C_k,C_{k-1}$.} 

Successively, for all $j$ 
from $1$ to $k-3$, we determine the column 
$C_j$ according to the following rule: 
for all $\textbf{u} \in C_j$, $x_{\textbf{u}}$ 
is $\begin{tikzpicture}[scale=0.3,baseline=0.4mm]
\fill[gray!90] (0,0) rectangle (1,1);
\draw (0,0) rectangle (1,1);
\end{tikzpicture}$ when $x_{\textbf{u}-(1,0)}$, 
$x_{\textbf{u}-(1,1)}$, $x_{\textbf{u}-(1,-1)}$ 
and $x_{\textbf{u}-(2,0)}$ are $\begin{tikzpicture}
[scale=0.3,baseline=0.4mm]
\draw (0,0) rectangle (1,1);
\end{tikzpicture}$ (see \Cref{figure.intermediate.columns}). Else, $x_{\textbf{u}}$ is set to $\begin{tikzpicture}[scale=0.3,baseline=0.4mm]
\draw (0,0) rectangle (1,1);
\end{tikzpicture}$.
Similarly, for $j=k$ and then $j=k-1$, 
we determine $x$ on any position $x_{\textbf{u}}$ 
for $\textbf{u} \in C_j$ by applying 
a symmetrical rule: $x_{\textbf{u}}$ 
is $\begin{tikzpicture}[scale=0.3,baseline=0.4mm]
\fill[gray!90] (0,0) rectangle (1,1);
\draw (0,0) rectangle (1,1);
\end{tikzpicture}$ when $x_{\textbf{u}+(1,0)}$, 
$x_{\textbf{u}+(1,1)}$, $x_{\textbf{u}+(1,-1)}$ 
and $x_{\textbf{u}+(2,0)}$ are
 $\begin{tikzpicture}[scale=0.3,baseline=0.4mm]
\draw (0,0) rectangle (1,1);
\end{tikzpicture}$. Else, $x_{\textbf{u}}$ is set to
$\begin{tikzpicture}[scale=0.3,baseline=0.4mm]
\draw (0,0) rectangle (1,1);
\end{tikzpicture}$.

\item \textbf{The central column $C_{k-2}$.}

We now 
determine $x$ on the central column $C_{k-2}$. For all $\textbf{u} \in C_{k-2}$, $x_{\textbf{u}}$ 
is $\begin{tikzpicture}[scale=0.3,baseline=0.4mm]
\fill[gray!90] (0,0) rectangle (1,1);
\draw (0,0) rectangle (1,1);

\end{tikzpicture}$ when $x_{\textbf{u}+(1,0)}$, 
$x_{\textbf{u}+(1,1)}$, $x_{\textbf{u}+(1,-1)}$ 
and $x_{\textbf{u}+(2,0)}$ are equal to 
$\begin{tikzpicture}[scale=0.3,baseline=0.4mm]
\draw (0,0) rectangle (1,1);
\end{tikzpicture}$,
or when $x_{\textbf{u}-(1,0)}$, 
$x_{\textbf{u}-(1,1)}$, $x_{\textbf{u}-(1,-1)}$ 
and $x_{\textbf{u}-(2,0)}$ are equal to 
$\begin{tikzpicture}[scale=0.3,baseline=0.4mm]
\draw (0,0) rectangle (1,1);
\end{tikzpicture}$. Else, it is 
$\begin{tikzpicture}[scale=0.3,baseline=0.4mm]
\draw (0,0) rectangle (1,1);
\end{tikzpicture}$.

\item \textbf{Eliminating non-domination errors 
in the central column.}
Choose any position $\bm{u_0} \in C_{k-2}$ and check if 
this position has a symbol $\begin{tikzpicture}[scale=0.3,baseline=0.4mm]
\fill[gray!90] (0,0) rectangle (1,1);
\draw (0,0) rectangle (1,1);
\end{tikzpicture}$ in its neighbourhood. 
If not, then set the symbol $\begin{tikzpicture}[scale=0.3,baseline=0.4mm]
\fill[gray!90] (0,0) rectangle (1,1);
\draw (0,0) rectangle (1,1);
\end{tikzpicture}$ on this position. Repeat this from $\bm{u_0}+(0,1)$ upwards, and in parallel from $\bm{u_0}-(0,1)$ downwards.

\end{enumerate}

See an illustration of this algorithm on 
\Cref{figure.completing.algorithm}.

\item \textbf{The obtained configuration 
is in $X^\mathrm{M}$.}

We have to check that the configuration $x$ we constructed satisfies the local 
rules of the minimal domination subshift. We divide this part of the proof according to whether or not the cells belong to $p$ or $q$, or lie outside. For concision, this division is approximate: columns $C_{-1}$ and $C_{k+2}$ are checked in the second and third point instead of the first one.

\begin{enumerate}

\item \textbf{The local rules are satisfied in the "interior" of the half-planes.}

By hypothesis, the patterns $p$ and $q$ 
are globally admissible in $X^\mathrm{M}$. As a consequence, for all $\textbf{u}$ 
in $\rrbracket - \infty,-2\rrbracket \times \mathbb{Z}$ or $\llbracket k+3,+\infty\llbracket \times \mathbb{Z}$, the pattern $x_{|\textbf{u}+\llbracket - 2,2 \rrbracket^2}$ is not forbidden.
It remains to check that no forbidden patterns are created through the execution of the algorithm described in the first point of the proof, for columns $C_{-1}$ to $C_{k+2}$.

\item \textbf{Every position in $C_{-1}, \cdots, C_{k+2}$ and not in $S$ is dominated.}

In the columns $C_{-1}$ and $C_{k+2}$, this 
comes from the fact that the patterns $p$ and $q$ 
are globally admissible. For $j$ between $0$ and $k-3$, 
and $\textbf{u} \in C_j$, if 
$\textbf{u}$ is not dominated by a position 
in $C_j$ or $C_{j-1}$
the position $\textbf{u}+(1,0)$ contains the symbol $\begin{tikzpicture}[scale=0.3,baseline=0.4mm]
\fill[gray!90] (0,0) rectangle (1,1);
\draw (0,0) rectangle (1,1);
\end{tikzpicture}$ (by the first and second steps of the algorithm), and thus $\textbf{u}$ 
is dominated. A symmetrical reasoning 
works for the positions in the columns $C_{k+1},C_{k},C_{k-1}$. For a position 
in the central column $C_{k-2}$, this is guaranteed by Step 3.

\item \textbf{Every dominant position in $C_{-1}, \cdots, C_{k+2}$ is isolated or has a private neighbour not in $S$.}

Let us consider a non-isolated dominant position $\textbf{u}$.
\begin{enumerate}
\item If it lies in $C_{-1}$ (resp. $C_{k+2}$), 
$\textbf{u}$ has a private neighbour
in a configuration $x \in X$ 
that extends $p$ (resp. $q$). 
If this private neighbour is in fact in $p$ (resp. $q$), in column $C_{-2}$ 
or $C_{-1}$ (resp. $C_{k+2}$ or $C_{k+3}$), then 
it stays a private neighbour of $\textbf{u}$ in $x$ since the intermediate columns cannot make it not private. If it is $\textbf{u} + (1,0)$ (resp. $\textbf{u} - (1,0)$), then it stays 
a private neighbour in $x$: since this position 
is dominated by $\textbf{u}$ according to the first step  (resp. second step)
of the algorithm, it 
is not dominated in $x$ by a position in $C_0$ 
(resp. $C_{k+1}$). The same reasoning is applied
to positions in columns $C_0$ and $C_{k+1}$.
\item In the other columns $C_j$ for $j < k-2$,
the first step guarantees, for any position $\textbf{u}$ 
in $C_j$ that is dominant, that the position 
$\textbf{u}-(1,0)$ is a private neighbour. 
A symmetric reasoning applies to column $C_{k}$ 
and $C_{k-1}: \textbf{u}+(1,0)$ is a private neighbour in that case.
\item If $\textbf{u}$ is in 
the column $C_{k-2}$, it means that it 
was introduced in either the second or the third step, 
meaning that it has a private neighbour in 
        column $C_{k-3}$ or $C_{k-1}$ (if introduced in the second step), or in $C_{k-2}$ (if introduced in the third step).
\end{enumerate}
\end{enumerate}

\item \textbf{The subshift $X^\mathrm{M}$ is not $4$-block-gluing.}

We consider the two half-plane patterns $p$ and $q$ on respective supports 
$\mathbb{Z}_{-} \times \mathbb{Z}$ and 
$\mathbb{Z}_{+} \times \mathbb{Z}$ such that
for all $j \le 0$, if $-j \equiv 0,1 \mod 4$,
then for all $\textbf{u} \in C_j$, $p_{\textbf{u}}$ is $\begin{tikzpicture}[scale=0.3,baseline=0.4mm]
\fill[gray!90] (0,0) rectangle (1,1);
\draw (0,0) rectangle (1,1);
\end{tikzpicture}$, else for all $\textbf{u} \in C_j$, it is $\begin{tikzpicture}[scale=0.3,baseline=0.4mm]
\draw (0,0) rectangle (1,1);
\end{tikzpicture}$, and $q$ is obtained from $p$ 
by vertical symmetry. It is easy to see that these patterns are globally admissible. We leave 4 columns between $p$ and $q$ (see \Cref{minimal-4-gluing}).
To ensure that the dominant positions in columns 0 and 5, which are not isolated, have private neighbours, every cell of the middle four columns needs to be $\begin{tikzpicture}[scale=0.3,baseline=0.4mm]
\draw (0,0) rectangle (1,1);
\end{tikzpicture}$, as in \Cref{minimal-4-gluing}. This filling implies that the cells in column 2 and 3, which are not dominant, are also not dominated. This shows that the subshift is not 4-block-gluing.

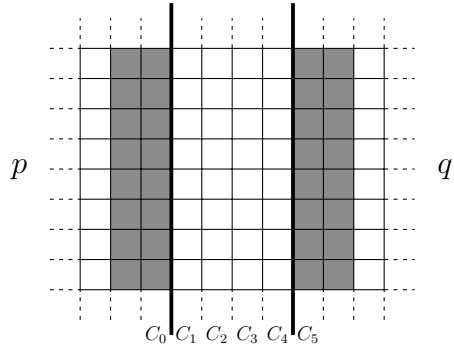
\begin{figure}[h!]
\centering
\begin{tikzpicture}[scale=0.4]
\fill[gray!90] (1,0) rectangle (3,8);
\fill[gray!90] (7,0) rectangle (9,8);

\foreach \x in {0,...,2} {\draw[mydashed] (\x, -1) -- (\x, 0);}
\foreach \x in {0,...,2} {\draw[mydashed] (\x, 8) -- (\x, 9);}
\foreach \x in {8,...,10} {\draw[mydashed] (\x, -1) -- (\x, 0);}
\foreach \x in {8,...,10} {\draw[mydashed] (\x, 8) -- (\x, 9);}

\foreach \x in {0,...,8} {\draw[mydashed] (-1,\x) -- (0,\x);}
\foreach \x in {0,...,8} {\draw[mydashed] (10,\x) -- (11,\x);}
\draw (0,0) grid (10,8); 

\node at (-2,4) {$p$};
\node at (12,4) {$q$};
\foreach \x in {4,...,6} {\draw[mydashed] (\x, -1) -- (\x, 0);}
\foreach \x in {4,...,6} {\draw[mydashed] (\x, 8) -- (\x, 9);}

\node[scale=0.625] at (2.5,-1.5) {$C_0$};
\node[scale=0.625] at (3.5,-1.5) {$C_1$};
\node[scale=0.625] at (4.5,-1.5) {$C_2$};
\node[scale=0.625] at (5.5,-1.5) {$C_3$};
\node[scale=0.625] at (6.5,-1.5) {$C_4$};
\node[scale=0.625] at (7.5,-1.5) {$C_5$};

\draw[line width =0.5mm] (3,-1.5) -- (3,9.5);
\draw[line width =0.5mm] (7,-1.5) -- (7,9.5);
\end{tikzpicture}
	
\caption{Illustration of the fact that $X^\mathrm{M}$
is not $4$-block-gluing: when attempting to glue $p$ and $q$, ensuring the existence of private neighbours according to the rules of $X^\mathrm{M}$ 
forces the presence of undominated positions 
in columns $C_2$ and $C_3$. It is also a counter-example for $X^\mathrm{MT}$ being 4-block-gluing.}
\label{minimal-4-gluing}
\end{figure}
As a consequence, $c(X^\mathrm{M}) > 4$. Since it is $5$-block-gluing, $c(X^\mathrm{M}) = 5$.
\end{itemize}
\end{proof}

\begin{thm}
The minimal total domination subshift is block gluing and $c(X^\mathrm{MT}) = 5$.
\label{total-minimum-gluing-th}
\end{thm}

\noindent \textbf{Idea of the proof:}
\textit{We follow the same 
scheme as in the proof of \Cref{minimal-gluing-th}, 
except that we have to take into account 
the variations in the definition of the 
subshift $X^\mathrm{MT}$. For the sake 
of readability, we reproduce the structure 
of the proof. We now give the proof.}

\begin{proof}
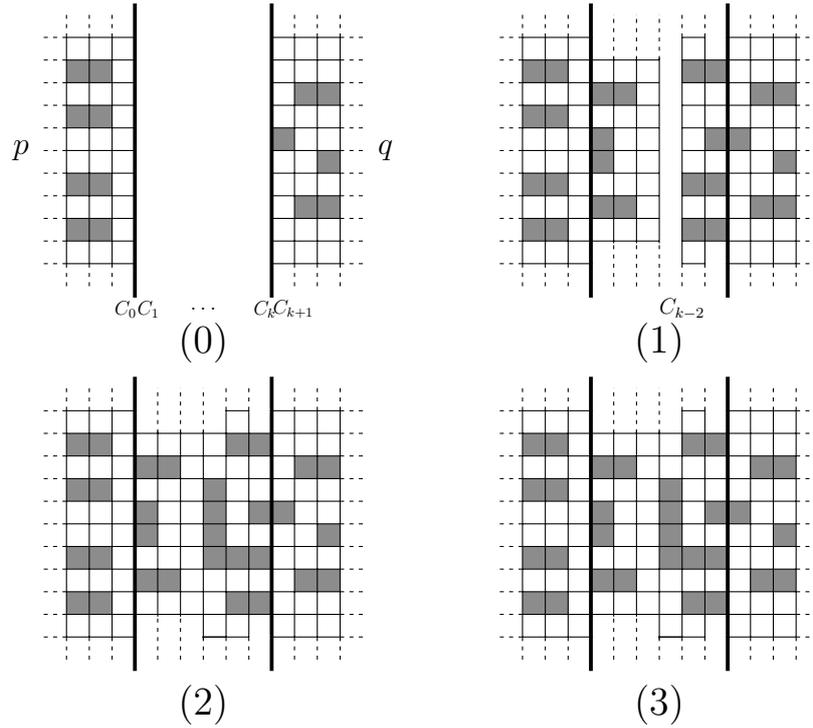
\begin{figure}[h!]

\centering
\begin{tikzpicture}[scale=0.3]

\fill[gray!90] (-1,2) rectangle (1,3);
\fill[gray!90] (-1,0) rectangle (1,1);
\fill[gray!90] (-1,5) rectangle (1,6);
\fill[gray!90] (-1,7) rectangle (1,8);
\fill[gray!90] (8,4) rectangle (9,5);
\fill[gray!90] (9,1) rectangle (11,2);
\fill[gray!90] (10,3) rectangle (11,4);
\fill[gray!90] (9,6) rectangle (11,7);

\foreach \x in {-1,...,2} {\draw[mydashed] (\x, -1) -- (\x, -2);}
\foreach \x in {-1,...,2} {\draw[mydashed] (\x, 9) -- (\x, 10);}
\foreach \x in {8,...,11} {\draw[mydashed] (\x, -1) -- (\x, -2);}
\foreach \x in {8,...,11} {\draw[mydashed] (\x, 9) -- (\x, 10);}

\foreach \y in {-1,...,9} {\draw[mydashed] (-1,\y) -- (-2,\y);}
\foreach \y in {-1,...,9} {\draw[mydashed] (11,\y) -- (12,\y);}

\draw (-1,-1) grid (2,9);
\draw (8,-1) grid (11,9);
\node at (-3,4) {$p$};
\node at (13,4) {$q$};

\node[scale=0.625] at (2.6,-3) {$C_1$};
\node[scale=0.625] at (1.6,-3) {$C_0$};
\node[scale=0.625] at (8.95,-3) {$C_{k+1}$};
\node[scale=0.625] at (7.7,-3) {$C_k$};
\node[scale=0.7] at (5,-3) {$\hdots$};

\draw[line width =0.5mm] (2,-2.5) -- (2,10.5);
\draw[line width =0.5mm] (8,-2.5) -- (8,10.5);

\node[scale=1.25] at (5,-4.5) {$(0)$};

\begin{scope}[xshift=20cm]

\fill[gray!90] (-1,2) rectangle (1,3);
\fill[gray!90] (-1,0) rectangle (1,1);
\fill[gray!90] (-1,5) rectangle (1,6);
\fill[gray!90] (-1,7) rectangle (1,8);
\fill[gray!90] (8,4) rectangle (9,5);
\fill[gray!90] (9,1) rectangle (11,2);
\fill[gray!90] (10,3) rectangle (11,4);
\fill[gray!90] (9,6) rectangle (11,7);

\fill[gray!90] (2,1) rectangle (4,2);
\fill[gray!90] (2,3) rectangle (3,5);
\fill[gray!90] (2,6) rectangle (4,7);

\fill[gray!90] (7,4) rectangle (8,5);
\fill[gray!90] (6,0) rectangle (8,1);
\fill[gray!90] (6,2) rectangle (8,3);
\fill[gray!90] (6,7) rectangle (8,8);

\foreach \x in {-1,...,2} {\draw[mydashed] (\x, -1) -- (\x,-2);}
\foreach \x in {-1,...,2} {\draw[mydashed] (\x, 9) -- (\x, 10);}
\foreach \x in {8,...,11} {\draw[mydashed] (\x, -1) -- (\x, -2);}
\foreach \x in {8,...,11} {\draw[mydashed] (\x, 9) -- (\x, 10);}

\foreach \x in {3,4,5,6,7} {\draw[mydashed] (\x, -1) -- (\x, -2);}
\foreach \x in {3,4,5,6,7} {\draw[mydashed] (\x, 9) -- (\x, 10);}

\foreach \y in {-1,...,9} {\draw[mydashed] (-1,\y) -- (-2,\y);}
\foreach \y in {-1,...,9} {\draw[mydashed] (11,\y) -- (12,\y);}
\draw (-1,-1) grid (5,9);
\draw (6,-1) grid (11,9);

\draw[line width =0.5mm] (2,-2.5) -- (2,10.5);
\draw[line width =0.5mm] (8,-2.5) -- (8,10.5);

\node[scale=0.7] at (6,-3) {$C_{k-2}$};
\node[scale=1.25] at (5,-4.5) {$(1)$};

\draw[white, line width=0.5mm] (7.03,-1) -- (7.91,-1);
\draw[white, line width=0.5mm] (2.09,-1) -- (5,-1);
\draw[white, line width=0.5mm] (3,-2) -- (3,-0.1); 
\draw[white, line width=0.5mm] (4,-2) -- (4,-0.1);
\draw[white, line width=0.5mm] (5,-2) -- (5,-0.1);

\draw[white, line width=0.5mm] (7.03,9) -- (7.91,9);
\draw[white, line width=0.5mm] (5.03,9) -- (5.96,9);
\draw[white, line width=0.5mm] (2.09,9) -- (5,9); 
\draw[white, line width=0.5mm] (3,8.1) -- (3,10); 
\draw[white, line width=0.5mm] (4,8.1) -- (4,10);
\draw[white, line width=0.5mm] (5,8.1) -- (5,10);

\draw[mydashed] (3, -2) -- (3,0);
\draw[mydashed] (4,-2) -- (4,0);
\draw[mydashed] (5,-2) -- (5,0);

\draw[mydashed] (3,8) -- (3,10);
\draw[mydashed] (4,8) -- (4,10);
\draw[mydashed] (5,8) -- (5,10);

\end{scope}

\begin{scope}[yshift=-16.5cm]

\fill[gray!90] (-1,2) rectangle (1,3);
\fill[gray!90] (-1,0) rectangle (1,1);
\fill[gray!90] (-1,5) rectangle (1,6);
\fill[gray!90] (-1,7) rectangle (1,8);
\fill[gray!90] (8,4) rectangle (9,5);
\fill[gray!90] (9,1) rectangle (11,2);
\fill[gray!90] (10,3) rectangle (11,4);
\fill[gray!90] (9,6) rectangle (11,7);

\fill[gray!90] (5,2) rectangle (6,6);
\fill[gray!90] (2,1) rectangle (4,2);
\fill[gray!90] (2,3) rectangle (3,5);
\fill[gray!90] (2,6) rectangle (4,7);

\fill[gray!90] (7,4) rectangle (8,5);
\fill[gray!90] (6,0) rectangle (8,1);
\fill[gray!90] (6,2) rectangle (8,3);
\fill[gray!90] (6,7) rectangle (8,8);

\foreach \x in {-1,...,11} {\draw[mydashed] (\x, -2) -- (\x, -1);}
\foreach \x in {-1,...,11} {\draw[mydashed] (\x, 9) -- (\x, 10);}

\foreach \y in {-1,...,9} {\draw[mydashed] (-1,\y) -- (-2,\y);}
\foreach \y in {-1,...,9} {\draw[mydashed] (11,\y) -- (12,\y);}
\draw (-1,-1) grid (11,9);

\draw[line width =0.5mm] (2,-2.5) -- (2,10.5);
\draw[line width =0.5mm] (8,-2.5) -- (8,10.5);

\node[scale=1.25] at (5,-4) {$(2)$};

\draw[white, line width=0.5mm] (7.03,-1) -- (7.91,-1);
\draw[white, line width=0.5mm] (2.09,-1) -- (4.96,-1); 
\draw[white, line width=0.5mm] (3,-2) -- (3,-0.1); 
\draw[white, line width=0.5mm] (4,-2) -- (4,-0.1);
\draw[white, line width=0.5mm] (5,-2) -- (5,-1.1);

\draw[white, line width=0.5mm] (7.03,9) -- (7.91,9);
\draw[white, line width=0.5mm] (5.03,9) -- (5.96,9);
\draw[white, line width=0.5mm] (2.09,9) -- (5.96,9); 
\draw[white, line width=0.5mm] (3,8.1) -- (3,10); 
\draw[white, line width=0.5mm] (4,8.1) -- (4,10);
\draw[white, line width=0.5mm] (5,8.1) -- (5,10);

\draw[mydashed] (3, -2) -- (3,0);
\draw[mydashed] (4,-2) -- (4,0);
\draw[mydashed] (5,-2) -- (5,-1);

\draw[mydashed] (3,8) -- (3,10);
\draw[mydashed] (4,8) -- (4,10);
\draw[mydashed] (5,8) -- (5,10);

\draw (4.96,-1) -- (6,-1);
\draw (5.96,9) -- (6,9);

\end{scope}

\begin{scope}[yshift=-16.5cm,xshift=20cm]

\fill[gray!90] (-1,2) rectangle (1,3);
\fill[gray!90] (-1,0) rectangle (1,1);
\fill[gray!90] (-1,5) rectangle (1,6);
\fill[gray!90] (-1,7) rectangle (1,8);
\fill[gray!90] (8,4) rectangle (9,5);
\fill[gray!90] (9,1) rectangle (11,2);
\fill[gray!90] (10,3) rectangle (11,4);
\fill[gray!90] (9,6) rectangle (11,7);

\fill[gray!90] (5,2) rectangle (6,6);
\fill[gray!90] (2,1) rectangle (4,2);
\fill[gray!90] (2,3) rectangle (3,5);
\fill[gray!90] (2,6) rectangle (4,7);

\fill[gray!90] (7,4) rectangle (8,5);
\fill[gray!90] (6,0) rectangle (8,1);
\fill[gray!90] (6,2) rectangle (8,3);
\fill[gray!90] (6,7) rectangle (8,8);

\foreach \x in {-1,...,11} {\draw[mydashed] (\x, -2) -- (\x, -1);}
\foreach \x in {-1,...,11} {\draw[mydashed] (\x, 9) -- (\x, 10);}

\foreach \y in {-1,...,9} {\draw[mydashed] (-1,\y) -- (-2,\y);}
\foreach \y in {-1,...,9} {\draw[mydashed] (11,\y) -- (12,\y);}
\draw (-1,-1) grid (11,9);

\draw[line width =0.5mm] (2,-2.5) -- (2,10.5);
\draw[line width =0.5mm] (8,-2.5) -- (8,10.5);

\node[scale=1.25] at (5,-4) {$(3)$};

\draw[white, line width=0.5mm] (7.03,-1) -- (7.91,-1);
\draw[white, line width=0.5mm] (2.09,-1) -- (4.96,-1); 
\draw[white, line width=0.5mm] (3,-2) -- (3,-0.1); 
\draw[white, line width=0.5mm] (4,-2) -- (4,-0.1);
\draw[white, line width=0.5mm] (5,-2) -- (5,-1.1);

\draw[white, line width=0.5mm] (7.03,9) -- (7.91,9);
\draw[white, line width=0.5mm] (5.03,9) -- (5.96,9);
\draw[white, line width=0.5mm] (2.09,9) -- (5.96,9); 
\draw[white, line width=0.5mm] (3,8.1) -- (3,10); 
\draw[white, line width=0.5mm] (4,8.1) -- (4,10);
\draw[white, line width=0.5mm] (5,8.1) -- (5,10);

\draw[mydashed] (3, -2) -- (3,0);
\draw[mydashed] (4,-2) -- (4,0);
\draw[mydashed] (5,-2) -- (5,-1);

\draw[mydashed] (3,8) -- (3,10);
\draw[mydashed] (4,8) -- (4,10);
\draw[mydashed] (5,8) -- (5,10);

\draw (4.96,-1) -- (6,-1);
\draw (5.96,9) -- (6,9);

\end{scope}
\end{tikzpicture}
\caption{Illustration of the algorithm filling the intermediate columns between two half-plane patterns $p$ and $q$ for the minimal total domination subshift. In the last step, the position $\bm{u_0}$ is the bottommost represented position of the central column, and the central column is coloured with a possible colouring. We chose $k=6$ for the illustration, still the proof works from $k=5$ upwards.}
\label{figure.completing.algorithm.MT}
\end{figure}

\leavevmode

\begin{itemize} 
\item \textbf{Filling the intermediate columns 
between two half-plane patterns.} 

We provide here an algorithm to fill
these columns between two patterns $p$ and $q$ respectively on 
$\mathbb{Z}_{-} \times \mathbb{Z}$ and 
$\mathbb{Z}_{+} \times \mathbb{Z}$ into 
a configuration $x \in X^\mathrm{MT}$. As in the proof of \Cref{minimal-gluing-th}, intuitively, it puts grey cells only when it is really necessary to dominate the neighbour in the previous column: 

\begin{enumerate}
\vspace{-0,1cm}
\setlength{\parskip}{0pt}
\setlength{\itemsep}{3pt}
\item \textbf{Filling the intermediate columns, from $C_1$ to $C_{k-3}$, then $C_k,C_{k-1}$.} 
Successively, for all $j$ 
from $1$ to $k-3$, we determine the column 
$C_j$ according to the following rule: 
for all $\textbf{u} \in C_j$, $x_{\textbf{u}}$ 
is $\begin{tikzpicture}[scale=0.3,baseline=0.4mm]
\fill[gray!90] (0,0) rectangle (1,1);
\draw (0,0) rectangle (1,1);
\end{tikzpicture}$ when
$x_{\textbf{u}-(1,1)}$, $x_{\textbf{u}-(1,-1)}$ 
and $x_{\textbf{u}-(2,0)}$ are $\begin{tikzpicture}[scale=0.3,baseline=0.4mm]
\draw (0,0) rectangle (1,1);
\end{tikzpicture}$ (the difference with the 
proof of \Cref{minimal-gluing-th} is 
that the symbol $x_{\textbf{u}-(1,0)}$ is 
not imposed). Else, $x_{\textbf{u}}$ is set to $\begin{tikzpicture}[scale=0.3,baseline=0.4mm]
\draw (0,0) rectangle (1,1);
\end{tikzpicture}$. This rule 
is illustrated in \Cref{figure.local.rules.completion}.

\begin{figure}[h!]
\[\begin{tikzpicture}[scale=0.4]
\draw (0,0) grid (1,3);
\draw (-1,1) rectangle (0,2);
\node at (0.5,-0.5) {$\vdots$};
\node at (0.5,4) {$\vdots$};

\draw[-latex] (2,1.5) -- (4,1.5);

\begin{scope}[xshift=6cm]
\node at (0.5,-0.5) {$\vdots$};
\node at (0.5,4) {$\vdots$};
\fill[gray!90] (1,1) rectangle (2,2);
\draw (0,0) grid (1,3);
\draw (-1,1) grid (2,2);
\end{scope}

\begin{scope}[xshift=15cm]
\fill[gray!90] (0,1) rectangle (1,2);
\draw (0,0) grid (1,3);
\draw (-1,1) rectangle (0,2);
\node at (0.5,-0.5) {$\vdots$};
\node at (0.5,4) {$\vdots$};

\draw[-latex] (2,1.5) -- (4,1.5);

\begin{scope}[xshift=6cm]
\node at (0.5,-0.5) {$\vdots$};
\node at (0.5,4) {$\vdots$};
\fill[gray!90] (0,1) rectangle (1,2);
\fill[gray!90] (1,1) rectangle (2,2);
\draw (0,0) grid (1,3);
\draw (-1,1) grid (2,2);
\end{scope}
\end{scope}

\end{tikzpicture}\]
\caption{Illustration of the local rules for the completion algorithm for the intermediate columns.}
\label{figure.local.rules.completion}
\end{figure}
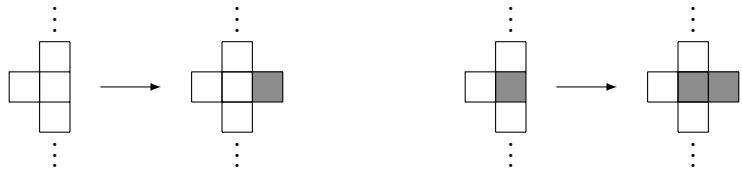

For $j=k$ and then $j=k-1$, 
we determine $x$ on any position $x_{\textbf{u}}$ 
for $\textbf{u} \in C_j$ by applying 
a symmetrical rule: $x_{\textbf{u}}$ 
is $\begin{tikzpicture}[scale=0.3,baseline=0.4mm]
\fill[gray!90] (0,0) rectangle (1,1);
\draw (0,0) rectangle (1,1);
\end{tikzpicture}$ when 
$x_{\textbf{u}+(1,1)}$, $x_{\textbf{u}+(1,-1)}$ 
and $x_{\textbf{u}+(2,0)}$ are $\begin{tikzpicture}[scale=0.3,baseline=0.4mm]
\draw (0,0) rectangle (1,1);
\end{tikzpicture}$. Else it is 
$\begin{tikzpicture}[scale=0.3,baseline=0.4mm]
\draw (0,0) rectangle (1,1);
\end{tikzpicture}$.

\newpage
\item \textbf{The central column ($j=k-2)$.} 

We then 
determine $x$ on the central column $C_{k-2}$. For all $\textbf{u} \in C_{k-2}$, $x_{\textbf{u}}$ 
is $\begin{tikzpicture}[scale=0.3,baseline=0.4mm]
\fill[gray!90] (0,0) rectangle (1,1);
\draw (0,0) rectangle (1,1);
\end{tikzpicture}$ when 
$x_{\textbf{u}+(1,1)}$, $x_{\textbf{u}+(1,-1)}$ 
and $x_{\textbf{u}+(2,0)}$ are equal to 
$\begin{tikzpicture}[scale=0.3,baseline=0.4mm]
\draw (0,0) rectangle (1,1);
\end{tikzpicture}$, 
or when
$x_{\textbf{u}-(1,1)}$, $x_{\textbf{u}-(1,-1)}$ 
and $x_{\textbf{u}-(2,0)}$ are equal to 
$\begin{tikzpicture}[scale=0.3,baseline=0.4mm]
\draw (0,0) rectangle (1,1);
\end{tikzpicture}$. Else, it is 
$\begin{tikzpicture}[scale=0.3,baseline=0.4mm]
\draw (0,0) rectangle (1,1);
\end{tikzpicture}$.

\item \textbf{Eliminating total-domination errors in the central column.}

Choose any position $\bm{u_0} \in C_{k-2}$. From this position upwards, check 
for every position $\textbf{u}$ if it is dominated.
If this is not the case, then change 
the symbol on $\textbf{u}+(0,1)$ into 
$\begin{tikzpicture}[scale=0.3,baseline=0.4mm]
\fill[gray!90] (0,0) rectangle (1,1);
\draw (0,0) rectangle (1,1);
\end{tikzpicture}$.
As soon as $\bm{u_0}$ has been processed, do the same symmetrically (change the symbol in $\textbf{u}-(0,1)$ when $\textbf{u}$ is not dominated) in parallel downwards, beginning from $\bm{u_0}-(0,1)$.

\end{enumerate}

See an illustration of this algorithm on 
\Cref{figure.completing.algorithm.MT}.

\item \textbf{The obtained configuration 
is in $X^\mathrm{MT}$.}

We have to check that the local 
rules of the minimal-total-domination subshift 
are satisfied over the whole constructed configuration $x$.

\begin{enumerate}

\item \textbf{The local rules are satisfied inside the "interior" of 
the half-planes.}

Same as the corresponding point 
in the proof of \Cref{minimal-gluing-th}.

\item \textbf{Every position in $C_{-1}, \cdots, C_{k+2}$ is dominated.}

Same as the corresponding point 
in the proof of \Cref{minimal-gluing-th} for the positions 
outside the central column $C_{k-2}$. 
In this column, let us assume 
that a position $\textbf{u}$ 
above $\bm{u_0}$ (without loss of 
generality) is not dominated. Then 
the last step of the algorithm, when 
examining this position, would have
changed the symbol on position $\textbf{u}+(0,1)$, which is a contradiction. In particular, no dominant positions are 
isolated.

\item \textbf{Every dominant position in $C_{-1}, \cdots, C_{k+2}$ has a private neighbour.}

\begin{enumerate}
\item[(a+b)] Outside the central column, the 
proof is similar to the corresponding points 
in the proof of \Cref{minimal-gluing-th}.
\item[(c)] In the column $C_{k-2}$, the 
dominant positions added in Step 2 necessarily
have a private neighbour in column $C_{k-3}$ 
or $C_{k-1}$. Let us take a dominant position $\textbf{u}$, assumed without loss of generality to be above $\bm{u_0}$, which was added in the last step. This implies that $\textbf{u}-(0,1)$ was not dominated when 
the algorithm checked this position. As a consequence, it is a private neighbour for $\textbf{u}$.
\end{enumerate}
\end{enumerate}

\item 

\textbf{The subshift $X^\mathrm{MT}$ is not $4$-block-gluing.}

Let us consider the patterns $p$ and $q$ 
defined in the corresponding point in the proof of \Cref{minimal-gluing-th}.
It is easy to see that these two patterns are also globally admissible in the subshift $X^\mathrm{MT}$.
Using the proof for $X^\mathrm{M}$, we only have to check that in the constructed configuration, no dominant positions are isolated, which is straightforward.

Using the same arguments as the ones for the minimal domination case, it is easy to see that any configuration in $\mathcal{A}^{\mathbb{Z}^2}$ where $p$ and $q$ are glued at distance 4 contains some forbidden patterns.

As a consequence $c(X^\mathrm{MT}) >4$. Since 
it is $5$-block-gluing, $c(X^\mathrm{MT}) = 5$.
\end{itemize}
\end{proof}

As a direct consequence of \Cref{theorem.computability.entropy}:

\begin{thm}
The numbers $\nu_\mathrm{D}$, $\nu_\mathrm{T}$, $\nu_\mathrm{M}$ and $\nu_\mathrm{MT}$
are computable with rate $n \mapsto 2^{O(n^2)}$.
\end{thm}

\section{Bounding the growth rates with computer resources}
\label{section-bounds-growth-rates}
Although the algorithm presented in \Cref{subsection.computability} provides a \uline{theoretical} way to compute the growth rates of various dominating sets of the finite square grids, it is not efficient enough for practical use on a computer. In this section, we use other tools which make it possible to obtain bounds for the growth rates.
These bounds are obtained using computer resources, by running a C++ program made for the occasion. The program is a modification of the program used in \Cref{domination-chapter}. Its modularity enabled us to only (re)write partly the local characterisation rules of domination problems and use the same main code. The main difference is that we work in the $(+,\times)$- algebra instead of the previously used $(\min, +)$-algebra.

We explain here the method differently, relating it to unidimensional SFTs. We do not dive much into the details of the specific problems, but give the abstract framework. As we have just mentioned, the technique relies on, for a fixed $m$, assimilating the dominating sets of $G_{n,m}$ to patterns of a unidimensional subshift of finite type, whose entropy is known to be computable through linear algebra computing.
Note that, in its use in \Cref{domination-chapter} there are no such things as lower or upper bounds we investigate now: we only enumerated sets which are precisely 2-dominating or Roman dominating. Since we were interested in finding the minimum size of such a set, we had the right to apply some optimisations to avoid having to enumerate all dominating sets, in particular the ones we knew were not of minimum size. Since we now want to count \emph{all} the different dominating sets, these optimisations do not apply here: we must enumerate all possible states.

\subsection{Nearest-neighbour unidimensional subshifts of finite type}
\label{section-nearest}
In this section $\mathcal{A}=(a_1,...,a_k)$ 
is a finite set, 
and $X$ is a unidimensional subshift of finite type
on alphabet $\mathcal{A}$. The elements of $\mathcal{A}$ are the cell values we defined in \Cref{domination-chapter}. Let us 
denote by $(e_1,...,e_k)$ the canonical basis 
of $\mathbb{R}^k$.

\begin{deff}
The subshift $X$ is said to have the \textbf{nearest-neighbour} property when it is defined 
by forbidding a set of patterns on support 
$\{0,1\}$.
\end{deff}

\begin{deff}
The \textbf{adjacency matrix} of a nearest-neighbour SFT $X$ is the matrix $M \in \mathcal{M}_k (\mathbb{R})$ 
such that $M[e_i][e_j]=1$ if the 
pattern $a_i a_j$ is not forbidden, or 0 otherwise. 
\end{deff}

The following is well known (see~\cite{Lind-Marcus}): 

\begin{prop}
Let $||.||$ be any matricial norm. The entropy of $X$ is equal to the spectral radius of $M$: 
\[h(X)= \log_2{\lim_n ||M^n||^{1/n}}.\]
\end{prop}

Our matrix has non-negative coeffecients and is primitive for the same reason as in \Cref{prop-primitive}. Thanks to these properties, the Perron-Frobenius theorem states that the matrix has a largest eigenvalue in $\RR^+$.

\subsection{Unidimensional versions of the domination subshifts}

We define here the unidimensional versions of 
the domination subshifts introduced in \Cref{subshifts-examples}. We use them to describe and prove the method we use to obtain the bounds on the growth rates. We recall that $n$ refers to the number of lines and $m$ to the number of columns. The first sequence of SFTs ($X^{D, n}$) is used to obtain the lower bound, whereas we use the second one ($X_*^{\mathrm{D},n}$) to obtain the upper bound.

\begin{notation}
Let us fix some integer $n \ge 1$. 
We denote by $X^{\mathrm{D},n}$ the 
undimensional subshift on alphabet $\mathcal{A}_0 ^n$ such that a configuration $x$ 
is in $X^{\mathrm{D},n}$ if and only if 
the set of positions $(j,i) \in \mathbb{Z} \times \llbracket 1,n\rrbracket$ such that the symbol $x_{(j,i)}$ is grey forms a dominating set of the grid $\mathbb{Z} \times \llbracket 1,n\rrbracket$.
\end{notation}

With arguments similar to the ones of the proofs of \Cref{lemma.comparison.minimal} and \Cref{lemma.comparison.minimal.total},
we obtain the following asymptotic formula when $n$ is fixed and $m$ grows 
to infinity. We recall \Cref{notation-counting-numbers}: $D_{n,m}$ denotes the number of dominating sets of $G_{n,m}$.
\[D_{n,m} = 2^{h(X^{\mathrm{D},n}) \cdot m + o(m)}.\]

\begin{notation}
For all $n \ge 3$, we also 
denote by $X_{*}^{\mathrm{D},n}$ the 
undimensional subshift on alphabet $\mathcal{A}_0 ^n$ such that a configuration $x$ 
is in $X_{*}^{\mathrm{D},n}$ if and only if 
the set of positions $(j,i) \in \mathbb{Z} \times \llbracket 2,n-1\rrbracket$ such that the symbol $x_{(j,i)}$ 
is grey forms a dominating set of the 
grid $\mathbb{Z} \times \llbracket 2,n-1\rrbracket$.
\end{notation}

\subsection{Recoding into nearest-neighbour 
subshifts}
\label{subsection.nearest.neighbour}

Let us set $\mathcal{A}_1 = \left\{ 
\begin{tikzpicture}[scale=0.3]
\draw (0,0) rectangle (1,1);
\end{tikzpicture},
\begin{tikzpicture}[scale=0.3]
\draw[pattern color=gray!90, pattern=my crosshatch dots] (0,0) rectangle (1,1);
\end{tikzpicture}, \begin{tikzpicture}[scale=0.3]
\draw[fill=gray!90] (0,0) rectangle (1,1);
\end{tikzpicture}\right\}$, 
and let us consider the map 
$\varphi: \left(\mathcal{A}^n_0\right)^{\mathbb{Z}} \rightarrow \left(\mathcal{A}^n_1\right)^{\mathbb{Z}}$ 
which acts on configurations of $\left(\mathcal{A}^n_0\right)^{\mathbb{Z}}$ by changing the $i^\text{th}$ 
symbol of any position $j \in \mathbb{Z}$ 
into $\begin{tikzpicture}[scale=0.3]
\draw[pattern color=gray!90, pattern=my crosshatch dots] (0,0) rectangle (1,1);
\end{tikzpicture}$ whenever it is not dominant and 
dominated by an element of $C_{j-1} \bigcap \left( \mathbb{Z} \times \llbracket 1,n\rrbracket \right)$ or $C_{j} \bigcap \left( \mathbb{Z} \times \llbracket 1,n\rrbracket \right)$. Informally, from lightest to darkest, the symbols stand for an undominated cell (which is not dominant), a dominated cell which is not dominant and a dominant cell. This is illustrated in \Cref{figure.illustration.conjugation}. The nearest-neighbour property makes it possible to count the dominating sets without enumerating them fully: it is enough to store the information about a small number of the latest columns, proceeding from left to right in the grid. We encode this information from the possibly several columns we need to recall into a single column, for practical reasons. This is why it coincides with the definition of the nearest-neighbour property.

\begin{figure}[h!]
\[\begin{tikzpicture}[scale=0.4]

\fill[gray!90] (1,0) rectangle (3,2);
\fill[gray!90] (3,3) rectangle (4,4);
\draw[dashed] (0,0) -- (6,0);
\draw[dashed] (0,4) -- (6,4);
\draw (1,0) -- (5,0);
\draw (1,4) -- (5,4);
\draw (1,0) grid (5,4);

\draw[line width=0.3mm,-latex] (7.5,2) -- (10.5,2);
\node at (9,1) {$\varphi$};

\begin{scope}[xshift=12cm]

\fill[pattern color=gray!90, pattern=my crosshatch dots] (3,0) rectangle (4,3);
\fill[pattern color=gray!90, pattern=my crosshatch dots] (1,2) rectangle (3,3);
\fill[pattern color=gray!90, pattern=my crosshatch dots](1,3) rectangle (2,4);
\fill[pattern color=gray!90, pattern=my crosshatch dots](4,3) rectangle (5,4);

\fill[gray!90] (1,0) rectangle (3,2);
\fill[gray!90] (3,3) rectangle (4,4);
\draw[dashed] (0,0) -- (6,0);
\draw[dashed] (0,4) -- (6,4);
\draw (1,0) -- (5,0);
\draw (1,4) -- (5,4);
\draw (1,0) grid (5,4);

\end{scope}
\end{tikzpicture}\]
\caption{Illustration of the map recoding $X^{\mathrm{D},n}$ into a nearest-neighbour SFT.}
\label{figure.illustration.conjugation}
\end{figure}
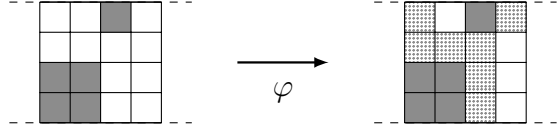

$\varphi$ is a conjugation and the entropy is stable by conjugation by \Cref{lemma-conjugation}, therefore:

\[h(X^{\mathrm{D},n}) = h(\varphi(X^{\mathrm{D},n})).\]

Moreover, the subshift $X^{\mathrm{D},n}$ has the nearest-neighbour property.

\subsection{Numerical approximations}
We give here the numerical bounds we can compute, and prove them. The proof is only given for the simple domination, the other ones use the same method, hence are very similar.

\begin{thm}[Domination]
The following inequalities hold:
    \begin{equation*}1.950022198 \le \nu_\mathrm{D} \le 1.959201684.\end{equation*}
\end{thm}

\begin{proof}\leavevmode
\begin{itemize} \item \textbf{Lower bound:}

\begin{enumerate}
\item For all integers $n_1,n_2$ and $m$, $D_{n_1+n_2,m} \ge D_{n_1,m} \cdot D_{n_2,m}.$

Indeed, let us consider two
sets dominating respectively $G_{n_1,m}$ and $G_{n_2,m}$. 
By gluing the first one on the top of 
the second one, we obtain a dominating set of 
$G_{n_1+n_2,m}$. This is true because any position in this grid 
is either in the copy of the grid $G_{n_1,m}$
and thus dominated by an element in this grid, 
or in the copy of $G_{n_2,m}$. Since this construction is invertible our remaping is indeed a conjugation. We obtain the announced inequality.

\item As a consequence, for all $k \ge 0$,
    \[D_{18k, m} \ge D_{18,m}^{k} = 2^{h(X^{\mathrm{D},18}) \cdot km + k \cdot o(m)},\]
        where the function $o(m)$ is related to the fact that we used $18$ lines\footnote{We could not do the computations with more lines out of lack of RAM.}.
This implies that
\[\lim_{n,m} \frac{\log_2 (D_{n,m})}{nm} 
        = \lim_{n,k} \frac{\log_2 (D_{18k,m})}{18km} \ge h(X^{\mathrm{D},18}).\]
\item This number is equal to 
    $h(\varphi(X^{D,18}))$, which is computed using \Cref{section-nearest} and \Cref{subsection.nearest.neighbour}. The lower bound follows: for any value of $n$ we may compute the transfer matrix $T_n$ counting $D_{n,m}$. The $n$-th root of the largest eigenvalue of $T_n$ is a lower bound for $\nu_D$. We could compute $T_n$ for at most 18 lines given the computing resources at our disposal.
\end{enumerate}
\newpage 
\item \textbf{Upper bound:}

\begin{enumerate}
\item 
For all $n,m$, let us denote by $D^{*}_{n,m}$ 
the number of sets of vertices of $G_{n,m}$ 
which dominate the middle $n-2$ lines (i.e. cells of the first and last lines might not be dominated). We have a direct inequality \[D_{n,m} \le D^{*}_{n,m}.\]
\item For a reason similar as the one in the first point 
of the proof of the lower bound, for all $n_1,n_2$, $D^{*}_{n_1+n_2, m} \le D^{*}_{n_1,m} \cdot D^{*}_{n_2,m}$.
\item For all $k \ge 0$ and $m \ge 0$, 
    \[D^{*}_{18k, m} \le (D^{*}_{18,m})^{k} = 2^{h(X_{*}^{\mathrm{D},18}) \cdot km + k \cdot o(m)}.\] 
As a consequence 
        \[\nu_\mathrm{D} \le h(X_{*}^{\mathrm{D},18}).\]
With the same method as for the lower bound, 
we obtain the upper bound.
\end{enumerate}

\end{itemize}
\end{proof}

\begin{rk}
With further numerical manipulations, we notice that the lower bound and the upper bound seem to get closer to each other rather slowly. To speed up the convergence, we had the idea of using the sequences of ratios $2^{h(X^{\mathrm{D},n+1})}/2^{h(X^{\mathrm{D},n})}$ and 
$2^{h(X_{*}^{\mathrm{D},n+1})}/2^{h(X_{*}^{\mathrm{D},n})}$. This seems to offer a much better convergence speed. Indeed, for both sequences, from $h=11$ on, the ratio seem to be stabilised around $1.954751195$. This ratio is not a bound on $\nu_\mathrm{D}$ but it is guaranteed to converge towards this quantity.
\label{rk-convergence}
\end{rk}

\begin{conj}[Domination]
$\nu_\mathrm{D} \approx 1.954751$
\end{conj}

Using the same method, we provide bounds for the other problems. For the total domination, we can make the computations up to $n = 17$. However, for the other problems (the ones with the minimality constraint) the number of patterns we enumerate grows (exponentially) at a much faster rate than for the non-minimal problems, thus the bounds are less good. We cannot go further than about $n=10$ for these problems. This comes from the fact that we need to encode information about more columns than from their non-minimal counterparts.
\begin{thm}[Total domination]
$ 1.904220376\le \nu_\mathrm{T} \le 1.923434191$.
$\nu_\mathrm{T} \approx 1.9153$
\end{thm}

\begin{rk}
As in \Cref{rk-convergence}, we can observe that the ratios $2^{h(X^{\mathrm{T},n+1})}/2^{h(X^{\mathrm{T},n})}$ and 
$2^{h(X_{*}^{\mathrm{T},n+1})}/2^{h(X_{*}^{\mathrm{T},n})}$ offer a much better convergence speed. Indeed, from $n=10$ they seem to stabilise, both around $1.915316$.
\end{rk}

\begin{conj}[Total domination]
$\nu_\mathrm{T} \approx 1.9153$ 
\end{conj}

\begin{thm}[Minimal domination]
$1.315870482 \le \nu_\mathrm{M} \le 1.550332154$.
\label{th-minimal-dom}
\end{thm}

\begin{thm}[Minimal total domination]
$1.275805204 \le \nu_\mathrm{MT} \le 1.524476040$.
\end{thm}

\label{counting-numerical-section}

\section{A (2\emph{k}+3)-block-gluing family: the minimal meta \emph{k} domination}
\label{section-meta-k-domination}
When we spoke about the block-gluing property, we only defined the \emph{constant} version of it. This property may be declined in finer-grained versions than choosing between being block gluing or not being block gluing.\footnote{"To be or not to be, that is the question." said a famous writer. Maybe he was thinking about block-gluing SFTs.}
 
\begin{deff}
Let $X$ be a subshift and $f : \NN \rightarrow \NN$ be a function. We say that $X$ is $f$-block-gluing when any two globally-admissible patterns $q$ and $q$ of size $n \times n$ may be glued when they are separated by at least $f(n)$ columns.
\end{deff}
This definition encompasses more SFTs because now the gap depends on the size of the pattern. This leads to interesting questions, such as:

\begin{question}
What functions can $f$ be? For instance, can $f$ be a non-constant sub-linear function, or a super-linear function?
\label{question-block-gluing}
\end{question}
We will answer partly these questions by recalling some known results in \Cref{counting-conclusion-gluing}.

We propose here an extension of the domination problems. The standard domination consists of two sets $S_1 = D$ and $S_0 = V \setminus D$: the members of $S_0$ need to be dominated by a vertex in $S_1$ while the members of $S_1$ do not need to be dominated. We introduce a new family of problems where we colour the vertices with $k+1$ colours,  which denotes a hierarchy of who may dominate whom. We will then show that each SFT associated to a problem of this family is block gluing, but the gluing constant linearly depends on the parameter $k$. We will also give some hints as to why it is difficult to find problems which are $f$-block-gluing with $f$ being something other than an affine function.

\begin{deff}
    $S = (S_0, S_1, \cdots, S_k)$ is a \textbf{meta-$\bm{k}$-dominating} tuple if it is a partition of $V$ and every $v \in V_i$ such that $i < k$ has a neighbour in some $V_j$ with $j > i$.
\end{deff}

\begin{notation}
If $S$ is a meta-$k$-dominating tuple and $x \in V$, we denote by $\textbf{\aarg}_S(x)$ the integer $0 \leq i \leq k$ such that $x \in S_i$.
\end{notation}

\begin{deff}
Similarly as in the other domination problems, we say that $u \in G_{n,m}$ is \textbf{dominated} by a neighbour $v$ or that $v$ \textbf{dominates} $u$ when $\arg_S(u) < \arg_S(v)$. $u$ is a \textbf{private neighbour} of $v$ if, in addition to being dominated by $v$, $u$ is not dominated by another neighbour.
\end{deff}

Now that we have generalised the notion of domination and total domination, we attack the minimal and minimal total domination. Since we no longer have a natural inclusion order, we choose one which generalises the inclusion on sets, for tuples of sets.

\begin{deff}
Let $S$ and $S'$ be two meta-$k$-dominating tuples. We say that $S' \leq S$ if for every $x \in V$, $\aarg_{S'}(x) \leq \aarg_S(x)$. We say that $S' < S$ if $S' \leq S$ and if there exists at least one $x \in V$ such that $\aarg_{S'}(x) < \aarg_S(x)$.\\
We say that a meta-$k$-dominating tuple $S$ is \textbf{minimal} if there is no meta-$k$-dominating tuple $S'$ such that $S' < S$.
\end{deff}

In this definition, $S' \leq S$ if we can reduce the labels of a set of vertices altogether simultaneously. We could define another order allowing the reduction of the label of a unique vertex at each step. We show that the two definitions are equivalent.

\begin{deff}
If $S$ and $S'$ are two meta-$k$-dominating tuples, we write $S' <_1 S$ if there exists a unique $x \in V$ such that $\aarg_{S'}(x) < \aarg_S(x)$ and $\aarg_{S'}(u) = \aarg_S(u)$ for every $u \neq x$.
\end{deff}

\begin{prop}
Let $S$ be a meta-$k$-dominating tuple. $S$ is minimal for the relation $<$ if and only if it is minimal for the relation $<_1$.
\label{equiv-reduction-prop}
\end{prop}

\begin{proof}
It suffices to show that if a tuple $S$ is reducible for one order, it is also reducible for the other, and vice versa.

It is clear that if $S$ is reducible for $<_1$, then it is reducible for $<$.

Now, let $S$ be such that there exists some $S' < S$. We show that $S$ is reducible for $<_1$. Let $D = \{x \in V \;|\; \aarg_{S'}(x) < \aarg_S(x)\}$ and $x_0 \in D$ with minimal $\aarg_{S'}$ value among the elements of $D$. Let $S''$ be such that $\aarg_{S''}(x_0) = \aarg_{S'}(x_0)$ and $\aarg_{S''}(x) = \aarg_S(x)$ for every $x \neq x_0$. Let us show that $S''$ is a meta-$k$-dominating tuple.

    First, any vertex in $V\setminus N[x_0]$ (see \Cref{def-neighbourhood} for the definition of the the closed neighbourhood $N[v]$) is dominated because itself and its neighbours have the same label as in $S$. $x_0$ is also dominated because its label decreased while the others stayed constant. Finally, let $v \in N(x_0)$. If $v$ belongs to $S_k$, then it still does not need to be dominated. Else, if $v$ is dominated by some vertex other than $x_0$ in $S'$, it is also the case in $S''$. Else, $v$ is only dominated by $x_0$ in $S'$. The definition of $x_0$ implies that $\aarg_{S}(v) = \aarg_{S'}(v)$. Since $v$ is dominated in $S'$, it implies that $\aarg_{S''}(x_0) = \aarg_{S'}(x_0) > \aarg_{S'}(v) = \aarg_{S}(v) = \aarg_{S''}(v)$, hence $v$ is dominated in $S''$.

Hence the minimal elements for $<$ are the same as the minimal elements for $<_1$.
\end{proof}

\begin{deff}
If $u \in S_m$ with $m > 0$, we say that any private neighbour of $u$ belonging to $S_{m-1}$ is a \textbf{good} private neighbour.
\end{deff}

\begin{prop}
Let $(S_0, \cdots, S_k)$ be a meta-$k$-dominating tuple. Then it is minimal meta-$k$-dominating if and only if every $x \in S_i$ with $i < k$ has a good private neighbour, and any $x \in S_k$ with a neighbour in $S_k$ has a good private neighbour.
\label{prop-mini-meta-k}
\end{prop}

\begin{proof}
\leavevmode

    \textbf{If a meta-$\bm{k}$-dominating tuple has the private-neighbour property, then it is minimal.}

Let $S$ be a meta-$k$-dominating tuple having the private neighbour property. By \Cref{equiv-reduction-prop} it is sufficient to show that there is no $S'$ such that $S' <_1 S$. We proceed by contradiction and assume that we have some $S' <_1 S$. Let $x_0$ be the vertex such that $\aarg_{S'}(x_0) < \aarg_S(x_0)$. There are two cases, the first one being that $x_0$ has label $k$ in $S$. Since $x$ is dominated in $S'$, it has a neighbour of label $k$ in $S'$ which has the same label in $S$: $x_0$ is not isolated in $S$. Since by hypothesis it has, in $S$, a private neighbour $y$ of label $k-1$, this implies that $y$ is not dominated in $S'$. Thus $S'$ is not meta $k$  dominating, a contradiction. The other case is when $\arg_S(x) < k$. Let $y$ be a good private neighbour of $x_0$ in $S$: $\aarg_S(y) = \aarg_S(x_0)-1$ and no neighbours of $y$ dominates it in $S$. Since $\aarg_S'(x_0) < \aarg_S(x_0)$ and the labels of the others vertices are the same, $y$ is not dominated in $S'$, which is also a contradiction. This proves that $S$ is minimal.\\
\newpage
    \textbf{If a meta-$\bm{k}$-dominating tuple is minimal then it has the private-neighbour property.}

We show the contrapositive implication: let us consider a meta-$k$-dominating tuple $S$ which does not have the private-neighbour property. We will prove that it is not minimal. Let $x_0$ be such a vertex with label $p > 0$ which does not have any good private neighbour, and if $p=k$ then has a neighbour in $S_k$. Let $S'$ be such that $\aarg_{S'}(x_0) = p-1$ and $\aarg_{S'}(x) = \aarg_S(x)$ for every $x \neq x_0$. In $S'$, $x_0$ is dominated by the same vertex it is dominated by in $S$, or by its neighbour of label $k$ if $p = k$. Any neighbour which is dominated by $x_0$ in $S$ has either label at most $p-2$ or is also dominated by another vertex in $S$, hence it is dominated in $S'$. 
Therefore $S'$ is a meta-$k$-dominating tuple and $S' < S$, which concludes the proof.
\end{proof}

\begin{thm}
The SFT associated to the minimal meta-$k$-domination problem is ($2k+3$)-block-gluing.
\label{meta-gluing-th}
\end{thm}

We can notice that this result coincides with the one about the minimal domination, for which $k = 1$. We can also see easily that the meta-$k$-domination subshift is 1-block-gluing for all $k$ (we put label $k$ for any cell in the gap column).

\begin{proof}
Let $X_k$ be the SFT associated to the meta $k$ domination, and $X_k^\mathrm{M}$ its minimal version.

We give here a proof similar to the one in the proof of \Cref{minimal-gluing-th}. The argument is more general and can be used to show that the minimal domination is 5-block-gluing. We give an algorithm to fill the middle columns, providing there are at least $2k+3$ of them, and then prove that the result is indeed minimal meta $k$ dominating.

\begin{itemize}
\item \textbf{Filling the intermediate columns 
between two half-plane patterns.}

Let $p$ and $q$ be two patterns respectively on, without loss of generality,
$\mathbb{Z}_{-} \times \mathbb{Z}$ and 
$\mathbb{Z}_{+} \times \mathbb{Z}$. Let 
us determine a configuration of $\mathcal{A}^{\mathbb{Z}^2}$ such that $x_{|\mathbb{Z}_{-} \times \mathbb{Z}}=p$ and $x_{|(k,0)+\mathbb{Z}_{+} \times \mathbb{Z}}=q$. 
The intermediate 
columns $C_1,...,C_n$ are split, like before, into the "left part" of the middle columns ($C_1$ to $C_{n-k-3}$, the central column $_{n-k-2}$ and the "right part" of the middle columns ($C_{n-k-2}$ to $C_n$). We determine their values by the following algorithm. For concision, we introduce the following definition:

\begin{figure}[h]
\centering
\begin{tikzpicture}[scale=0.4, xscale = 1.5]

\def\testarray{{{1,0.5},{4,2}}}

\def\x{1.5}
\foreach \v in {4,3,2,1,0,0,1,2,3,5,0,5,4,3,2,1,0,0} {\node[scale=0.5] at (\x, 6.5) {$\v$}; \xdef\x{\x+1}}

\xdef\x{1.5}
\foreach \v in {4,3,2,1,0,0,1,2,3,4,0,0,0,0,0,0,0,0} {\node[scale=0.5] at (\x, 5.5) {$\v$}; \xdef\x{\x+1}}

\xdef\x{1.5}
\foreach \v in {0,0,0,0,0,0,0,0,0,2,1,0,0,5,4,3,2,1} {\node[scale=0.5] at (\x, 4.5) {$\v$}; \xdef\x{\x+1}}

\xdef\x{1.5}
\foreach \v in {5,5,4,3,2,1,0,0,0,2,1,0,0,5,4,3,2,1} {\node[scale=0.5] at (\x, 3.5) {$\v$}; \xdef\x{\x+1}}

\xdef\x{1.5}
\foreach \v in {0,0,0,0,0,0,0,1,2,3,0,0,0,0,0,0,0,0} {\node[scale=0.5] at (\x, 2.5) {$\v$}; \xdef\x{\x+1}}

\xdef\x{1.5}
\foreach \v in {5,0,0,0,0,0,0,0,0,4,5,4,3,2,1,0,0,5} {\node[scale=0.5] at (\x, 1.5) {$\v$}; \xdef\x{\x+1}}

\xdef\x{1.5}
\foreach \v in {0,0,1,2,3,4,5,0,0,1,0,0,0,0,0,0,0,0} {\node[scale=0.5] at (\x, 0.5) {$\v$}; \xdef\x{\x+1}}

\foreach \x in {1,...,19}{ \draw[mydashed](\x, -1) -- (\x, 0);}
\foreach \x in {1,...,19}{ \draw[mydashed](\x, 7) -- (\x, 8);}

\foreach \x in {0,...,7} {\draw[mydashed] (0,\x) -- (1,\x);}
\foreach \x in {0,...,7} {\draw[mydashed] (19,\x) -- (20,\x);}

\draw (1,0) grid (19,7);

\node at (-1,4) {$p$};
\node at (20.8,4) {$q$};

\node[scale=0.625] at (2.5,-1.5) {$C_0$};
\node[scale=0.625] at (3.5,-1.5) {$C_1$};
\node[scale=0.625] at (4.5,-1.5) {$C_2$};

\node[scale=0.625] at (10.5,-1.5) {$C_{j_0}$};
\node[scale=0.625] at (15.5,-1.5) {$C_{n-1}$};
\node[scale=0.625] at (16.5,-1.5) {$C_{n}$};
\node[scale=0.625] at (17.5,-1.5) {$C_{n+1}$};

\draw[line width =0.5mm] (3,-1.5) -- (3,9.5);
\draw[line width =0.5mm] (17,-1.5) -- (17,9.5);
\end{tikzpicture}
	
\caption{Illustration of the algorithm filling the middle columns for the minimal meta $k$ domination for $k=5$. We assume the line above the values and the one below them are filled with zeroes. Here $n=14 = 2k+4$, although a gap of $13$ is enough.}
\label{meta-k-algo-example}
\end{figure}

\begin{enumerate}
\vspace{-0,1cm}
\setlength{\parskip}{0pt}
\setlength{\itemsep}{3pt}
\item \textbf{Filling the middle left columns, from $C_1$ to $C_{n-k-2}$.}

The following algorithm, to fill the values of the middle columns, is illustrated through an example in \Cref{meta-k-algo-example}.
Successively, for all $j$ 
from $1$ to $n-k-2$, we determine the column 
$C_j$ according to the following rule: for each $i \in \ZZ$,
\[ C_j[i] = 
\left\{ \begin{array}{ll}
C_{j-1}[i]+1 &\text{ if } C_{j-1}[i] \text{ is not dominated;} \\
C_{j-1}[i]-1 &\text{ if } C_{j-1}[i] \neq 0 \text{ and has no private neighbours in $S_{C_{j-1}[i]-1}$;}\\
0 &\text{ otherwise.}
\end{array}
\right.\]
Note that since $p$ is globally admissible, no cells in $C_0$ (or $C_{n+1}$) can need both a good private neighbour and some neighbour to be dominated. For the same reason, if some $C_0[i]$ needs a good private neighbour then $C_0[i-1]$ must, if it needs a good private neighbour, be equal to $C_0[i]$: otherwise its neighbour in $C_1[i-1]$ (resp. $C_1[i+1]$) would either dominate or be dominated by $C_1[i]$, hence one of the two would not play its role of good private neighbour. It must also be dominated, either by $C_0[i]$ if lesser, or by another cell in $p$ if greater or equal (otherwise its neighbour in $C_1[i-1]$ would dominate $C_1[i]$. The same applies to $C_0[i+1]$, so that the algorithm is well defined for the first step. These properties propagate to each new column filled, so that the algorithm is completely well defined.

\textbf{Filling the middle right columns, from $C_n$ to $C_{n-k-1}$.}

We apply the symmetric rule: instead of considering the value of $C_{j-1}[i]$ to determine the one of $C_j[i]$, we use the one of $C_{j+1}[i]$.

\item \textbf{The central column $C_{n-k-1}$.}

We first let $j_0 = n-k-1$ be the index of the central column. We first set the value $C_{j_0}[0]$, then the values of $C_{j_0}[-1]$ and $C_{j_0}[1]$. After this, we give the way to fill the values of the column from index 2 to infinity, and in parallel from index -2 to minus infinity. We define $\max(\emptyset) = -1$.

First let $R_0$ be the set of values among $C_{j_0-1}[0]$ and $C_{j_0+1}[0]$ which are different from $k$ and are not dominated so far. Then $C_{j_0}[0] = \max(R_0)$+1.\\
Now let\footnote{We named it $R'_1$ instead of $R_1$ here because it is not exactly how the others $R_i$s are defined.} $R'_1$ (resp. $R'_{-1}$) be the set of values among $C_{j_0-1}[1]$ and $C_{j_0+1}[1]$ (resp. $C_{j_0-1}[-1]$ and $C_{j_0+1}[-1]$) which are different from $k$ and are not dominated so far. Up to symmetry, we may assume that $\max(R'_1) \geq \max(R'_{-1})$. We then set $R"_1 = R'_1$ if $C_{j_0}[0]$ is already dominated, or $R"_1 = R'_1 \cup \{C_{j_0}[0]\}$ otherwise. We define $C_{j_0}[1] = \max(R"_1)+1$ and $C_{j_0}[-1] = \max(R'_{-1})+1$. We do the same thing in symmetric if $\max(R'_{-1}) > \max(R'_1)$.

Now for $i = 2$ to infinity, we define $R_i$ to be the set of values among $C_{j_0-1}[i]$, $C_{j_0+1}[i]$ and $C_{j_0}[i-1]$ 
which are different from $k$ and not dominated so far and we set $C_{j_0}[i] = \max(R_i)+1$. In parallel downwards for $i$ from -2 to minus infinity, we define similarly $R_i$ considering $C_{j_0}[i+1]$ instead of $C_{j_0}[i-1]$ and we set $C_{j_0}[i] = \max(R_i)+1$ as well.
\end{enumerate}

\item \textbf{The obtained configuration is in $X_k^\mathrm{M}$.}

We have to check that the configuration $x$ we constructed satisfies the local 
rules of the minimal domination subshift. We divide this part of the proof according to whether or not the cells belong to $p$ or $q$, or lie outside.
\begin{enumerate}

\item \textbf{The local rules are satisfied inside the half planes.}

By hypothesis, the patterns $p$ and $q$ 
are globally admissible in $X^\mathrm{M}_k$. As a consequence, all symbols in $p$ or $q$ except the columns $C_0$ and $C_{n+1}$ are dominated, and all symbols in $p$ or $q$ except the columns $C_{-1}, C_0, C_{n+1}$ and $C_{n+2}$ have a private neighbour with the right value. We prove that this is also the case for $C_{-1}$ and $C_0$. The cases for $C_{n+1}$ and $C_{n+2}$ are symmetric. Let $i \in \ZZ$. If $C_0[i] \neq k$ is not dominated by an element inside $p$ then the algorithm sets $C_1[i] = C_0[i]+1$, which dominates $C_0[i]$. Remember that $C_0[i]$ cannot both need being dominated and a good private neighbour, for $p$ would not be globally admissible if this were the case.
Now we know that every cell in $C_0$ is dominated, and has a potential good private neighbour in $C_1$. It remains to show that this neighbour is not dominated by $C_2$. Let $i \in \ZZ$ such that $C_0[i]$ does not have a good private neighbour in $p$. We showed that both $C_1[i-1]$ and $C_1[i+1]$ are less than or equal to $C_1[i]$. $C_1[i]$ is dominated by $C_0[i]$, hence $C_2[i] < C_1[i]$ and $C_1[i]$ is indeed a good private neighbour for $C_0[i]$. If $C_{-1}[i]$ needs a good private neighbour, either it has it in $C_{-2}$ or $C_{-1}$ and we are done, or it needs to be $C_0[i]$. But then $C_0[i]$ is dominated by $C_{-1}[i]$ so that $C_{1}[i] < C_0[i]$ and $C_0[i]$ stays a good private neighbour for $C_{-1}[i]$.

\item \textbf{Every position outside $\text{supp}(p) \cup \text{supp}(q)$  with label less than $k$ is dominated.}

We prove the case for columns $C_1$ to $C_{n-k-1}$ (the central column), the proof is identical for the columns at the right.
At the left of the central column, this is true by the definition of the algorithm: if $C_j[i]$ is not dominated, then $C_{j+1}[i]$ is equal to $C_j[i]$+1, or has at least this value if $C_{j+1}$ is the central column.

For the central column it is also true for the same reason by the definition of the algorithm for this column.

\item \textbf{Every dominant position outside $\text{supp}(p) \cup \text{supp}(q)$ with label greater than 0 has a good private neighbour.}

We recall that $j_0 = n-k-1$ is the index of the central column. Once again we only prove this for columns $C_1$ to $C_{j_0}$.
We first state an intermediate result before proving this.
\begin{claim}
Let $0 \leq j < j_0$ and $i \in \ZZ$.
If $C_j[i]$ needs a good private neighbour in $C_{j+1}$ then  $C_j[i]$ is the unique good private neighbour of $C_{j-1}[i]$, i.e. its predecessor was in the same situation.
\label{claim-algo-meta}
\end{claim}

\begin{proof}
    Let $0 \leq j < j_0$ and $i \in \ZZ$ be (if any) such that $C_j[i]$ needs a good private neighbour in $C_{j+1}$. This iplies, by the definition of the algorithm, that $C_{j-1}[i]$ did not need to be dominated by $C_j[i]$ (otherwise $C_{j-1}[i]$ would have been a good private neighbour). Since $C_j[i] \neq 0$ (a position with value 0 does not need a good private neighbour) then the definition of the algorithm implies that we are in the case $C_j[i] = C_{j-1}[i]-1$, hence $C_{j-1}$ needed a private neighbour.
\end{proof}

We now go back to the proof of \Cref{meta-gluing-th}. Let $0 \leq j < j_0-1$ and $i \in \ZZ$. Let us show that $C_j[i]$ has a good private neighbour if it is different from 0. From \Cref{claim-algo-meta} and an easy induction, we know that $C_0[i]$ also needed a good private neighbour in $C_1[i]$. We mentioned that each of $C_0[i-1]$ and $C_0[i+1]$ is dominated in $p$, and if any needed a good private neighbour in $C_1$ then this neighbour would have the same value as $C_0[i]$. By an easy induction, we can see that $C_j[i-1] = 0$ or $C_j[i-1] = C_j[i]$ and the same applies to $C_j[i+1]$. For the same reasons that $C_0[i]$ was guaranteed to have a good private neighbour in $C_1$, $C_j[i]$ has a good private neighbour in $C_{j+1}[i]$.

If now $j = j_0 - 1$, then $j = n-k-1-1 \geq 2k+3 - k - 2 \geq k+1$. Let us show that $C_{j_0-1}$ does not need a good private neighbour in $C_{j_0}$. Let us assume that this is false: it does need one. Then using the same argument as before, we know that each $C_u[i]$ needed a good private neighbour in $C_{u+1}[i]$ for $-1 \leq u < j$. This implies that $C_1[i] = C_0[i]-1$, and so on, so that $C_j[i] \leq C_{k+1}[i] = C_0[i]-k-1$. However, $C_0[i] \leq k$, which implies that this is not possible. Therefore, at most column $C_{k-1}$ may have cells which need a good private neighbour.

It only remains to show that positions in $C_{j_0}$, if containing values greater than 0, have a good private neighbour. This is true thanks to the definition of the algorithm for the central column: every set $S_i$ only contains values of undominated neighbours. This implies that any new defined value in the central column is one above its maximum undominated neighbour, which will not have any new neighbour to also dominate it. The exception is with the cells at lines 1 and -1. However, the way their values are defined in the central column ensures that at most one of them may have $C_{j_0}[0]$ as a good private neighbour, and the other one does not dominate it. This proves that both have a good private neighbour in column $C_{j_0}$, $C_{j_0-1}$ or $C_{j_0+1}$.
\end{enumerate}

\begin{figure}[h!]
\centering
\begin{tikzpicture}[scale=0.4, xscale = 1.5]

\foreach \y in {0,...,7} {\node[scale=0.5] at (1.5, \y+0.5) {$k$};}
\foreach \y in {0,...,7} {\node[scale=0.5] at (2.5, \y+0.5) {$k$};}

\foreach \y in {0,...,7} {\node[scale=0.5] at (3.5, \y+0.5) {$k-1$};}

\foreach \y in {0,...,7} {\node[scale=0.5] at (4.5, \y+0.5) {$k-2$};}

\foreach \y in {0,...,7} {\node[scale=0.5] at (7.5, \y+0.5) {$1$};}
\foreach \y in {0,...,7} {\node[scale=0.5] at (8.5, \y+0.5) {$0$};}
\foreach \y in {0,...,7} {\node[scale=0.5] at (9.5, \y+0.5) {$0$};}
\foreach \y in {0,...,7} {\node[scale=0.5] at (10.5, \y+0.5) {$0$};}
\foreach \y in {0,...,7} {\node[scale=0.5] at (11.5, \y+0.5) {$0$};}
\foreach \y in {0,...,7} {\node[scale=0.5] at (12.5, \y+0.5) {$1$};}

\foreach \y in {0,...,7} {\node[scale=0.5] at (15.5, \y+0.5) {$k-2$};}

\foreach \y in {0,...,7} {\node[scale=0.5] at (16.5, \y+0.5) {$k-1$};}

\foreach \y in {0,...,7} {\node[scale=0.5] at (17.5, \y+0.5) {$k$};}
\foreach \y in {0,...,7} {\node[scale=0.5] at (18.5, \y+0.5) {$k$};}

\foreach \x in {1,...,19} { (\ifthenelse{\x=6 \OR \x = 14}{}{ \draw[mydashed](\x, -1) -- (\x, 0);}}
\foreach \x in {1,...,19} { (\ifthenelse{\x=6 \OR \x = 14}{}{ \draw[mydashed](\x, 8) -- (\x, 9);}}

\foreach \x in {0,...,8} {\draw[mydashed] (0,\x) -- (1,\x);}
\foreach \x in {0,...,8} {\draw[mydashed] (19,\x) -- (20,\x);}

\draw (1,0) grid (5,8); 
\draw (7,0) grid (13,8); 
\draw (15,0) grid (19,8);

\node at (-1,4) {$p$};
\node at (20.8,4) {$q$};

\node at (6,4) {$\cdots$};
\node at (14,4) {$\cdots$};

\node[scale=0.625] at (2.5,-1.5) {$C_0$};
\node[scale=0.625] at (3.5,-1.5) {$C_1$};
\node[scale=0.625] at (4.5,-1.5) {$C_2$};

\node[scale=0.625] at (7.5,-1.5) {$C_{k-1}$};
\node[scale=0.625] at (8.5,-1.5) {$C_{k}$};
\node[scale=0.625] at (9.5,-1.5) {$C_{k+1}$};
\node[scale=0.625] at (10.5,-1.5) {$C_{k+2}$};
\node[scale=0.625] at (11.5,-1.5) {$C_{k+3}$};
\node[scale=0.625] at (12.5,-1.5) {$C_{k+4}$};
\node[scale=0.625] at (15.5,-1.5) {$C_{2k+1}$};
\node[scale=0.625] at (16.5,-1.5) {$C_{2k+2}$};

\draw[line width =0.5mm] (3,-1.5) -- (3,9.5);
\draw[line width =0.5mm] (17,-1.5) -- (17,9.5);
\end{tikzpicture}
	
\caption{Illustration of the fact that $X_k^M$ is not $2k+2$-block-gluing: when attempting to glue $p$ and $q$, the values of the next $k+1$ columns are forced, and leave the cell of the middle columns $C_{k+1}$ and $C_{k+2}$ undominated.}
\label{meta-k-counter-example}
\end{figure}
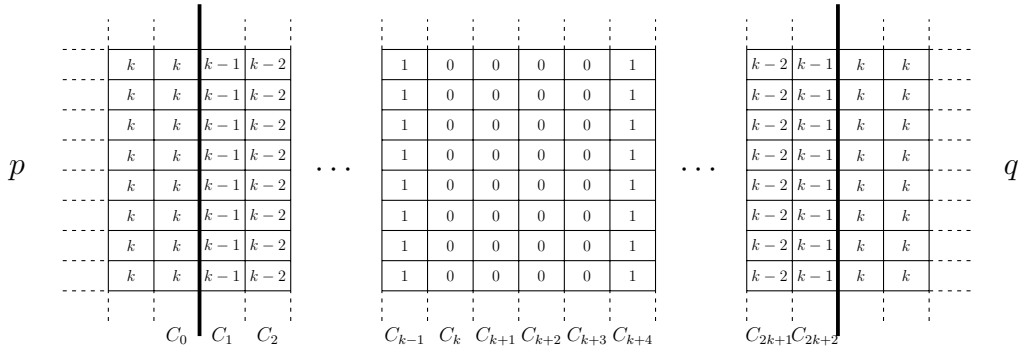

\item \textbf{$X_k^\mathrm{M}$ is not $2k$+2-block-gluing.}

We take a globally-admissible half-plane pattern $p$ of $\ZZ_{-} \times \ZZ$ such that its rightmost two columns are filled with value $k$. It is globally admissible because the preceding column may be filled with values $k-1$, the one before by $k-2$, and so on until a column of 0s. The column preceding it must also be filled with 0s and the one at its left can be filled by ones, then by twos, and so on. We take $q$ as the vertical symmetric of $p$.

By \Cref{prop-mini-meta-k} since the cells of $C_0$ (see \Cref{meta-k-counter-example}) have no neighbours of value $k-1$ in $p$, column $C_1$ must be filled with value $k-1$. In fact for the same reason, for every $1 \leq i \leq k$, the column $C_i$ must be filled with the value $k-i$. Now, to guarantee for each cell of column $C_{k-1}$ that its neighbour in $C_k$ is its private neighbour, column $C_{k+1}$ must not dominate column $C_k$, hence it must be filled with zeroes. A symmetric reasoning forces the values of the cells of the other columns between $p$ and $q$. Now columns $C_{k+1}$ and $C_{k+2}$ are not dominated. This proves that the minimal-meta-$k$-dominating subshift $S_k^\mathrm{M}$ is not ($2k+2$)-block-gluing. Hence it is block gluing with $c = 2k+3$.
\end{itemize}
\end{proof}

We can notice, for the minimal meta $k$ domination, that there is some linear decrease on the values of the forced cells. Indeed, in \Cref{meta-k-counter-example} for instance,  we see that when the label is $k$, we managed to need roughly $k$ steps to get down to 0, so that a minimum gap cannot be less than around $2k$. We could easily have only needed a gap of $k$ if we had slightly modified the rules, so that a private neighbour of a vertex with label $i$ may have either label $i-1$ or $i-2$. We may even implement a roughly $2k/q$-block-gluing SFT if a vertex of label $i$ needs a private neighbour of label between $i-q$ and $i-1$. This behaviour is similar to the one in \Cref{rk-sagging} when we mentioned the sagging of the set of forced cells. In the example showing that the minimal-meta-$k$-domination subshift is not ($2k$+2)-block-gluing, we took half-plane patterns to avoid some technicalities. Had we not done that, at each new column one cell at the top (and one at the bottom) would not have been forced, since they have upper (and lower) neighbours outside the band of the same height of the pattern. These extra cell give some liberty for the symbol in the top and bottom cells, so that with each new column we may lose two lines which were forced but which are independent now. This constitutes, as in \Cref{rk-sagging}, a linear sagging. Hopefully, here it does not prevent the SFT from being block-gluing because the values decrease by themselves.

\section{Conclusions}
\subsection{Counting dominating sets}
Here we both proved the existence and the computability of an asymptotic growth rate for four variants of the domination problem in grid graphs. We gave some bounds and some values we think approximates rather accurately each of these growth rates. The bounds for the domination and total domination problems are rather good: respectively 0.5\% and 1\% of the computed value. 
However, the ones for their minimal counterparts are looser: the gaps between the upper and the lower bounds are around 20\% of the lower bound.

\Cref{th-minimal-dom} improves the bound given by Fomin et al.~\cite{enum-minimal-dom-paper}, when the graph is a grid. They provide an algorithm enumerating the minimal dominating sets of a graph. By analysing its complexity, they show that there are at most $1.7159^n$ minimal dominating sets for a graph on $n$ vertices. We reduce this bound by approximately $10\%$ in the case of grids.

As for the minimal domination and minimal total domination, the associated bounds could be improved by using a more powerful computer (mainly one with more than 1.5TB of memory), or by optimising the technique or finding one more efficient.\\

Also, the bounds we give are not numerically certified because of two reasons. First, the computations are done in floating-point arithmetic, hence some rounding errors may  propagate. However, we only do additions and one division, so this should not occur. Also, we checked one value using arbitrary-precision numbers and it gave the exact same results. Second, we compute the largest eigenvalue of the matrix by using the power iteration method: start with a vector $V$, compute the iterates ($M^kV$) and see how the norm of the vector evolves. We observe that this method converges rather quickly (around 20 iterations), but it does not certify any digit of the value as being the right one: we have no guarantee on the precision of the numerical values. However, when we used the arbitrary-precision numbers, we computed $M^{150}V$ and there seemed to be 85 digits which stabilised, so it really looks like it is converging here. One direction of work could be to find a way it to certify the digits of the computed eigenvalues and use it. 

\subsection{Around the block-gluing property}
\label{counting-conclusion-gluing}
In the proofs of the block-gluing property, we mentioned some phenomenon: the sagging of the values in the middle cells between two pattern we try to glue, or the one about the number of lines which are still forced when the patterns are finite (instead of half planes). We discuss here the matter further, for there exist subshifts for which this sagging prevents the block-gluing property: the values or number of forced cells decrease too slowly, in a linear manner. This prevents two patterns to be glued with a gap less than something linear in the heights of the patterns, as we shall see.\\

Among others (and stronger) results, Gangloff and Sablik introduced the concept of linearly-block-gluing SFTs in~\cite{gangloff-sablik} and showed that some SFTs indeed have this property. One example they show kinds of encodes an integer by having a pattern with a column of $k$ black cells. By forbidding a list of three small patterns, they force the next column to be a column of $k-2$ black cells, the black cells being centred compared to the ones of the previous column. The next column must be a column of $(k-2)-2 = k-4$ black cells also centred, and so on. 
This means that if we begin with a column pattern of height $k$ full of black cells we need a gap of size around $k/2$ to glue it to a column of white cells. Hence no constant block-gluing gaps are possible, but only a gap the size of which is linear in the height of the pattern. We can see this as if the height of the pattern encodes a number, which then decreases is a linear way.

This reminds us of the \Cref{question-block-gluing} to know if there exists subshifts, or even SFTs which are $f(n)$-block-gluing, but where $f$ is not an affine nor a constant function. Let us now think about a potential SFT which would $\log(n)$-block-gluing for instance. We can think that each new column we build between two patterns should have halved the number of forced cells from the number in the previous column. For instance, with the previous example we would need, beginning with a column with $k$ adjacent black cells, that only $k/2$ cells are black in the next columns, then $k/4$, and so on. This could be done by several ways. We could, for instance in one dimension, encode an arbitrarily large alphabet with a finite alphabet but considering arbitrarily large patterns: if the alphabet is $\{0,1\}$ a word $1111$ could represent 15. For instance, we can define a subshift which would allow any word which forbids $01^k0^q1$ if $q < \log_2{k}$. This subshift is not an SFT and is not constant block gluing, but it is $\log(n)$-block-gluing.

However, for a SFT to be $\log(n)$-block-gluing appears to be a much more difficult problem: we have a finite alphabet and only a finite number of finite patterns which we forbid. For a subshift to have a sagging other than linear, it seems to need to be able to encode arbitrarily large values: if the sagging is logarithmic, it should work for patterns arbitrarily large so that it can effectively "compute" the logarithm of arbitrarily large numbers, but only with local rules of fixed radius. This seems rather difficult to achieve in the world of SFTs, but Gangloff and Sablik also showed in~\cite{gangloff-sablik} that this was possible by constructing a $\log(n)$-block-gluing SFT. They basically implement an $+1$ adder component (like in CPUs) with a system of carry. The number of digits only increases by one when we reach a new power of two, that is after $2^k$ iterations if the pattern is of size $2^k$. This settles one case of the question we asked in \Cref{question-block-gluing}, but all the spectrum of other subliner functions (except for the logarithm) or superlinear ones remains open. For instance, are there SFTs or subshifts which are $\sqrt{n}$-block-gluing?

\resetlinenumber
\chapter{Tiling rectangles with polyominoes}
\label{polys-chapter}
Even though we can wonder how we let this possible and even happen, some parts of the Pacific ocean are currently covered by 80 000 tons of plastic, tiling around three times the surface of France. Some other people are untiling big and ancient forests right now: over the past year, for instance, the Amazon rainforest in Brasil suffered the loss of between 500 000 and one million football pitches. The rate of deforestation of this area has increased eight times as much it was before Bolsonaro was elected.\\ 

Fortunately here, while we try to tile, it is not with plastic but instead with objects called polyominoes, looking like the inoffensive Tetris pieces. We try to fill a rectangle with such pieces, a bit like in a jigsaw puzzle, except that all the pieces are copies of the same one.

Informally, a polyomino is a connected finite set of unit squares. This means that the polyomino is made of only "one part": see \Cref{ex-polys-def}. We already talked about them in \Cref{domination-chapter}. We showed there how to encode a domination problem into a polyomino. A dominating set of the grid is then a \textit{covering} of the rectangle with the right polyomino: placing copies of the polyomino so that each cell is covered by at least one polyomino. A minimum dominating set is a covering of the rectangle with as few polyominoes as possible. The number of polyominoes used in such a covering gives the domination number of the grid.

In this chapter, we are interested in a notion close to the one of coverings: the notion of \emph{tiling}. Tiling a surface with a shape means covering the surface with copies of the shape such that no copy go over the surface, and no two copies overlap. We want to know, given one particular polyomino, if we can tile a rectangle with it. There are many questions in the topic of tilings with polyominoes. The natural question is to wonder which polyominoes can tile the plane. In the case when we allow only translations, that is when we choose one orientation for each polyomino and all its copies are obtained from it by translation only, the case is settled. Indeed, there exists an algorithm to decide if a given polyomino can tile the plane $\ZZ^2$ by translations. However, in the case when we allow translations, rotations and mirrors of the polyomino, we have little knowledge. We do not know whether or not this problem of tiling is decidable or whether there exists a polyomino which tiles only the plane in an aperiodic way. However, we know that the problem of tiling a plane with a set of polyominoes is not decidable. We will develop a bit more on these results in \Cref{section-defs-polys}. Unless written otherwise, when we speak of tiling, we mean a tiling by copies allowing translations, rotations and vertical mirrors.

A funny concept about the polyominoes is the one of \textbf{reptiles}.\footnote{Don't worry, they don't bite!} It stands for auto-replicating tiles. This means that we can assemble several copies of a polyomino $P$ in such a way that the obtained shape is $P$ zoomed in. This implies that the polyomino can tile the infinite plane $\ZZ^2$. Indeed, just take your shape, assemble copies of it so that you obtain that shape again but bigger in size. This way, you tile a bigger and bigger connected area by extending step by step your partial tiling, and never changing what is already done. Hence you will tile the whole plane if you can wait indefinitely, with no need for the axiom of choice.

This chapter is focused on a stronger property a polyomino can have: to be \textit{rectifiable}, i.e. to be able to tile (allowing rotations, translations and mirrors) some rectangle. This property is stronger than tiling the plane since when a polyomino tiles a rectangle then we can put together copies of this arrangement to tile the plane. We can also assemble the rectangles to obtain a big square, with which we can obtain a bigger shape of the polyomino, hence any rectifiable polyomino is also a reptile. If a polyomino is rectifiable, we want to find the smallest rectangle which can be tiled by this polyomino, or in fact the smallest number of copies we need to tile a rectangle. This minimum number of copies is called the \textbf{order} of a polyomino, if it exists. This problem looks rather simple, or at least can be expressed very easily, however there are big gaps in our knowledge of the subject. For instance, we do not know if there exists a polyomino the order of which is five, that is one which can tile a rectangle with five copies but cannot tile a rectangle with a smaller number of copies.

In 1989, in a paper studying polyominoes, their tileability and their rectifiability, Golomb~\cite{golomb-vieux} asked a question which is still unresolved. It was our goal for studying polyominoes.

\begin{question*}[\cite{golomb-vieux}]
Is there any polyomino with odd order greater than one?
\end{question*}

In the first chapter, we first give basic formal definitions and a bit of history on the subject of tilings using polyominoes. We will for instance give more details of the hierarchy Golomb made about the tiling properties polyominoes can have (tiling the plane, being rectifiable, being a reptile for instance). We will also say a bit about tilings with sets of polyominoes. In \Cref{algos-basiques}, we will give the methods which can be used to find the order of polyominoes. We give there simple versions of the algorithm, to be refined in \Cref{section-optim-polys}. It begins with \Cref{section-gen-polys} where we explain how we enumerate the polyominoes, which is a first step before testing them for rectifiability. We then focus on how to test if a rectangle can be tiled by some polyomino. We present two methods: one looking like a DFS in \Cref{section-dfs} which is rather slow, and another one which looks like a breadth-first search (BFS) in \Cref{section-bfs} which is much faster. We even mention in \Cref{section-gurobi} an approach using integer linear programming, which is slower than the two previous methods. In \Cref{section-optim-polys} we explain the methods we used to detect polyominoes which are not rectifiable, in order to rule them out. We also provide some optimisations mainly for the BFS algorithm. In \Cref{section-stats}, we give some statistics on the number of polyominoes which are rectifiable, and on their orders. We also give some ideas for further work in order to tackle this problem.

\section{Definition and some history}
\label{section-defs-polys}
\begin{deff}
    A \textbf{polyomino} $P$ is a finite and connected\footnote{Two unit cells of $\ZZ^2$ are \textbf{connected} when they share an edge.} union of unit cells of the $\ZZ^2$ lattice.
\end{deff}

\begin{figure}[h]
\centering
\includegraphics[scale=1]{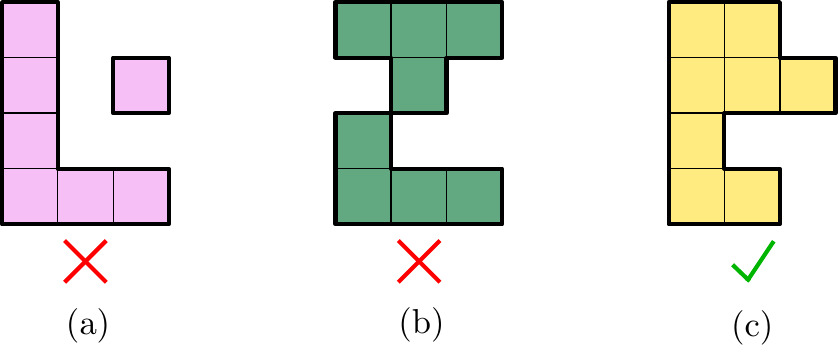}
\caption{Illustration of the polyomino definition. (a) is not a polyomino because it is not connected. The same goes for (b): connectivity is done through edges and not corners. (c) is a valid polyomino.}
\label{ex-polys-def}
\end{figure}

\begin{deff}
An \textbf{isometry} of the plane is a geometric transformation which preserves the distance between any pair of points. If $a$ and $b$ were at distance $d$, their images must be at distance $d$ as well.
\label{def-isometry}
\end{deff}

\begin{fact}
Including itself, a polyomino may have up to 8 isometric copies in total (see \Cref{example-copies}).
All isometric copies of $P$ are polyominoes obtained from $P$ by a succession of rotations by $\pi/2$ and vertical symmetries.
\label{fact-copies}
\end{fact}

\begin{deff}
Given a (finite or not) surface $X \subset \ZZ^2$ and a polyomino $P$ we say that $P$ \textbf{tiles} $X$ when $X$ can be decomposed into isometric copies of $P$.
\label{def-tile}
\end{deff}

In \Cref{def-tile} $X$ may be finite of infinite. In this chapter, we will mostly consider finite rectangles to be tiled, but sometimes we speak of tiling the plane $\ZZ^2$.

\begin{figure}[H]
\centering

\includegraphics[scale=1]{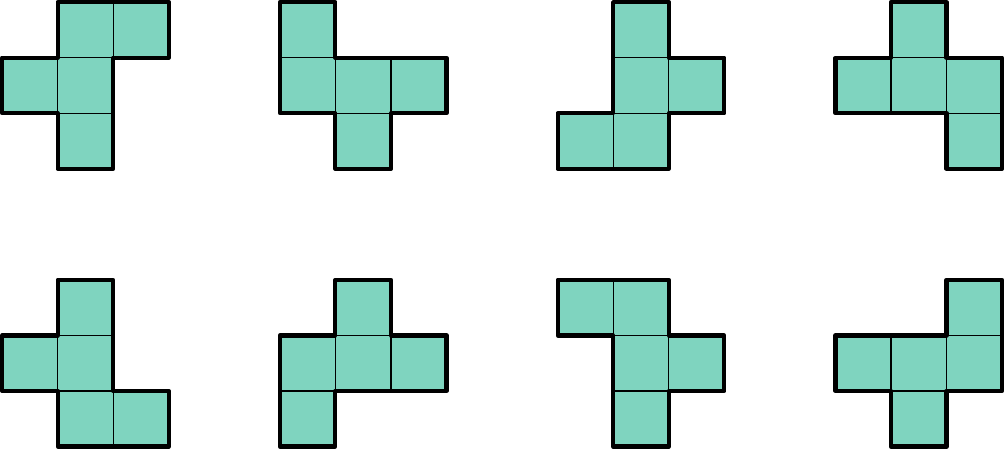}

\caption{The eight copies of one heptomino. The second line is the vertical flip of the first one.}
\label{example-copies}
\end{figure}

\begin{deff}
A polyomino $P$ is \textbf{rectifiable} when there exists a finite rectangle which can be tiled with copies of $P$. The \textbf{order} of a polyomino is the minimum number of copies with which we can tile a rectangle, or $+\infty$ if it is not rectifiable.
\end{deff}
Golomb~\cite{golomb-hierarchy} established in 1966 a hierarchy of polyominoes according to their ability to tile some parts of the planes. For instance, we mentioned that reptiles can tile the plane, thus being a reptile is a stronger property than tiling the plane. Also, being rectifiable implies that a polyomino can tile an infinite strip, which implies it can tile a half-plane, which in turn implies that it can tile the whole plane. We mentioned in the introduction that a rectifiable polyomino can tile a big square by combining the rectangles, hence it is a reptile: a bigger version of the polyomino can be constructed with copies of this big squares. In fact, the rectifiability is the strongest class of the hierarchy Golomb defined because of this ability to tile a square: when you tile a square you can tile any shape decomposable into squares. The full set of categories and implications can be found in~\cite{golomb-hierarchy}.
One big question remains open in this domain:
\begin{question*}
Is tiling the plane with one polyomino decidable?
\end{question*}

Golomb also studied the tiling of different surfaces when, instead doing it with just one polyomino, we are allowed to use several polyominoes from a given set. In 1970, he gave in~\cite{golomb-sets-hierarchy} a hierarchy of the different tiling problems when we tile with a set of polyominoes. He also answered the above question in this context.

\begin{thm}[Golomb,~\cite{golomb-sets-hierarchy}]
The problem of tiling the plane with a set of polyominoes is equivalent to Wang's domino problem and is therefore undecidable.
\label{undecidable-th}
\end{thm}

This means that there are no Turing machines which, on every possible set of polyominoes given as input, would output, in finite time, 1 if we can tile $\ZZ^2$ with the polyominoes in this set, and 0 otherwise. \Cref{undecidable-th} was further improved in 2008:

\begin{thm}[Ollinger,~\cite{undecidable-five}]
The problem of tiling the plane with a set of five polyominoes is undecidable.
\end{thm}

Ollinger also proved this result by reducing the domino problem to the problem of tiling the plane with a set of five polyominoes. This narrows the gap between the undecidability of tiling the plane by a set of polyomino and the question about doing so with a single polyomino. Maybe tiling the plane with a set of four polyominoes would turn out to be decidable, hence implying the same for sets of size less than four, including tiling with one polyomino.

In 1991, Beauquier and Nivat showed in~\cite{beauquier-nivat} an interesting characterisation of the \linebreak polyominoes which can tile the plane \uline{by translation} (forbidding copies obtained by rotation or mirror). This gives an algorithm to check if a polyomino can tile $\ZZ^2$ by translation: all we needed is to look at its frontier and use some combinatorics on words. However this does not help us here since we consider also rotations and symmetries of the polyominoes. Recently, Nitica investigated how to translate Golomb's hierarchies of polyominoes and sets of polyominoes when only translation is allowed. This resulted in different classes and some inclusion relations between these classes, as it can be seen in~\cite{nitica-polys} and~\cite{nitica-sets-polys}.

Shortly after, in 1992, Stewart and Wormstein~\cite{article-polys-order-3} put a first stone towards the resolution of this question: they answered the question about a polyomino of odd order for the smallest odd integer greater than one. We still do not know the answer for any other odd number.
\begin{thm}[\cite{article-polys-order-3}]
No polyominoes of order 3 exist.
\end{thm}

This theorem can be restated in the following way: if a polyomino can tile a rectangle with three copies then this polyomino is necessarily a rectangle. However, there are no strong reasons for ruling out any odd order greater than three. Indeed, Golomb has shown that there are infinitely many polyominoes which can tile some rectangle with an odd number of copies; the problem is that it might not be the minimum number of copies needed, hence not their orders in that case. Let us call such a polyomino, which can tile some rectangles with an odd number of copies, an \textbf{odd polyomino}. In 1997, Reid~\cite{article-no-upper-bound-odd} showed that the minimum number of odd copies needed to tile a rectangle with a polyomino can be arbitrarily large. He showed this by giving a way to construct, for each prime number $p$, a polyomino of odd-order $3(p+2)$: a polyomino which can tile a rectangle with $3(p+2)$ copies but cannot tile a rectangle with an odd number of copies less than that.

The funny thing is that while we have no clues about half of the possible orders, namely the odd numbers, we have knowledge for half of the even orders.

\begin{thm}[\cite{golomb-vieux}]
Each positive multiple of 4 is the order of some polyomino.
\end{thm}
Golomb showed this theorem by providing, for each number $4s$, a polyomino of order precisely this number.
In 1989, Dahlke \cite{dahlke} gave the order of a polyomino mentioned by Golomb: it was 92 (see \Cref{fig-92}). He also answered in another paper, with the same program, that another polyomino the ordered of which Golomb had also asked, had order 76.
To summarise, regarding which numbers are the order of some polyomino, we know the answer for a quarter of $\NN$ plus a few isolated cases. So there is still much space for discoveries. We also do not know the answer for small numbers: are there polyominoes of order 5? Of order 6, 14 or 22?\\

\begin{figure}[h]
\centering
\includegraphics[scale=1.4]{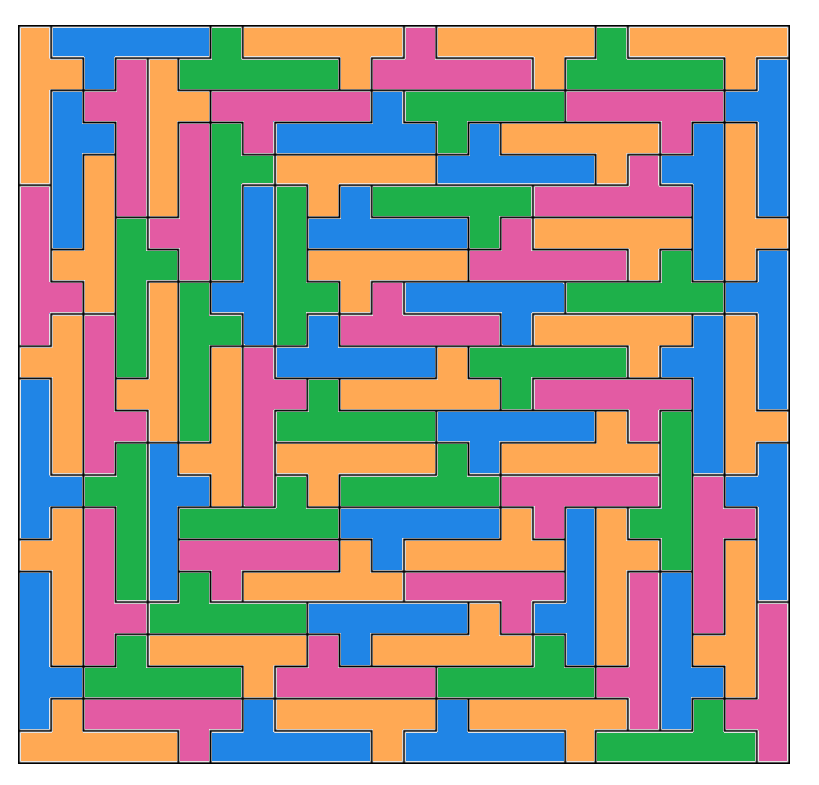}

\caption{A tiling with 92 copies of a polyomino of order 92. It also illustrates the four-colour theorem: two polyominoes sharing an edge always have different colours.}
\label{fig-92}
\end{figure}

\section{Finding the order of a polyomino: the basic algorithms}
\label{algos-basiques}
We explain in this section how to look for rectifiable polyominoes. We first show how to enumerate them, or in fact, enumerate the ones which have no holes since any polyomino with a hole cannot tile the plane. After this we present two methods, given a rectifiable polyomino, to find its order. This section will be completed by \Cref{section-optim-polys}, which provides some optimisations of our algorithms, as well as techniques to show that some polyominoes are not rectifiable.
In all this section we assume that we have a polyomino $P$ and try to find its order, by trying to tile rectangles of size $n \times m$ with isometric copies of $P$.

\subsection{Enumerating all the polyominoes}
\label{section-gen-polys}

Counting polyominoes is not a new topic at all. Several tables or sequences giving the number of a certain type of polyominoes according to their sizes can easily be found on some papers and on the Internet. The On-Line Encyclopedia of Integer Sequences (\url{oeis.org}) lists such tables: see for instance the sequences A000105, A001168, A000988. There are several sequences because there are several ways to list polyominoes, according for instance to whether or not we count isometric copies as different polyominoes. The sequences respectively focus on free, one-sided and fixed polyominoes. The first one considers that two isometric polyominoes are the same; the second one considers that two polyominoes are the same if one can be obtained from the other by combining rotations and translations; to the third one two polyominoes are the same only if they differ by translation.

In this chapter, we are interested in enumerating all the polyominoes. We consider, since we authorise isometric copies of a polyomino for our tilings, that two isometric copies of a polyomino are the same polyomino, and we enumerate just one of them. As mentioned in \Cref{fact-copies}, a polyomino can have up to 8 "different" forms: we may obtain them for instance by applying from zero to three rotations by $\pi/2$ and, on top of this, also applying  zero or one vertical flip\footnote{also called \uline{horizontal} symmetry} (see \Cref{example-copies}).

\begin{deff}
We say that $X$ is a set of \textbf{free polyominoes} when no two elements are isometric.\footnote{We keep this phrase because it is widespread, though we are not quite satisfied with it: the polyominoes are not free. Saying that a polyomino is free means nothing. It would be better to speak of a free set of polyominoes: two sets of polyominoes are equivalent if any polyomino in one has an isometric copy in the second, and vice versa.}
\end{deff}

\begin{deff}[see \Cref{figure-holes}]
Let $P$ be a polyomino placed on a grid. Any finite connected region of cells not in $P$ which is disconnected from the rest of the grid by $P$ is called a \textbf{hole}.
\end{deff}

\begin{figure}[h]
\centering
\includegraphics[scale=1]{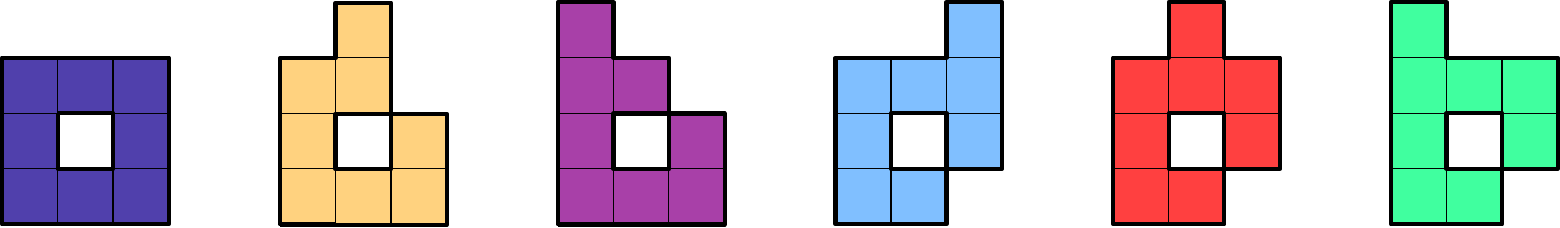}
\caption{The six octominoes containing a hole.}
\label{figure-holes}
\end{figure}

\begin{notation}
Let us denote by $\bm{\PP_n}$ the set of free polyominoes of size $n$ \uline{with no holes}.
\label{notation-holes}
\end{notation}
We now introduce two new notions we will use to show how to choose a \emph{representative} for equivalent polyominoes. As we will see, doing so will improve the running time of the enumeration.

\begin{deff}
Let $L$ be the list of coordinates of a copy of a polyomino $P$.
We say that $L$, or the copy of $P$ it represents, is \textbf{lifted} when $L$ is sorted in lexicographic order and the first coordinate of the list is $(0,0)$.\\
The copy represented by $L$ is the \textbf{representative} of $P$ when $L$ is lexicographically smaller than the lifted list of coordinates of the other copies of $P$.\\
\label{def-representative}
\end{deff}

To enumerate the polyominoes of size $n$, we use a rather simple method. It is based upon the fact that in any polyomino of size $n > 1$ there is at least one square we can remove to obtain a polyomino of size $n-1$. This fact is trivial, but here we deal with hole-free polyominoes, which is trickier. The proof of \Cref{lemma-gen-polys} will explain why this assertion is still true in the context of hole-free polyominoes. Our algorithm (see \Cref{algo-gen-poly} on next page) takes as input the set $\PP_{n-1}$ of polyominoes of size $n-1$: for each $P \in \PP_{n-1}$, we add a square at every possible location (maintaining the connectivity of the polyomino). We will explain the hole detection, done with a \emph{flood-and-fill} algorithm, in \Cref{section-optim-polys}.

In order to avoid storing several copies of the same polyomino, we use \Cref{def-representative} to keep a unique representative for each copy we enumerate. Let us reformulate the definition: a copy is \emph{lifted} when the first element of its list of cells is $(0,0)$ and this list is sorted by increasing x-values, and in case of equality, increasing order of y-values. A lifted polyomino does not have negative $x$-coordinates, but may have negative $y$-coordinates. This notion is central to compare two polyominoes: any copy of a polyomino has a unique lifted list of coordinates, which is a representative of this copy. Two copies with the same lifted list of coordinates are in fact the same up to translation.
The \emph{representative} copy of a polyomino is the copy which, once lifted, has the list of coordinates which is minimum for the lexicographic order. For instance, in \Cref{example-copies} the leftmost cell (and topmost in case there are several) of every copy will have coordinate $(0,0)$ once the (coordinate list of the) copy is lifted. Then the two copies of the second column will have two cells of $x$-coordinate 0 so they are smaller than the other ones. The third cell in the (sorted) list will have $x$-coordinate 0 for both copies, but $y$-coordinate -1 for the bottom one, whereas the copy on the first row will have $y = 0$. Hence the copy of the second row and second column is the smallest copy of \Cref{example-copies}. Again, this concept of representative of a polyomino is crucial: one can detect if two copies of polyominoes are copies of the same polyomino by comparing their representatives. We will detail in \Cref{other-optimisations} the \textit{hash table} data structure we used to achieve a fast insertion of new polyominoes without having duplicates of them. 

\begin{center} \begin{algorithm}[H]
\SetAlgoLined
\SetKwData{Left}{left}
\SetKwData{This}{this}
\SetKwData{Up}{up}
\SetKwFunction{Union}{Union}
\SetKwFunction{FindCompress}{FindCompress}

\SetKwInput{Input}{Input}
\SetKwInOut{Output}{Output}
\Input{An integer $n$ and the set $\PP_{n-1}$}
\Output{$\PP_n$}

 $X \leftarrow \emptyset$\\
 \ForEach{$P \in \PP_{n-1}$}
 {
 	\ForEach{$(x,y) \in P$}
 	{
 		\ForEach{\normalfont{$(x',y')$ neighbour of $(x,y)$}}
 		{
 		$Q \leftarrow P \cup (x',y')$\\
 		\If{\normalfont{$Q$ has no holes}}
            {$Q' \leftarrow$ representative of $Q$ \tcp{the smallest lifted copy of $Q$}
 			Insert $Q'$ into $X$}
 		}
 	
 	}
 } 
 
 \KwRet $X$
 \caption{Computing $\PP_n$ from $\PP_{n-1}$}
 \label{algo-gen-poly}
\end{algorithm}
\end{center}

\subsection{Tiling a rectangle with a DFS (inefficient)}
\label{section-dfs}
We present here a classical backtracking algorithm, described in \Cref{algo-dfs-poly}, which features the characteristics of a depth-first search (DFS). The idea is, when examining one specific (partial) tiling, to continue it as far as we can. When we find a contradiction we rollback one step before, try another choice, continue, and so on. DFS stands for "depth-first search": this means that when we have several options, we first examine the first one as far as we can go before considering the second one, and so on. By doing this, we explore all possible tilings, hence if there exists a tiling we will find it. If we exhaust all the possible choices for where to put the next copy of $P$ without managing to tile the rectangle, this constitutes a proof of the fact that $P$ does not tile the rectangle.

\begin{center} \begin{algorithm}[H]
\SetAlgoLined
\SetKwData{Left}{left}
\SetKwData{This}{this}
\SetKwData{Up}{up}
\SetKwProg{Fn}{Function}{:}{}

\SetKwFunction{dfsfunc}{DFS\_tile}{}{}

\SetKwInput{Input}{Input}
\SetKwInOut{Output}{Output}
\Input{The dimensions $n$ and $m$ of the rectangle, the polyomino $P$}
\Output{True if the $n\times m$ rectangle can be tiled by $P$, False otherwise}

\Fn{\dfsfunc{grid, $P$, nbLeft}}
{
\If{$nbLeft = 0$}
{
	\KwRet True
}

$(x,y) \leftarrow \texttt{choose\_free}(grid)$\\
\ForEach{\normalfont{partial tiling $grid'$ extending $grid$ by exactly one copy of $P$ covering $(x,y)$}}
{
	\If{\dfsfunc{grid', $P$, nbLeft-1}}
	{
		\KwRet True
	}

}
\KwRet False
}
\vspace{0.3cm}

\KwRet \dfsfunc{empty\_grid, $P$, $nm/|P|$}
 \caption{Trying to tile an $n\times m$ rectangle, in a DFS fashion.}
 \label{algo-dfs-poly}
\end{algorithm}
\end{center}

One important point in the algorithm is not to enumerate the same partial tiling several times: we want to avoid putting a copy in $P$ at cell $(0,0)$ then one copy at cell $(3,3)$, realise that there is a contradiction... and then try to put a copy at $(3,3)$ and then one at $(0,0)$ and come up with the same conclusion. To enforce this, at each depth we choose a unique cell to be tiled at that step (the role of the function \texttt{choose\_free}). We then force the chosen cell $(x,y)$ to be covered at this step backtrack directly, should the one we chose fail to be covered. We are sure to still enumerate all possible tilings: the chosen cell must be covered at some point so it might as well be covered now, and we try every possible way to cover it.  When we go back because we fail to cover the cell we chose at the current step, it is possible that another option for a choice made at an earlier step will lead to a tiling of the rectangle.

In our program, we chose the function \texttt{choose\_free} to return the leftmost free cell, and in case of ties, the topmost one. However, it could be chosen differently as we mention below. 
The choice for this function \texttt{choose\_free} is crucial. Indeed, choosing  judiciously the next cell to be covered can impact the performances a lot. There is one main approach. It consists in trying one cell which has very few possibilities to be covered, so that we do not have to branch a lot. We hope that we can find successive positions for which the number of ways to cover them is indeed very small: ideally only one possibility to cover it, or zero so that we backtrack directly. If we are lucky, by doing so we only have one choice at each time and the running time would in fact be "linear": the first tiling we explore would work, with no need to backtrack. However, nothing guarantees that by choosing a spot with only a few possibilities we do not double (or even more) the number of possibilities for the next spots. The only sure thing is: if there is only one way to cover one tile, we may as well cover it now so that we realise some contradictions sooner, and we cannot make things worse since there was a unique choice to cover the cell here no matter when it is done. Other approaches based on other heuristics can be tried, as for instance tiling in "spiral": prioritise cells according to their distance to the closest edge of the rectangle, and in clockwise direction if there are ties. This choice was less efficient than the one prioritising the cell with the lest number of possibilities. Another approach is to prioritise cells to be tiled according to their distance from the top-left corner, that is tiling kinds of diagonal waves one after the other.\\

This approach is very fast for polyominoes of small order, for instance it takes 0.08s for a specific polyomino of order 50. However it led to very long running times for greater orders: it takes one minute on a polyomino of order 76, and 24 minutes for a polyomino of order 96. This is why we describe a second more efficient method in \Cref{section-bfs}: the BFS approach. This other approach is outperformed by the DFS for small orders: it takes one second for the polyomino of order 50. However, for the other two, the BFS method (with the optimisations of \Cref{bfs-optimisations}) takes respectively 6.5s and 4.8s.

\subsection{BFS on the frontiers}
\label{section-bfs}
We present here an approach which may seem slower at first: we try iteratively and "simultaneously" all the partial tilings with $k$ copies. Basically, we try to fill the rectangle column by column and enumerate all the partial tilings with one tile, then with two tiles, and so on. We do it on a BFS fashion: all the partial tilings with $k$ copies of $P$ are tried before trying the ones with $k+1$ copies. This seems slower because we enumerate all the partial tilings sequentially whereas a good heuristic might have discarded a lot of them early and explored a promising tiling first. However, if the rectangle cannot be tiled, we are forced to try all possible tilings, hence the BFS approach might not be more costly 

\begin{deff}
    Let $T$ be a partial tiling of a rectangle. Let $C_l$ (resp. $C_r$) be the leftmost (resp. rightmost) non-empty yet non-completely filled column. The \textbf{frontier} is the set of cells of the columns $C_l, C_{l+1}, \cdots, C_r$ covered by the polyomino (shown in red boxes in \Cref{fig-frontier-bfs}).
\end{deff}

\begin{rk}
Assume that at each step the next cell to be covered is chosen to be the leftmost one still uncovered, and in case of ties, the topmost one. The surface covered by any partial tiling $T$ obtained in this fashion is connected.
\end{rk}

This is true because of the design of the algorithm. Apart from the first one, every new copy we put in the current partial tiling $T$ necessarily shares an edge with $T$. In case the obtained partial tiling does not contain a hole, the surface it covers constitutes a polyomino.

We now introduce the basic algorithm without the technical details and without some optimisations: we enumerate the frontiers resulting of the partial tilings with first one tile, then two, and so on.

\begin{figure}[h]
\centering
\includegraphics[scale=0.8]{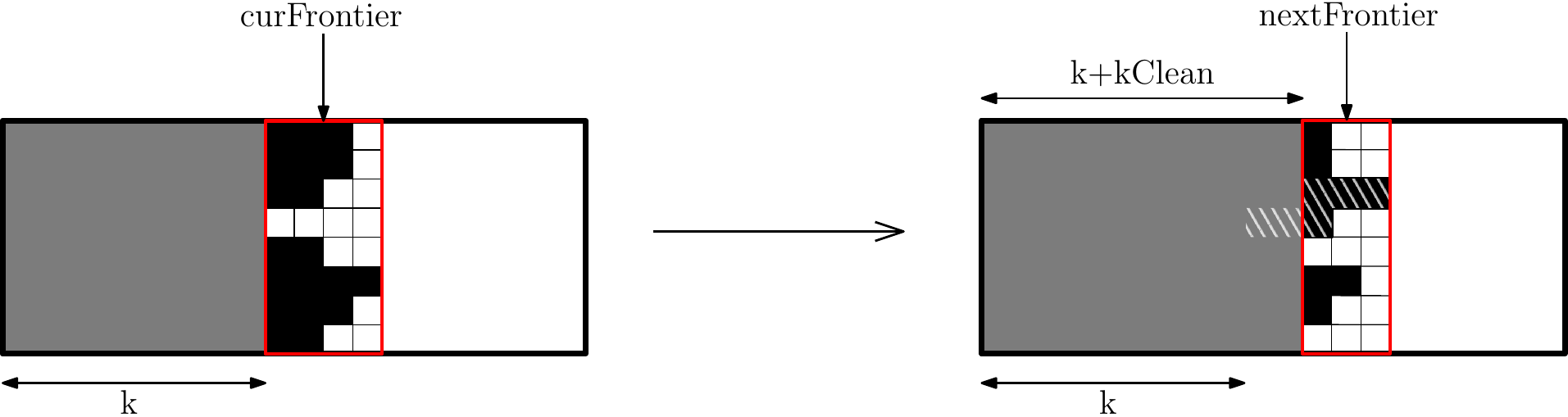}
\caption{Illustration of the creation of a new frontier from en existing one (\Cref{foreach-bfs} of \Cref{algo-bfs-poly}). The cells hatched with white lines represent the copy of $P$ which was added at this step.}
\label{fig-frontier-bfs}
\end{figure}

\begin{center} \begin{algorithm}[H]
\SetAlgoLined
\SetKwData{Left}{left}
\SetKwData{This}{this}
\SetKwData{Up}{up}
\SetKwFunction{bfsfunc}{BFS\_tile}

\SetKw{Break}{break}
\SetKwInput{Input}{Input}
\SetKwInOut{Output}{Output}
\Input{The dimension $n$ and $m_\mathrm{max}$ of the rectangle, the polyomino $P$}
\Output{(True, $m$) if the $n\times m$ rectangle can be tiled by $P$, False otherwise}

	nbColumnsTiledRectangle $\leftarrow +\infty$ \tcp{No tiled rectangle so far.}
	newFrontiers $ \leftarrow \emptyset\quad$ \tcp{This is a queue.}
	emptyColumn $ \leftarrow [0, \cdots, 0]$ \tcp{$n$ zeroes: an empty column; covered cells contain 1's.}
	newFrontiers.push((0, emptyColumn)) \tcp{0 means there were no filled columns at the left}
	\While{newFrontiers is not empty}
	{
		(k,curFrontier) $\leftarrow$ newFrontiers.pop() \tcp{curFrontier was seen with $k$ filled columns at its left}

			(x,y) $\leftarrow \texttt{choose\_leftmost\_free}(curFrontier)$ \tcp{In case of ties, the topmost one.}
	\ForEach{nextFrontier extending curFrontier by a copy of $P$ covering (x,y)\label{foreach-bfs} \tcp{See \Cref{fig-frontier-bfs}.}}
	{
		nextFrontierClean $\leftarrow$ nextFrontier stripped of its fully filled columns\\
		kClean $\leftarrow$ number of removed filled columns\\

		\If{nextFrontierClean = $[\;]$}
		{\tcp{We tiled a rectangle.} 
			nbColumnsTiledRectangle $\leftarrow$ k+kClean\\
			\Break \tcp{We get out of the loop.}
		}
		\If{nextFrontierClean was not seen yet\label{line-test-seen}}
		{
			Mark nextFrontierClean as seen\\
			newFrontiers.push((k+kClean, nextFrontierClean))
		}

	}
	}

\KwRet nbColumnsTiledRectangle
    \caption{Trying to tile a rectangle with $n$ lines and at most $m_\mathrm{max}$ columns, in a BFS fashion.}
 \label{algo-bfs-poly}
\end{algorithm}
\end{center}

The general principle of this algorithm was already used by Karl Dahlke. In this algorithm, we first explore all frontiers consisting of one copy of $P$, then the ones with two copies, then the ones with three copies, and so on. Hence this algorithm has the characteristic of a BFS. Indeed, BFS stands for "Breadth-first search": we explore step by step all possible solutions, as opposed to the DFS which explores fully each one after the other.\\

We can observe that any frontier may consists of at most $w$ columns, where $w$ is the maximum width of the copies of the polyomino. We speak here of the number of columns which are not completely filled. Indeed, since we always try to fill the leftmost free cell, any copy we add must fill this cell. The result follows since it spans at most $w$ columns. This means that the number of frontiers is bounded by $2^{nw}$. However, since we forbid a lot of configurations, we have fewer than $2^{nw}$ possible configurations: see \Cref{bfs-optimisations} to see how we discard some frontiers, apart from the obvious cases when the configuration contains a hole. Since the maximum width of a polyomino is fixed, it is better, when trying to tile a $n \times m$ rectangle, to choose $\min(n,m)$ as the number of lines of the rectangle we actually try to tile. The method works well for rectangles with disproportionate widths and heights, and less with the ones close to squares.

Keeping the frontiers in memory enables us to reduce the running time of the program (see \Cref{line-test-seen}). Indeed, we may obtain a frontier by several different ways. By keeping them in memory we know that when we find some frontier $F$ which we had already seen (necessarily with fewer filled columns since we enumerate the frontiers in a BFS way, with fewer copies needed first), we may skip it now instead of investigating again what other frontiers it leads to. To save memory and for easier access, we only remember the frontiers without the completely filled columns, and instead store this number and link the frontier to it: in \Cref{fig-frontier-bfs} the four (on the left) and three (on the right) columns, surrounded by a red rectangle, are the real frontiers we store. For the left one we associate the number k of filled columns, and the number k+kClean for the right one. By doing so, if we find a frontier $F$ with 42 preceding filled columns and see again later $F$ with 44 filled column, we overlook it. We used a \emph{hash table} to perform fast lookups for already seen frontiers. We detail this data structure in \Cref{other-optimisations}.

\subsection{Another approach: solving a linear program}
\label{section-gurobi}
At the beginning, we decided to give a chance to some solvers of (integer) linear programs. Indeed, the problem of knowing if a polyomino tiles a given rectangle can be expressed as an integer linear program. The idea is that we encode the locations of the copies of the polyomino and ensure that each cell is covered exactly once. To so do, we first consider the 8 different orientations of our polyomino $P$ (applying rotations and horizontal symmetry): $P_1, \dots, P_8$. For each $P_i$ and each cell $(x,y)$ of the rectangle, we denote by $C(P_i, x, y)$ the set of coordinates covered by placing the top-left cell of $P_i$ on cell $(x,y)$ or $\varnothing$ if doing so causes the polyomino to go over the rectangle. Our integer linear program can then be defined as:

\begin{flalign*}
\text{For each } x_0, y_0: \sum_{(x_0,y_0) \in C(P_i, x, y)}{A_{i,x,y}} = 1.
\end{flalign*}
Each sum is done for a fixed value $(x_0, y_0)$: we are summing over $i, x$ and $y$. To each $(x_0, y_0)$ corresponds an equation for our linear program. If the set of feasible solutions is not empty, this means that a tiling exists, and the variables set to one give this tiling. If, on the contrary, it is empty, this is a proof of the fact that the given polyomino does not tile the rectangle.

However, even by using Gurobi, one of the best solvers available, we did not obtain any interesting results with this method. The program using Gurobi was much too slow:
4 minutes for the polyomino of order 76, and 24 minutes for another one of order 96. The method in \Cref{section-bfs} (with the optimisations of \Cref{bfs-optimisations}) took respectively 6.5s for the polyomino of order 76 and 4.8s for the one of order 96. In addition to this, Gurobi used the 24 cores available in the test machine whereas the time for the BFS method was achieved using only one core. Gurobi was even (slightly) outperformed by the slow DFS algorithm!

\section{Refinements of the algorithms and other optimisations}
\label{section-optim-polys}
We present here two main things. On the one hand we introduce other techniques to improve our search. We begin by giving and explaining methods to show that a polyomino is not rectifiable, so that we may avoid to lose time by not attempting to find its order. We then present optimisations both in the design and the implementation of the algorithms of the \Cref{algos-basiques}, which we kept at the time as simple as possible for pedagogical purposes.

\subsection{Ruling out non rectifiable polyominoes}

\label{section-non-rectifiable}
One crucial point is to find which polyominoes can be discarded because they are not rectifiable. Indeed, it would take an infinite amount of time to try to tile all possible rectangles with a polyomino which is not rectifiable. This is why we need efficient methods which can discard as many non-rectifiable polyominoes as possible. Some of the techniques were used by Karl Dahlke (see~\cite{dahlke-website}), and some are "new" or were improved by us.

\paragraph{Playing chess.\\}
This method does not, properly speaking, detect non-rectifiable polyominoes, but rather shows that some classes of rectangles with specific properties cannot be tiled by some polyominoes.
This well-known parity argument consists in overlaying a checkerboard on the rectangle and deducing some properties of the tiling. For instance, let us assume that the polyomino $P$ is formed out of a rectangle with 2 lines and three columns by removing the cell of the middle column in the first row. Wherever on the checkerboard we place it, horizontally or vertically, it consumes either three white cells and one black, or the contrary, once put on the checkerboard. Since $P$ has an even number of cells, the rectangle it might tile must have, for instance an even number of lines, therefore it contains as many black cells as white cells. We deduce that there should be an even number of copies of $P$ since each one consumes an odd number of either black or white cells: the set of copies covering 3 black cells must be of even size, as must its white counterpart. This shows that our polyomino cannot tile any rectangle of size $n \times l$ when $kl$ is not a multiple of 10. For instance tiling a rectangle of size $5 \times 5$ would require 5 copies of the polyomino, hence either the number of white cells covered or the one of black cells covered would be odd. This argument implies that the polyomino is less interesting: its possible order cannot be odd. We also tried to obtain other equations showing that some other types of rectangles are impossible to be tiled by some polyominoes. We tried to put alternating columns of white and black cells, or even different moduli: instead of having white and black cells, that is reasoning modulo 2, we can try with other prime numbers. We managed to discard some rectangles which could not be discarded with the classical checkerboard argument, but not many more.

The checkerboard argument can also be exploited in another way: when tiling any rectangle there must be as many pieces covering three black cells as pieces covering three white cells since there are as many black cells as white cells. This could be used to improve the DFS algorithm: if, out of the $k = nm/|P|$ copies, we already placed more than $k/2$ pieces covering three white cells out of the $k$ pieces to put, we may backtrack.

\paragraph{Tiling a quarter of the plane.\\}
As we mentioned in \Cref{section-defs-polys}, being rectifiable implies tiling a half-plane. Using the same argument, it is easy to prove that in fact being rectifiable also implies tiling a quarter of the plane ($\{(x,y) \;|\; x \geq 0, y \leq 0\}$ for instance). We use this property to show the non-rectifiability of some polyominoes: if they do not tile a quarter of the plane then they are not rectifiable. We also use another similar property which we will define just below.\\

We now present one property which can be tested as soon as when we enumerate the polyominoes and helps us enumerating fewer of them, hence speeding up the process.

\begin{deff}
We say that a polyomino $P$ is \textbf{corner compatible} when there is a way to place it such that it covers the corner cell of an arbitrarily large rectangle without disconnecting the set of empty cells of the rectangle.
\label{def-corner-compatible}
\end{deff}

Informally, $P$ is not corner compatible when, however the way we place it, it separates the set of empty cells into several connecting components. This implies that the polyomino cannot tile a quarter of the plane, thus we can overlook them. Note that if we cannot cover the top-left corner, it implies that the polyomino is not corner compatible.

\begin{notation}
Let us denote by $\bm{\PP_n^*}$ the set of free polyominoes of size $n$, with no holes, \uline{which are corner compatible}.
\end{notation}

The cross with 5 squares (see \Cref{fig-poly-croix}) is an example of a polyomino which is not corner compatible: because it creates a hole which cannot be filled so that the corner cannot be tiled. Our program uses a modified version of \Cref{algo-gen-poly} to enumerate elements of $\PP^*_{n}$: it discards any polyomino which is not corner compatible. This is done in the same way as the elimination polyominoes with holes: we use a flood-and-fill algorithm\footnote{The flood-and-fill algorithm is a classical graph algorithm. It starts from one vertex and fills it with one colour, along with any vertex accessible: it colours its connected component with the same colour. Repeat this, with a new colour each time, for any vertex not yet coloured so that we compute the different connected components of the graph, in linear time.}. We assume (up to rotating the copy of the polyomino) that we want to tile the top-left corner. We put each copy of the polyomino into the smallest rectangle hull and add a column at the right and one line at the bottom. We then compute the set of connected components of the free cells (the ones which do not belong to the polyomino), considering that the polyomino is a "wall": two cells are connected if they are connected in the grid and none of them belong to the polyomino. The polyomino contains a hole or is not corner compatible if and only if there are several components.

We must now ensure that any element of $\PP^*_n$ can be obtained from an element $\PP^*_{n-1}$ to which we add a square, so that computing $\PP^*_n$ from $\PP^*_{n-1}$  indeed enumerates all the polyominoes we are interested in.

\begin{lemma}
Let $n > 1$ and  $P$ be a polyomino of size $n$ which is corner compatible and contains no holes. Then there exists a polyomino $Q$ of size $n-1$ with the same properties such that $P \setminus Q$ is a single square.
    \label{lemma-gen-polys}
\end{lemma}

What follows also proves that the algorithm we explained in \Cref{algos-basiques} for the enumeration of hole-free polyominoes is valid. Indeed, the lemma remains true if we remove the corner-compatibility requirement and only keep the hole-free one, as we mention at the end of the proof.

\begin{proof}
    Let us place $P$ on the top-left corner of a smallest rectangle containing it and augmented by a column on the right and a line below, in a way such that the flood-and-fill algorithm would find a single component outside $P$.
    We first observe that if the polyomino is corner compatible, this implies that there exists some $x_0$ such that every cell of the top row with $x \leq x_0$ belongs to $P$. Otherwise, this would create a component of free cells disjoint from the one containing the rightmost cell of the bottom line. The same goes for the leftmost column: there exists some $y_0$ such that every cell of the first column with $y \leq y_0$ belongs to $P$. Note that both sets we defined contain the top-left cell, so they are connected.
    
    We know, since $n > 1$, that $\max(x_0, y_0) > 0$. We assume without loss of generality that $x_0 > 0$: at least two cells of the top row are covered by $P$. We define a procedure which will terminate, and find the a cell we can remove. We first choose the cell $(x_0,0)$ and try to remove it. We know that we are neither creating a hole nor obtaining a non-corner-compatible polyomino since this cell has a free neighbour at its right. If the polyomino we obtain is still connected, we set $Q = P \setminus (x_0,0)$.

    Otherwise, we consider the polyomino $P'_1$ defined as the component disconnected from $P$ by removing $(x_0,0)$, and we set $P_1 = P'_1 \cup (x_0, 0)$. Note that $P'_1$ does not contain any cell of the form $(0,y)$ or $(x, 0)$: we mentioned that all the cells of this shape which belong to $P$ are connected to $(x_0, 0)$, hence not in $P_1$. This means that any $Q$ we will define by removing a cell from $P_1$ will necessarily be corner compatible. If $P'_1$ is reduced to a single cell, we may remove it to obtain $Q$ with the desired properties. Otherwise, we choose a cell $(x_1,y_1) \in P_1 \setminus (x_0,0)$ with fewer than four neighbours in $P_1$. Since it has fewer than 4 neighbours, removing it does not create a hole in the polyomino. If removing it does not disconnect $P_1$, then it does not disconnect $P$: the only way to disconnect $P_1$ from $P$ is to remove $(x_0,y_0)$ and we define $Q = P \setminus (x_1,y_1)$. Otherwise, let $P'_2$ be a connected component of $P \setminus (x_1, y_1)$ which does not contain $(x_0,y_0)$. We set $P_2 = P'_2 \cup (x_1,y_1)$. $P_2$ is smaller than $P_1$ so that our procedure selects smaller and smaller parts of $P$. This means that it will terminate at some point, finding some cell $(x,y)$ with fewer than four neighbours which fulfils our conditions.
    
If we only want a cell which does not create a hole or disconnects the polyomino, to show that our generation of hole-free polyominoes in \Cref{section-gen-polys} works, it suffices to choose for $(x_0,y_0)$ any cell with fewer than four neighbours. The rest of the procedure is not modified, and comes up with some cell such that $P \setminus (x,y)$ is a connected hole-free polyomino.
\end{proof}

\paragraph{Trying to "fully" tile a corner.\\}
This method looks a bit like the method we have just explained, to test if the polyomino is corner compatible. In fact, it is an extension of the corner-compatibility concept. Trying to fully tile a corner consists in taking the corner of a quarter of a plane. We then try to tile this corner, by tiling the $k$ cells which are the closest to the corner point. If we want to cover all cells at distance at most $d$, we will have to cover $(d+1)(d+2)/2$ cells. We do this using our DFS algorithm, covering cells according to the distance to the corner. If at any step it turns out that it is impossible to cover these $k$ cells, then we know that the polyomino will not be able to tile any rectangle (but it might tile the plane like the one in  \Cref{fig-poly-croix}).
This method extends the one of the corner compatibility because by tiling cells at distance at most $|P|+1$ for instance we would have realised that the position of the first copy of the polyomino induces an unfillable hole. To see the relative efficiency of this method, see \Cref{table-ruled-out} where we can observe that almost every polyomino ruled out by this method also does not "tile" a bottom band of a certain size. One could ask why introduce two concepts if tiling a corner implies being corner compatible. The answer is that the corner-compatibility test can be used during the enumeration phase, discarding a lot of polyominoes (see the data of \Cref{table-orders}). This saves us a lot of time generating the polyominoes.

\begin{figure}[h!]
\centering
\includegraphics[scale=0.5]{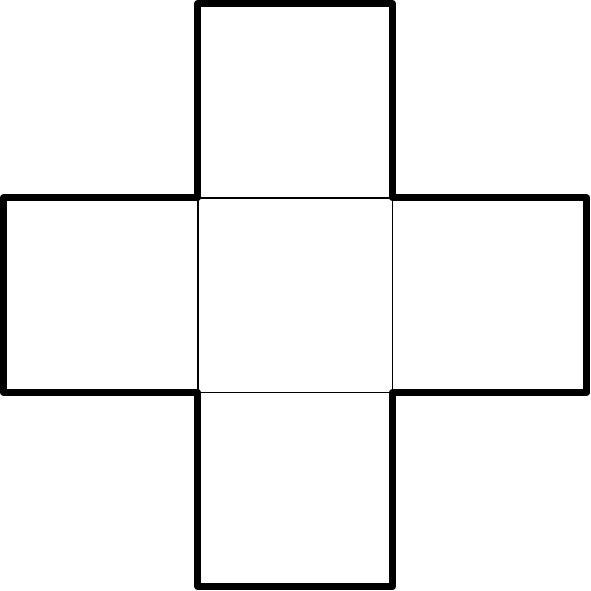}
\caption{The cross polyomino. It cannot tile the rectangle of a rectangle but can tile the plane $\ZZ^2$.}
\label{fig-poly-croix}

\end{figure}

\paragraph{Trying to "tile" the bottom band.\\}
The most efficient technique to rule-out polyominoes which are not rectifiable is to try to tile the bottom band of a rectangle. Indeed, it may not only show that a particular rectangle cannot be tiled by a polyomino, but also show that the polyomino is not rectifiable at all. The scheme is simple: we try to tile a band of height $h$ but we allow the tiling to go over $h$. However, it must not exceed the strip neither by the left or by the bottom. This way, we simulate the bottommost $h$ lines of any rectangle to be tiled by the polyomino. We enumerate all the possible frontiers in a BFS manner, like the one for tiling a rectangle. We keep all the frontiers in memory and stop when either we found a "full frontier": all the columns are filled, or when we no longer enumerate new frontiers. In the latter case, this means that any rectangle with a height greater than $h$ cannot be tiled: indeed this outcome of the algorithm shows that in any possible tiling the bottom-right corner cannot be covered. It is easy to show, by using this technique with $h=2$, that the $z$-shaped pieces in Tetris cannot tile any rectangle.
Sometimes we have no contradictions with $h = 2$ or $h = 3$, but when we increase the height we find one. For the efficiency of this method, see \Cref{table-ruled-out}.

\subsection{Improving the BFS approach}
\label{bfs-optimisations}

After presenting some ways to mark a polyomino as not rectifiable or showing that it cannot tile some family of rectangles, we focus here on the BFS algorithm described in \Cref{algo-bfs-poly} for which we present some optimisations.

\vspace{2cm}

\paragraph{Looking for a complementary frontier.\\}
    In addition to using the BFS approach, Dahlke used a clever idea: when trying to tile a rectangle, we may at each step look for the \textbf{complement} of the current frontier among the ones we have seen. If we have already enumerated the complement (the rectangle deprived of the current tiling), then assembling the two frontiers yield to a tiling of the rectangle and we are finished.
    Dahlke uses this idea to stop his search at roughly half the width of the desired rectangle: if a tiling exists, then it can be split into two parts of almost equal sizes, the biggest of the two being roughly half of the width of the rectangle. Therefore it is possible to look for this biggest part, and to stop the search when we no longer have any chance to find it.
    However, due to the way we generate the frontiers, we cannot guarantee that we would enumerate this biggest part of the particular decomposition we spoke of. Hence we always look for a complementary frontier, but we do not stop the search at half the width of the rectangle. Yet, in our runs, we could always find the order of the known rectifiable polyominoes by stopping halfway of when trying to tile a rectangle.

\paragraph{Using symmetries.\\}
The notion of symmetries can generally lead to large reductions in the running time of some algorithms. Here we use it to reduce the running time of \Cref{algo-bfs-poly} but also its memory usage. Let us say that two frontiers $F$ and $F'$ are \textbf{equivalent} if $F'$ is obtained from $F$ by applying a horizontal symmetry. For each frontier, we maintain only one representative for each equivalence class\footnote{They are of size two, or one if the frontier is invariant by symmetry.}. In our case, we choose the lexicographically smallest (see the paragraph after \Cref{notation-holes}).
This way, we approximately halve (some states are symmetric) the number of frontiers we store. Since we also process half the number of states, hence also the running time.
When we consider a new frontier, we now also check if we have seen the symmetric of the representative, or the complement of its symmetric.

\paragraph{Forbidden relative positions.\\}

Some copies of a polyomino $P$ may need to avoid some particular cells. For instance, let us consider a square of width three from which we remove the top cell of the middle column. It is obvious that if we place it in this orientation at the top of the rectangle, the middle free cell will never be covered. If we reverse the orientation (by applying a horizontal symmetry), the same problem happens: it cannot be placed at the bottom of a rectangle for the same reason. To avoid exploring and storing useless frontiers leading to or containing such contradictions,
we compute at the beginning, for each copy of $P$, the set of cells on which we may place the copy such that that the frontier will not be trivially not extendible. To test this, we may for instance place the copy on a location and then try to cover the next $k$ cells for some value of $k$. If this is not possible, we forbid this position for the given copy. The bigger $k$ is, the more time it takes to do this pre-computation but the most efficient it will turn out later so we have to choose a good trade-off. One could worry that if we place the copy near the last column of the rectangle we want to tile, then there is no need for some $k$ other copies to be placed. This could lead us to some false forbidden positions. However this is not true: let us assume $P$ tiles a $n \times m$ rectangle. Then it also tiles a $2n \times 2m$ rectangle, so this is not a problem at all because with a wide enough rectangle, we would be able to place $k$ copies of $P$ for any value of $k$.

\subsection{Other optimisations}
We begin by giving an optimisation for the DFS approach, and then speak of more technical optimisations, which are a matter of implementation.
\label{other-optimisations}

\paragraph{Forbidden pairs.\\}

In order to avoid placing one copy of $P$ and realise later on that there is no way to cover some neighbouring cells, we do some other pre-computations. We trade again some extra time at the beginning of the program for benefits each time some particular function is run. This time the idea is to know, for every two copies $P_1, P_2$ of our polyomino, at which relative positions they may be put from each other to guarantee that the obtained partial tiling is extendible. Let us assume that we put $P_1$ at $(x_1, y_1)$. For each free position $(x_2, y_2)$ we check whether, if $P_2$ is put at $(x_2, y_2)$, there is a way to cover all cells neighbouring $P_1 \cup P_2$. If it is not possible we forbid $P_1$ and $P_2$ to be placed at these positions relatively to each other.
For instance, if putting $P_1$ at $(2,3)$ and $P_2$ at $(6,3)$ would necessarily cause one neighbouring cell to be uncovered however the way we extend the partial tiling, we will forbid $P_2$ to be separated by the vector $(6,3)-(2,3) = (4,0)$ from where any $P_1$ is placed. This is very useful because it saves us time (we realise sooner that placing $P_2$ there leads nowhere) and memory (we store fewer frontiers). This optimisation is much more suited to the DFS approach, because we explore one partial tiling at a time. For the BFS approach, we would need to store a list of forbidden cells for specific copies, and detect when it is time to remove elements from this list.

\paragraph{Data structures used.\\}

In order to have good performance, one needs to use the right data structures. Some of them may be asymptotically optimal but have poor performance when used with few elements. For instance, a priority queue achieves insertion and deletion in $\OO(\log(n))$ while keeping the elements ordered but there is a constant in the $\OO$ which can make it worse than using an array and inserting or deleting (at an arbitrary position) in time $\OO(n)$, when $n$ is not big enough. Here, we used a hash table (\texttt{std::unordered\_map} in the C++ code) to store and remember the frontiers. This allows us to make operations (search, insertion for instance) in time $\OO(1)$ on average. To further optimise, since the complement of any frontier $F$ must have the same number of partially filled columns to be able to match $F$, we may store the frontiers according to their sizes. We create an array of size $w$ of hash tables, one for every possible number of partially filled columns, and store and look for a complement in the right hash table. When we look up for a frontier of width $k$, we look it up in the hash table of index $k$. This way, each hash table has fewer elements than a single one would have. This makes the insertions and lookups empirically faster.\\ 

Let us describe a bit more this data structure. Each element is assigned a \emph{hash}, for instance a 64-bit value computed from the value of the element. The hash is then some index to access the element in the hash table. Since there are many 64-bit integers, we cannot allocate the memory for a full array, so different hashes may be attributed the same \emph{bucket}. This also happens for two elements which would have the same hash. In case, when we look for the presence in the hashtable of some element $x$, two scenarii can occur. If the hash of $x$ is not present, we know that $x$ is necessarily absent from the table. If the hash exists, all the elements with the same hash (the point in the structure is to choose a good hash function so that there are few collisions) are checked until one equal to $x$ is found, if any. This makes almost all operations on a hashtable to take $\OO(n)$ in the worst case (all elements are put in the same bucket) but an average $\OO(1)$ practical complexity. In C++, this structure corresponds to \texttt{std::unordered\_set<>} and \texttt{std::unordered\_map<>}.\\

For the frontiers, we use a vector of boolean values: each cell is either occupied (\texttt{true}) or empty (\texttt{false}). In C++, the corresponding \texttt{std::vector<bool>} is in fact a \emph{bitset}: instead of having each bool element take its usual one byte size, here one byte stores eight elements. This is saving a lot of memory, but may or may not slow down the program depending on the operations made.

\paragraph{Too many polyominoes.\\}

When enumerating polyominoes, at some point there are just too many of them. In this case, even if they can fit in the very large RAM we had at our disposal (around 1.5TiB), we may not want to write such a file to the disk. Therefore we decided to compress the files we were writing. One simple way was to use the gzip format and the gzstream library\footnote{\url{https://www.cs.unc.edu/Research/compgeom/gzstream/}} which enables us to directly write in the compressed format with very little modification of the C++ code. We achieved a good gain. For instance, the list of the about 5 million of free hole-free and corner-compatible polyominoes of size 16 takes 358Mib uncompressed versus 45Mib if compressed. If we compare the other sizes, it seems that we reduce the size of the file by a factor around 9.

\section{Statistics and perspectives}
\label{section-stats}

\subsection{Statistics on the polyominoes and their orders}
We give here statistics about the orders of polyominoes, and also details about polyominoes we could discard as not rectifiable. Given the way we decided to enumerate only polyominoes which had a chance to tile a rectangle, the statistics we give are to be read as data about polyominoes in $\mathcal{P}^*_n$: \textbf{with no holes} and which \textbf{are corner compatible}. Polyominoes with holes, or which would split the set of free cells of a rectangle into several connected components when put on a corner (however the way we put them) are not counted here.

We recall that $\PP_n$ is the set of hole-free free polyominoes, that is no two polyominoes in this set are isometric; $\PP^*_n$ is the set of polyominoes we enumerate: they contain no holes, are corner compatible (see \Cref{def-corner-compatible}). In \Cref{section-non-rectifiable} we explained two methods to rule out a polyomino as not rectifiable, after they are enumerated. We show here the efficiency two methods: the \emph{band method} which consists in trying to tile the bottom band of a rectangle and the \emph{corner method} which tries to tile a corner with some number $k$ of pieces, covering cells by prioritising at each step the ones closest to the corner.\\

As indicated in the caption of \Cref{table-ruled-out}, for lines 17 and 18 the percentage of ruled out polyominoes is 100.00. In fact there remain a few polyominoes, but very very few: 18 out of the 18 637 273 we generated (using the corner-compatibility optimisation) for size 17, for instance. This illustrates completely the scarcity of rectifiable polyominoes. Another interesting property is that when the size is a prime number there seems to be much fewer possibly rectifiable polyominoes: the number of remaining polyominoes is less than 20 for orders 11, 13, 17 when it is at least 97 for orders 12, 14, 15 and 16. The numbers for 15 and 16 are even 210 and 385, before dropping down to 18 for size 17... and it increases again to 686 for size 18.

\begin{landscape}
\begin{table}[H]
\centering
\begin{tabular}{ |c||c|c|c|c|c|c|c|c|c| } 
\hline
$n$ & $\#\PP_n$ & $\#\PP^*_n$ & Maybe  & Order $\leq 10$ & $10 < $ Order $ \leq 100$ & $100 < $ Order $ \leq 300$ & ? ($> 150$)  & ? ($> 300$)\\
\hline
1 & 1 & 1 & 1 & 1 & 0 & 0 & 0 &0\\
\hline
2 & 1 & 1 & 1 & 1 & 0 & 0 & 0 & 0\\
\hline
3 & 2 & 2 & 2 & 2 & 0 & 0 & 0 & 0\\
\hline
4 & 5 & 5 & 4 & 4 & 0 & 0 & 0 & 0\\
\hline
5 & 12 & 11 & 4 & 4 & 0 & 0 & 0 & 0\\
\hline
6 & 35 & 32 & 10 & 8 & 2 & 0 & 0 & 0\\
\hline
7 & 107 & 91 & 7 & 4 & 2 & 0 & 1 & 0\\
\hline
8 & 363 & 288 & 16 & 10 & 1 & 2 & 2 & 1\\
\hline
9 & 1248 & 923 & 36 & 33 & 0 & 0 & 3 & 0\\
\hline
10 & 4460 & 3062 & 33 & 26 & 1 & 1 & 2 & 3\\
\hline
11 & 16094 & 10296 & 13 & 6 & 1 & 1 & 4 & 1\\
\hline
12 & 58937 & 35175 & 97 & 79 & 0 & 1 & 8 & 9\\
\hline
13 & 217117 & 121349 & 10 & 7 & 0 & 0 & 1 & 2\\
\hline
14 & 805475 & 422665 & 101 & 84 & 0 & 0 & 9 & 8\\
\hline
15 & 3001127 & 1483274 & 210 & 192 & 1 & 0 & 6 & 11\\
\hline
16 & 11230003 & 5241856 & 385 & 370 & 0 & 0 & 4 & 11\\
\hline
17 & 42161529 & 18637273 & 18 & 9 & 0 & 0 & 0 & 9\\
\hline
18 & 158781106 & 66635182 & 686 & 650 & 0 & 0 & 18 & 18\\
\hline
\end{tabular}
\caption{Statistics on the order of polyominoes up to size 18.\\
$\mathcal{P}_n$ is the number of free polyominoes of size $n$ with no holes. $\mathcal{P}^*_n \subset \mathcal{P}_n$ contains the polyominoes without holes which are corner-compatible (see \Cref{def-corner-compatible}).\\
"Maybe" indicates the number of polyominoes that we could not discard with the corner method with 50 copies and the bottom "tiling" method of height up to 11 (one more than in \cref{table-ruled-out}.\\
    The "? $> x$" columns mean that we proved the polyomino either is not rectifiable or has order greater than $x$. These columns are not cumulative but mutually exclusive: the column "? $(> 100)$" does not include the polyominoes of the column "? $( > 200)$".}
\end{table}
\label{table-orders}
\end{landscape}

\begin{landscape}
\begin{table}[H]
\centering
\begin{tabular}{ |c||c|c|c|c|c|c|c|c|c|c|c|c|c| } 
\hline
$n$ & Ruled out & C & B & C \& B & B2 & B3 & B4 & B5 & B6 & B7 & B8 & B9 & B10\\
\hline

1 & 0.00 & 0.00 & 0.00 & 0.00 & 0.00  & 0.00  & 0.00  & 0.00  & 0.00  & 0.00  & 0.00  & 0.00  & 0.00 \\
\hline

2 & 0.00 & 0.00 & 0.00 & 0.00 & 0.00  & 0.00  & 0.00  & 0.00  & 0.00  & 0.00  & 0.00  & 0.00  & 0.00 \\
\hline

2 & 0.00 & 0.00 & 0.00 & 0.00 & 0.00  & 0.00  & 0.00  & 0.00  & 0.00  & 0.00  & 0.00  & 0.00  & 0.00 \\
\hline

4 & 20.00 & 0.00 & 20.00 & 0.00 & 20.00  & 0.00  & 0.00  & 0.00  & 0.00  & 0.00  & 0.00  & 0.00  & 0.00 \\
\hline

5 & 63.64 & 45.45 & 63.64 & 45.45 & 18.18  & 36.36  & 9.09  & 0.00  & 0.00  & 0.00  & 0.00  & 0.00  & 0.00 \\
\hline

6 & 68.75 & 50.00 & 68.75 & 50.00 & 18.75  & 31.25  & 9.38  & 6.25  & 0.00  & 0.00  & 0.00  & 3.12  & 0.00 \\
\hline

7 & 92.31 & 86.81 & 91.21 & 85.71 & 19.78  & 45.05  & 15.38  & 7.69  & 1.10  & 1.10  & 1.10  & 0.00  & 0.00 \\
\hline

8 & 94.44 & 88.54 & 94.44 & 88.54 & 21.53  & 44.44  & 16.32  & 7.99  & 3.12  & 0.69  & 0.00  & 0.00  & 0.35 \\
\hline

9 & 96.10 & 93.17 & 96.10 & 93.17 & 24.05  & 46.05  & 17.77  & 5.31  & 1.84  & 0.76  & 0.00  & 0.22  & 0.11 \\
\hline

10 & 98.92 & 97.26 & 98.92 & 97.26 & 26.09  & 47.68  & 17.41  & 5.23  & 1.60  & 0.56  & 0.23  & 0.07  & 0.07 \\
\hline

11 & 99.87 & 99.78 & 99.87 & 99.78 & 28.20  & 48.47  & 17.03  & 4.24  & 1.31  & 0.38  & 0.16  & 0.07  & 0.03 \\
\hline

12 & 99.72 & 99.36 & 99.71 & 99.35 & 30.00  & 48.51  & 15.97  & 3.73  & 1.01  & 0.33  & 0.11  & 0.04  & 0.01 \\
\hline

13 & 99.99 & 99.98 & 99.99 & 99.97 & 31.71  & 48.93  & 14.91  & 3.19  & 0.86  & 0.26  & 0.09  & 0.03  & 0.01 \\
\hline

14 & 99.98 & 99.94 & 99.97 & 99.93 & 33.21  & 48.84  & 14.01  & 2.89  & 0.71  & 0.21  & 0.07  & 0.02  & 0.01 \\
\hline

15 & 99.99 & 99.97 & 99.98 & 99.96 & 34.58  & 48.68  & 13.19  & 2.64  & 0.63  & 0.18  & 0.06  & 0.02  & 0.01 \\
\hline

16 & 99.99 & 99.98 & 99.99 & 99.98 & 35.78  & 48.50  & 12.46  & 2.45  & 0.58  & 0.16  & 0.05  & 0.02  & 0.01 \\
\hline

17 & 100.00 & 100.00 & 100.00 & 100.00 & 36.86  & 48.25  & 11.85  & 2.29  & 0.54  & 0.15  & 0.04  & 0.01  & 0.00\\
\hline

18 & 100.00 & 100.00 & 100.00 & 100.00 & 37.83 & 47.98  & 11.33  & 2.17  & 0.50  & 0.14  & 0.04  & 0.01  & 0.00\\
\hline
\end{tabular}
\caption{How the polyominoes are ruled out as not rectifiable.\\'C' means the percentages of polyominoes (out of $|\PP^*_n$) ruled out as not being able to put 50 copies to tile a corner.\\
'B$i$' means the ones which cannot "tile" a band of size $i$ but can tile a band of size $i-1$. \\
"C \& B" means both the band and the corner methods rule out these polyominoes.\\
The methods for ruling out polyominoes are explained in \Cref{section-non-rectifiable}. For $n = 17$ and $n=18$ the 100.00 and 0 .00 are rounded: some polyominoes are rectifiable, and the band method for $h=10$ is useful: it rules out 2464 polyominoes for $n = 18$ for example.}
\label{table-ruled-out}
\end{table}
\end{landscape}

\subsection{Perspectives for future works}
\paragraph{Trying other types of polyominoes.\\}

Since we could not find a polyomino of odd order, we thought about more general types of objects. First we thought about 3D polyominoes, which would have required quite some work to rewrite the program to test them. A less costly approach, which we tried, was to test the \textbf{extended polyominoes} (see \Cref{fig-ex-extended-polys}): sets of unit squares of $\ZZ^2$ connected either by edges or by corners. This means that the cells of an extended polyomino need not be connected by edges but can be connected through their corners. We did not have more success with extended polyominoes: we found none with an odd order, despite the fact that there are much more of them than of regular polyominoes. We explored extended polyominoes of size up to 11 (there are around 1.6 millions of them) but found none of odd order either.

\begin{figure}
\centering
\includegraphics{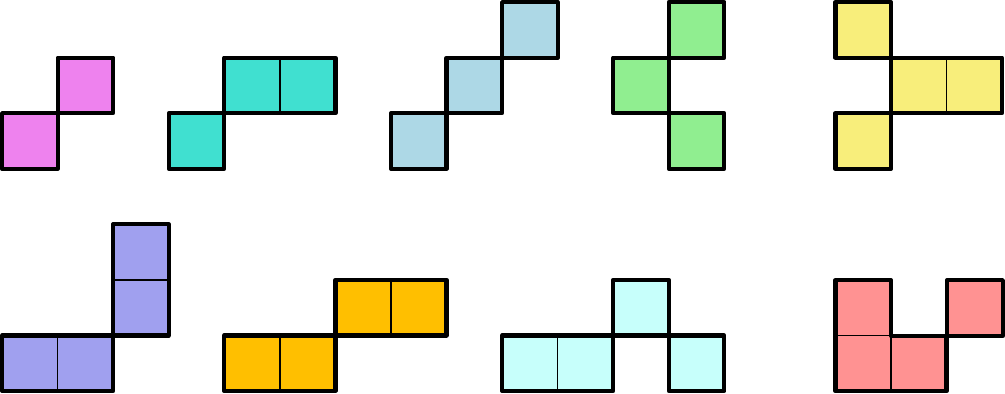}
\caption{Example of extended polyominoes which are not classical polyominoes. Six of them have order two.}
\label{fig-ex-extended-polys}
\end{figure}

\paragraph{Finding other impossible orders.\\}

So far, we only know that the number three cannot be the order of a polyomino. It would be interesting to check if this fact still holds for greater odd number, like five, seven and so on. The problem with the argument in~\cite{article-polys-order-3} is that it is an ad-hoc geometric argument which does not seem, as the authors write, to be generalisable. It would be nice to simplify their argument so that showing the same for five, should it be true, would be less tedious.

\cleardoublepage
\phantomsection
\addcontentsline{toc}{chapter}{Conclusions and perspectives}
\begin{conclusion}

    The main goal of this thesis was to use the power of computers to solve some problems of graph theory. We attacked four main problems, two of which are much related: the dominations numbers and the growth rates of various dominating sets. We solved totally some problems we attacked, some others partially, and did not make much raw progress in the polyomino problem. We recall in what follows our contribution to each problem we tackled, chapter by chapter.\\

In \Cref{discharging-chapter} we reproduced the proof of the four-colour theorem by Robertson et al. \cite{4-col-paper}. It did not need verification since it was proved in Coq. However, this gave us an occasion to give another explanation of their proof, in a more abstracted way, but keeping some technical details. A piece of software we developed was released to the community to make it easier for anyone to investigate a problem using the discharging method. We gave small hints as to how the search for the discharging rules could be automated. Achieving some kind of automation on this part would be a very good improvement to the method, and would contribute to make it more widespread and performing.\\

\Cref{domination-chapter} provides an alternate explanation of the method of the loss introduced by Gonçalves et al. \cite{rao}. We reuse this method for the first time to solve the 2-domination and the Roman-domination numbers on grid graphs: we give closed formulas computing these numbers for any size of grids. This proves that the lower bound for the Roman domination of grids given by Currò \cite{curro} was tight. We also give values for the total-domination and distance-2-domination numbers for small number of lines on grids. This confirms the first formulas of the work of Crevals and Ostergård~\cite{total-dom-article-28}. We also give bounds for arbitrary grids for the total domination, as well as a conjecture on the real formula. According to our conjectures, the total domination number is out of reach of the loss method. What method will be found to solve this problem? For all the problems, we may also wonder how to solve the problems on cylinders: Cartesian products between a path and a cycle.\\

In \Cref{dom-counting-chapter}, we study the counting problem for various dominations problems: the domination, total domination, and their minimal counterparts. To solve it, we link it to the study of SFTs. We show some properties on the associated subshifts: they are block gluing. We prove that each number of dominating sets grows exponentially, at a specific growth rate. We show that these growths rate are computable, and give numerical bounds on each of them. The bounds on the domination and total domination are quite good, and we are able to conjecture the actual value of their growth rates. Will these conjecture be proved someday? Also, we may resort to other methods to improve the bounds for the minimal domination and the minimal total domination.\\

Concerning the polyominoes, in \Cref{polys-chapter}, we do our little part in tackling the question of whether or not a polyomino of odd order exists. We pursue notably the work done by Dahlke to find new orders empirically. We recall some algorithmical techniques to discard a polyomino as being not rectifiable, and to find its order. We also improve some of them. We compare the effectiveness of the main methods to show that a polyomino is not rectifiable. We also sum up some data about the orders we could find for polyominos of size up to 18. If there exists one polyomino of odd order, the only proof needed is this polyomino and the reasons why it does not tile a smaller rectangle. This involves checking a finite number of cases. It seems much harder to prove that no odd number greater than one is the order of some polyomino, given that so far we only proved it for three. Also, the fact that some polyominoes can tile some rectangle with an odd number of copies (but have a smaller even order) rather seems to give more credit to the existence of an odd-order polyomino than to its impossibility. Also, due to parity reasons there are many polyominoes which cannot tile any rectangle with an odd number of copies. This could explain their possible scarcity and our difficulties in finding one.\\

This PhD was also an occasion to think about the proofs using results from some computer program. Some doubt them because a bug could occur, or even some bit in the memory of the computer could flip because of cosmic rays\footnote{yes, it happens, but not that much}. One first answer to this objection is to ask if a very long mathematical proof (with no use of a computer) can truly be verified with a high degree of scrutiny. The more the proof is long and complicated, the more small (or even bigger ones) flaws can be contained in it. Also, a source code may be hard to verify, because it requires simulating it somehow. It is not necessarily more prone to risks than a long proof. Yet, we could get some inspiration from experimental sciences to convince sceptics: it is possible to strengthen the confidence in some program by  having other people program it their own way, independently. Should one or more teams be able to independently make a program giving the same results, we could consider the risk of bugs or errors is negligible.\\

Something else we may wonder is how much computers will help us solve problems in the future. Or, rather, which problems will be within reach, and which could remain forever out of reach of computer solving. In 2010, Google put a lot of computing power to solve a problem on Rubik's cube. They showed that, from any of the $43 \cdot 10^{18}$ starting configurations, one could solve the cube with 20 moves or less. Other problems seem out of reach of today's computers, like finding a winning strategy for chess, because the combinatorial complexity of this game is huge: there are an enormous amount of configurations of the game. However, it seems impossible from today's knowledge. Maybe the next years will see a shift on the computing models: some people for instance think that quantum computers will help us solve some problems which were previously thought out of reach. On the contrary, we do not know if the computing resources will continue to grow forever. For instance, the frequencies of processors has stopped to grow for some years, because of limits from physics. This forces people to use parallelism instead of raw power, but comes with limitations: the more units you put together, the more time you spend in communications between them. So the question remains: which problems will be solved in 100 years\footnote{if we are still here trying to solve these kind of abstract problems, and with computers...} which were not possible to solve today?\\

Finally, each chapter was the occasion, in this manuscript, to sensitize the readers to some problems of our world. The major one is ecological: global warming and the huge loss of biodiversity. Maybe it is also time to reconsider everyday life, and some choices made a long ago when we were not informed. It should require us to reconsider and update our values and lives accordingly.
 
 \end{conclusion}

\appendix

\printindex
\end{document}